\pdfoutput=1

\documentclass[a4paper,twoside,11pt]{article}

\usepackage{BeaCed_article_1}

\title{The $Z$-Dirac and massive Laplacian operators in the $Z$-invariant Ising model}
\author{B\'eatrice de Tili\`ere
\thanks{{\small Universit\'{e} Paris-Est, Laboratoire d'Analyse et de Math\'{e}matiques Appliqu\'{e}es (UMR 8050) 
UPEM, UPEC, CNRS, F-94010, Créteil, France.}
{\small\texttt{beatrice.taupinart-de-tiliere@u-pec.fr.}}
{\small Supported by l'Institut Universitaire de France.}
}}

\begin{document}

\maketitle

\begin{abstract}
\hspace{0.2cm} Consider an elliptic parameter $k$; we introduce a family of $Z^u$-Dirac operators 
$(\mathsf{K}(u))_{u\in\Re(\mathbb{T}(k))}$,
relate them to the $Z$-massive Laplacian of~\cite{BdTR1}, and 
extend to the full $Z$-invariant case
the results of Kenyon~\cite{Kenyon3} on discrete holomorphic and harmonic functions, which correspond to the case $k=0$.
We prove, in a direct statistical mechanics way, how and why the $Z^u$-Dirac and $Z$-massive Laplacian operators appear in the 
$Z$-invariant Ising model, considering the case of infinite and finite isoradial graphs.
More precisely, consider the dimer model on the Fisher graph $\GF$ arising from a $Z$-invariant Ising model.
We express coefficients of the inverse Fisher Kasteleyn operator as a function of the inverse $Z^u$-Dirac operator and 
also as a function of the $Z$-massive Green function; in particular this proves a (massive) random walk representation of important observables of the 
Ising model. We prove that the squared partition function of the Ising model is equal, up to a constant, 
to the determinant of the $Z$-massive Laplacian operator with specific boundary conditions, the latter being
the partition function of rooted spanning forests. To show these results,
we relate the inverse Fisher Kasteleyn operator and that of the dimer model on the bipartite graph $\GQ$ arising from 
the XOR-Ising model, and we prove matrix identities between the Kasteleyn matrix of $\GQ$ and the $Z^u$-Dirac operator, 
that allow to reach inverse matrices as well as determinants.
\end{abstract}

\section{Introduction}

This paper is inspired by three sets of results suggesting connections between the Ising model on a planar graph $\Gs$
and (massive) random walks on $\Gs$ and its dual $\Gs^*$. 

$\bullet$ Messikh~\cite{Messikh} observes that large deviation estimates of 
a massive random walk occur when computing the correlation length of the super-critical Ising model on $\ZZ^2$; a result later
proved by Beffara and Duminil-Copin~\cite{BeffaraDuminil} using the FK-Ising observable of~\cite{Smirnov2} away from the critical point.

$\bullet$
In Smirnov and Chelkak-Smirnov's proof of conformal invariance of the critical $Z$-invariant Ising 
model~\cite{Smirnov3,Smirnov2,ChelkakSmirnov:ising},
the key discrete tools are observables - spin or FK (see also~\cite{KadanoffCeva}) - that are holomorphic.
Discrete holomorphic functions in turn are naturally related to harmonic functions~\cite{Duffin,Mercat:ising,Kenyon3,
ChelkakSmirnov:toolbox}; 
the paper~\cite{MakarovSmirnov} extends some of the above to the massive case. 

$\bullet$ We prove, through combinatorial constructions,
that the squared partition function of the critical $Z$-invariant Ising model is equal, up to a multiplicative constant, to the partition function of spanning 
trees~\cite{deTiliere:mapping,deTiliere:partition}. An abstract proof of this identity is given in the toroidal $Z$-invariant case 
in~\cite{BoutillierdeTiliere:iso_gen,BdtR2}. 

The main contribution of this paper is to provide a unified framework for all of the above, which 
holds in the \emph{full $Z$-invariant} case, in the infinite \emph{and} finite cases. Our main results are obtained as a combination
of intermediate steps that are interesting in their own respect. We nevertheless feel that, before listing statements leading to the 
principal Ising results, we should convey the main ideas.

Let us first be more precise about operators underlying our ``inspiration'' papers. Large deviation estimates of 
massive random walks are related to the massive Green function, the latter being the inverse of the massive Laplacian operator. By definition 
discrete holomorphic functions are in the kernel of the Dirac operator, which is a Kasteleyn matrix/operator of the double graph $\GD$~\cite{Kenyon3}; 
harmonic functions are in the kernel of the Laplacian operator. The spanning tree partition function is equal to the determinant
of the Laplacian operator~\cite{Kirchhoff}. Summarizing, a central role is played by the Dirac operator (at criticality)
and the (massive) Laplacian in the (super) critical Ising model. 

Our first contribution is to introduce one of the missing pieces of the puzzle, namely the full $Z$-invariant version 
of the (critical) Dirac operator of~\cite{Kenyon3}, 
referred to as the \emph{$Z^u$-Dirac operator}, $u$ being a natural free parameter disappearing
at criticality. This is the subject of Section~\ref{sec:Z_Dirac_Z_Lap}, as well as its connections to the $Z$-invariant massive Laplacian 
of~\cite{BdTR1} and the study of the corresponding dimer model on the double graph $\GD$.

To study the Ising model, we use Fisher's correspondence~\cite{Fisher} relating the low (or high) temperature expansion of the 
model~\cite{KramersWannier1,KramersWannier2} to the dimer model on the Fisher graph $\GF$ with associated Kasteleyn matrix/operator 
$\KF$. The partition function of the dimer model is the Pfaffian of $\KF$, and the Boltzmann/Gibbs measures are explicitly expressed
using coefficients of $\KF$ and its inverse $(\KF)^{-1}$~\cite{TF,Kenyon0,CKP,KOS,BoutillierdeTiliere:iso_perio}. This means that knowing the determinant of 
$\KF$ and its inverse amounts to fully understanding the partition function of the Ising model and probabilities of its low (or 
high) temperature expansion. Notably, coefficients of the inverse Kasteleyn operator $(\KF)^{-1}$ are also related to other 
important observables of the Ising model as the spin-Ising observable of~\cite{ChelkakSmirnov:ising}, see~\cite{CCK}, and
the fermionic spinor observable of~\cite{KadanoffCeva}, see~\cite{Dubedat}.

Consider the dimer model on $\GF$ arising from a $Z$-invariant Ising model. Our main contribution is to 
prove matrix identities relating the Kasteleyn operator $\KF$ and
the $Z^u$-Dirac operator and also the $Z$-massive Laplacian of~\cite{BdTR1}. The strength of these identities is that they allow
to reach inverse operators and also, after some extra work, determinants. As a consequence, in the finite and infinite cases,
we express coefficients of $(\KF)^{-1}$ using the inverse $Z^u$-Dirac operator and also using the $Z$-massive Green function; this is the subject of
Section~\ref{sec:KFKQ}, see also the corresponding part of the introduction. In essence, this proves that the contour Ising Boltzmann/Gibbs measures 
can 
be computed from (massive) random walks (with specific boundary conditions in the finite case). In the finite case we also prove
that the squared Ising partition function is equal, up to an explicit constant, to the determinant of the massive Laplacian, that is
to the partition function of rooted spanning forests, see Corollary~\ref{cor:partition_function} and also Theorem~\ref{cor:partition_function_intro} of the 
introduction. Comments on how
these results connect to our ``inspiration'' and other papers are given at the end of this section.

Section~\ref{sec:prelim} contains preliminaries. 
Section~\ref{sec:GQ_Z_Dirac} contains the main intermediate step: we consider the dimer model on the bipartite graph $\GQ$
arising from the XOR-Ising model~\cite{WilsonXOR} constructed from two independent $Z$-invariant Ising models~\cite{Dubedat,BoutillierdeTiliere:XORloops}. 
We prove matrix 
identities relating its Kasteleyn matrix/operator $\KQ$ and the $Z^u$-Dirac operator.  In Section~\ref{sec:KFKQ}, building on the work of Dub\'edat~\cite{Dubedat}, 
we express coefficients of the inverse operator $(\KF)^{-1}$
using coefficients of the inverse operator $(\KQ)^{-1}$; this result holds for the dimer model on the Fisher graph $\GF$ arising from 
any 2d-Ising model, not necessarily $Z$-invariant; combining this with the results of Section~\ref{sec:GQ_Z_Dirac} then allows us to deduce the 
Ising results. In Section~\ref{sec:Examples}, we specify some of our results in two important cases: the $Z$-invariant \emph{critical} case, and the 
full $Z$-invariant case when the underlying graph is $\ZZ^2$.

To give detailed statements, let us be more precise about $Z$-invariant models~\cite{Onsager,Kennelly}, fully developed by
Baxter~\cite{Baxter:8V,Baxter:Zinv,Baxter:exactly}, see also~\cite{Perk:YB,Perk3,Perk4}. 
A \emph{$Z$-invariant model} is naturally defined on an 
\emph{isoradial} graph $\Gs=(\Vs,\Es)$; parameters are chosen so that
the partition function only changes by a constant when performing a star-triangle transformation of the underlying graph, 
\emph{i.e.}, they are required to satisfy the \emph{Yang-Baxter equations}. The solution to this set of equations for the Ising model
has, given the embedding of the graph, a free \emph{elliptic parameter} $k$, such that $(k')^2:=1-k^2\in(0,\infty)$, and 
the \emph{coupling constants} $\Js$ are~\cite{Baxter:exactly}:
\begin{equation*}
\forall\,e\in\Es,\quad \Js_e=\frac{1}{2}\ln\Bigl(\frac{1+\sn(\theta_e|k)}{\cn(\theta_e|k)}\Bigr),  
\end{equation*}
where $\sn,\cn$ are two of the \emph{Jacobi elliptic trigonometric functions}, and $\bar{\theta}_e=\theta_e \frac{\pi}{2K}$ is an angle
associated to the edge $e$ in the isoradial embedding. When $k=0$, \emph{i.e.} $k'=1$, the elliptic functions $\sn,\cn$, are the trigonometric 
functions $\sin,\cos$, and the Ising model is \emph{critical}~\cite{Li:critical,CimasoniDuminil,Lis}. 
As $(k')^2$ varies from $0$ to $\infty$, the coupling constants range from $\infty$ to $0$~\cite{BdtR2} thus covering the whole range of 
inverse temperatures. In the paper~\cite{BdTR1}, we introduce \emph{$Z$-invariant rooted spanning forests} with associated operator 
the massive Laplacian $\Delta^{m}$; when $k=0$, we recover the critical Laplacian of~\cite{Kenyon3}. We also prove an explicit~\emph{local}
expression for its inverse, the $Z$-massive Green function $G^m$, using the \emph{discrete massive exponential function}~\cite{BdTR1}.
We are now ready to give a detailed overview of this paper.

\paragraph{Section~\ref{sec:Z_Dirac_Z_Lap}: $Z^u$-Dirac and $Z$-massive Laplacian operators.} Fix an elliptic parameter $k$. We
introduce a family of \emph{$Z^u$-Dirac operators} $(\KD(u))_{u\in\Re(\TT(k))}$ 
on the double graph 
$\GD=(W\cup B,\ED)$ associated to pairs of dual directed spanning trees, extending to  
the full $Z$-invariant case the Dirac operator $\bar{\partial}$ of~\cite{Kenyon3}, corresponding to $k=0$.
In the finite case, we introduce a family of operators $(\KD^\spartial(u))_{u\in\Re(\TT(k))'}$, with boundary conditions 
tuned for the Ising model. Although these operators play a key role in the $Z$-invariant Ising model, they are interesting in their own 
respect. In the specific case $k=0$, results we obtain can be found in~\cite{Kenyon3,Temperley,Kennelly}. 
We prove, see also Theorem~\ref{prop:KDtKD}:
\begin{thm}\label{prop:KDtKD_intro}$\,$
\begin{itemize}
 \item[$\bullet$] \emph{Infinite case.} Let $u\in\Re(\TT(k))$, then the $Z^u$-Dirac operator $\KD(u)$, the $Z$-massive Laplacian 
 $\Delta^m$ and the dual 
 $\Delta^{m,*}$ of~\cite{BdTR1} satisfy the following identity:
\begin{equation*} 
\overline{\KD(u)}^{\,t\,}\KD(u)=k'
\begin{pmatrix}
\Delta^m&0\\
0&\Delta^{m,*}
\end{pmatrix}.
\end{equation*}\item[$\bullet$] \emph{Finite case.} Let $u\in\Re(\TT(k))'$, then the $Z^u$-Dirac operators $\KD(u)$, $\KD^\spartial(u)$, 
the $Z$-massive Laplacian
$\Delta^{m,\spartial}(u)$ and the dual $\Delta^{m,*}$ satisfy the following identity:
\begin{equation*} 
\overline{\KD^\spartial(u)}^{\,t\,}\KD(u)=k'
\begin{pmatrix}
\Delta^{m,\spartial}(u)& Q(u)\\
0&\Delta^{m,*}
\end{pmatrix}.
\end{equation*}
\end{itemize} 
\end{thm}
A function $F\in\CC^{B}$ is said to be \emph{$Z^u$-holomorphic} if $\KD(u)F=0$. As a consequence of Theorem~\ref{prop:KDtKD_intro}, if 
$F$ is $Z^u$-holomorphic, then $F_{| \Vs}$ is $Z$-massive harmonic on $\Gs$ and $F_{| \Vs^*}$ is $Z$-massive harmonic on $\Gs^*$, thus 
explaining the part ``Dirac'' in ``$Z^u$-Dirac operator''.

In the infinite case, Theorem~\ref{prop:KDtKD_intro} yields the following relations for inverse operators, see also Corollary~\ref{cor:KD_G};
the statement in the finite case is given in Corollary~\ref{cor:KD_G_finite}.
\begin{cor}[Infinite case]\label{cor:KD_G_intro}
For every $u\in\Re(\TT(k))$, consider the operator $\KD(u)^{-1}$
mapping $\CC^{W}$ to $\CC^{B}$ whose coefficients are defined by, for every $\ubar{v},\ubar{f},w$ as in Figure~\ref{fig:corKDG},
\begin{align*}
\KD(u)^{-1}_{\ubar{v},w}&=e^{-i\frac{\alphafb+\betafb}{2}}(k')^{-1}\sc(\thetaf)^{\frac{1}{2}}
\Bigl([\dn(u_{\alphaf})\dn(u_{\betaf})]^{\frac{1}{2}} G^m_{\ubar{v},v_2}-[\dn(u_{\alphaf+2K})\dn(u_{\betaf+2K})]^{\frac{1}{2}}G^m_{\ubar{v},v_1}\Bigr)\\
\KD(u)^{-1}_{\ubar{f},w}&=-ie^{-i\frac{\alphafb+\betafb}{2}}(k')^{-1}\sc(\thetaf^*)^{\frac{1}{2}}
\Bigl([\dn((u_{\betaf})^*)\dn((u_{\alphaf+2K})^*)]^{\frac{1}{2}}G^{m,*}_{\ubar{f},f_2}+\\
&\hspace{7.5cm} -[\dn((u_{\betaf-2K})^*)\dn((u_{\alphaf})^*)]^{\frac{1}{2}}G^{m,*}_{\ubar{f},f_1} \Bigr),
\end{align*}
where $G^m$ and $G^{m,*}$ are the $Z$-massive and dual $Z$-massive Green functions of~\cite{BdTR1}.
Then $\KD(u)^{-1}$ is an inverse of the $Z^u$-Dirac operator $\KD(u)$. When the graph $\GD$ is moreover $\ZZ^2$-periodic, it is 
the unique inverse decreasing to zero at infinity.
\end{cor}
This gives, in Theorem~\ref{thm:Gibbs_KD}, an explicit \emph{local} expression for a Gibbs measure of the dimer model on the double graph
$\GD$, where the locality property is inherited from that of the $Z$-massive Green functions of~\cite{BdTR1}. Using the KPW-Temperley
bijection~\cite{Temperley,KPW}, probabilities of pairs of dual directed spanning trees are computed using the Green function of 
a massive, non-directed random walk. Apart from the locality property which is specific,
a similar result is obtained by Chhita~\cite{Chhita} in the case of $\ZZ^2$ with a specific choice of weights.

In Theorem~\ref{thm:det} and Corollary~\ref{cor:det}, we restrict to the finite case and prove relations on determinants; we show,
\begin{thm}\label{thm:det_intro}
Let $\Ms_0$ be a dimer configuration of $\GDro$. 
Then, for every $u\in\Re(\TT(k))'$, we have
\begin{align*}
|\det\KD(u)|&=
(k')^{\frac{|\Vs^*|}{2}}\Bigl(\prod_{w\in W}\sc(\theta_w)^\frac{1}{2}\Bigr)
\Bigl(\prod_{e=wx\in\Ms_0} [\dn(u_{\alpha_e})\dn(u_{\beta_e})]^{\frac{1}{2}}  \Bigr)
\det \Delta^{m,*},\\
|\det\KD^\spartial(u)|&=
(k')^{\frac{|\Vs|-1}{2}}\Bigl(\prod_{w\in W}\cs(\theta_w)^\frac{1}{2}\Bigr)  
\Bigl(\prod_{e=wx\in\Ms_0} [k'\nd(u_{\alpha_e})\nd(u_{\beta_e})]^{\frac{1}{2}}  \Bigr)
\det\Delta^{m,\spartial}(u).
\end{align*}
\end{thm}
As a consequence, the partition function of pairs of dual directed spanning trees is equal, up to a constant, to the partition function of rooted spanning 
forests. In the critical case, $k=0$, this is an easy consequence of Temperley's bijection~\cite{Temperley}, but the correspondence does not extend 
when $k\neq 0$. The main tools of the proof are gauge equivalences on bipartite adjacency matrices and on~\emph{adjacency matrices of digraphs},
see also Appendix~\ref{app:gauge}. 

The $Z^u$-Dirac operator is equivalent to a model of directed spanning trees. In
Proposition~\ref{prop:Zinv}, we prove that the latter is $Z$-invariant, thus explaining the part ``$Z^u$'' of the 
terminology ``$Z^u$-Dirac operator''.

\paragraph{Section~\ref{sec:GQ_Z_Dirac}: Kasteleyn operator of the graph $\GQ$ and $Z^u$-Dirac operator.}
We consider the dimer model on the graph $\GQ$ arising from the $Z$-invariant XOR-Ising model, with Kasteleyn matrix $\KQ$. The main result
of this section is Theorem~\ref{thm:main} proving the following relations between the matrix $\KQ$ and the $Z^u$-Dirac operator. The matrices
$S(u)$ and $T(u)$ are defined in Section~\ref{sec:KQ_KD_relation} and the statement is as follows.
\begin{thm}\label{thm:main_intro}$\,$
\begin{itemize}
\item[$\bullet$] \emph{Infinite case.} Let $u\in\Re(\TT(k))$, then the Kasteleyn matrix $\KQ$, the $Z^u$-Dirac operator
$\KD(u)$ and the matrices $S(u)$, $T(u)$ are related by the following identity:
\begin{equation*}
\KQ\,T(u) = S(u)\, \KD(u). 
\end{equation*}
\item[$\bullet$] \emph{Finite case.} Let $u\in\Re(\TT(k))'$, then the Kasteleyn matrix $\KQu$, the $Z^u$-Dirac operator $\KD^\spartial(u)$
and the matrices $S(u)$, $T(u)$ are related by the following identity:
\begin{equation*}
\KQu\,T(u) = S(u)\, \KD^\spartial(u).
\end{equation*}
\end{itemize} 
\end{thm}
In the infinite case, Theorem~\ref{thm:main_intro} yields the following relations on inverse matrices, see also Corollary~\ref{cor:KD_KQ}; the 
statement in the finite 
case is the subject of Corollary~\ref{cor:KD_KQ_finite}.
\begin{cor}[Infinite case]\label{cor:KD_KQ_intro}
For every $u\in\Re(\TT(k))$, for every $\ubar{\ws},\ubar{v},\ubar{f}$, and every $w,\bs,\bs'$ as in Figure~\ref{fig:cor},
\begin{equation*}
(\KQ)^{-1}_{\ubar{\ws},\bs}\cn(u_{\betaf})-i(\KQ)^{-1}_{\ubar{\ws},\bs'}\sn(u_{\betaf})\dn(u_{\alphaf})
=\frac{e^{i\frac{\betafb-\betaib}{2}}}{\Lambda(u_{\alphaf},u_{\betaf})}
\left[\cn(u_{\betai})\KD(u)^{-1}_{\ubar{v},w}-i\sn({u_{\betai}})\KD(u)^{-1}_{\ubar{f},w}\right],
\end{equation*}
where $\Lambda(u_\alpha,u_\beta)=[\sn\theta\cn\theta\nd(u_{\alpha})\nd(u_{\beta})]^{\frac{1}{2}}$, and 
$\theta=u_\alpha-u_\beta$.
\end{cor}
Specifying the value of the parameter $u$ allows to express coefficients of the inverse Kasteleyn operator $(\KQ)^{-1}$
using the inverse $Z^u$-Dirac operator; combining this with Theorem~\ref{cor:KD_G_intro} yields an expression using the $Z$-massive 
Green function of~\cite{BdTR1}. We obtain, see also Corollary~\ref{cor:gibbsGQ_GD} and Corollary~\ref{cor:BoltzmannGQ_GD} for the finite case,
\begin{cor}[Infinite case]\label{cor:gibbsGQ_GD_intro}
For every $\ubar{\ws},\bs$ of $\GQ$ as in Figure~\ref{fig:cor_dimerKQ},
\begin{align*}
(\KQ)^{-1}_{\ubar{\ws},\bs}&=\textstyle\frac{
e^{i\frac{\betafb-\betaib}{2}}}
{[\cn(\thetaf)\sn(\thetaf)\nd(\thetaf)]^{\frac{1}{2}}}
\left(\cn\bigl(\frac{\betaf-\betai}{2}\bigr)\KD(\betaf)^{-1}_{\ubar{v},w}
-i\sn\bigl(\frac{\betaf-\betai}{2}\bigr)\KD(\betaf)^{-1}_{\ubar{f},w}
\right)\\
(\KQ)^{-1}_{\ubar{\ws},\bs}&=\textstyle e^{-i\frac{\betaib+\alphafb}{2}}(k')^{-1}
\Bigl(\frac{\cn\bigl(\frac{\betaf-\betai}{2}\bigr)}{\cn(\thetaf)}
\bigl[\dn(\thetaf) G^m_{\ubar{v},v_2}-k'G^m_{\ubar{v},v_1}\bigr]
-\frac{\sn\bigl(\frac{\betaf-\betai}{2}\bigr)}{\sn(\thetaf)}
\bigl[\dn(\thetaf)G^{m,*}_{\ubar{f},f_2}-G^{m,*}_{\ubar{f},f_1}\bigr]\Bigr).
\end{align*}
\end{cor}
This proves, in an alternative way, an explicit \emph{local} expression for a Gibbs measure of the dimer model on the graph $\GQ$~\cite{BdtR2},
where the locality property is seen as directly inherited from that of the $Z$-massive Green function.

Using Theorem~\ref{thm:main_intro} and additional combinatorial arguments, we prove in Theorem~\ref{thm:part_function} that the determinants of $\KQ$ and of 
the $Z^u$-Dirac operator are equal, up to an explicit constant. By~\cite{Dubedat}, the determinant of $\KQ$ is equal up to a constant, 
to the squared partition function of the Ising model. Combining this with Theorem~\ref{thm:det_intro} gives, see 
also Corollary~\ref{cor:partition_function} for the explicit value of $C(u)$,
\begin{thm}\label{cor:partition_function_intro}
For every $u\in\Re(\TT(k))''$,
\begin{align*}
[\Zising^+(\Gs,\Js)]^2&=C(u) |\det \Delta^{m,\spartial}(u)|.
\end{align*}
\end{thm} 
When $k=0$, we essentially recover the result of~\cite{deTiliere:partition}; see Section~\ref{sec:ex_Z_inv_critical}.

\paragraph{Section~\ref{sec:KFKQ}: Dimer model on the Fisher graph $\GF$ and the Kasteleyn matrix $\KQ$.}
Consider the dimer model on $\GF$ with Kasteleyn matrix $\KF$ arising from an Ising model with coupling constants $\Js$, not necessarily 
$Z$-invariant, and the corresponding dimer model on the bipartite graph $\GQ$, with (real) Kasteleyn matrix $\KQt$. Building 
on the work of Dub\'edat~\cite{Dubedat} and proving additional matrix relations, we express coefficients of the inverse operator 
$(\KF)^{-1}$ using coefficients of the inverse operator $(\KQt)^{-1}$. Partitioning 
vertices of $\GF$ as $A\cup B$ as in~\cite{Dubedat}, we obtain, see also Theorem~\ref{thm:KFKQinv},
\begin{thm}[Finite and infinite cases]\label{thm:KFKQinv_intro}
Using the notation of~Figure~\ref{fig:thmKFKQinv}, there are four cases to consider:

\emph{1.} For every $\ubar{a}\in A$ and every $b\in B$ such that, when the graph $\GF$ is moreover finite, $b$ is not a boundary vertex:
\begin{align*}\label{form:KABm1_intro}
(\KF)^{-1}_{\ubar{a},b}&=\frac{1}{1+e^{-4\Jf_e}}
\bigl[(\KQt)^{-1}_{\ubar{\ws},\bs} + (\KQt)^{-1}_{\ubar{\ws},\bs'}\eps_{b',b}e^{-2\Jf_e}\bigr].
\end{align*}
\emph{2.} When the graph $\GF$ is finite, for every $\ubar{a}\in A$ and every boundary vertex $b$ of $B$, we have
\begin{equation*}\label{form:KABm1_finite_intro}
(\KF)^{-1}_{\ubar{a},b}=(\KQt)^{-1}_{\ubar{\ws},\bs}.
\end{equation*}
\emph{3.} For every $\ubar{a},\,a\in A$,
\begin{equation*}\label{form:KAAm1_intro}
(\KF)^{-1}_{\ubar{a},a}=-\frac{1}{2}(\KQt)^{-1}_{\ubar{\ws},\bs}\eps_{b,a}+\kappa_{\ubar{a},a},
\end{equation*}
where $\kappa_{\ubar{a},a}=0$ if $\ubar{a}$ and $a$ do not belong to the same decoration, and to $\pm \frac{1}{4}$ if they do.

\emph{4.} For every $\ubar{b},\,b\in B$,
\begin{align*}\label{form:KBBm1_intro}
(\KF)^{-1}_{\ubar{b},b}&=-\eps_{\ubar{b},\ubar{a}_1}(\KF)^{-1}_{\ubar{a}_1,b}+\eps_{\ubar{b},\ubar{a}_2}(\KF)^{-1}_{\ubar{a}_2,b},
\end{align*}
where $(\KF)^{-1}_{\ubar{a}_1,b},(\KF)^{-1}_{\ubar{a}_2,b}$ are given by~Case 1.
\end{thm}

As a consequence, the Boltzmann/Gibbs measures of the dimer model on the \emph{non-bipartite} graph $\GF$ can be computed using the inverse
Kasteleyn operator of the \emph{bipartite} graph $\GQ$. Note that in the finite case, we do not need positivity of the coupling constants $\Js$.
When $\Js<0$, the dimer model on the Fisher graph $\GF$ has positive weights $1$ and $e^{-2\Js_e}$ on edges, and 
is related to a bipartite dimer model with some negative weights, see also Remark~\ref{rem:positivity_dimer}. 
As mentioned in~\cite{Dubedat}, bozonisation identities somehow prove the existence of such linear relations, but working them out 
requires more work, which is the subject of the above theorem. 

Next we restrict to the $Z$-invariant case. The coefficient $(\KF)^{-1}_{\ubar{a},a}$ is equal, up to an additive constant, to the 
coefficient $(\KQ)^{-1}_{\ubar{\ws},\bs}$, and is thus expressed using the inverse $Z^u$-Dirac operator using 
Corollary~\ref{cor:gibbsGQ_GD_intro} in the infinite case, and Corollary~\ref{cor:BoltzmannGQ_GD} in the finite case. The same holds for
$(\KF)^{-1}_{\ubar{a},b}$ when $b$ is a boundary vertex. The coefficient $(\KF)^{-1}_{\ubar{b},b}$ is a simple linear combination of two
coefficients $(\KF)^{-1}_{\ubar{a}_1,b}$, $(\KF)^{-1}_{\ubar{a}_2,b}$, so we are left with expressing the coefficient 
$(\KF)^{-1}_{\ubar{a},b}$. Choosing a specific value of $u$ in Corollary~\ref{cor:KD_KQ_intro}, and using 
Corollary~\ref{cor:KD_G_intro} gives, see also Corollary~\ref{thm:KFKQinv_Zinv} for the finite case,
\begin{cor}[Infinite case]\label{thm:KFKQinv_Zinv_intro}
Let $\ubf=\frac{\alphaf+\betaf}{2}+K$. Then,
\begin{align*}
\textstyle
(\KF)^{-1}_{\ubar{a},b}
=&q_{\bs,\ubar{\ws}}\textstyle e^{i\frac{\betafb-\betaib}{2}}(k')^{\frac{1}{2}}  \frac{\cn\bigl(\frac{K-\thetaf}{2}\bigr)(1+(k')^{-1}\dn(\thetaf))}{
2[\cn(\thetaf)\sn(\thetaf)]^{\frac{1}{2}}}
\left(\cn(\ubf_{\betai})\KD(\ubf)^{-1}_{\ubar{v},w}-i\sn({\ubf_{\betai}})\KD(\ubf)^{-1}_{\ubar{f},w}\right)\\
=&q_{\bs,\ubar{\ws}}\textstyle e^{-i\frac{\alphafb+\betaib}{2}} \frac{\cn\bigl(\frac{K-\thetaf}{2}\bigr)(1+(k')^{-1}\dn(\thetaf))}{
2}\times\\
&\textstyle\times \left(\frac{\cn(\ubf_{\betai})}{\cn(\thetaf)}(G^m_{\ubar{v},v_2}-G^m_{\ubar{v},v_1})
-\frac{\sn({\ubf_{\betai}})}{\sn(\thetaf)}\Bigl(\nd\bigl(\frac{K-\thetaf}{2}\bigr)G^{m,*}_{\ubar{f},f_2}-
\nd\bigl(\frac{K+\thetaf}{2}\bigr)G^{m,*}_{\ubar{f},f_1}\Bigr)\right).
\end{align*}
\end{cor}

\emph{Connection to previously known results.} 
Apart from allowing to compute the contour Ising Boltzmann/Gibbs measures,
coefficients of the inverse Kasteleyn matrix $(\KF)^{-1}$ are important observables of the Ising model:
$(\KF)^{-1}_{\ubar{b},b}$ is related to the spin-Ising observable~\cite{ChelkakSmirnov:ising}, see for example~\cite{CCK}. 
Dub\'edat~\cite{Dubedat} proves that $(\KF)^{-1}_{\ubar{a},a}$ is the fermionic spinor observable of~\cite{KadanoffCeva} and,
referring to Nienhuis-Knops~\cite{NK}, mentions that it 
is related to the FK-Ising observable of~\cite{Smirnov3} (up to normalization). As a consequence, 
in the specific case $k=0$ (the critical case),
Theorem~\ref{thm:KFKQinv_intro}, Corollary~\ref{thm:KFKQinv_Zinv_intro} and Corollary~\ref{cor:gibbsGQ_GD_intro} 
 are deeply related to the discrete part of~\cite{ChelkakSmirnov:ising} proving that these
observables are holomorphic, and integrating the square to obtain close to harmonic functions, see also Section~\ref{sec:ex_Z_inv_critical}.
Our results have two important features: they prove that in the infinite \emph{and} finite cases, these observables have an exact explicit expression 
involving Green functions, and also that these expressions not only hold at criticality but in the \emph{full} $Z$-invariant regime.

The paper~\cite{Lis2} gives a \emph{non-backtracking random walk} representation of the inverse Kac-Ward operator, the latter being 
connected to the inverse Kasteleyn operator. In this paper, we give a \emph{(massive) random walk} representation of the inverse Kasteleyn 
operator where, in the finite case, this random walk has some vortices along the boundary. In the critical case, and for one choice of $u$ 
(namely $u=i\infty)$, part of the relation of Theorem~\ref{thm:main_intro} was obtained in~\cite{Cimasoni:KacWard2}. 
Let us end this introduction with a comment on the paper~\cite{BeffaraDuminil} based on an observation by Messikh~\cite{Messikh}
about the occurrence of large deviation estimates of a massive random walk in the correlation length of the super-critical Ising model on $\ZZ^2$. The proof 
consists in showing that, in the super-critical regime, spin correlations are approximated by the FK-Ising observable, and then analyzing 
the latter. By Theorem~\ref{thm:KFKQinv_intro} and Corollary~\ref{cor:gibbsGQ_GD_intro}, the latter a directly related to the massive 
Green function, thus explaining the occurrence of the massive random walk, see also Section~\ref{sec:ex_Z_inv_ZZ2} specifying our results 
to the case where $\Gs=\ZZ^2$.

\textbf{Acknowledgments.} We thank C\'edric Boutillier, Dima Chelkak, David Cimasoni, Adrien Kassel, Marcin Lis, Paul Melotti, 
Sanjay Ramassamy and Kilian Raschel for interesting conversations
in the course of writing this paper.


\section{Preliminaries}\label{sec:prelim}

This section contains all the preliminaries required for this paper. We give the definitions of the Ising model, the dimer model per se, 
the dimer model on decorated graphs arising from the Ising model, from the XOR-Ising model and from pairs of dual directed spanning 
trees; we also define the rooted directed spanning forests model. We end with isoradial graphs, $Z$-invariance and the $Z$-invariant 
versions of the above models.

\subsection{The Ising model}\label{sec:defIsing}

Consider a finite, planar, simple graph $\Gs=(\Vs,\Es)$. Suppose that edges of $\Gs$ are assigned positive 
\emph{coupling constants} $\Js=(\Js_e)_{e\in\Es}$.
The \emph{Ising model on $\Gs$ with free boundary conditions} is defined as follows. A \emph{spin configuration} is a function on
vertices of $\Gs$ taking values in $\{-1,1\}$. The probability on the set of spin configurations $\{-1,1\}^{\Vs}$ is given by
the \emph{Ising Boltzmann measure} $\PPising$, defined by:
\[
\forall\,\sigma\in\{-1,1\}^\Vs,\quad \PPising(\sigma)=\frac{1}{\Zising(\Gs,\Js)}
\exp\Bigl(\sum_{e=vv'\in\Es}\Js_e \sigma_v\sigma_{v'}\Bigr),
\]
where $\Zising(\Gs,\Js)=\sum_{\sigma\in\{-1,1\}^\Vs}\exp\Bigl(\sum_{e=vv'\in\Es}\Js_e \sigma_v\sigma_{v'}\Bigr)$ is the normalizing
constant known as the \emph{Ising partition function}.

From now on, we suppose that the planar graph $\Gs$ is embedded and simply connected. \emph{Boundary vertices} of $\Gs$ are vertices on the boundary of 
the unbounded face of $\Gs$. The Ising model with \emph{$+$ boundary conditions} has the additional restriction that boundary vertices
have +1 spin. Denote by $\PPising^+$
and $\Zising^+(\Gs,\Js)$ the corresponding Boltzmann measure and partition function\footnote{Note that the 
Ising model with $+$ boundary conditions on the graph $\Gs$ 
can be seen as the Ising model with free boundary conditions on the graph $\Gs'$ obtained from $\Gs$ by merging all 
boundary edges and vertices into a single vertex.}.

Denote by $\bar{\Gs}^*=(\bar{\Vs}^*,\bar{\Es}^*)$ the dual graph of $\Gs$, and by $\outer$ the vertex of $\bar{\Gs}^*$ corresponding
to the unbounded face of $\Gs$. Consider also the \emph{restricted dual graph} 
$\Gs^*=(\Vs^*,\Es^*)$ obtained from $\bar{\Gs}^*$ by removing the vertex $\outer$ and all of its incident edges. 
A \emph{polygon configuration} of $\Gs^*$ is a subset of edges such that every vertex has even degree;
let $\P(\Gs^*)$ denote the set of polygon configurations of $\Gs^*$. Then, the \emph{low temperature expansion} (LTE) of the Ising 
partition function with + boundary conditions is~\cite{KramersWannier1,KramersWannier2}:
\begin{equation}\label{equ:LTE}
\Zising^+(\Gs,\Js)=
\Bigl(\prod_{e\in\Es}e^{\Js_e}\Bigr)
\sum_{\Ps\in\P(\Gs^*)}\prod_{e^*\in\Ps} e^{-2\Js_e}.
\end{equation}
Polygon configurations of this expansion separate clusters of $\pm 1$ spins of the Ising model.

In this paper we consider the case where the graph is finite or infinite. The definition of the Boltzmann measure does not hold in 
the infinite case but extends naturally, and this will be clarified as we go along. 

In the finite case, we consider the Ising model with + boundary conditions. It will be crucial
to use the \emph{boundary trick} of Chelkak and Smirnov~\cite{ChelkakSmirnov:ising}
consisting
in adding one extra vertex with +1 spin on every boundary edge of the graph. This has no effect on the Ising model,
but the graph gains geometric freedom along the boundary, which will be key to handling boundary terms in Theorem~\ref{thm:main}.
In order not to introduce too many graphs and confuse the reader, from 
now on we let $\Gs$ be the graph we started from \emph{with the extra vertex on every boundary edge}, 
then $\bar{\Gs}^*$ is its dual graph and $\Gs^*$ its restricted dual. Figure~\ref{fig:low_temp} provides an example of:
a graph $\Gs$, its restricted dual $\Gs^*$, a spin configuration with + boundary conditions and the corresponding low temperature 
polygon configuration of~$\Gs^*$.


In the infinite case, we suppose that the embedded graph together with its faces cover the whole plane.  
So as not to have too many notation, and since it will be clear from the setting, we also denote by $\Gs$ the infinite graph; 
the dual graph is denoted $\Gs^*$.

\begin{figure}[ht]
\centering
\begin{overpic}[width=7.5cm]{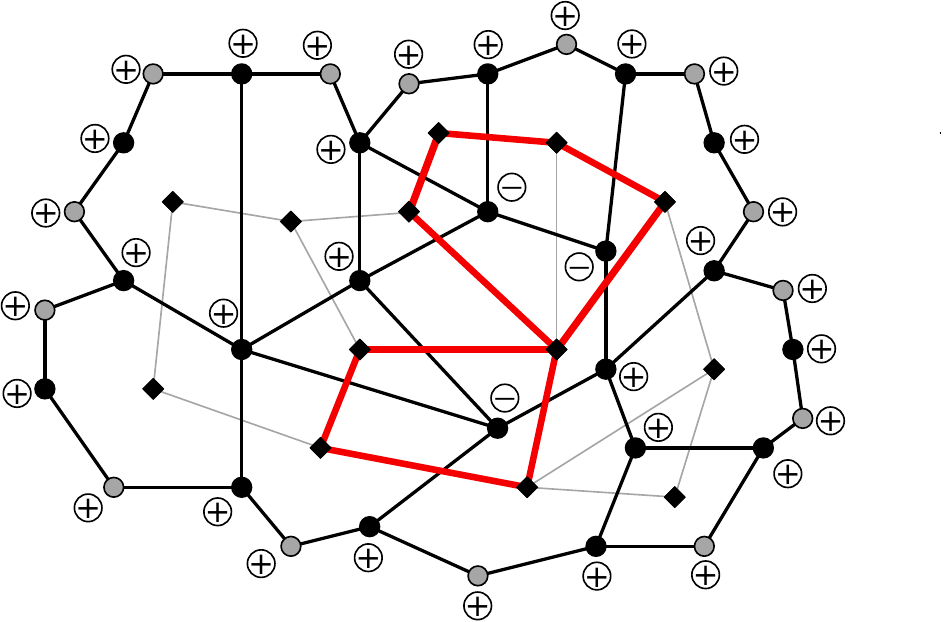}
\end{overpic}
\caption{An example of a graph $\Gs$ (black), of its restricted dual $\Gs^*$ (grey), of a spin configuration with + boundary conditions on 
$\Gs$ and its corresponding polygon configuration on $\Gs^*$ (red). Vertices of $\Gs$ are pictured as bullets -- black ones represent vertices of the
original graph and grey ones are the additional vertices on boundary edges -- vertices of $\Gs^*$ are pictured as diamonds.}
\label{fig:low_temp}
\end{figure}

\subsection{The dimer model}\label{sec:dimer_model}

Throughout the paper, we use the dimer model defined on three decorated versions of the graph $\Gs$. Prior to defining these 
decorated graphs, we recall the definition of the dimer model per se, as well as that of the
\emph{Kasteleyn matrix}.
We also recall the founding results that we will use.

Consider a planar, simple, simply connected graph $G=(V,E)$.
A \emph{dimer configuration} of $G$,
also known as a \emph{perfect matching}, is a subset of edges 
such that every vertex is incident to exactly one edge of this subset. Denote by $\M(G)$ the set of dimer configurations
of the graph $G$. Suppose that a positive weight function $\nu$ is assigned to edges of $G$. 

\subsubsection{Finite case}

Suppose that the graph $G$ is finite, and that $|V|$ is even. Then, the probability of 
occurrence of a dimer configuration, chosen with respect to the \emph{dimer Boltzmann measure} $\PPdimer$, is given by
\begin{equation*}
\forall\,\Ms\in\M(G),\quad \PPdimer(\Ms)=\frac{\prod_{e\in \Ms}\nu_e}{\Zdimer(G,\nu)},
\end{equation*}
where $\Zdimer(G,\nu)=\sum_{\Ms\in\M(G)}\prod_{e\in \Ms}\nu_e$ is the normalizing constant known as the \emph{dimer partition function}.

The main tool used to study the dimer model is the \emph{Kasteleyn matrix}~\cite{Kasteleyn1,Kasteleyn2,TF},
it is defined as follows. A \emph{face-cycle} is a cycle of $G$ bounding a bounded face of the graph. A
\emph{Kasteleyn orientation} is an orientation of the edges such that every face-cycle is \emph{clockwise odd}, meaning
that when traveling clockwise around a face-cycle, the number of co-oriented edges is odd. By the results 
of~\cite{Kasteleyn2}, a Kasteleyn orientation always exists for planar graphs. A \emph{Kasteleyn matrix},
denoted by $K$, is a 
weighted, directed, adjacency matrix of the graph $G$ associated to the weight function $\nu$ and to a Kasteleyn orientation. More
precisely, rows and columns of the matrix $K$ are indexed by vertices of $G$, and non-zero coefficients of $K$ are defined by,
\begin{align*}
\forall\text{ edge $(x,y)$ of $\Gs$},\quad K_{x,y}=\eps_{x,y}\nu_{xy},
\end{align*}
where
\begin{align*}
\eps_{x,y}=
\begin{cases}
1&\text{ if $xy\in E$ and $x\rightarrow y$}\\
-1&\text{ if $xy\in E$ and $x\leftarrow y$}.
\end{cases}
\end{align*}
Note that the matrix $K$ is skew symmetric. 

When the graph $G$ is bipartite, the set of vertices can be split into $V=W\cup B$, where $W$
represents the set of white vertices, $B$ the set of black ones, and vertices in $W$ are only adjacent to vertices in $B$.
Suppose that $|W|=|B|$ for otherwise $G$ has no dimer configurations.
The Kasteleyn matrix $K$ is naturally block diagonal with two 0 blocks corresponding to rows/columns indexed by 
$W/W$ or $B/B$. It thus suffices to consider the \emph{bipartite}, weighted, directed, adjacency matrix of the graph $G$, denoted 
by $\tilde{K}$. It has rows indexed by white vertices of $G$ and column by black ones. Non-zero coefficients are defined by:
\begin{align*}
\forall\,\text{ edge $wb$ of $G$},\quad \tilde{K}_{w,b}=\eps_{w,b}\nu_{wb}.
\end{align*}
Note that the bipartite Kasteleyn matrix can also be defined 
as minus the transpose of the above matrix $\tilde{K}$; rows are then indexed by black vertices and columns by 
white ones. Actually both bipartite Kasteleyn matrices are considered in this paper.

The two founding results of the dimer model are: an explicit expression for the partition 
function~\cite{Kasteleyn1,Kasteleyn2,TF} and for the dimer Boltzmann measure~\cite{Kenyon0}. Here are their statements. 

\begin{thm}[\cite{Kasteleyn1,Kasteleyn2,TF}]
The dimer partition function of the graph $G$ with weight function $\nu$ is equal to:
\begin{equation*}
\Zdimer(\Gs,\nu)=|\Pf K |.
\end{equation*}
When the graph $G$ is moreover bipartite, we have:
\begin{equation*}
\Zdimer(\Gs,\nu)=|\det \tilde{K}|.
\end{equation*}
\end{thm}

\begin{thm}[\cite{Kenyon0}]~\label{thm:Kenyon0}
The probability of occurrence of a subset $\E=\{e_1=x_1 y_1,\dots,e_l=x_l y_l\}$ of edges of $G$, chosen with respect to the dimer Boltzmann
measure $\PPdimer$ is equal to
\begin{equation*}
\PPdimer(e_1,\dots,e_l)=\Bigl(\prod_{j=1}^l K_{x_i,y_i}\Bigr)\Pf (K^{-1})^t_{\E},
\end{equation*}
where $(K^{-1})_{\E}$ is the sub-matrix of the inverse Kasteleyn matrix $K^{-1}$ whose rows and columns are indexed by vertices 
$x_1,y_1,\dots,x_l,y_l$.

When the graph $G$ is moreover bipartite, the subset of edges $\E$ is written as 
$\E=\{e_1=w_1b_1,\dots,e_l=w_l b_l\}$, and we also have,
\begin{equation*}
\PPdimer(e_1,\dots,e_l)=\Bigl(\prod_{j=1}^l \tilde{K}_{w_i,b_i}\Bigr)\det (\tilde{K}^{-1})_{\E},
\end{equation*}
where $(\tilde{K}^{-1})_{\E}$ is the sub-matrix of the inverse bipartite Kasteleyn matrix $\tilde{K}^{-1}$ whose rows are indexed
by black vertices $b_1,\dots,b_l$ and columns by white vertices $w_1,\dots,w_l$.
\end{thm}

\subsubsection{Infinite case}\label{sec:dimer_infinite}

Suppose that the graph $G$ is infinite. The dimer Boltzmann measure is not well defined and is replaced by the notion of 
Gibbs measure. A 
\emph{Gibbs measure} is a probability measure on $\M(G)$ satisfying the \emph{DLR-conditions}:
when one fixes a dimer configuration in an annular region, then
perfect matchings inside and outside of the annulus are independent; moreover, the probability of a dimer configuration in the finite
region separated by the annulus is proportional to the product of the edge-weights. 

Consider a Kasteleyn orientation of the graph $G$ and the corresponding Kasteleyn matrix $K$, then $K$ 
can also be seen as an operator acting on $\CC^V$:
\[
\forall\,F\in\CC^V,\quad (KF)_x=\sum_{y\sim x}K_{x,y}F_y.
\]
When $G$ is bipartite, the bipartite Kasteleyn matrix $\tilde{K}$ is an operator mapping $\CC^W$ to~$\CC^B$:
\[
\forall\,F\in\CC^W,\quad (\tilde{K}F)_b=\sum_{w\sim b}\tilde{K}_{b,w}F_w.
\]
Explicit expressions for Gibbs measures involve inverse Kasteleyn matrices/operators. An \emph{inverse Kasteleyn operator}
$L$ is asked to satisfy the following conditions:
\begin{align*}
\bullet &\, KL=\mathrm{Id} \text{ or } LK=\mathrm{Id},\\
\bullet &\, L_{x,y}\rightarrow 0 \text{ as the distance between $x$ and $y$ tends to infinity}.
\end{align*}
\emph{Existence} of an inverse Kasteleyn operator and explicit expressions for coefficients are proved for:
$\ZZ^2$-periodic bipartite graphs using Fourier techniques~\cite{CKP,KOS};  $\ZZ^2$-periodic (non-bipartite) 
Fisher graphs~\cite{BoutillierdeTiliere:iso_perio,Dubedat}; non-periodic, bipartite isoradial graphs, bipartite quadri-tiling graphs,
and Fisher graphs, all with specific weights arising from $Z$-invariance~\cite{Kenyon3,BoutillierdeTiliere:iso_gen,BdtR2}, 
see Sections~\ref{sec:def_Fisher_graph},~\ref{sec:def_graph_GQ},~\ref{sec:isoradial_Z_invariance} for definitions;
coefficients of the inverse then have the remarkable property of being \emph{local}. We refer to the original papers for the explicit expressions. 

\emph{Uniqueness} is established when the graph $G$ is $\ZZ^2$-periodic~\cite{Sheffield0,BoutillierdeTiliere:iso_perio}. When the inverse
Kasteleyn operator exists and is unique, it is denoted by $K^{-1}$. Note that uniqueness and the fact that the product 
$(K K^{-1})K=K(K^{-1}K)$ is associative implies that if $K^{-1}$ is a right, resp. left, inverse it is also a left, resp. right, 
inverse~\cite{Cooke}.

Consider the $\sigma$-field generated by cylinder sets of $\M(G)$. In all of the above cases, there is
an explicit expression for a Gibbs measure $\PPdimer$ on $(\M(G),\F)$
whose probabilities on cylinder sets is given by the formulas of Theorem~\ref{thm:Kenyon0} with $K^{-1}$ being the inverse Kasteleyn 
operator above. When the graph $G$ is moreover $\ZZ^2$-periodic, this Gibbs measure is obtained as weak limit of the Boltzmann measures 
on the toroidal exhaustion $(G_n)_{n\geq 1}$, where $G_n=G/n\ZZ^2$. We refer to the original papers for an exact statement, see also 
Theorem~\ref{thm:Gibbs_KD} which has the same form.

\subsection{Dimer models on decorated graphs}

In this paper, an important role is played by the dimer model on the \emph{double graph $\GD$},
a model in correspondence with random pairs of dual directed spanning trees~\cite{Temperley,BurtonPemantle,KPW}. Furthermore, 
we consider two dimer representations of the Ising model. The first is related
to the LTE of the Ising model~\cite{KramersWannier1,KramersWannier2}, while the second arises from the XOR-Ising 
model, built from two independent copies of the Ising model~\cite{Dubedat,BoutillierdeTiliere:XORloops}. 
The two corresponding dimer models live on the \emph{Fisher graph} $\GF$ and 
the bipartite graph $\GQ$, respectively. The three graphs $\GD,\GF$ and $\GQ$ are decorated versions of the graph~$\Gs$. 

In the next three sections, we define these decorated graphs and the mappings considered. We treat the case where the graph $\Gs$ is 
infinite or finite. In the 
finite case, the graph $\Gs$, the dual graph $\bar{\Gs}^*$ and the restricted dual $\Gs^*$ are those defined in Section~\ref{sec:defIsing}, where recall
that $\Gs$ has an additional vertex on every boundary edge, and that $\outer$ denotes the vertex of $\bar{\Gs}^*$ corresponding to the 
unbounded face of $\Gs$. In the infinite case, the dual graph is $\Gs^*$. Figures illustrate the finite case; a local picture of the infinite case is
obtained by looking at the interior of the finite case.

\subsubsection{Dimers on the double graph $\GD$ and Temperley's bijection}\label{sec:def_double_graph}

The \emph{double graph} is denoted by $\GD=(\VD,\ED)$. It is defined as follows, see also Figure~\ref{fig:G_Gdouble_tree} (left).

\emph{Infinite case.} Embed the dual graph $\Gs^*$ so that edges of the primal 
and the dual intersect at a single point. The \emph{double graph}
is obtained by superimposing $\Gs$ and $\Gs^*$ and adding an extra vertex at the crossing of each primal 
and dual edge.

\emph{Finite case.} It is constructed similarly to the infinite case from the superimposition of $\Gs$ and the dual graph 
$\bar{\Gs}^*$. Edges incident to the vertex $\outer$ are then removed.

In the infinite and fine cases, the double graph $\GD$ is bipartite and face-cycles are quadrangles. 
The set of black vertices of $\GD$, denoted by $B$, consists of 
vertices of $\Gs$ and $\Gs^*$; the set of white vertices of $\GD$, denoted by $W$, consists of vertices at the crossing of 
edges of $\Gs$ and $\Gs^*$ in the infinite case, and of $\Gs$ and $\bar{\Gs}^*$ in the finite case. 
White vertices are in bijection with edges of the graph $\Gs$, or equivalently with edges of the dual graph. 
We thus have, $\VD=B\cup W$, where $B=\Vs\cup \Vs^*$ and $W\leftrightarrow \Es$.

Suppose again that $\Gs$ is finite, fix a vertex $\rs$ of $\Gs$ amongst the additional vertices on boundary edges, 
and let $\Vs^\rs=\Vs\setminus\{\rs\}$.
Denote by $\GDro$ the graph obtained from $\GD$ by removing the vertex
$\rs$ and all edges incident to it. The graph $\GDro$ is also bipartite;
its set of black vertices is $B^\rs$, where $B^\rs=\Vs^\rs\cup\Vs^*$
and its set of white vertices is $W^\rs=W\leftrightarrow\Es$, see Figure~\ref{fig:G_Gdouble_tree} (right) for an example.
Note that $\GDro$ has the same number of black and white vertices:
$|B^\rs|=|W^\rs|$.

\paragraph{Bijection between pairs of dual directed spanning trees and dimers.} Suppose that $\Gs$ is finite. 
Prior to stating the bijection, we need a few definitions.
A \emph{tree} of $\Gs$ is an acyclic connected subset of edges. A \emph{spanning tree} is 
a tree spanning all vertices of the graph. Let $v$ be a vertex of $\Gs$, then
a \emph{$v$-directed spanning tree} ($v$-dST) is obtained from a spanning tree 
by directing all edges towards the vertex $v$, referred to as the \emph{root}; with such an orientation, every vertex has 
exactly one outgoing edge except the root which has none. Given a spanning tree, the set of dual edges of the edges 
absent in the spanning tree form a spanning tree of the dual graph $\bar{\Gs}^*$, known as the \emph{dual spanning tree}.

Consider the fixed boundary vertex $\rs$ of $\Gs$ as above.
Denote by $\T^\rs(\Gs)$ the set of $\rs$-dST of $\Gs$, by $\T^\outer(\bar{\Gs}^*)$ the set of 
$\outer$-dST of $\bar{\Gs}^*$, and by $\T^{\rs,\outer}(\Gs,\bar{\Gs}^*)$ the set of pairs of dual 
directed spanning trees (dST-pairs) of $\Gs$ and $\bar{\Gs}^*$ such that the primal tree is 
rooted at $\rs$ and the dual tree is rooted at $\outer$, see Figure~\ref{fig:G_Gdouble_tree} (left) for an example.

The result of Temperley~\cite{Temperley}, extended by~\cite{BurtonPemantle} to general non-directed graphs and by~\cite{KPW} 
to the directed case, proves a weight 
preserving bijection between dimer configurations of the double graph $\GDro$ and dST-pairs
of $\T^{\rs,\outer}(\Gs,\bar{\Gs}^*)$. It relies on the following bijection
between edges of $\GDro$ and directed edges of $\Gs$ and $\bar{\Gs}^*$. Let $w\in W\ro=W$, $x\in\Vs\ro\cup\Vs^*$ such that $wx$ is an edge of 
$\GDro$, then 
\begin{equation}\label{equ:bij_Temperley_edges}
wx\, \longleftrightarrow\,  
\begin{cases}
(v,v') \text{ of }\Gs & \text{if $x=v\in\Vs\ro$ and $v'$ is s.t. $w$ belongs to the edge $(v,v')$}\\
(f,f') \text{ of }\bar{\Gs}^* & \text{if $x=f\in\Vs^*$ and $f'$ is s.t. $w$ belongs to the edge $(f,f')$}.
\end{cases}
\end{equation}
Note that there are no directed edges of $\Gs$ exiting the vertex $\rs$, and no directed edges of $\bar{\Gs}^*$ exiting the vertex $\outer$.
Using this bijection, a subset of edges of $\GDro$ corresponds to a
subset of directed edges of $\Gs$ and $\bar{\Gs}^*$; Temperley's bijection states that subsets defining dimer configurations
are in correspondence with subsets defining dST-pairs of $\T^{\rs,\outer}(\Gs,\bar{\Gs}^*)$. 
An example is provided in Figure~\ref{fig:G_Gdouble_tree}, the vertex $\outer$ is represented in a spread-out way, \emph{i.e.},
the dotted line should be thought of as being the single vertex $\outer$.

\begin{figure}[ht]
\begin{minipage}[b]{0.5\linewidth}
\begin{center}
\begin{overpic}[width=7.5cm]{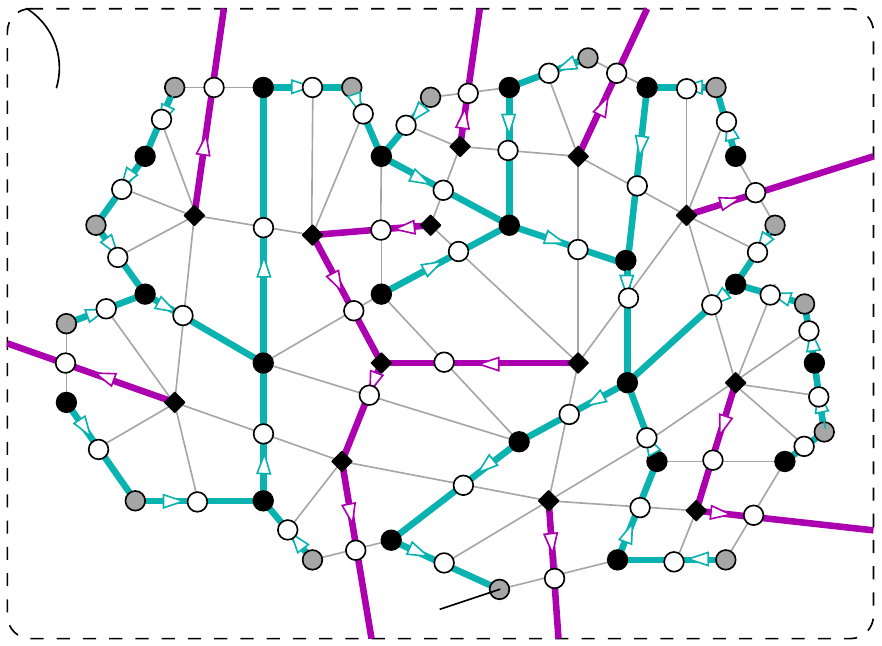}
\put(4,61){\scriptsize $\outer$}
\put(48,3){\scriptsize $\rs$}
\end{overpic}
\end{center}
\end{minipage}
\begin{minipage}[b]{0.5\linewidth}
\begin{center}
\begin{overpic}[width=7.5cm]{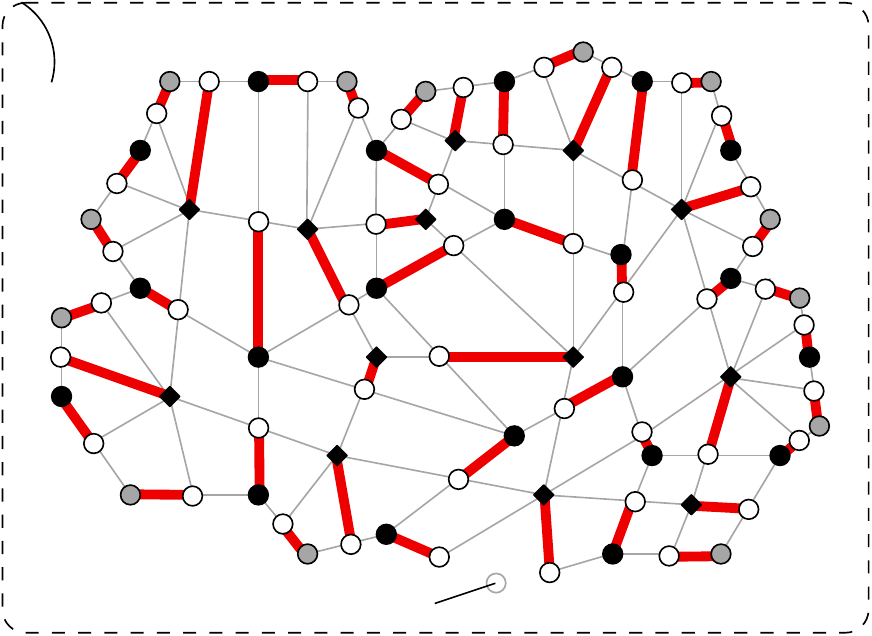}
\put(48,3){\scriptsize $\rs$}
\put(4,61){\scriptsize $\outer$}
\end{overpic}
\end{center}
\end{minipage}
\caption{Left: double graph $\GD$ of the graph $\Gs$ of Figure~\ref{fig:low_temp} (grey lines), 
$\rs$-directed spanning tree of $\Gs$ (turquoise) 
and dual $\outer$-directed spanning tree of $\bar{\Gs}^*$ (purple). Right: graph $\GDro$ (grey) and the dimer configuration (red)
in bijection with the pair of dual directed spanning trees of the left figure.}
\label{fig:G_Gdouble_tree}
\end{figure}
Let $c$ be a weight function on edges of $\GDro$ and $\tilde{c}$ a weight function on directed edges of $\Gs,\bar{\Gs}^*$. 
The relation between $c$ and $\tilde{c}$ which makes Temperley's bijection weight preserving naturally arises from 
the bijection between edges of $\GDro$ and directed edges of $\Gs,\bar{\Gs}^*$. Let $w\in W\ro$, $x\in \Vs^\rs\cup \Vs^*$, such that $wx$ is an edge 
of $\GDro$. Using the notation of~\eqref{equ:bij_Temperley_edges}, we have 
\begin{equation}\label{equ:relation_c_ctilde}
c_{wx}=
\begin{cases}
\tilde{c}_{v,v'} & \text{ if $x=v\in\Vs\ro$}\\
\tilde{c}_{f,f'} & \text{ if $x=f\in\Vs^*$},
\end{cases}
\end{equation}
and $\tilde{c}_{\rs,v}=0$ for every vertex $v\in \Vs^\rs$ adjacent to $\rs$, $\tilde{c}_{\outer,f'}=0$ for every vertex $f'\in\Vs^*$
adjacent to $\outer$.

\paragraph{Model on pairs of dual directed spanning trees.} Suppose that directed edges of $\Gs,\bar{\Gs}^*$ are assigned the weight function 
$\tilde{c}$. Consider the \emph{Boltzmann measure on dST-pairs}, denoted 
$\PPtreepair$, defined by
\begin{equation*}
\forall\,(\Ts,\Ts^*)\in\T^{\rs,\outer}(\Gs,\bar{\Gs}^*),
\quad 
\PPtreepair(\Ts,\Ts^*)=
\frac{\bigl(\prod_{(v,v')\in \Ts} \tilde{c}_{v,v'}\bigr)\bigl(\prod_{(f,f')\in\Ts^*}\tilde{c}_{f,f'}\bigr)}{\Ztreepair(\Gs,\bar{\Gs}^*)},
\end{equation*}
where $
\Ztreepair((\Gs,\bar{\Gs}^*),\tilde{c})=\sum_{(\Ts,\Ts^*)\in\T^{\rs,\outer}(\Gs,\bar{\Gs}^*)} 
\Bigl(\prod_{(v,v')\in \Ts} \tilde{c}_{v,v'}\Bigr)\Bigl(\prod_{(f,f')\in\Ts^*}\tilde{c}_{f,f'}\Bigr)$,
is the \emph{dST-pairs partition function}. 
As a consequence of the KPW-Temperley bijection~\cite{Temperley,KPW}, we have
\begin{equation*}
\Ztreepair((\Gs,\bar{\Gs}^*),\tilde{c})=\Zdimer(\GDro,c).
\end{equation*}
There is also a natural correspondence between the dST-pairs Boltzmann measure $\PPtreepair$ and
the dimer Boltzmann measure $\mathbb{P}_{\mathrm{dimer}}^{\scriptscriptstyle{\,\mathrm{D}}}$ on $\GDro$ with weight function $c$.

Note that if $\tilde{c}\equiv 1$ on edges of $\bar{\Gs}^*$, resp. on edges of $\Gs$, then 
$\Ztreepair((\Gs,\bar{\Gs}^*),\tilde{c})$ is equal to the partition function $\Ztree^\rs(\Gs,\tilde{c})$
of $\rs$-directed spanning trees of $\Gs$, resp. $\Ztree^\outer(\bar{\Gs}^*,\tilde{c})$ of $\outer$-directed spanning 
trees of $\bar{\Gs}^*$.

\subsubsection{Dimers on the Fisher graph $\GF$ and the LTE of the Ising model}\label{sec:def_Fisher_graph}

The \emph{Fisher graph} is denoted by $\GF=(\VF,\EF)$. It is constructed as follows~\cite{Fisher,Dubedat},
see Figure~\ref{fig:G_GFisher} for an example.

\emph{Infinite case.} Start from the dual graph $\Gs^*$ and
replace every vertex of $\Gs^*$ by a \emph{decoration} made of triangles, where each of the triangles corresponds to 
an edge incident to this vertex, then join the triangles in a circular way.

\emph{Finite case.} Start from the dual graph $\bar{\Gs}^*$ and do the same procedure as in the infinite case. Then,
remove the decoration of the vertex $\outer$ as well as all edges of $\bar{\Gs}^*$ incident to 
this decoration.

In both the infinite and finite case, the Fisher graph consists of \emph{internal edges}, which are edges of the decorations, and \emph{external edges} which are in
bijection with edges of $\Gs^*$ and will often be identified with them. Each decoration has a dual vertex in its center, giving a way of
identifying decorations and vertices of $\Gs^*$.

\begin{figure}[ht]
\begin{center}
\begin{overpic}[width=7.5cm]{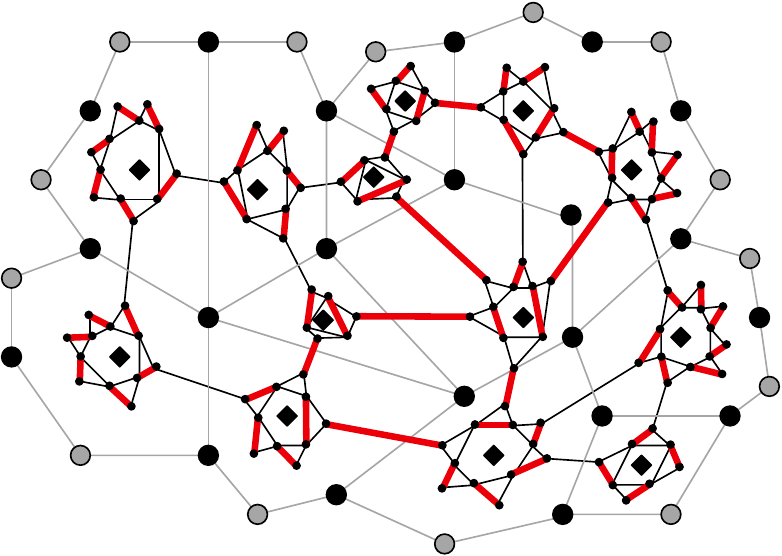}
\end{overpic}
\end{center}
\caption{The Fisher graph $\GF$ for the LTE expansion of the Ising model on $\Gs$ with + boundary conditions (black); one of
the $2^{13}$ dimer configurations corresponding to the polygon configuration of Figure~\ref{fig:low_temp}.}
\label{fig:G_GFisher}
\end{figure}

\paragraph{Mapping between LTE polygon configurations and dimers.} Suppose that $\Gs$ is finite.
Fisher~\cite{Fisher} introduces a mapping between polygon configurations of $\Gs^*$ and dimer
configurations of the corresponding Fisher graph $\GF$. To a given polygon configuration of $\Gs^*$, there corresponds $2^{|\Vs^*|}$ dimer configurations of $\GF$: edges of the polygon 
configuration are exactly the external edges of the dimer configurations and given these external edges, there is 
exactly two ways of filling each decoration so as to have a dimer configuration~\cite{Fisher}, see Figure~\ref{fig:G_GFisher} for an example.
This mapping naturally extends when the graph $\Gs$ is infinite.

We consider polygon configurations arising from the LTE expansion of the Ising model on $\Gs$ with + boundary conditions and coupling 
constants $\Js$. In order for this correspondence to be weight preserving, the dimer weight function $\muJ$ on edges of $\GF$ is 
defined to be, see Equation~\eqref{equ:LTE}:
\begin{equation*}
\forall\,\text{ edge $\es$ of $\GF$},\quad 
\muJ_\es=
\begin{cases}
1&\text{ if $\es$ is an internal edge}\\
e^{-2\Js_e}&\text{ if $\es$ is an external edge arising from a dual edge $e^*$ of $\Gs^*$}.
\end{cases}
\end{equation*}
Let $\PPdimerF$ and $\Zdimer(\GF,\muJ)$ be the corresponding dimer Boltzmann measure and partition function. Then, as a consequence
of Fisher's correspondence we have,
\begin{equation}\label{equ:PF_Ising_1}
\Zising^+(\Gs,\Js)=2^{-|\Vs^*|}\Bigl(\prod_{e\in\Es}e^{\Js_e}\Bigr)\Zdimer(\GF,\muJ).
\end{equation}

\subsubsection{Dimers on the bipartite graph $\GQ$ and the XOR-Ising model}\label{sec:def_graph_GQ}

The \emph{quadri-tiling graph} is denoted by $\GQ=(\VQ,\EQ)$, where the name comes from the paper~\cite{deTiliere:quadri}. 
In both the 
finite and infinite cases, we start from the 
preceding definition of the double graph $\GD$. Recall that face-cycles of $\GD$ are quadrangles
consisting of two black and two white vertices, then add the edges joining opposite black vertices in quadrangles.

\emph{Infinite case}. The graph $\GQ$ is the dual of the modified graph $\GD$.

\emph{Finite case}. The graph $\GQ$ is the restricted dual of the modified graph $\GD$, see Figure~\ref{fig:G_GQ}.

Vertices of $\GQ$ are partitioned as $\VQ=\BQ\cup\WQ$, and the bipartite coloring is fixed as in Figure~\ref{fig:G_GQ}. Black, resp. white,
vertices of $\GQ$ are denoted by $\bs$, resp. $\ws$, with or with sub/super-scripts.

In the infinite case, the graph $\GQ$ consists of \emph{quadrangles} that are joined by \emph{external edges}. Quadrangles
are in bijection with edges of $\Gs$, or equivalently edges of $\Gs^*$, or equivalently white vertices of $\GD$: each quadrangle
has a white vertex of $\GD$ in its interior, two of its edges are ``parallel'' to an edge of $\Gs$ and the two other edges are
``parallel'' to the dual edge of $\Gs^*$. Face-cycles of $\GQ$ other than quadrangles either have a vertex of $\Gs$ or a 
vertex of $\Gs^*$ in their interior. 

In the finite case, the description is similar away from the boundary. Along the boundary 
``quadrangles'' in bijection with boundary edges of $\Gs$, or equivalently with boundary white vertices of $\GD$, are actually reduced to single edges
``parallel'' to boundary edges of $\Gs$. We refer to those degenerate quadrangles as \emph{boundary quadrangles} of $\GQ$, keeping in mind
that they actually are \emph{edges}. Note that some quadrangle edges of $\GQ$ are boundary edges of $\GQ$ (in the sense that they belong to the boundary
of the unbounded face) but still belong to ``full'' quadrangles; as such they are \emph{not} boundary quadrangle edges.

\begin{figure}[ht]
\begin{center}
\begin{overpic}[width=7.5cm]{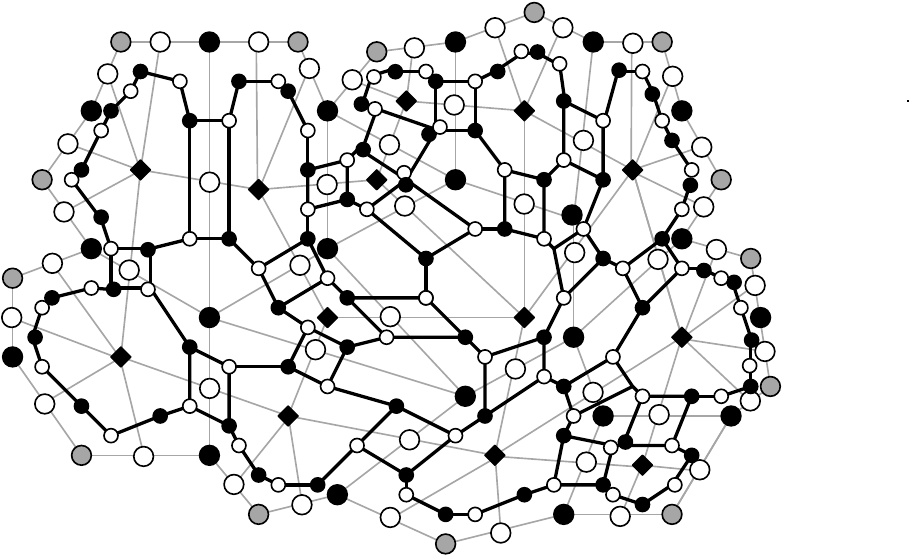}
\end{overpic}
\end{center}
\caption{The quadri-tiling graph $\GQ$: take the double graph $\GD$ of 
Figure~\ref{fig:G_Gdouble_tree} (left) and add edges joining opposite black vertices in quadrangle-faces; the restricted dual of this graph
is~$\GQ$.}
\label{fig:G_GQ}
\end{figure}

We consider the dimer model on the bipartite graph $\GQ$ arising from the \emph{XOR-Ising model} \cite{WilsonXOR}, also known as 
the \emph{polarization} of the Ising model~\cite{KadanoffBrown}, obtained by taking the product of the spins of two independent
Ising models. There are two mappings leading to the dimer model on $\GQ$, both of them are rather long to describe so that 
we refer to the original papers:~\cite{Dubedat} based on results of~\cite{KadanoffWegner,Wu71,FanWu,WuLin} for the first approach,
and~\cite{BoutillierdeTiliere:XORloops} based on results of~\cite{Nienhuis,WuLin} for the second one. The dimer weight function 
$\nuJ$ on $\GQ$ is defined by, for every edge $\es$ of $\GQ$,
\begin{equation*}
\nuJ_{\es}=
\begin{cases}
1&\text{ if $\es$ is an external edge}\\
1&\text{ if $\es$ is a boundary quadrangle edge in the finite case}\\
\tanh(2\Js_e)&\text{ if $\es$ is a quadrangle edge/non-boundary quadrangle edge}\\
&\text{ in the infinite/finite case, ``parallel'' to an edge $e$ of $\Gs$}\\
\cosh^{-1}(2\Js_e)&\text{ if $\es$ is a quadrangle edge, ``parallel'' to a dual edge $e^*$ of an edge $e$ of $\Gs$}.
\end{cases}
\end{equation*}
When the graph $\GQ$ is finite, we let $\PPdimerQ$ and $\Zdimer(\GQ,\nuJ)$ be the corresponding dimer Boltzmann measure and partition function.
As a consequence of~\cite{Dubedat}, see also Corollary~\ref{cor:det_KF_KQ}, we have:
\begin{equation*}
\Zdimer(\GF,\muJ)^2=2^{|\Vs^*|}\prod_{e^*\in \Es^*}(1+e^{-4\Js_e})\Zdimer(\GQ,\nuJ).
\end{equation*}
Combining this with Equation~\eqref{equ:PF_Ising_1} for the Ising partition function, and denoting by $\Es^\spartial$
the set of boundary edges of the graph $\Gs$, we obtain
\begin{align}\label{equ:PF_Ising_2}
[\Zising^+(\Gs,\Js)]^2&=
2^{-|\Vs^*|}\Bigl(\prod_{e\in\Es} e^{2\Js_e}\Bigr)\Bigl(\prod_{e^*\in\Es^*}(1+e^{-4\Js_e})\Bigr)\Zdimer(\GQ,\nuJ)\nonumber \\
&=2^{-|\Vs^*|+|\Es^*|}\Bigl(\prod_{e\in\Es^\spartial} e^{2\Js_e}\Bigr)\Bigl(\prod_{e^*\in\Es^*}\cosh(2\Js_e)\Bigr)\Zdimer(\GQ,\nuJ)\nonumber \\
&= 2^{|\Vs|-1}\Bigl(\prod_{e\in\Es^\spartial} \frac{e^{2\Js_e}}{2}\Bigr)\Bigl(\prod_{e^*\in\Es^*}\cosh(2\Js_e)\Bigr)\Zdimer(\GQ,\nuJ),
\end{align}
where in the last line we used that $|\Es|=|\Es^*|+|\Es^\spartial|$ and Euler's formula: $|\Es|=|\Vs|+|\Vs^*|-1$.

\subsection{Rooted directed spanning forests and directed spanning trees}\label{sec:spanning_forests}

We also need the model of \emph{rooted directed spanning forests} on the graphs $\Gs$ and $\Gs^*$, resp. $\Gs^\rs$ and $\Gs^*$,
in the infinite, resp. finite, case. So as to include both the primal and the dual graphs, we now define this model on a simple graph 
$G=(V,E)$.

Suppose that vertices are assigned non-negative \emph{masses}, denoted $m=(m_{x})_{x\in V}$, and that directed
edges have positive \emph{conductances}, denoted $\rho$, meaning that every directed edge $(x,x')$
has conductance $\rho_{x,x'}$. 

A \emph{rooted directed spanning forest} (rdSF) of $G$ is a subset of edges spanning all vertices of the graph, 
such that each connected component is a directed tree $\Ts$ rooted at a vertex of $G$, denoted $x_\Ts$. Let 
$\F(G)$ denoted the set of rdSF of the graph $G$. 

Suppose that $G$ is finite and consider the \emph{Boltzmann measure on rdSF}, denoted $\PPforest$, defined by:
\begin{equation*}
\forall\,\Fs\in\F(G),\quad \PPforest(\Fs)=
\frac{\prod_{\Ts\in\Fs}\Bigl(m_{x_\Ts} \prod_{(x,x')\in\Ts}\rho_{x,x'}\Bigr)}{\Zforest(G,\rho,m)},
\end{equation*}
where $\Zforest(G,\rho,m)=\sum_{\Fs\in \F(G)} \prod_{\Ts\in\Fs}m_{x_\Ts} \prod_{(x,x')\in\Ts}\rho_{x,x'}$ 
is the \emph{rdSF partition function}. Whenever conductances are symmetric, \emph{i.e.}, $\rho_{x,x'}=\rho_{x',x}$, we will remove 
the ``d'' in rdSF. 

As a consequence of the directed version of Kirchhoff's matrix-tree theorem~\cite{Kirchhoff,Tutte}, the rdSF partition function is computed 
using the massive Laplacian operator/matrix as follows.
The \emph{massive 
Laplacian operator} $\Delta^m:\CC^V\rightarrow\CC^V$ is defined by:
\begin{equation*}
\forall\, F\in\CC^V,\quad (\Delta^m F)_x=\sum_{x'\sim x} \rho_{x,x'}(F_x-F_{x'}) + m_x F_x.
\end{equation*}
The operator $\Delta^m$ is represented by a matrix, also denoted $\Delta^m$, whose non-zero coefficients are given by:
\begin{equation*}
\Delta^m_{x,x'}=
\begin{cases}
-\rho_{x,x'}& \text{ if $(x,x')$ is an edge of $G$}\\ 
m_x+\sum_{x'\sim x} \rho_{x,x'} & \text{ if $x'=x$ is a vertex of $G$}.
\end{cases}
\end{equation*}
A function $F\in\CC^V$ is said to be \emph{massive-harmonic}, if $\Delta^m F=0$.

Consider the graph $G_\dag$ constructed from $G$ by adding a \emph{cemetery vertex} $\dag$ and an edge $(x,\dag)$ for every 
vertex $x$ such that $m_x\neq 0$. Define the modified weight function $\rho^m$ on (directed) edges of $G_\dag$ by, 
\begin{equation*}
\forall\text{ edge $(x,x')$ of }G_\dag,\quad
\rho^m_{x,x'}=
\begin{cases}
\rho_{x,x'} & \text{ if $x'\neq\dag$}\\
m_x & \text{ if $x'=\dag$}.
\end{cases}
\end{equation*}
There is a natural weight-preserving bijection between $\T^\dag(G_\dag)$ and $\F(G)$: a $\dag$-directed spanning tree of $G_\dag$ corresponds to the 
rooted directed spanning forest of $G$
obtained by replacing every edge $(x,\dag)$ of the dST by a root of the rdSF. 

Denote by $\Delta_\dag$ the (non-massive) Laplacian matrix of $G_\dag$ with weight function $\rho^m$ on the edges. 
Then, $\Delta^m$ is the Laplacian matrix $\Delta_\dag$ from which one has removed
the row and column corresponding to the cemetery $\dag$ and thus, by Kirchhoff's matrix-tree theorem~\cite{Kirchhoff,Tutte}, 
the determinant of $\Delta^m$ counts $\rho^m$ weighted $\dag$-dST of $G_\dag$. Using the bijection between $\dag$-dST of $G_\dag$ and rdSF of $G$,
we thus have,
\begin{thm}[\cite{Kirchhoff,Tutte}]\label{thm:Kirchhoff_Tutte}
\begin{equation*}
\Zforest(G,\rho,m)=\Ztree^\dag(G_\dag,\rho^m)=\det(\Delta^m).
\end{equation*}
\end{thm}

When there is at least one vertex with positive mass, the~\emph{massive Green function}, denoted $G^m$, 
is the inverse of the massive Laplacian $\Delta^m$. In the chore of the paper, graphs are written with the letter $\Gs$ with or without
superscripts, so that we believe that the notation $G^m$ will not create confusion. The massive Green function is naturally related to
the expected number of visits of the network random walk associated to the conductances $\rho$ and masses $m$, see for example Appendix 
D of~\cite{BdTR1}, where a number of facts are recalled.

\subsection{Isoradial graphs and $Z$-invariance}\label{sec:isoradial_Z_invariance}

Sections~\ref{sec:Z_Dirac_Z_Lap},~\ref{sec:GQ_Z_Dirac} and~\ref{sec:GQ_GF_Zinv} use $Z$-invariant models defined on isoradial graphs.
We recall these notions, related concepts and more specifically give the definitions of the $Z$-invariant versions of the Ising model on $\Gs$, of the dimer model
on the decorated graphs $\GF$ and $\GQ$ and of rooted spanning forests on $\Gs$ or $\Gs^*$.

\subsubsection{Isoradial graphs, diamond graphs and angles}

Isoradial graphs naturally appear when considering a discrete version of the Cauchy-Riemann equations, see~\cite{Duffin} 
and~\cite{Mercat:ising,Kenyon3,ChelkakSmirnov:toolbox}; they also arise in $Z$-invariant models when solving the corresponding
Yang-Baxter equations~\cite{Baxter:exactly,CostaSantos}; the name \emph{isoradial} comes from the paper~\cite{Kenyon3}.

Suppose that $\Gs$ is an infinite, planar graph. Then $\Gs$ is said to be \emph{isoradial} if it can be embedded in the plane 
in such a way that every face is inscribable in a circle of the same radius, and such that the circumcircles
are in the interior of the faces. We consider $\Gs$ as an embedded graph and take the common radius to be 
2. Note that the dual $\Gs^*$ of an isoradial graph is also isoradial, an embedding of $\Gs^*$ is 
obtained by taking as vertices the circumcenters of the circles.

This definition also holds when the graph is finite. Recall that in this case, the notation $\Gs$ is used for the graph 
having an additional vertex on each boundary edge. We now fix
the isoradial embedding of $\Gs$ when the original graph (the one without the additional vertices) is isoradial. This is done in the same way as 
in~\cite{ChelkakSmirnov:ising}: each additional boundary vertex corresponds to a boundary edge $xy$ of the original graph
and we embed this additional vertex in the middle of the arc joining $x$ and $y$, see Figure~\ref{fig:iso} (left and right).

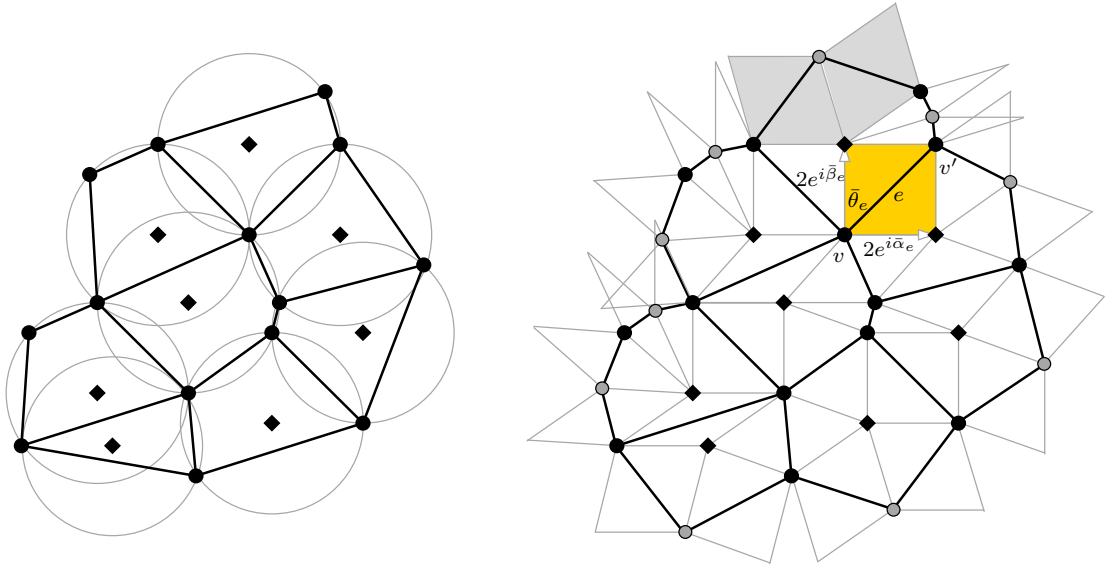
\begin{figure}[ht]
\begin{minipage}[b]{0.5\linewidth}
\begin{center}
\begin{overpic}[width=8cm]{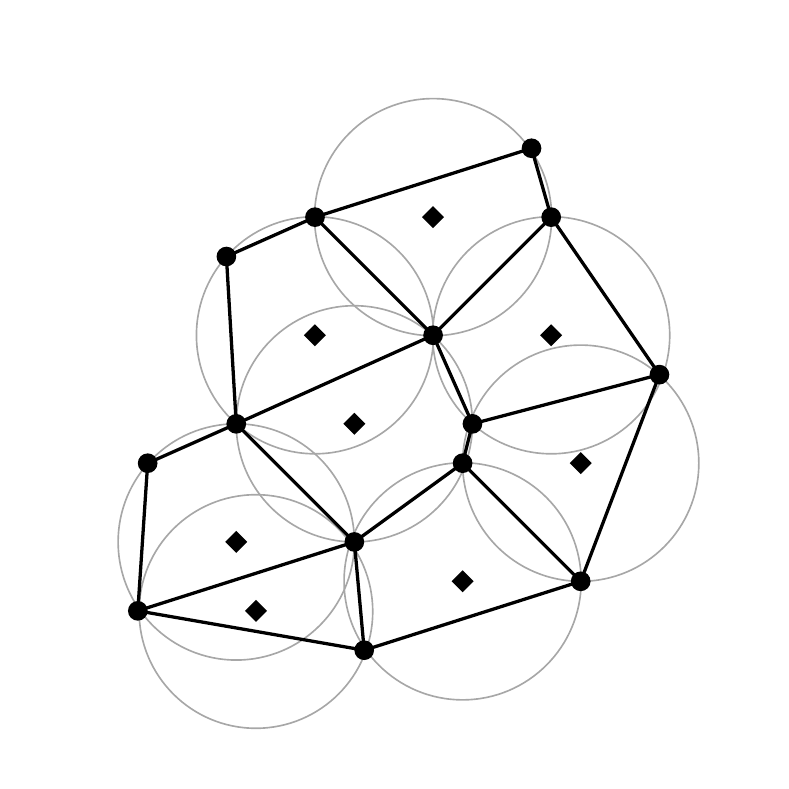}
\end{overpic}
\end{center}
\end{minipage}
\begin{minipage}[b]{0.5\linewidth}
\begin{center}
\begin{overpic}[width=8cm]{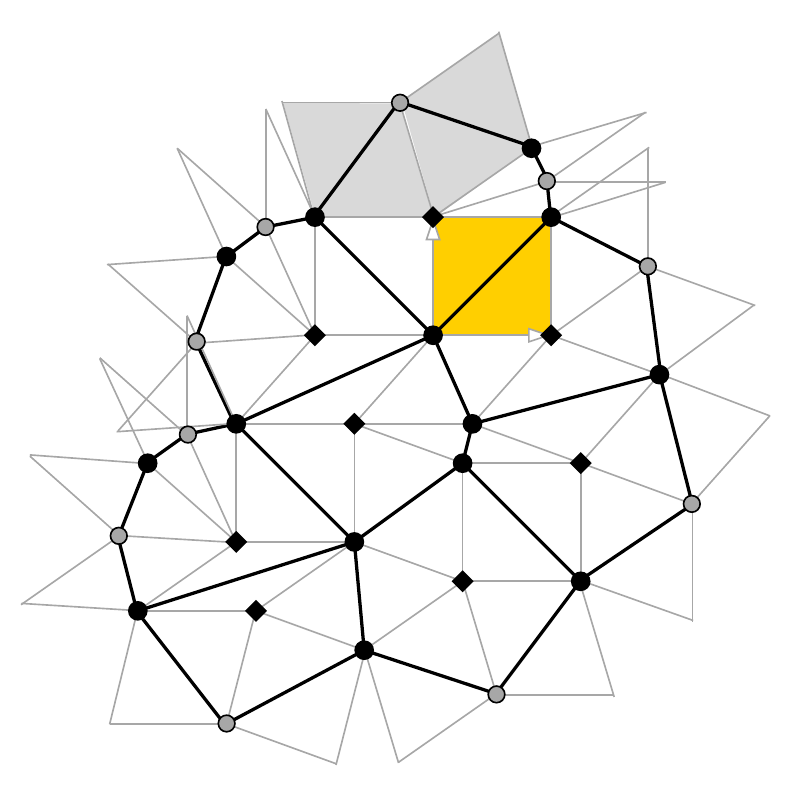}
\put(63,63){\scriptsize $e$}
\put(53,53){\scriptsize $v$}
\put(70.5,67.5){\scriptsize $v'$}
\put(58,54){\scriptsize $2e^{i\bar{\alpha}_e}$}
\put(47,65.5){\scriptsize $2e^{i\bar{\beta}_e}$}
\put(55.5,62){\scriptsize $\bar{\theta}_e$}
\end{overpic}
\end{center}
\end{minipage}
\caption{Left: original isoradial graph (black lines) with circumcircles (grey) and dual vertices embedded as circumcenters (diamonds).
Right: isoradial graph $\Gs$ (black lines) with the additional boundary vertices (grey bullets); diamond graph $\GR$ (grey lines);
rhombus (yellow) assigned to an edge $e=(v,v')$ with the corresponding half-angle $\bar{\theta}_e$ and rhombus vectors 
$2e^{i\bar{\alpha}_e}$, $2e^{i\bar{\beta}_e}$; a boundary rhombus pair of $\R^\spartial$ (light grey).}
\label{fig:iso}
\end{figure}

In the infinite case, the \emph{diamond graph}, denoted~$\GR$, is constructed from an isoradial graph $\Gs$ and its dual $\Gs^*$ as follows:
its vertex set is $\Vs\cup\Vs^*$, the vertices of $\Gs$ and $\Gs^*$; and each dual vertex is joined to all vertices bounding the face 
it corresponds to. Since the graph $\Gs$ is isoradial, faces of the diamond graph $\GR$ are side-length-2 rhombi. 

There is a bijection between rhombi of $\GR$ and pairs $e,e^*$ of primal and dual edges, the latter being the two diagonals of 
the rhombi. To every edge $e$, one assigns an \emph{angle} $\bar{\theta}_e\in(0,\frac{\pi}{2})$ 
defined to be the half-angle of the corresponding rhombus at the edge $e$. We furthermore ask that $\bar{\theta}_e\in(\eps,\frac{\pi}{2}-\eps)$,
for some $\eps>0$. The rhombus angle of the dual edge $e^*$ is 
$\bar{\theta}_{e^*}=\frac{\pi}{2}-\bar{\theta}_e:=\bar{\theta}_e^*$.
To a directed edge $e=(v,v')$ of $\Gs$ we further assign two rhombus vectors $2e^{i\bar{\alpha}_e}$, $2e^{i\bar{\beta}_e}$ of $\GR$, such that
$2e^{i\bar{\alpha}_e}$ is on the right of the edge $(v,v')$, see Figure~\ref{fig:iso} (right). The angles $\bar{\alpha}_e$ and $\bar{\beta}_e$ are 
defined so that $\frac{\bar{\beta}_e-\bar{\alpha}_e}{2}=\bar{\theta}_e$. Whenever no confusion occurs, we remove the subscript $e$ from the notation.
In the finite case, the \emph{diamond graph}, also denoted $\GR$, is constructed in a similar way from $\Gs$ and 
its restricted dual $\Gs^*$. One then adds the missing half-rhombi along the boundary; they may overlap but this causes no 
problem, see Figure~\ref{fig:iso} (right). Angles and rhombus vectors assigned to edges are defined in the same way. 

Because of the additional vertex on each boundary edge 
and because of our choice of embedding, rhombi along the boundary of 
$\GR$ come in pairs, both having the same rhombus half-angle; let us denote by $\R^\spartial$ the set of boundary rhombus pairs,  
an instance is highlighted in Figure~\ref{fig:iso} (right, light grey).

\subsubsection{Isoradial embeddings of the decorated graphs $\GD$ and $\GQ$}\label{sec:iso_GDGQ}

Consider an isoradial graph $\Gs$, its dual $\Gs^*$ in the infinite case and its restricted dual $\Gs^*$ in the finite case.
The double graph $\GD$ is embedded so that the black vertices are those of $\Gs$ and $\Gs^*$ and the white 
vertices are at the crossing of the diagonals of the rhombi of $\GR$, see Figure~\ref{fig:iso_2} (left).

\begin{figure}[ht]
\begin{minipage}[b]{0.5\linewidth}
\begin{center}
\begin{overpic}[width=7cm]{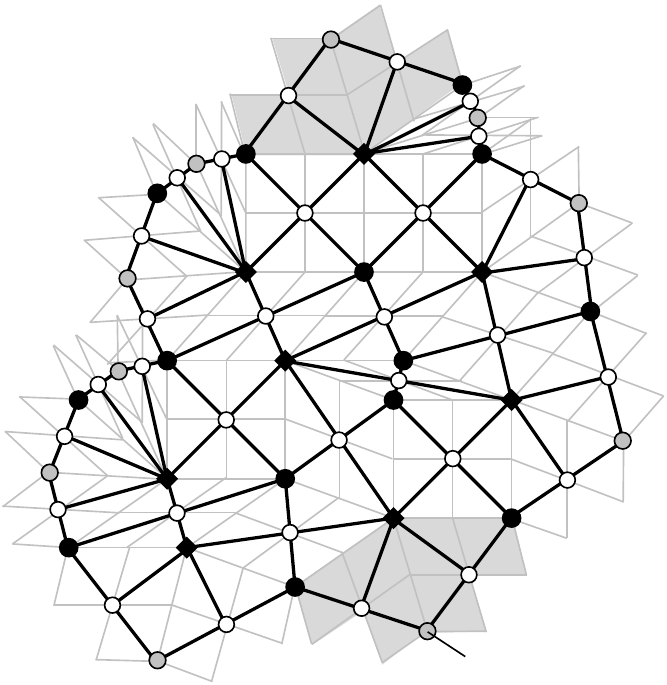}
\put(69,2.5){\scriptsize $\rs$}
\end{overpic}
\end{center}
\end{minipage}
\begin{minipage}[b]{0.5\linewidth}
\begin{center}
\begin{overpic}[width=7cm]{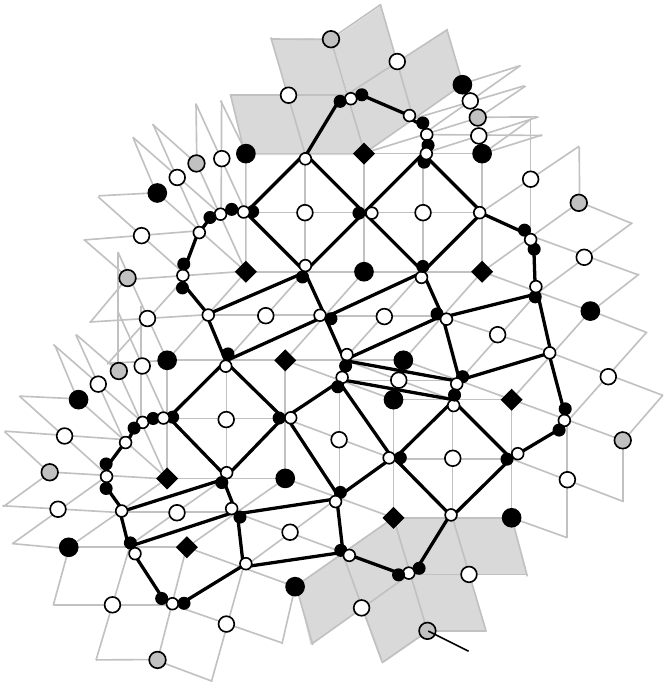}
\put(69.5,3){\scriptsize $\rs$}
\end{overpic}
\end{center}
\end{minipage}
\caption{Left: isoradial embedding of the graph $\GD$ (black lines).
Right: isoradial embedding of the graph $\GQ$ (black lines). In both cases is also pictured the diamond graph $\GRR$ (grey lines),
a boundary rhombus pair of $\R^\spartial$ and the root pair (light grey).}
\label{fig:iso_2}
\end{figure}

In the infinite case and in the finite non-boundary case, 
consider the embedding of the bipartite graph $\GQ$ where external edges have length-0
and their endpoints become a single vertex in the middle of the rhombus edges of $\GR$, see Figure~\ref{fig:iso_2} 
(right, inner vertices); then, inner quadrangles of $\GQ$ are rectangles. 
Note that although external edges are embedded as single vertices, 
they still consist of two vertices joined by (a length-0) edge, \emph{i.e.}, the combinatorics of the graph does not change.

When the graph $\GQ$ is \emph{finite}, the procedure along 
the boundary is different. Consider a boundary rhombus pairs of $\R^\spartial$,
the following notation will be used throughout the paper and is illustrated in Figure~\ref{fig:rhombus_pair} below. 
Let
$v^\ell,v^c,v^r$ be the vertices of $\Gs$ in cw order and let $f^c$ be the vertex of $\Gs^*$; note that 
$v^c$ is the additional vertex on the edge 
$v^\ell v^r$ of the original graph. Denote by $w^\ell,w^r$ the white vertices of the double graph $\GD$ and by 
$\ws^\ell,\bs^\ell,\ws^c,\bs^r,\ws^r$ the vertices of $\GQ$. Then, 
taking the same convention as in the infinite case 
for the embedding gives Figure~\ref{fig:rhombus_pair} (left); but it turns out that the appropriate embedding to obtain Theorem~\ref{thm:main}
is that of Figure~\ref{fig:rhombus_pair} (center), see also Figure~\ref{fig:iso_2} (right),
where the boundary quadrangle edge $\bs^\ell\ws^c$ has length-0
and is ``replaced'' by the external edge $\ws^\ell \bs^\ell$. This change
of embedding preserves the combinatorics of the graph; it has the effect of exchanging the colors of the bipartite coloring of $\GQ$ 
in the left rhombus of the rhombus pair.

We will often be using the fact that vertices/edges of the boundary rhombus pairs of $\R^\spartial$ 
encode: boundary vertices/edges of $\Gs$, boundary vertices of the restricted dual $\Gs^*$, where a \emph{boundary vertex} of $\Gs^*$ is
defined to be a vertex adjacent to the vertex $\outer$ in $\bar{\Gs}^*$; boundary quadrangle vertices/edges of $\GQ$, where recall that 
boundary quadrangles are degenerate and reduced to edges in bijection with boundary edges of $\Gs$. 
We will use the notation $\{v^\ell\in\R^\spartial\}$ for the set of vertices of type $v^\ell$ belonging to boundary rhombus pairs, and 
similarly for other vertices or edges of $\R^\spartial$.

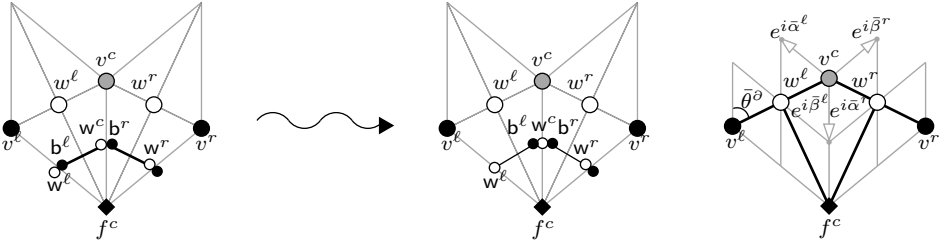
\begin{figure}[ht]
\begin{minipage}{0.66\linewidth}
\begin{center}
\begin{overpic}[width=8.5cm]{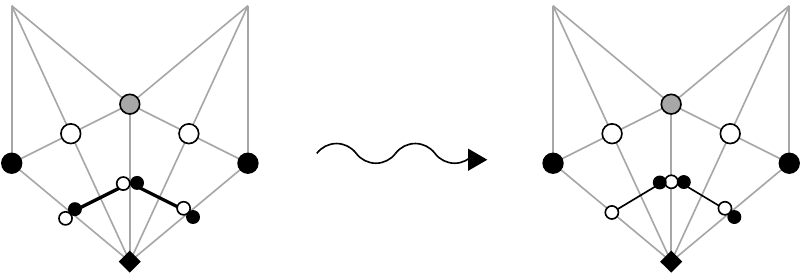}
\put(0,10.5){\scriptsize $v^\ell$}
\put(14.5,23.5){\scriptsize $v^c$}
\put(30,10.5){\scriptsize $v^r$}
\put(8,20){\scriptsize $w^\ell$}
\put(20,20){\scriptsize $w^r$}
\put(7,4){\scriptsize $\ws^\ell$}
\put(7.5,9.8){\scriptsize $\bs^\ell$}
\put(12,13){\scriptsize $\ws^c$}
\put(16.5,12.8){\scriptsize $\bs^r$}
\put(22,9.8){\scriptsize $\ws^r$}
\put(14.5,-2.8){\scriptsize $f^c$}
\put(68,10.5){\scriptsize $v^\ell$}
\put(82,23.5){\scriptsize $v^c$}
\put(97,10.5){\scriptsize $v^r$}
\put(76,20){\scriptsize $w^\ell$}
\put(87.5,20){\scriptsize $w^r$}
\put(74.5,4.5){\scriptsize $\ws^\ell$}
\put(78.2,13){\scriptsize $\bs^\ell$}
\put(82,13.5){\scriptsize $\ws^c$}
\put(86,13){\scriptsize $\bs^r$}
\put(88.5,9.8){\scriptsize $\ws^r$}
\put(82,-2.8){\scriptsize $f^c$}
\end{overpic}
\end{center}
\end{minipage}
\begin{minipage}{0.33\linewidth}
\vspace{0.75cm}
\begin{overpic}[width=2.8cm]{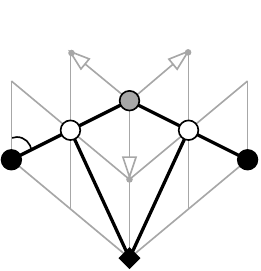}
\put(1,31){\scriptsize $v^\ell$}
\put(43,69){\scriptsize $v^c$}
\put(88,30.5){\scriptsize $v^r$}
\put(27,59){\scriptsize $w^\ell$}
\put(58,59){\scriptsize $w^r$}
\put(43,-8){\scriptsize $f^c$}
\put(21,81){\scriptsize $e^{i\bar{\alpha}^\ell}$}
\put(30.5,45){\scriptsize $e^{i\bar{\beta}^\ell}$}
\put(48,45){\scriptsize $e^{i\bar{\alpha}^r}$}
\put(58,81){\scriptsize $e^{i\bar{\beta}^r}$}
\put(8,49.5){\scriptsize $\bar{\theta}^\spartial$}
\end{overpic}
\begin{center}
\end{center}
\end{minipage}
\caption{Notation for vertices of boundary pairs of rhombi of $\R^\spartial$. 
Left: isoradial embedding of $\GQ$ under the convention that external edges have length 0.
Center: isoradial embedding used in this paper. Right: rhombus vectors of $\GRR$ and half-angle assigned to the edges $(v^c,w^\ell),(v^c,w^r)$ of $\GD$.
}
\label{fig:rhombus_pair}
\end{figure}

The embeddings of $\GD$ and $\GQ$ are both isoradial with circumcircles having common radius $1$. 
Consider the graph obtained from the diamond graph $\GR$ by cutting rhombi into four equal length-1 rhombi. Denote this graph 
by $\GRR$ in the infinite case and, in the finite case, let $\GRR$ be this graph where 
the boundary quarter rhombi crossed by no edge of $\GD$ are removed.
Then $\GRR$ is the diamond graph of $\GD$. Note that $\GRR$ is nearly the diamond graph of $\GQ$: it is slightly extended along the 
boundary and one should think of it as having flat rhombi associated to length-0 edges of $\GQ$. We will nevertheless refer to it
as the \emph{diamond graph} of $\GD$ or $\GQ$.
An example of graph $\GRR$ is given in Figure~\ref{fig:iso_2} (left and right, grey).

Consider a boundary rhombus pair of $\R^\spartial$, and
let $e^{i\bar{\alpha}^\ell},e^{i\bar{\beta}^\ell}$, resp. $e^{i\bar{\alpha}^r},e^{i\bar{\beta}^r},$ be the rhombus vectors of 
the diamond graph $\GRR$ assigned to the edge $(v^c,w^\ell)$, resp. $(v^c,w^r)$, of $\GD$ see Figure~\ref{fig:rhombus_pair} (right).
By definition we have
$\frac{\bar{\beta}^\ell-\bar{\alpha}^\ell}{2}$, $\frac{\bar{\beta}^r-\bar{\alpha}^r}{2}\in(\eps,\frac{\pi}{2}-\eps)$, and 
by construction the two rhombi have the same half-angle denoted $\bar{\theta}^\spartial=\frac{\bar{\beta}^\ell-\bar{\alpha}^\ell}{2}=
\frac{\bar{\beta}^r-\bar{\alpha}^r}{2}$.
\emph{We further impose that
$\bar{\beta}^\ell=\bar{\alpha}^r+2\pi$ or equivalently that $\frac{\bar{\alpha}^\ell-\bar{\beta}^r}{2}\in(2\eps,\pi-2\eps)$.}

Amongst the boundary rhombus pairs of $\R^\spartial$ the one containing 
the fixed vertex $\rs$, \emph{i.e.} the one for which  $v^c=\rs$, plays a special role; it will be referred to as the 
\emph{root-boundary rhombus pair} or simply \emph{root-pair}, see Figure~\ref{fig:iso_2} where the root pair is highlighted in 
light grey. We 
denote by  
$\R^{\spartial,\rs}$ the set $\R^\spartial$ without the root pair.
The isoradial embedding of the graph $\GDro$ is obtained from $\GD$ by removing the vertex $\rs$ and the edges 
$w^\ell \rs$, $w^r \rs$ of the root pair. Whenever needed, we add a superscript $\rs$ to the notation of Figure~\ref{fig:rhombus_pair} to 
specify vertices of the root pair. 

\subsubsection{Train-tracks}\label{sec:train_tracks}

A \emph{train-track} of a finite isoradial graph $\Gs$, also known 
as a \emph{de Bruijn line}~\cite{deBruijn1,deBruijn2} or a \emph{rapidity line}~\cite{Baxter:Zinv} is 
a maximal chain of edge-adjacent rhombi of the diamond graph $\GR$,
such that when entering a rhombus one exits along the opposite edge~\cite{KeSchlenk}; each train-track $\tau$ has a parallel direction 
$\pm 2e^{i\bar{\alpha}_\tau}$. Consider the simply connected domain $D(\Gs)$ obtained by taking the union of the faces of $\Gs$.
Then a train-track $\tau$ enters and exits $D(\Gs)$, and there are exactly two parallel edges of $\tau$ outside of $D(\Gs)$, see Figure~\ref{fig:iso},
(right). The boundary rhombus vectors $\{2e^{i\bar{\alpha}^\ell}, 2e^{i\bar{\beta}^r}\in\R^\spartial\} $,
$\{2e^{i\bar{\alpha}^r}, 2e^{i\bar{\beta}^\ell}\in\R^\spartial\}$
come in parallel pairs, and all the parallel directions of the train-tracks are encoded in 
$\{\pm 2e^{i\bar{\alpha}^\ell},\pm 2e^{i\bar{\beta}^\ell}\in\R^\spartial\}=\{\pm 2e^{i\bar{\alpha}^r},\pm 2e^{i\bar{\beta}^r}\in\R^\spartial\}$.

\subsubsection{Elliptic angles}\label{sec:elliptic_angles}

Consider an \emph{elliptic modulus} $k$, then $k'=(1-k^2)^{\frac{1}{2}}$ is the \emph{complementary elliptic modulus}. Suppose that
$k$ is such that $(k')^2\in(0,\infty)$. The \emph{complete elliptic integral of the first kind}, denoted 
$K=K(k)$ is
\[
K=\int_{0}^{\frac{\pi}{2}} \frac{1}{(1-k^2\sin\tau)}\ud\tau,
\]
and for later purposes we also need $K'=K(k')$.
As in~\cite{BdTR1,BdtR2}, we need the following linear transformation of rhombus angles and vectors of the diamond graph $\GR$ 
associated to edges:
\begin{equation*}
\theta=\bar{\theta}\frac{2K}{\pi},\quad \alpha=\bar{\alpha}\frac{2K}{\pi},\quad \beta=\bar{\beta}\frac{2K}{\pi}. 
\end{equation*}

\subsubsection{$Z$-invariant Ising model and corresponding dimer models}

Underlying $Z$-invariance is the \emph{star-triangle transformation}, also known as the $\mathsf{Y}$-$\Delta$ \emph{move}, on isoradial graphs.
Suppose that an isoradial graph $\Gs$ has a triangle, then this triangle can be transformed into a three-legged star while preserving 
isoradiality. This amounts to performing a cubic flip in the underlying diamond graph $\GR$; the embedding of the additional vertex of the 
triangle is given by the cubic flip, see Figure~\ref{fig:Zinv_0}. 

\begin{figure}[ht]
\begin{minipage}[b]{0.5\linewidth}
\begin{center}
\begin{overpic}[width=3cm]{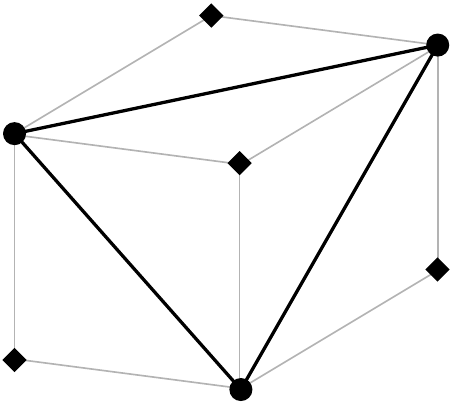}
\end{overpic}
\end{center}
\end{minipage}
\begin{minipage}[b]{0.5\linewidth}
\begin{center}
\begin{overpic}[width=3cm]{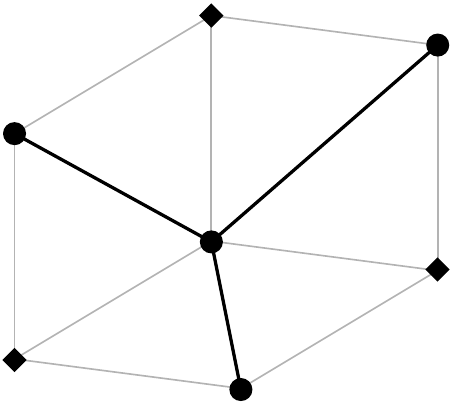}
\end{overpic}
\end{center}
\end{minipage}
\caption{A star-triangle transformation of an isoradial graph.}
\label{fig:Zinv_0}
\end{figure}

\emph{$Z$-invariance}~\cite{Baxter:8V,Baxter:Zinv}
phrased in the context of the Ising model requires that when decomposing the partition function according to the 
$2^3$ possible spin configurations at the three vertices bounding the star/triangle, it only changes by an overall constant when performing a 
$\mathsf{Y}$-$\Delta$ move, and this constant is independent of the choice of spin configuration. This yields a set of equations
for the coupling constants, known as the \emph{Ising Yang-Baxter} equations, see also~\cite{Perk:YB}. Extending the form of the solutions to the whole of the graph
naturally leads to introducing isoradial graphs: the solution is parametrized  by the rhombus half-angles assigned to edges and by 
the \emph{elliptic modulus} $k$, where $k$ is such that $(k')^2=1-k^2 \in(0,\infty)$, see~\cite{Baxter:exactly}. The 
\emph{Z-invariant coupling constants} are explicitly given by:
\begin{align}\label{equ:def_J_Zinv}
\forall\,e\in\Es,\quad \Js_e=\frac{1}{2}\ln\left(\frac{1+\sn(\theta_e|k)}{\cn(\theta_e|k)}\right),
\end{align}
where $\sn,\,\cn$ are two of the twelve Jacobi elliptic trigonometric functions. We refer the reader 
to~\cite{AS,Lawden} for more on elliptic and related functions.

Suppose that the $Z$-invariant coupling constants are chosen for the Ising model. Then, the dimer weight function $\muJ$ 
on the Fisher graph $\GF$ arising from Fisher's correspondence is given by, for every edge $\es$ of $\GF$,
\begin{equation}\label{equ:def_mu_Zinv}
\muJ_\es=
\begin{cases}
1&\text{ if $\es$ is an internal edge}\\
\frac{\cn(\theta_e|k)}{1+\sn(\theta_e|k)}&\text{ if $\es$ is an external edge arising from an edge $e^*$ of $\Gs^*$}.
\end{cases}
\end{equation}
 
The dimer weight function $\nuJ$ on the bipartite graph $\GQ$ arising from the XOR-Ising model is given by, for every edge $\es$ of $\GQ$,
\begin{equation}\label{eq:def_nu_Zinv}
\nuJ_{\es}=
\begin{cases}
1&\text{ if $\es$ is an external edge}\\
1&\text{ if $\es$ is a boundary quadrangle edge in the finite case}\\
\sn(\theta_e|k)&\text{ if $\es$ is a quadrangle edge/non-boundary quadrangle edge}\\
&\text{ in the infinite/finite case, parallel to an edge $e$ of $\Gs$}\\
\cn(\theta_e|k)&\text{ if $\es$ is parallel to a dual edge $e^*$ of an edge $e$ of $\Gs$}.
\end{cases}
\end{equation}

\subsubsection{The Z-invariant massive Laplacian}\label{sec:def_Lapmass_0}

In the paper~\cite{BdTR1}, we consider an infinite, isoradial graph $\Gs$ and introduce conductances and masses defining the 
\emph{$Z$-invariant massive Laplacian operator} or simply \emph{$Z$-massive Laplacian}, related to \emph{$Z$-invariant rooted spanning forests}.
Recall 
that every edge $(v,v')$ of $\Gs$ is assigned two rhombus vectors $2e^{i\bar{\alpha}},2e^{i\bar{\beta}}$ and a half-angle $\bar{\theta}$
of the diamond graph $\GR$. Denote by $v_1,\dots,v_d$ the neighbors of a vertex $v$ of degree $d$, and for every edge $(v,v_j)$ use the notation
$2e^{i\bar{\alpha}_j},2e^{i\bar{\alpha}_{j+1}},\bar{\theta_j}$ for the associated rhombus vectors and half-angle. Fix an elliptic modulus $k$ such that 
$(k')^2\in(0,\infty)$. Then, for every edge $(v,v')$ and every vertex $v$ of $\Gs$, the conductances $\rho^k$ and masses $m^k$ of~\cite{BdTR1}
are defined by:
\begin{equation*}
\rho^k_{v,v'}=\sc(\theta|k),\quad 
m^k_v=\sum_{j=1}^d [A(\theta_j|k)-\sc(\theta_j|k)],
\end{equation*}
where $\sc=\frac{\sn}{\cn}$, and 
\[A(u|k)=(k')^{-1}\Bigl(\mathrm{Dc}(u|k)+\frac{E-K}{K}u\Bigr),\]

$\mathrm{Dc}(u|k)=\int_{0}^u \dc^2(v|k)\ud v$, and
$E=E(k)=\int_{0}^{\frac{\pi}{2}} (1-k^2\sin^2(\tau))^{\frac{1}{2}}\ud\tau,$ is the \emph{complete elliptic integral of the second kind}.
Note that $m^k\geq 0$, by Proposition 6 of~\cite{BdTR1}.

The corresponding $Z$-massive Laplacian matrix $\Delta^{m(k)}$ has non-zero coefficients given by:
\begin{equation}\label{equ:defLap}
\forall\,v,v'\in \Vs,\quad \Delta^{m(k)}_{v,v'}=
\begin{cases}
-\sc(\theta|k)&\text{ if $(v,v')$ is an edge of $\Gs$}\\
\sum_{j=1}^d A(\theta_j|k)&\text{ if $v'=v$}.
\end{cases}
\end{equation}
The dual graph $\Gs^*$ is also isoradial, we denote by $\rho^{k,*}, m^{k,*}$ the associated conductances and masses, and by  
 $\bar{\theta}^*$ the half-angle of a dual edge.
We let $\Delta^{m(k),*}$ be the corresponding $Z$-massive Laplacian operator and refer to it as the \emph{dual $Z$-massive Laplacian}.

The inverse of the $Z$-massive Laplacian $\Delta^{m(k)}$ is the~\emph{$Z$-massive Green function}; it
is denoted $G^{m(k)}$. In~\cite{BdTR1}, we prove the following explicit \emph{local} expression for coefficients of $G^{m(k)}$: 
for every pair of vertices $x,y$ of $\Gs$, 
\begin{equation}\label{equ:massiv_Green}
G^{m(k)}=\frac{k'}{4\pi}\int_{\Gamma_{x,y}} \expo_{(x,y)}(u|k)\ud u, 
\end{equation}
where $\Gamma_{x,y}$ is a vertical contour on the torus $\TT(k):=\CC/(4K\ZZ+4iK'\ZZ)$, and
$\expo(\, \cdot\, |k):\Vs\times\Vs\times\CC\rightarrow \CC$ is the~\emph{massive exponential function} defined in~\cite{BdTR1}. To compute
$\expo_{(x,y)}(u|k)$, consider a path $x=x_1,\dots,x_n=y$ of the diamond graph $\GR$ from $x$ to $y$, let $2e^{i\bar{\alpha}_j}$ be
the rhombus vector corresponding to the edge $x_jx_{j+1}$, then
\begin{equation}\label{equ:expo_function}
\expo_{(x_j,x_{j+1})}(u|k)=i(k')^{\frac{1}{2}}\sc(u_{\alpha_j}|k),\quad \text{ and } \expo_{(x,y)}(u|k)=\prod_{j=1}^{n-1}\expo_{(x_j,x_{j+1})}(u|k),
\end{equation}
where $u_{\alpha_j}:=\frac{u-\alpha_j}{2}$.

Up to an explicit multiplicative constant $G^{m(k)}(x,y)$ is the expected number of visits to $y$
of the associated massive random walk on the infinite graph $\Gs$ started at $x$, see for example Appendix D4 of~\cite{BdTR1}.

From now on, we consider $k$ such that $(k')^2\in(0,\infty)$ as fixed and omit all reference to $k$ in the notation.

\subsubsection{Bipartite dimer models on isoradial graphs}\label{sec:Kast_complex}

When considering a dimer model on a bipartite isoradial graph $G=(W\cup B,E)$ with weight function $\nu$, instead of a Kasteleyn orientation, one can multiply edge-weights
by a complex phase~\cite{Kuperberg,Kenyon3}. This defines the \emph{complex, bipartite Kasteleyn matrix}, denoted $\tilde{K}$
in the context of this section, whose non-zero coefficients are given by:
\[
\forall\,\text{ edge $wb$ of $G$},\quad \tilde{K}_{w,b}=e^{i^{\frac{\bar{\alpha}_e+\bar{\beta}_e}{2}}}\nu_{wb},
\]
where $e^{i\bar{\alpha}_e},e^{i\bar{\beta}_e}$ are the rhombus vectors of $\GR$ associated to the edge 
$e=(w,b)$. The real and complex bipartite Kasteleyn matrices satisfy the alternating cycle condition around every inner face of 
$G$ and are gauge equivalent, see Section~\ref{sec2:app} of Appendix~\ref{app:gauge}. The results of~\cite{Kasteleyn1,Kasteleyn2,Kenyon0}
recalled in Section~\ref{sec:dimer_model} also hold with the complex, bipartite Kasteleyn matrix.
In the sequel the graph $G$ will be the double graph $\GD$ or the bipartite graph $\GQ$.


\section{$Z^u$-Dirac and $Z$-massive Laplacian operators}\label{sec:Z_Dirac_Z_Lap}

We let $\Gs$ be an infinite/finite isoradial graph; its dual graph is $\Gs^*$/$\bar{\Gs}^*$, and in the finite case $\Gs^*$ is
its restricted dual. We consider the isoradial embedding of the double graph $\GD=(\VD,\ED)$ given in 
Section~\ref{sec:iso_GDGQ}, see also Figure~\ref{fig:iso_2}.
Fix $k$ and recall the definition of the torus $\TT(k)=\CC/(4K\ZZ+4iK'\ZZ)$. 

In Section~\ref{sec:def_Dirac} we introduce a 
family of bipartite Kasteleyn matrices/operators $(\KD(u))_{u\in\Re(\TT(k))}$ on the double graph $\GD$, referred to as the 
\emph{$Z^u$-Dirac operators}. Fixing $u\in\Re(\TT(k))$, a function 
$F\in\CC^{B}$ is said to be \emph{$Z^u$-holomorphic} if, $\KD(u)F=0$. When $k=0$, the dependence in $u$ disappears and we recover the 
discrete Dirac operator $\bar{\partial}$ introduced in~\cite{Kenyon3}, see Remark~\ref{rem:Zu_Dirac}. As a consequence of Theorem~\ref{prop:KDtKD} of Section~\ref{sec:KD_Lap_mas}, we have that if
$F$ is a $Z^u$-holomorphic function, then $F_{|\Vs}$ is massive harmonic on $\Gs$
and $F_{|\Vs^*}$ is massive harmonic on $\Gs^*$,
for the $Z$-massive Laplacian $\Delta^m$ and its dual $\Delta^{m,*}$ of~\cite{BdTR1}; explaining the part \emph{Dirac} of the terminology. 
In the finite case, we moreover introduce the operator $\KD^\spartial(u)$ with  
specific boundary conditions arising from the forthcoming Theorem~\ref{thm:main} related to the Ising model,
that are different from the natural dimer ones.

In Section~\ref{sec:det_Kpartial_Lap} we restrict to the finite case. Theorem~\ref{thm:det} proves that, for every $u\in\Re(\TT(k))$,
the determinant of the $Z^u$-Dirac operator $\KD(u)$ is equal, up to an explicit multiplicative constant, to the determinant of the dual massive Laplacian $\Delta^{m,*}$; we show a similar result
for the operator $\KD^\spartial(u)$ and the massive Laplacian $\Delta^{m,\spartial}(u)$, where $\Delta^{m,\spartial}(u)$ has
specific boundary conditions and depends on $u$ \emph{along the boundary only}. 
Interpreting these determinants as partition functions, this proves that the weighted sum of pairs of dual directed spanning 
trees is equal, up to an explicit constant, to the weighted sum of rooted spanning forests. In the case $k=0$, pairs of directed spanning trees
become undirected and rooted spanning forests are un-rooted so that this theorem is a consequence of Temperley's bijection~\cite{Temperley} and 
of the matrix-tree theorem~\cite{Kirchhoff}. For $k\neq 0$, this result is non-trivial; the proof uses gauge equivalences on bipartite Kasteleyn matrices and on
\emph{weighted adjacency matrices of directed graphs (digraphs)}, see Appendix~\ref{app:gauge}.

For every $u\in\Re(\TT(k))$, the operator $\KD(u)$ is gauge equivalent to an operator $\KD^\gs(u)$ associated to a model of directed spanning trees. 
In Proposition~\ref{prop:Zinv} of
Section~\ref{sec:Z_invariance}, we prove that this model of directed spanning trees is $Z$-invariant, thus explaining the part ``$Z^u$'' 
in the terminology \emph{$Z^u$-Dirac operator}. 

Using Theorem~\ref{prop:KDtKD}, in Corollary~\ref{cor:KD_G} of Section~\ref{sec:KD_Lap_inv}, we express the inverse $Z^u$-Dirac operator
using the $Z$-massive Green function $G^m$ and its dual $G^{m,*}$ of~\cite{BdTR1} in the infinite case. This proves in Theorem~\ref{thm:Gibbs_KD}
an explicit \emph{local} expression for a Gibbs measure for the dimer model on $\GD$ with operator $\KD(u)$, generalizing 
to the full $Z$-invariant case the results of~\cite{Kenyon3} proved in the case $k=0$. In Corollary~\ref{cor:KD_G_finite}, we explicitly express the inverse of the 
$Z^u$-Dirac operator $\KD^\spartial(u)$ as a function of the finite versions of the $Z$-massive Green functions.
Theorem~\ref{thm:Gibbs_KD} is a planar, directed version of the transfer-impedance theorem of~\cite{BurtonPemantle},
where probabilities of pairs of dual directed spanning trees are computed using the Green functions of massive non-directed random walks. 
Apart from the locality property which is specific to $Z$-invariance,
a similar result is obtained by Chhita~\cite{Chhita} in the case of the square lattice with a specific choice of weights. 
Sun~\cite{Sun} expresses probabilities of directed spanning trees using the Green function of directed random walks, and Kenyon~\cite{Kenyon8}
proves that probabilities of rooted spanning forests are determinantal, without connecting them to directed spanning trees. It might 
be that the techniques of this paper, in particular gauge transformations on weighted adjacency matrices of digraphs, extend in some respect and 
allow to relate probabilities of pairs of dual directed spanning trees to massive non-directed Green functions.

\subsection{Family of $Z^u$-Dirac operators $(\KD(u))_{u\in\Re(\TT(k))}$}\label{sec:def_Dirac}

We will be using Section~\ref{sec:def_double_graph} on the double graph and Temperley's bijection.
Recall that vertices of the double graph $\GD$ are partitioned as 
$\VD=B\cup W$, where $B=\Vs\cup\Vs^*$ and $W\leftrightarrow \Es$. The diamond graph of $\GD$ is $\GRR$,
see Section~\ref{sec:iso_GDGQ} and Figure~\ref{fig:iso_2}. 
Let $u\in\Re(\TT(k))$ be fixed; we now define the \emph{$Z^u$-Dirac operator}, first in the infinite case, then in the finite case. 

\emph{Infinite case.}
Consider the weight function $c(u)$ on edges of $\GD$ defined by, $\forall\,w\in W,\,\forall\,x\in\Vs\cup\Vs^*$ such that 
$wx$ is an edge of $\GD$,
\begin{equation}\label{equ:defKD}
c(u)_{wx}=
\begin{cases}
\fs(u_{\alpha_e},u_{\beta_e})&\text{ if $x\in\Vs$}\\
\fs((u_{\alpha_e})^*,(u_{\beta_e})^*)&\text{ if $x\in\Vs^*$},
\end{cases}
\end{equation}
where $e^{i\bar{\alpha}_e},\,e^{i\bar{\beta}_e}$ are the rhombus vectors of $\GRR$ associated to the edge $e=(w,x)$;
$u_\alpha:=\frac{u-\alpha}{2}$, $u^*:=K-u$, and
\begin{equation*}
\fs(u_{\alpha},u_{\beta})= 
[\sc(u_{\alpha}-u_{\beta})\dn(u_{\alpha})\dn(u_{\beta})]^{\frac{1}{2}}.
\end{equation*}
\begin{rem}$\,$\label{rem:positivit_c}
The function $\dn$ is periodic in two directions and naturally defined on the torus $\CC/(2K\ZZ+4iK'\ZZ)$~\cite{AS}.
Since $\dn(u+2iK')=-\dn(u)$, and since the function $c(u)$ involves products of two $\dn$'s and half arguments, it is defined 
on the torus $\TT(k)$. This argument is similar to that used in~\cite{BdTR1} to define the domain of the massive exponential function.

We restrict to $u\in\Re(\TT(k))$ because the weight function $c(u)$ is then positive; indeed the function $\dn$ is, and 
$\alpha_e,\beta_e$ are such that
$\frac{\bar{\beta}_e-\bar{\alpha}_e}{2}\in(\eps,\frac{\pi}{2}-\eps)$. Also, on $\Re(\TT(k))$ the (pure imaginary) poles of $c(u)$ are 
avoided and the weights are thus finite.
Results in the sequel which use elliptic trigonometric identities 
actually hold for all $u\in\TT(k)$; it is when considering the corresponding dimer model that we use positivity of the weights.
\end{rem}

Let $\KD(u)$ be the complex, bipartite Kasteleyn matrix defined in Section~\ref{sec:Kast_complex} corresponding to the weight function $c(u)$, with rows indexed by 
white vertices of $\GD$. Non-zero coefficients of $\KD(u)$ are given by
\begin{equation}\label{equ:defKD_1}
\forall\,\text{edge $wx$ of $\GD$},\quad 
\KD(u)_{w,x}=
e^{i\frac{\bar{\alpha}_e+\bar{\beta}_e}{2}}c(u)_{wx}.
\end{equation}
Recall that this Kasteleyn matrix can also be interpreted as 
an operator mapping $\CC^B$ to $\CC^W.$ We refer to this matrix/operator as the \emph{$Z^u$-Dirac operator}. As an example, we 
explicitly compute $\KD(u)$ around a white vertex $w$ of $\GD$. 

\begin{exm} Figure~\ref{fig:GD} below sets the notation. 
A white vertex $w$ of $\GD$ is adjacent to the black vertices $v_1,f_1,v_2,f_2$ of $\GD$ defining a rhombus of the diamond graph 
$\GR$, such that $v_1,v_2$ belong to $\Gs$ and $f_1,f_2$ to $\Gs^*$.
Denote by $2e^{i\bar{\alpha}}$, $2e^{i\bar{\beta}}$ the rhombus vectors of $\GR$ associated to
the edge $(v_1,v_2)$, and 
by $\bar{\theta}=\frac{\bar{\beta}-\bar{\alpha}}{2}$ the rhombus half-angle of this edge.

\begin{figure}[ht]
\begin{minipage}[b]{0.5\linewidth}
\begin{center}
\begin{overpic}[width=4cm]{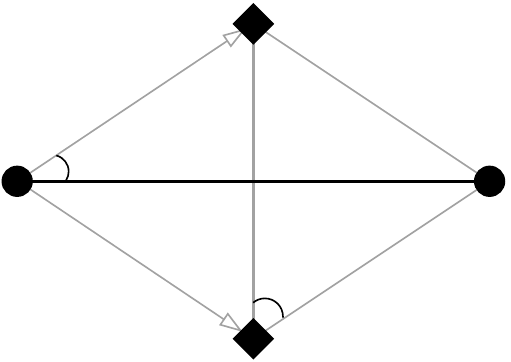}
  \put(15,12){\scriptsize $2e^{i\bar{\alpha}}$}
  \put(15,50){\scriptsize $2e^{i\bar{\beta}}$}
  \put(-8,34){\scriptsize $v_1$}
  \put(101,34){\scriptsize $v_2$}
  \put(48,-6){\scriptsize $f_1$}
  \put(48,72){\scriptsize $f_2$}
  \put(15,36){\scriptsize $\bar{\theta}$}
  \put(50.5,16){\scriptsize $\frac{\pi}{2}\!-\!\bar{\theta}$}
\end{overpic}
\end{center}
\end{minipage}
\begin{minipage}[b]{0.5\linewidth}
\begin{center}
\begin{overpic}[width=4cm]{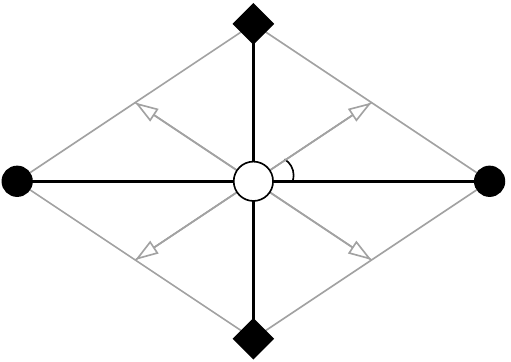}
  \put(67,24){\scriptsize $e^{i\bar{\alpha}}$}
  \put(11,25){\scriptsize $e^{i(\bar{\beta}+\pi)}$}
  \put(11,40){\scriptsize $e^{i(\bar{\alpha}+\pi)}$}
  \put(67,42){\scriptsize $e^{i\bar{\beta}}$}
  \put(-8,34){\scriptsize $v_1$}
  \put(101,34){\scriptsize $v_2$}
  \put(48,-6){\scriptsize $f_1$}
  \put(48,72){\scriptsize $f_2$}
  \put(51,25){\scriptsize $w$}
  \put(59.5,36){\scriptsize $\bar{\theta}$}
\end{overpic}
\end{center}
\end{minipage}
\caption{Rhombus $v_1,f_1,v_2,f_2$ of the diamond graph $\GR$ (left) and rhombi of the diamond 
graph $\GRR$ associated to the edges $wv_1,wf_1,wv_2,wf_2$, with corresponding rhombus vectors and half-angles.}
\label{fig:GD}
\end{figure}
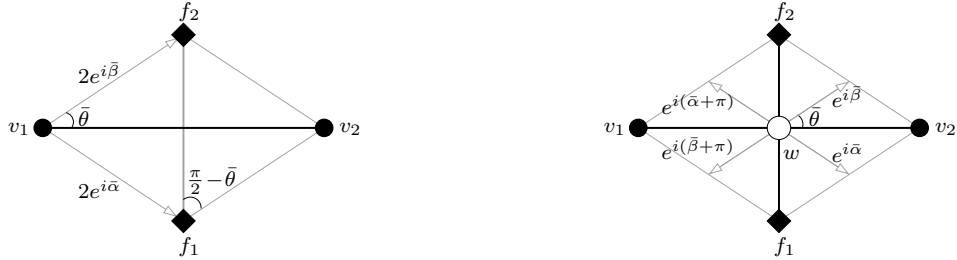
Then, we have:
\begin{align*}
\KD(u)_{w,v_1}&=-e^{i\frac{\overline{\alpha}+\overline{\beta}}{2}}[\sc(\theta)\dn(u_{\alpha+2K})\dn(u_{\beta+2K})]^{\frac{1}{2}}\\
\KD(u)_{w,v_2}&=e^{i\frac{\overline{\alpha}+\overline{\beta}}{2}}[\sc(\theta)\dn(u_{\alpha})\dn(u_{\beta})]^{\frac{1}{2}}\\ 
\KD(u)_{w,f_1}&=-ie^{i\frac{\overline{\alpha}+\overline{\beta}}{2}}[(k')^2\sc(\theta^*)\nd(u_{\beta-2K})\nd(u_\alpha)]^{\frac{1}{2}}=
-ie^{i\frac{\overline{\alpha}+\overline{\beta}}{2}}[k'\cs(\theta)\nd(u_{\beta-2K})\nd(u_\alpha)]^{\frac{1}{2}}\\
\KD(u)_{w,f_2}&=ie^{i\frac{\overline{\alpha}+\overline{\beta}}{2}}
[(k')^2\sc(\theta^*)\nd(u_{\beta})\nd(u_{\alpha+2K})]^{\frac{1}{2}}=
ie^{i\frac{\overline{\alpha}+\overline{\beta}}{2}}[k'\cs(\theta)\nd(u_{\beta})\nd(u_{\alpha+2K})]^{\frac{1}{2}},
\end{align*} 
using that, $\sc(\theta^*)=\sc(K-\theta)=(k')^{-1}\cs(\theta)$, $\dn(K-u)=k'\nd(u)$.
\end{exm}

\begin{defi}
A function $F\in\CC^B$ is said to be \emph{$Z^u$-holomorphic} if 
\begin{equation*}
\Ks(u)F=0.
\end{equation*}
With the notation of Figure~\ref{fig:GD}, this is equivalent to asking that the function $F$ satisfies, 
\begin{align*}
\forall\,w \in W,&\quad\Ks(u)_{w,v_1}\cdot F_{v_1}+\Ks(u)_{w,v_2}\cdot F_{v_2}+\Ks(u)_{w,f_1}\cdot F_{f_1}+\Ks(u)_{w,f_2}\cdot F_{f_2}=0\\
\Leftrightarrow\quad  \forall\,w \in W,& \quad \sc(\theta)^\frac{1}{2}\left( [\dn(u_{\alpha})\dn(u_{\beta})]^{\frac{1}{2}}\cdot F_{v_2}-
 [\dn(u_{\alpha+2K})\dn(u_{\beta+2K})]^{\frac{1}{2}}\cdot F_{v_1}\right)+\\
&\quad +i[k'\cs(\theta)]^\frac{1}{2}\left(\nd(u_{\beta})\nd(u_{\alpha+2K})]^{\frac{1}{2}}\cdot F_{f_2}
-[\nd(u_{\beta-2K})\nd(u_\alpha)]^{\frac{1}{2}}\cdot F_{f_1}\right)=0.
\end{align*}
\end{defi}

\begin{rem}\label{rem:Zu_Dirac}
When $k=0$, then $k'=1$, $\dn\equiv 1$, $\sc=\tan$, and the $Z^u$-Dirac operator is the discrete Dirac operator $\bar{\partial}$ 
introduced in~\cite{Kenyon3}. A
$Z^u$-holomorphic function is then a \emph{holomorphic} function as defined in~\cite{Duffin,Mercat:ising,Kenyon3,ChelkakSmirnov:toolbox}:
\[
\forall\,w\in W,\quad 
\tan(\theta)^\frac{1}{2}\left(F_{v_2}-F_{v_1}\right)
+i[\cot(\theta)]^\frac{1}{2}\left(F_{f_2}-F_{f_1}\right)=0.
\]
\end{rem}

\emph{Finite case.} We will be using Sections~\ref{sec:def_double_graph} and~\ref{sec:iso_GDGQ}: 
the graph $\GDro$ is obtained from $\GD$ by removing the vertex $\rs$ and its incident edges; 
the set of black vertices of $\GDro$ is $B^\rs=\Vs^\rs\cup\Vs^*$, the set of white vertices is $W^\rs=W\leftrightarrow \Es$, and $|B^\rs|=|W^\rs|$,
see Figure~\ref{fig:iso_2} (left) for an example. The set of boundary rhombus pairs of the diamond graph $\GR$ is $\R^\spartial$; $\R^{\spartial,\rs}$
is the set $\R^\spartial$ without the root pair containing the vertex $\rs$. For boundary rhombus pairs of $\R^{\spartial}$, we use the notation of Figure~\ref{fig:rhombus_pair},
which we recall here for convenience of the reader.

\begin{figure}[ht]
\begin{center}
\begin{overpic}[width=3.8cm]{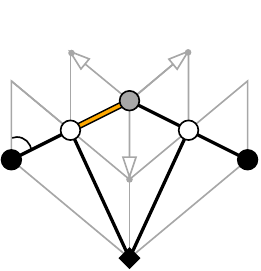}
\put(-7,40){\scriptsize $v^\ell$}
\put(44,68){\scriptsize $v^c$}
\put(95,39){\scriptsize $v^r$}
\put(18,42){\scriptsize $w^\ell$}
\put(68,42){\scriptsize $w^r$}
\put(44,-5){\scriptsize $f^c$}
\put(26,64){\scriptsize $e^{i\bar{\alpha}^\ell}$}
\put(34,49){\scriptsize $e^{i\bar{\beta}^\ell}$}
\put(48,49){\scriptsize $e^{i\bar{\alpha}^r}$}
\put(56,64){\scriptsize $e^{i\bar{\beta}^r}$}
\put(8,49.5){\scriptsize $\bar{\theta}^\spartial$}
\end{overpic}
\end{center}
\caption{Notation for vertices of $\GD$ in a boundary rhombus pair of $\R^\spartial$, and rhombus vectors
$e^{i\bar{\alpha}^\ell}$, $e^{i\bar{\beta}^\ell}$, resp. $e^{i\bar{\alpha}^r}$, $e^{i\bar{\beta}^r}$, of $\GRR$ assigned to the 
edge $(v^c,w^\ell)$, resp. $(v^c,w^r)$. When this pair is not the root one,
the marked black edge is where the operator $\KD^\spartial(u)$ differs from $\KD(u)$.}
\label{fig:def_Kpartial}
\end{figure}

We introduce two versions $\KD(u)$ and $\KD^\spartial(u)$ of the $Z^u$-Dirac operator, corresponding 
to different boundary conditions; both operators map $\CC^{B^\rs}$ to $\CC^{W^\rs}$.

The operator $\KD(u)$ is the finite version of the operator $\KD(u)$ defined in the infinite case, 
thus justifying the notation; it is the complex, bipartite Kasteleyn matrix corresponding to the weight function $c(u)$ of~\eqref{equ:defKD}
restricted to $\GDro$, and we do not repeat the definition here. 

The operator $\KD^\spartial(u)$ has specific boundary conditions arising from the forthcoming Theorem~\ref{thm:main} related 
to the Ising model. 
It is the complex, bipartite Kasteleyn matrix 
associated to the weight function $c^\spartial(u)$ differing from $c(u)$ along 
edges $(w^\ell, v^c)$ of the boundary rhombus pairs of $\R^{\spartial,\rs}$. For every edge $wx$ of $\GD$, we have
\begin{equation}\label{equ:def_c_partial}
c^\spartial(u)_{wx}=
\begin{cases}
c(u)_{wx}&\text{ if $(w,x)\notin\{(w^\ell,v^c)\in\R^{\spartial,\rs}\}$}\\
c(u)_{w^\ell v^c}\frac{\cd(u_{\beta^r})}{\cd(u_{\alpha^\ell})}& \text{ if $(w,x)\in\{(w^\ell,v^c)\in\R^{\spartial,\rs}\}$}.
\end{cases}
\end{equation}
In order for $c^\spartial(u)$ to be finite, we restrict the domain of $u$ to:
\begin{align}
\Re(\TT(k))'&=\Re(\TT(k))\setminus\{(\alpha^\ell+2K)[4K]:e^{i\bar{\alpha}^\ell}\in \R^{\spartial}\}\label{equ:Re_Tk'}\\
&=\Re(\TT(k))\setminus\{(\beta^r+2K)[4K]:e^{i\bar{\alpha}^r}\in \R^{\spartial}\},\nonumber
\end{align}
where the second equality is a consequence of properties of train-tracks, see Section~\ref{sec:train_tracks}.
As a consequence, the weight function $c^{\spartial}(u)$ is everywhere non-zero. 

The coefficients of $\KD^\spartial(u)$ differing from those of $\KD(u)$ are those corresponding to
edges $(w^\ell, v^c)$ of $\R^{\spartial,\rs}$, and we have:
\begin{align}\label{equ:defKpartial}
\KD^\spartial(u)_{w^\ell,v^c}&=\KD(u)_{w^\ell,v^c}\frac{\cd(u_{\beta^r})}{\cd(u_{\alpha^\ell})}=
-e^{i\frac{\overline{\alpha}_\ell+\overline{\beta}_\ell}{2}}[\sc(\theta^\spartial)\dn(u_{\alpha^\ell+2K})\dn(u_{\beta^\ell+2K})]^{\frac{1}{2}}
\frac{\cd(u_{\beta^r})}{\cd(u_{\alpha^\ell})}.
\end{align}

\begin{rem}
The weight $c^\spartial(u)_{w^\ell v^c}$ might be negative. Using the physics terminology, see also Remark~\ref{rem:positivity_dimer}, 
this means that the
corresponding dimer model has vortices on boundary faces $v^c,w^\ell,f^c,w^r$ where this is the case.
\end{rem}

\subsection{$Z$-massive Laplacian operators}\label{sec:def_Lap}

In the \emph{infinite case}, we consider the $Z$-massive Laplacian associated to the isoradial graph $\Gs$
introduced in~\cite{BdTR1}, whose definition is recalled in Section~\ref{sec:def_Lapmass_0}. We also consider the 
dual $Z$-massive Laplacian $\Delta^{m,*}$ of the dual isoradial graph $\Gs^*$.

The purpose of this section is to define the \emph{finite} versions of the $Z$-massive Laplacian and dual 
$Z$-massive Laplacian that are used in this paper;
we introduce two operators. The first is the finite version of the dual $Z$-massive Laplacian, denoted 
$\Delta^{m,*}$ as in the infinite case. It acts on $\CC^{\Vs^*}$, where recall that $\Vs^*$ is the vertex set of the restricted dual $\Gs^*$,
whose boundary vertices are defined to be those adjacent to the vertex $\outer$ in $\bar{\Gs}^*$.
For every vertex $f$ of $\Gs^*$ of degree $d$ in $\bar{\Gs}^*$, denote its neighbors by $f_1,\dots,f_d$. Let $\ubar{d}$ be the degree of $f$ in 
$\Gs^*$, then if $f$ is a boundary vertex of $\Gs^*$ we have $\ubar{d}\neq d$, and we label the vertices so that the first $\ubar{d}$ ones
are common to $\Gs^*$ and $\bar{\Gs}^*$. For every $j\in\{1,\dots,\ubar{d}\}$, $(f,f_j)$ is an edge of $\Gs^*$
and we let $\bar{\theta}_j^*$ denote its half-angle in $\GR$; for every $j\in\{\ubar{d}+1,\dots,d\}$, $(f,f_j)=(f,\outer)$ corresponds
to an edge $fw_j$ of the double graph $\GD$, where $w_j$ is the unique vertex of $\GD$ which belongs to the edge $f\outer$, we let 
$\bar{\theta}_j^*$ denote the half-angle of the edge $(f,w_j)$ in $\GRR$. Then, non-zero coefficients of the finite version of the 
dual $Z$-massive Laplacian $\Delta^{m,*}$ are given by,
\begin{equation*}
\Delta^{m,*}_{f,f'}=
\begin{cases}
-\sc(\theta^*)& \text{ if $(f,f')$ is an edge of $\Gs^*$}\\
\sum_{j=1}^d A(\theta_j^*)&\text{ if $f'=f$ is a vertex of $\Gs^*$}.
\end{cases}
\end{equation*}
The corresponding conductances $\rho^*$ and masses $m^*$ are as defined in Section~\ref{sec:spanning_forests}.

The second operator acts on $\CC^{\Vs^\rs}$; it is denoted $\Delta^{m,\spartial}(u)$ and we restrict to $u\in\Re(\TT(k))'$.
It is \emph{not} the natural finite version of the operator $\Delta^m$ but has boundary conditions
inherited from those of the $Z^u$-Dirac operator $\KD^\spartial(u)$, see the proof of Theorem~\ref{prop:KDtKD}
below. We use the notation of the infinite $Z$-massive Laplacian operator $\Delta^m$, see Section~\ref{sec:def_Lapmass_0}, and 
also the notation of Figure~\ref{fig:def_Kpartial} for the boundary rhombus pairs of $\R^\spartial$. 

In the following, the two boundary vertices $v^r,v^\ell$ of $\R^\spartial$
enter the same framework, and we denote such a vertex by $v^\spartial$. For every $v^\spartial$ of $\R^\spartial$ of degree $d$ in $\Gs$,
let $v_1,\dots,v_d$ denote its neighbors in cclw order with
$v_1$ on the boundary of $\Gs$ on the left of $v^\spartial$, see Figure~\ref{fig:Kvv} (right). Note that if $v^\spartial=v^{\spartial,\rs}$ belongs to the 
root pair, then $v^c=\rs$ is considered as a neighbor of $v^{\spartial,\rs}$. 
%
The coefficients of $\Delta^{m,\spartial}(u)$ differing from those of $\Delta^m$ are those corresponding to edges $(v^c,v^\ell)$
of $\R^{\spartial,\rs}$, and to vertices $v^r,v^\ell,v^c$ of $\R^\spartial$ such that $v^c\neq \rs$:
\begin{equation}\label{equ:def_Lap_bry}
\Delta^{m,\spartial}(u)_{v,v'}=
\begin{cases}
-\sc(\theta^\spartial)\frac{\cd(u_{\beta^r})}{\cd(u_{\alpha^\ell})}&
     \text{ if $(v,v')\in\{(v^c,v^\ell)\in \R^{\spartial,\rs}\}$}\\
k'\sc(\theta^\spartial)
\frac{\nd(u_{\beta^\ell})\nd(u_{\beta^r})}{\cn(u_{\alpha^\ell})}\bigl(\cn(u_{\beta^r})+\cn(u_{\alpha^\ell})\bigr)
&\text{ if $v'=v=v^c\in\{v^c\in \R^{\spartial,\rs}\}$}\\
k'\sum_{j=1}^d \sc(\theta_j)\nd(u_{\alpha_j})\nd(u_{\alpha_{j+1}})&\text{ if $v'=v=v^\spartial\in\{v^r,v^\ell\in\R^\spartial\}$}.
\end{cases}
\end{equation}
The corresponding conductances and masses are denoted $\rho^\spartial(u)$ and $m^\spartial(u)$.
Note that the diagonal term $\Delta^{m,\spartial}(u)_{v^\spartial,v^\spartial}$ has the same form as that of the natural 
finite version of the $Z$-massive Laplacian.

In order to state Theorem~\ref{prop:KDtKD} relating the $Z^u$-Dirac and the $Z$-massive Laplacian operators, in the finite case we 
need to introduce the \emph{matrix $Q(u)$}. It has rows indexed by black vertices of $\Gs\ro$ and columns by 
those of the restricted dual $\Gs^*$. 
The only non-zero coefficients of $Q(u)$ are for rows corresponding to vertices $v^c$ of boundary rhombus pairs of $\R^{\spartial,\rs}$.
For such a vertex $v^c$, 
there is one non-zero coefficient, corresponding to the column $f^c$, given by
\[
Q(u)_{v^c,f^c}=-i\frac{\nd(u_{\beta^\ell})}{\cd(u_{\alpha^\ell})}\bigl(\cd(u_{\beta^r})-\cd(u_{\alpha^\ell})\bigr).
\]

\subsection{Relating the $Z^u$-Dirac and the $Z$-massive Laplacian operators}\label{sec:KD_Lap_mas}

Theorem~\ref{prop:KDtKD} below relates the $Z^u$-Dirac operator, the $Z$-massive Laplacian and the dual $Z$-massive Laplacian in the infinite
and finite cases. In the infinite case, this extends to the full $Z$-invariant case the results of Section 6 of~\cite{Kenyon3} which correspond to $k=0$.

\begin{thm}\label{prop:KDtKD}$\,$
\begin{itemize}
 \item[$\bullet$] \emph{Infinite case.} Let $u\in\Re(\TT(k))$, then the $Z^u$-Dirac operator $\KD(u)$, the $Z$-massive Laplacian 
 $\Delta^m$ and the dual 
 $\Delta^{m,*}$ satisfy the following identity:
\begin{equation*} 
\overline{\KD(u)}^{\,t\,}\KD(u)=k'
\begin{pmatrix}
\Delta^m&0\\
0&\Delta^{m,*}
\end{pmatrix}.
\end{equation*}\item[$\bullet$] \emph{Finite case.} Let $u\in\Re(\TT(k))'$, then the $Z^u$-Dirac operators $\KD(u)$, $\KD^\spartial(u)$, 
the $Z$-massive Laplacian
$\Delta^{m,\spartial}(u)$ and the dual $\Delta^{m,*}$ satisfy the following identity:
\begin{equation*} 
\overline{\KD^\spartial(u)}^{\,t\,}\KD(u)=k'
\begin{pmatrix}
\Delta^{m,\spartial}(u)& Q(u)\\
0&\Delta^{m,*}
\end{pmatrix}.
\end{equation*}
\end{itemize}
\end{thm}

\begin{rem}\label{rem:Lap}$\,$
\begin{itemize}
\item[$\bullet$] As a consequence of Theorem~\ref{prop:KDtKD} in the infinite case, we have that if $F\in\CC^B$ is a $Z^u$-holomorphic function, 
then $F_{|\Vs}$ is $Z$-massive harmonic on $\Gs$ and $F_{|\Vs^*}$ is $Z$-massive harmonic on $\Gs^*$.
\item[$\bullet$]
From the proof of Theorem~\ref{prop:KDtKD} it follows that, in the finite case, we also have the following matrix relation:
\begin{equation*} 
\overline{\KD(u)}^{\,t\,}\KD(u)=k'
\begin{pmatrix}
\Delta^{m}(u)& 0 \\
0&\Delta^{m,*}
\end{pmatrix},
\end{equation*} 
where $\Delta^{m}(u)$ is the natural finite version of the $Z$-massive Laplacian, \emph{i.e.}, where diagonal coefficients 
of boundary vertices of $\Delta^{m}(u)$ are defined as in the last line of~\eqref{equ:def_Lap_bry}.
\end{itemize}
\end{rem}

\begin{proof}
As in the critical case~\cite{Kenyon3}, the proof consists in showing that matrix coefficients on both sides are equal.
We separate the infinite case together with the part of the finite case which enters 
the infinite framework, from the part of the finite case which is specific. To simplify notation,
we omit the argument $u$ from the operators.

When two black vertices $x,y$ of $\GD$, resp. $\GDro$, are at distance more than two, the corresponding coefficient
$[\overline{\KD}^{\,t\,}\KD]_{x,y}$, resp. $[\overline{\KD^\spartial}^{\,t\,}\KD]_{x,y}$, is trivially equal to 0, we thus suppose that
$x,y$ are at distance 2 or 0. 

\emph{Infinite case and part of the finite case}.
Let $(x,y)=(v,f)$ or $(f,v)$, with $v$ a vertex of $\Gs$ and $f$ a vertex of $\Gs^*$, at distance two. Denote by $w,w'$ the two white vertices of the quadrangle
of $\GD$/$\GDro$ containing $v$ and $f$; let $e^{i\bar{\alpha}},e^{i\bar{\beta}},$ resp. 
$e^{i\bar{\alpha}'},e^{i\bar{\beta}'},$ be the rhombus vectors of $\GRR$ associated to the edge $(w,v)$, resp. $(w',v)$, and 
let $\bar{\theta}$, resp. $\bar{\theta}'$, be the corresponding half-angle, see Figure~\ref{fig:GD1}.
In the finite case, we moreover suppose that 
$(v,f)\neq (v^c,f^c)$ for all boundary rhombus pairs of $\R^{\spartial,\rs}$
, then we have $\KD^\spartial=\KD$ for all coefficients involved. Let us prove that 
$[\overline{\KD}^{\,t}\KD]_{v,f}=\overline{\KD}_{w,v}\KD_{w,f}+\overline{\KD}_{w',v}\KD_{w',f}=0$.

\begin{figure}[H]
\begin{center}
\centering
\begin{overpic}[width=5.5cm]{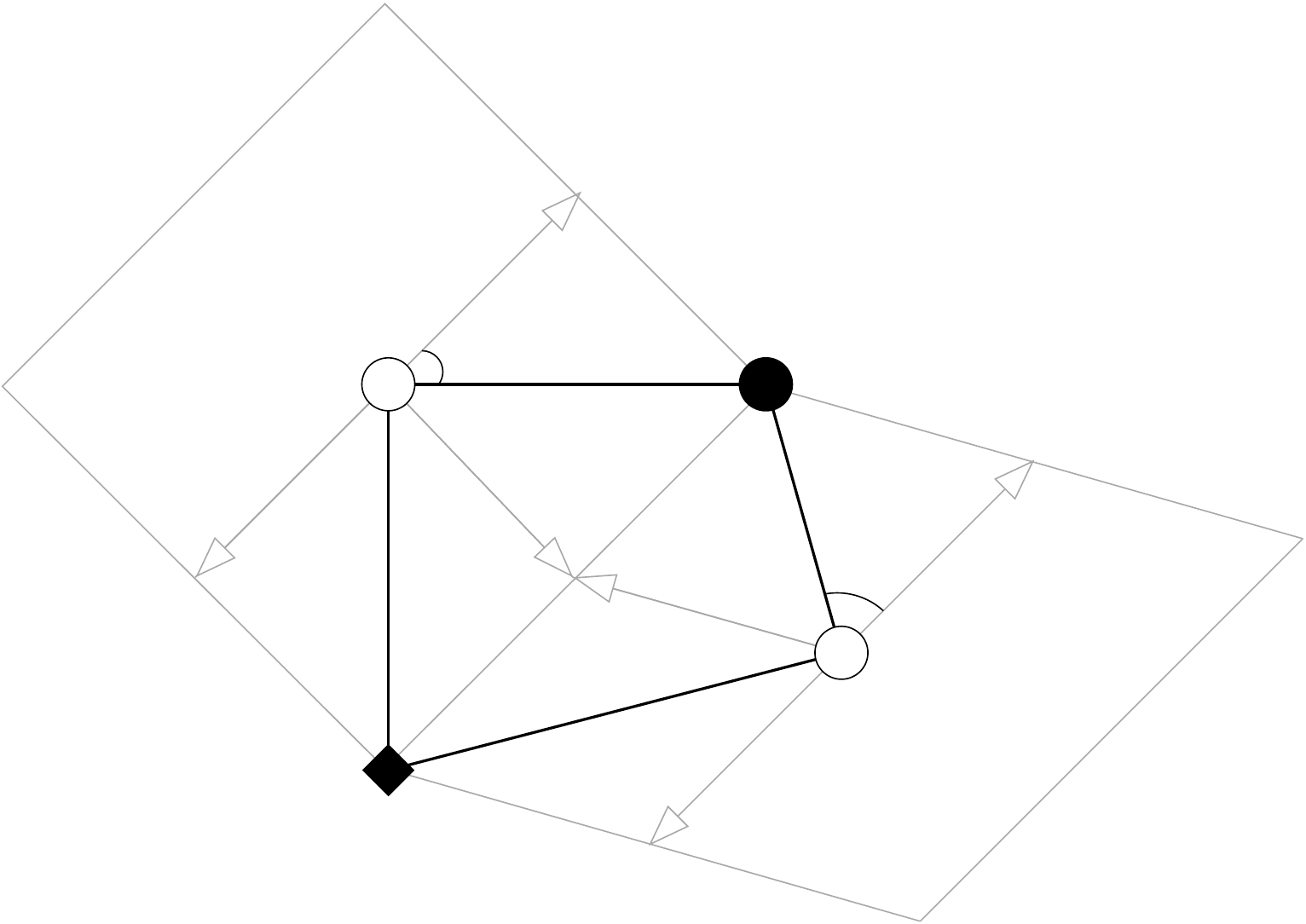}
\put(61,42){\scriptsize $v$}
\put(28,5){\scriptsize $f$}
\put(23,40){\scriptsize $w$}
\put(63,15.5){\scriptsize $w'$}
\put(39,32){\scriptsize $e^{i\bar{\alpha}}$}
\put(29,46){\scriptsize $e^{i\bar{\beta}}$}
\put(72,25){\scriptsize $e^{i\bar{\alpha}'}$}
\put(52.5,25){\scriptsize $e^{i\bar{\beta}'}$}
\put(36,42){\scriptsize $\bar{\theta}$}
\put(65,26){\scriptsize $\bar{\theta}'$}
\end{overpic}
\caption{Notation for computing $[\overline{\KD}^{\,t\,}\KD]_{v,f}$.}
\label{fig:GD1}
\end{center}
\end{figure}
By definition of $\KD$, we have
\begin{align*}
\overline{\KD}_{w,v}\KD_{w,f}&=e^{-i\frac{\overline{\alpha}+\overline{\beta}}{2}}[\sc(\theta)\dn(u_{\alpha})\dn(u_\beta)]^{\frac{1}{2}}
e^{i\frac{\overline{\beta}-\pi+\overline{\alpha}}{2}}
[k'\cs(\theta)\nd(u_{\beta-2K})\nd(u_\alpha)]^{\frac{1}{2}}
=-i\dn(u_\beta)\\
\overline{\KD}_{w',v}\KD_{w',f}&=
e^{-i\frac{\bar{\alpha}'+\bar{\beta}'}{2}}[\sc(\theta')\dn(u_{\alpha'})\dn(u_\beta')]^{\frac{1}{2}}
e^{i\frac{\bar{\beta}'+\bar{\alpha}'+\pi}{2}}
[k'\cs(\theta')\nd(u_{\beta'})\nd(u_{\alpha'+2K})]^{\frac{1}{2}}
=i\dn(u_{\alpha'}).
\end{align*}
Since $\bar{\beta}=\bar{\alpha}'[2\pi]$, we deduce that $[\overline{\KD}^{\,t\,}\KD]_{v,f}=0$. We also have 
$[\overline{\KD}^{\,t\,}\KD]_{f,v}=0$, because $[\overline{\KD}^{\,t\,}\KD]_{f,v}=\overline{[\overline{\KD}^{\,t\,}\KD]_{v,f}}$. 

Next, let $(x,y)=(v_1,v_2)$ be an edge of the graph $\Gs$/$\Gs^\rs$ corresponding to a path $v_1,w,v_2$ of $\GD$/$\GDro$. In the finite
case, we moreover suppose that $(v_1,v_2)\neq (v^c,v^\ell)$ for all boundary rhombus pairs of $\R^{\spartial,\rs}$, 
then we have $\KD^\spartial=\KD$ for all coefficients involved. 
Using the notation of Figure~\ref{fig:GD}, we have
\begin{align*}
[\overline{\KD}^{\,t\,}\KD]_{v_1,v_2}&=\overline{\KD}_{w,v_1}\KD_{w,v_2}\\
&=e^{-i\frac{\bar{\alpha}+\bar{\beta}+2\pi}{2}}[\sc(\theta)\dn(u_{\alpha+2K})\dn(u_{\beta+2K})]^{\frac{1}{2}}
e^{i\frac{\bar{\alpha}+\bar{\beta}}{2}}[\sc(\theta)\dn(u_{\alpha})\dn(u_\beta)]^{\frac{1}{2}}\\
&=-k'\sc(\theta)=k'\Delta^m_{v_1,v_2}.
\end{align*}
We now handle the case where $x,y$ are at distance $0$, \emph{i.e.}, when $x=y$ and we suppose that $x$ is a vertex $v$ of $\Gs$ of degree $d$. 
In the finite case, we moreover suppose that $v\neq v^c$ for all boundary rhombus pairs of $\R^{\spartial,\rs}$, then we have $\KD^\spartial=\KD$ for all coefficients involved.
Using the notation of Figure~\ref{fig:Kvv} we have:
\begin{align*}
[\overline{\KD}^{\,t\,}\KD]_{v,v}&=\sum_{j=1}^{d} \overline{\KD}_{w_j,v}\KD_{w_j,v}=
\sum_{j=1}^{d}|\KD_{w_j,v}|^2\\
&=\sum_{j=1}^{d} \sc(\theta_j)\dn(u_{\alpha_j+2K})\dn(u_{\alpha_{j+1}+2K})=
(k')^2 \sum_{j=1}^{d} \sc(\theta_j)\nd(u_{\alpha_j})\nd(u_{\alpha_{j+1}}).
\end{align*}

\begin{figure}[H]
\begin{minipage}[b]{0.5\linewidth}
\begin{center}
\begin{overpic}[width=4cm]{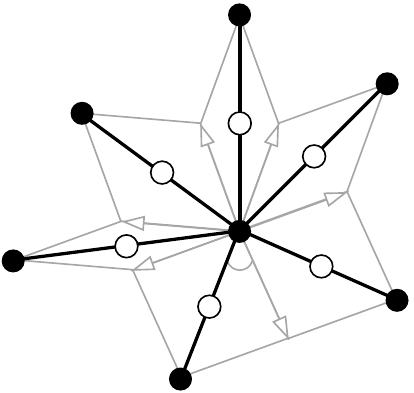}
\put(60.5,32.5){\scriptsize $v$}
\put(97,76){\scriptsize $v_1$}
\put(78,51){\scriptsize $w_1$}
\put(55.5,95.5){\scriptsize $v_2$}
\put(46,66){\scriptsize $w_2$}
\put(39,-4){\scriptsize $v_j$}
\put(39,18){\scriptsize $w_j$}
\put(37,27){\scriptsize $2e^{i\bar{\alpha}_j}$}
\put(66,21){\scriptsize $2e^{i\bar{\alpha}_{j\!+\!1}}$}
\put(55,22){\scriptsize $\bar{\theta}_j$}
\end{overpic}
\end{center}
\end{minipage}
\begin{minipage}[b]{0.5\linewidth}
\begin{center}
\begin{overpic}[width=4.3cm]{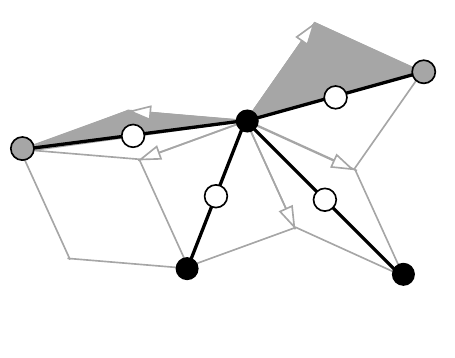}
\put(50,55){\scriptsize $v^\spartial$}
\put(-5,44){\scriptsize $v_1$}
\put(21,53){\scriptsize $w_1$}
\put(50.5,31){\scriptsize $w_j$}
\put(97,62){\scriptsize $v_d$}
\put(42,13){\scriptsize $v_j$}
\put(72,51){\scriptsize $w_d$}
\put(50,63){\scriptsize $2e^{i\bar{\alpha}_d}$}
\put(34,55){\scriptsize $2e^{i\bar{\alpha}_{1}}$}
\end{overpic}
\end{center}
\end{minipage}
\caption{Notation for computing $[\overline{\KD}^{\,t\,}\KD]_{v,v}$, when $v$ is not a boundary vertex 
(left), and when $v$ is a boundary vertex $v^\spartial\in\{v^\ell,v^r\in\R^\spartial\}$
(the additional half-rhombi along the boundary are pictures in grey).}
\label{fig:Kvv}
\end{figure}
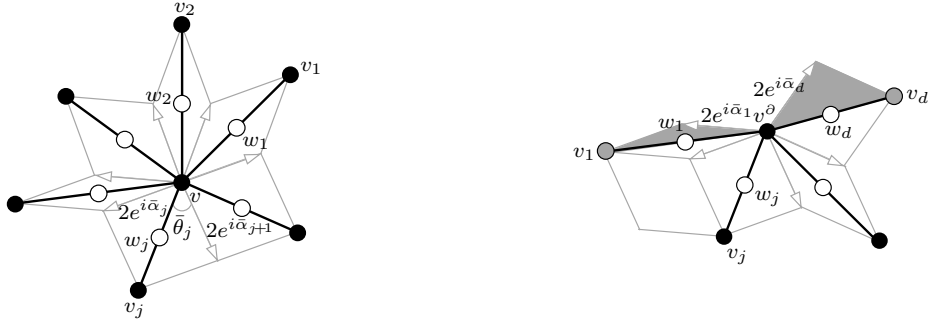

In the finite case, when $v\in\{v^\ell,v^r\in\R^\spartial\}$, both vertices
enter the same framework and we denote $v$ by $v^\spartial$. We stop the computation here, and
returning to~\eqref{equ:def_Lap_bry}, we have that $[\overline{\KD}^{\,t\,}\KD]_{v^\spartial,v^\spartial}=
k'\Delta^{m,\spartial}_{v^\spartial,v^\spartial}$. In the infinite case or in the finite case when $v$ is not a boundary vertex
we write, $\nd(u_{\alpha})=\nd((u-2iK')_{\alpha}+iK')=i\sc((u-2iK')_{\alpha})$ and obtain
\begin{align*}
[\overline{\KD}^{\,t\,}\KD]_{v,v}&=-(k')^2\sum_{j=1}^{d}\sc(\theta_j)\sc((u-2iK')_{\alpha_j})\sc((u-2iK')_{\alpha_{j+1}}).
\end{align*}
Since in this case we have $\bar{\alpha}_{d+1}=\bar{\alpha}_1+2\pi$, we can use 
Proposition 11 of~\cite{BdTR1} and obtain that, for every $u\in\TT(k)$,
\begin{align*}
-(k')^2\sum_{j=1}^{d} \sc(\theta_j)\sc(u_{\alpha_j})\sc(u_{\alpha_{j+1}})=k'[\sum_{j=1}^{d} A(\theta_j)].
\end{align*}
Evaluating this identity at $u-2iK'$, and recalling the definition of $\Delta^m$, see~\eqref{equ:defLap},
we deduce that 
$
[\overline{\KD}^{\,t\,}\KD]_{v,v}=k'[\sum_{j=1}^{d} A(\theta_j)]
=k'\Delta^m_{v,v}.
$

We are left with handling the cases where $(x,y)=(f_1,f_2)$ is an edge of $\Gs^*$, and where $x=y=f$ is a vertex of $\Gs^*$.
First note that coefficients of $\KD,\KD^\spartial$ involved in these computations arise from dual edges, and that 
for these edges we have $\KD^\spartial=\KD$. Next, 
observe that coefficients of the operator $\KD$ on edges arising from the primal and from the dual differ in that 
the weight function $c$ is evaluated at $u$ or $u^*=K-u$, see~\eqref{equ:defKD}. Given that the value of 
$[\overline{\KD}^{\,t\,}\KD]_{v_1,v_2}$ is independent of $u$, we immediately deduce the corresponding result for
$[\overline{\KD}^{\,t\,}\KD]_{f_1,f_2}$. In the infinite case and in the finite case if 
$v$ is an inner vertex of $\Gs$, $[\overline{\KD}^{\,t\,}\KD]_{v,v}$ is also 
independent of $u$ thus giving the corresponding result for $[\overline{\KD}^{\,t\,}\KD]_{f,f}$. 
The computation of $[\overline{\KD}^{\,t\,}\KD]_{v,v}$ when 
$v$ is a boundary vertex $v^\spartial\in\{v^\ell,v^r\}$, yields a result which depends on $u$. Recall that what prevents us from 
proceeding with the computation as in the non-boundary case is the fact that $\bar{\alpha}_{d+1}\neq \bar{\alpha}_1+2\pi$.
Since for boundary vertices of the restricted dual the condition $\bar{\alpha}_{d+1}= \bar{\alpha}_1+2\pi$ is satisfied,
we proceed with the computation as for inner vertices and obtain the corresponding result for $[\overline{\KD}^{\,t\,}\KD]_{f,f}$.
This ends the proof of Theorem~\ref{prop:KDtKD} in the infinite case and part of the finite case.

\emph{Remaining part of the finite case.}
When $(x,y)=(v,f)$ or $(f,v)$, with $v,f$ at distance 2, we are left with the case where $(v,f)=(v^c,f^c)$ for some boundary rhombus pair 
of $\R^{\spartial,\rs}$. 
Since the operator $[{\overline{\KD^\spartial}}^{\,t\,}\KD]$ is not skew-symmetric, 
we have to consider the coefficients $(v^c,f)$ and $(f,v^c)$ separately. 
For the coefficient $(f^c,v^c)$, using the notation of Figure~\ref{fig:def_Kpartial}, we have
$
[{\overline{\KD^\spartial}}^{\,t\,}\KD]_{f^c,v^c}=\overline{\KD^\spartial}_{w^\ell,f^c}\KD_{w^\ell,v^c}+
\overline{\KD^\spartial}_{w^r,f^c}\KD_{w^r,v^c}.
$
This is equal to 0 as in the finite non-boundary case, 
because $\KD^\spartial$ is equal to $\KD$ on the edges $(w^\ell,f^c)$ and $(w^r,f^c)$. 
For the coefficient $(v^c,f^c)$, we have
\begin{align*}
[{\overline{\KD^\spartial}}^{\,t\,}\KD]_{v^c,f^c}&=
\overline{\KD^\spartial}_{w^\ell,v^c}\KD_{w^\ell,f^c}+\overline{\KD^\spartial}_{w^r,v^c}\KD_{w^r,f^c}\\
&=\frac{\cd(u_{\beta^r})}{\cd(u_{\alpha^\ell})}\overline{\KD}_{w^\ell,v^c}\KD_{w^\ell,f^c}+\overline{\KD}_{w^r,v^c}\KD_{w^r,f^c},
\text{ by definition of $\KD^\spartial$, see~\eqref{equ:defKpartial}}\\
&=-i\dn(u_{\beta^\ell+2K})\Bigl(\frac{\cd(u_{\beta^r})}{\cd(u_{\alpha^\ell})}-1\Bigr),\text{ by the finite non-boundary computation}\\
&=-ik'\frac{\nd(u_{\beta^\ell})}{\cd(u_{\alpha^\ell})}\bigl(\cd(u_{\beta^r})-\cd(u_{\alpha^\ell})\bigr)=k'Q_{v^c,f^c}.
\end{align*}
Now consider an edge $(v_1,v_2)$ of $\Gs^\rs$. Then, we are left with the case where $(v_1,v_2)=(v^c,v^\ell)$
for some boundary rhombus pair of $\R^{\spartial,\rs}$. We have,
\begin{align*}
[{\overline{\KD^\spartial}}^{\,t\,}\KD]_{v^c,v^\ell}&=\overline{\KD^\spartial}_{w^\ell,v^c}\KD_{w^\ell,v^\ell}
=\frac{\cd(u_{\beta^r})}{\cd(u_{\alpha^\ell})}\overline{\KD}_{w^\ell,v^c}\KD_{w^\ell,v^\ell},\text{ by definition of $\KD^\spartial$
, see~\eqref{equ:defKpartial}}\\
&=-k'\frac{\cd(u_{\beta^r})}{\cd(u_{\alpha^\ell})}\sc(\theta^\spartial)
=k'\Delta^{m,\spartial}_{v^c,v^\ell},\text{ by the finite non-boundary computation}.
\end{align*}
When $x,y$ are at distance 0, and $x=y=v$ is a vertex of $\Gs$, we are left with the case where $v=v^c$ for some boundary rhombus pair of 
$\R^{\spartial,\rs}$. We have,
\begin{align*}
[{\overline{\KD^\spartial}}^{\,t\,}\KD]_{v^c,v^c}&=
\overline{\KD^\spartial}_{w^\ell,v^c}\KD_{w^\ell,v^c}+\overline{\KD^\spartial}_{w^r,v^c}\KD_{w^r,v^c}\\
&=|\KD_{w^\ell,v^c}|^2\frac{\cd(u_{\beta^r})}{\cd(u_{\alpha^\ell})}+ |\KD_{w^r,v^c}|^2,
\text{ by definition of $\KD^\spartial$, see~\eqref{equ:defKpartial}}\\
&=(k')^2\sc(\theta^\spartial)\Bigl(\nd(u_{\alpha^\ell})\nd(u_{\beta^\ell})\frac{\cd(u_{\beta^r})}{\cd(u_{\alpha^\ell})}
+\nd(u_{\alpha^r})\nd(u_{\beta^r})\Bigr),\\
&\hspace{6.5cm}\text{ by the finite non-boundary computation}\\
&=(k')^2\sc(\theta^\spartial)\frac{\nd(u_{\beta^\ell})\nd(u_{\beta^r})}{\cn(u_{\alpha^\ell})}
\bigl(\cn(u_{\beta^r})+\cn(u_{\alpha^\ell})\bigr)=k' \Delta^{m,\spartial}_{v^c,v^c},
\text{ since $\overline{\beta}_\ell=\overline{\alpha}_r[2\pi]$}.\qedhere
\end{align*}
\end{proof}

\subsection{Determinants of the $Z^u$-Dirac and $Z$-massive Laplacian operators}\label{sec:det_Kpartial_Lap}

We restrict to the finite case and consider the $Z^u$-Dirac operators $\KD(u)$ and $\KD^\spartial(u)$
of Section~\ref{sec:def_Dirac} (finite case), the $Z$-massive Laplacian $\Delta^{m,\spartial}(u)$, and the dual $Z$-massive Laplacian
$\Delta^{m,*}$ of Section~\ref{sec:def_Lap} (finite case). 
Theorem~\ref{thm:det} below proves that the 
determinants of $\KD(u)$ and $\Delta^{m,*}$ are equal up to an explicit constant depending on $u$. Corollary~\ref{cor:det} establishes a 
similar result for the determinants of $\KD^{\spartial}(u)$ and $\Delta^{m,\spartial}(u)$, thus implying identities between 
partition functions.

\subsubsection{Determinants as partition functions}\label{sec:determinant_pf}

Returning to Section~\ref{sec:def_double_graph} on Temperley's bijection, in particular to Equation~\eqref{equ:relation_c_ctilde},
let $\tilde{c}(u)$, be the weight function on directed edges of 
$\Gs,\bar{\Gs}^*$ corresponding to the weight function $c(u)$ of Equation~\eqref{equ:defKD}. Then,
the partition function of pairs of dual $\rs$-rooted and $\outer$-rooted directed spanning trees of $\Gs,\bar{\Gs}^*$ with weight function $\tilde{c}(u)$,
is equal to
the partition function of the dimer model on the graph $\GDro$ with weight function $c(u)$~\cite{Temperley,KPW},
which is equal to $|\det\KD(u)|$~\cite{TF,Kasteleyn1}. We proceed in a similar way with dimers weighted by $c^\spartial(u)$ and thus have, 
\begin{equation}\label{equ:partition_function_1}
\Ztreepair((\Gs,\bar{\Gs}^*),\tilde{c}(u))=|\det \KD(u)|,\quad 
\Ztreepair((\Gs,\bar{\Gs}^*),\tilde{c}^\spartial(u))=|\det \KD^\spartial(u)|.
\end{equation}
Returning to Section~\ref{sec:spanning_forests}, we have that $\det\Delta^{m,*}$ counts weighted rooted spanning forests of $\Gs^*$ 
with conductances $\rho^*$ and masses $m^*$,
where conductances are symmetric. In a similar way, $\det \Delta^{m,\spartial}(u)$ counts weighted rooted directed spanning forests
of $\Gs\ro$ with conductances $\rho^\spartial(u)$,
$m^\spartial(u)$, where the dependence in $u$ is along the boundary only, and the conductances are symmetric away from the boundary. That is,
\begin{equation}\label{equ:partition_function_2}
 \ZforestA(\Gs^*,\rho^*,m^*)=\det \Delta^{m,*},\quad 
 \Zforest(\Gs\ro,\rho^\spartial(u),m^\spartial(u))=\det \Delta^{m,\spartial}(u).
\end{equation}

\subsubsection{Statements}

Recall that every edge $e=(w,x)$ of $\GDro$ is assigned two rhombus vectors $e^{i\bar{\alpha}_e}$, $e^{i\bar{\beta}_e}$
of the diamond graph $\GRR$ and a half-angle $\bar{\theta}_e$. Moreover, every white vertex $w$ of $\GDro$ is in the
center of a rhombus of the diamond graph $\GR$; we let $\bar{\theta}_w$ be the half-angle of this rhombus at one of the two 
\emph{primal} vertices. 

\begin{thm}\label{thm:det}
Let $\Ms_0$ be a dimer configuration of $\GDro$. 
Then, for every $u\in\Re(\TT(k))$, we have
\begin{equation}\label{equ:thm_det}
|\det\KD(u)|=
(k')^{\frac{|\Vs^*|}{2}}\Bigl(\prod_{w\in W}\sc(\theta_w)^\frac{1}{2}\Bigr)
\Bigl(\prod_{e=wx\in\Ms_0} [\dn(u_{\alpha_e})\dn(u_{\beta_e})]^{\frac{1}{2}}  \Bigr)
\det \Delta^{m,*},
\end{equation}
where, $|\det\KD(u)|=\Ztreepair((\Gs,\bar{\Gs}^*),\tilde{c}(u))$, $\det \Delta^{m,*}=\ZforestA(\Gs^*,\rho^*,m^*)$.
\end{thm}

\begin{rem}\label{rem:partition_funciont_KD_Lap}$\,$
Comments on this theorem are given in the introduction to Section~\ref{sec:Z_Dirac_Z_Lap} especially how, in the critical case $k=0$, it is an 
easy consequence of Temperley's bijection and the matrix-tree theorem, and how the argument does not directly extend to the non-critical case.

As one expects, the quantity $\prod_{e=wx\in\Ms_0} [\dn(u_{\alpha_e})\dn(u_{\beta_e})]^{\frac{1}{2}}$
is independent of the choice of perfect matching $\Ms_0$. To see this, it suffices to check that the alternating product around every inner quadrangle face of $\GDro$
is equal to 1; this is proved in Lemma~\ref{lem:K_Kg} below.

\end{rem}

The following is an immediate corollary of Theorem~\ref{thm:det}.
\begin{cor}\label{cor:det}
Let $\Ms_0$ be a dimer configuration of $\GDro$. Then, for every $u\in\Re(\TT(k))'$ we have,
\[
|\det\KD^\spartial(u)|=
(k')^{\frac{|\Vs|-1}{2}}\Bigl(\prod_{w\in W}\cs(\theta_w)^\frac{1}{2}\Bigr)  
\Bigl(\prod_{e=wx\in\Ms_0} [k'\nd(u_{\alpha_e})\nd(u_{\beta_e})]^{\frac{1}{2}}  \Bigr)
\det\Delta^{m,\spartial}(u),
\] 
where $|\det\KD^\spartial(u)|=\Ztreepair((\Gs,\bar{\Gs}^*),\tilde{c}^\spartial(u))$, 
$\det\Delta^{m,\spartial}(u)=\Zforest(\Gs\ro,\rho^\spartial(u),m^\spartial(u))$.
\end{cor}
\begin{proof}
From Theorem~\ref{prop:KDtKD}, we have:
\begin{equation*}
\det \KD^\spartial(u)\det\KD(u)=(k')^{|\Vs^\rs|+|\Vs^*|}\det\Delta^{m,\spartial}(u)\det\Delta^{m,*}. 
\end{equation*}
Replacing $\det\KD(u)$ in the above by Equation~\eqref{equ:thm_det} and recalling that $|W^\rs|=|B^\rs|=|\Vs^\rs|+|\Vs^*|$ 
and that $|\Vs^\rs|=|\Vs|-1$, yields the result. Note that an argument similar to that used for proving Theorem~\ref{thm:det} 
would give an alternative direct proof of Corollary~\ref{cor:det}.
\end{proof}

\subsubsection{Proof of Theorem~\ref{thm:det}}

The proof of Theorem~\ref{thm:det} is postponed until the end of this section. It is a consequence of four intermediate results,
see Equations~\eqref{equ:lem_K_Kg}, \eqref{equ:Kg_arbres}, \eqref{equ:egalite_KgD}, \eqref{lem:D*Dm*} below, which we now establish. We will
use Section~\ref{sec:def_double_graph} on Temperley's bijection and Appendix~\ref{app:gauge} on gauge transformations.

\paragraph{From pairs of directed spanning trees to directed spanning trees.}

We start by defining a gauge transformation $\KD^\gs(u)$ of the matrix $\KD(u)$. Being gauge equivalent, the determinants of 
$\KD^\gs(u)$ and $\KD(u)$ are equal up to an explicit constant. 
We then use Temperley's bijection to deduce
that the determinant of $\KD^\gs(u)$ counts weighted $\outer$-directed spanning trees of the dual graph $\bar{\Gs}^*$.

Let $\gs(u)$ be the weight function on white-to-black edges of $\GDro$ defined by, for every edge $wx$ of $\GDro$, 
\begin{equation*}
\gs(u)_{wx}=
[\cs(\theta_w)\nd(u_{\alpha_e})\nd(u_{\beta_e})]^{\frac{1}{2}}.
\end{equation*}

Consider the matrix $\KD^{\gs}(u)$ obtained from $\KD(u)$ by multiplying edge-weights by $\gs(u)$. Returning to the definition 
of $\KD(u)$, see~\eqref{equ:defKD} and~\eqref{equ:defKD_1}, we obtain that non-zero coefficients of $\KD^\gs(u)$ are given by,
\begin{align*}
\KD^\gs(u)_{w,x}=
\begin{cases}
e^{i\frac{\overline{\alpha}_e+\overline{\beta}_e}{2}} &\text{ if $x\in\Vs^\rs$}\\
e^{i\frac{\overline{\alpha}_e+\overline{\beta}_e}{2}}
(k')^\frac{1}{2} \cs(\theta_w)\nd({u_{\alpha_e}})\nd({u_{\beta_e}})
&\text{ if $x\in\Vs^*$}.
\end{cases}
\end{align*}

\begin{lem}\label{lem:K_Kg}
Consider a perfect matching $\Ms_0$ of $\GDro$.
The bipartite, weighted adjacency matrices $\KD^\gs(u)$ and $\KD(u)$ are gauge equivalent and we have,
\begin{equation}\label{equ:lem_K_Kg}
|\det\KD(u)|=\Bigl(\prod_{w\in W}\sc(\theta_w)^\frac{1}{2}\Bigr)
\Bigl(\prod_{e=wx\in\Ms_0} [\dn(u_{\alpha_e})\dn(u_{\beta_e})]^{\frac{1}{2}}  \Bigr)|\det\KD^\gs(u)|. 
\end{equation}
\end{lem}
\begin{proof}
Details on gauge equivalences for bipartite, weighted adjacency matrices are given in Section~\ref{sec2:app} of Appendix~\ref{app:gauge}.
To prove that $\KD$ and $\KD^\gs$ are gauge equivalent, it suffices to show that the alternating products
of $\KD$ and $\KD^\gs$ around inner face-cycles of $\GDro$ are equal. Inner face-cycles of $\GDro$ are quadrangles; using the notation of 
Figure~\ref{fig:GD1} we have,
\[
\frac{\KD^\gs_{w,v}\KD^\gs_{w',f}}{\KD^\gs_{w,f}\KD^\gs_{w',v}}=
\frac{\KD_{w,v}\KD_{w',f}}{\KD_{w,f}\KD_{w',v}}\quad \Leftrightarrow\quad
\frac{{\gs}_{w,v}}{{\gs}_{w',f}}\frac{{\gs}_{w,f}}{{\gs}_{w',v}}=1.
\]
By definition of $\gs$ we have,
\begin{align*}
\frac{{\gs}_{w,v}}{{\gs}_{w,f}}\frac{{\gs}_{w',f}}{{\gs}_{w',v}}&=
\frac{\cs(\theta)\nd(u_\alpha)\nd(u_\beta)}{\cs(\theta)\nd(u_{\beta-2K})\nd(u_{\alpha})} 
\frac{\cs(\theta')\nd(u_{\beta'})\nd(u_{\alpha'+2K})}{\cs(\theta')\nd(u_{\alpha'})\nd(u_{\beta'})
}=1,
\end{align*}
since $\bar{\alpha}'=\bar{\beta}[2\pi]$. The equality between determinants comes from~\cite{Kuperberg}, see also
Corollary~\ref{cor:gauge_bip_det},
and from the fact that, $\prod_{wx\in\Ms_0}\sc(\theta_w)=\prod_{w\in W}\sc(\theta_w)$, since a dimer configuration covers all 
white vertices $W^\rs$ of $\GDro$, and $W^\rs=W$.
\end{proof}

Now, the matrix $\KD^\gs(u)$ is a complex, bipartite Kasteleyn matrix of the graph $\GDro$, where edges are assigned the positive
weight function $c^{\,\gs}(u)$ given by:
\begin{equation*}
c^{\,\gs}(u)_{w,x}=
\begin{cases}
1&\text{ if $x\in \Vs^\rs$}\\
(k')^{\frac{1}{2}}\cs(\theta_w)\nd(u_{\alpha_e})\nd(u_{\beta_e})&\text{ if $x\in\Vs^*$}.
\end{cases}
\end{equation*}
Returning to Section~\ref{sec:def_double_graph}, the weight function $c^{\,\gs}(u)$ determines a weight function 
$\tilde{c}^{\,\gs}(u)$ on directed edges of $\Gs,\bar{\Gs}^*$, see~\eqref{equ:relation_c_ctilde}. 
Then, $\tilde{c}^{\,\gs}(u)_{v,v'}=1$, for all directed edges $(v,v')$ of $\Gs$ such that $v$ is a vertex of $\Gs\ro$.
Let us now express the weight function $\tilde{c}^{\,\gs}(u)$ on directed edges of $\bar{\Gs}^*$ using the
rhombus vectors and half-angles assigned to those edges. Recall that an edge 
$(f,f')$ of the restricted dual $\Gs^*$ is assigned two rhombus vectors $2e^{i\bar{\alpha}}$, $2e^{i\bar{\beta}}$
and a half-angle $\bar{\theta}^*$ of $\GR$. 
An edge $(f,\outer)$ of $\bar{\Gs}^*$ corresponds to an edge $fw$ of $\GD$, and one assigns to $(f,\outer)$ the rhombus 
vectors $e^{i\bar{\alpha}}$, $e^{i\bar{\beta}}$ and half-angle $\bar{\theta}^*$ 
of $\GRR$ associated to the edge $(f,w)$. Then, for every directed edge $(f,f')$ of $\bar{\Gs}^*$ such that $f$ is 
a vertex of $\Gs^*$ we have, using the notation of~\eqref{equ:relation_c_ctilde}:
\begin{align}
\tilde{c}^{\,\gs}(u)_{f,f'}=c^{\,\gs}(u)_{w,f}&=
(k')^{\frac{1}{2}}\cs(\theta_w)\nd(u_{\alpha_e})\nd(u_{\beta_e})\nonumber \\
&=(k')^{-\frac{1}{2}}\sc(\theta^*)\dn(u_{\alpha})\dn(u_{\beta}) \label{equ:weightZ},
\end{align}
using that $\bar{\theta}^*=\frac{\pi}{2}-\bar{\theta}_w$, $\bar{\alpha}_e=\bar{\alpha}\pm \pi$, 
$\bar{\beta}_e=\bar{\beta}\pm \pi$; and $\tilde{c}^{\,\gs}(u)_{\outer,f'}=0$, for every edge $(\outer,f')$ of $\bar{\Gs}^*$. 

Since edges of the primal graph have weight 1, by the KPW-Temperley bijection~\cite{Temperley,KPW},
the determinant of $\KD^\gs(u)$ counts weighted $\outer$-directed spanning trees of $\T^\outer(\bar{\Gs}^*)$, 
where directed edges of $\bar{\Gs}^*$ are assigned the weight function $\gamma^*(u)$ given by, for every directed edge $(f,f')$ of $\bar{\Gs}^*$,
\begin{equation}\label{equ:gamma_star}
\gamma^*(u)_{f,f'}:=\tilde{c}^{\,\gs}(u)_{f,f'}=
\begin{cases}
(k')^{-\frac{1}{2}}\sc(\theta^*)\dn(u_{\alpha})\dn(u_{\beta})&\text{ if $f$ is a vertex of $\Gs^*$}\\
0 & \text{ if $f=\outer$}.
\end{cases}
\end{equation}
That is,
\begin{equation}\label{equ:Kg_arbres}
|\det\KD^\gs(u)|=\Ztree^\outer(\bar{\Gs}^*,\gamma^*(u)).
\end{equation}

\paragraph{Matrix-tree theorem.} An alternative way of computing the partition function~\eqref{equ:Kg_arbres} is to use the directed version
of the matrix-tree theorem~\cite{Kirchhoff,Tutte}. Let $\Delta^{*}(u)$ be the (non-massive) Laplacian matrix of the graph $\bar{\Gs}^*$
with conductances $\gamma^*(u)$ on directed edges. The non-zero coefficients of the matrix $\Delta^{*}(u)$ are given by:
\begin{equation}\label{equ:Lap_oriente}
\Delta^{*}(u)_{f,f'}=
\begin{cases}
-\gamma^*(u)_{f,f'}&\text{ if $(f,f')$ is an edge of $\bar{\Gs}^*$}\\
\sum_{j=1}^d \gamma^*(u)_{f,f_j} &\text{ if $f'=f$ is a vertex of $\bar{\Gs}^*$}.
\end{cases}
\end{equation}
Let $\Delta^{*,\outer}(u)$ be the matrix obtained from $\Delta^{*}(u)$ by removing the row and column corresponding to the vertex 
$\outer$. Then,
\begin{equation}\label{equ:egalite_KgD}
\Ztree^\outer(\bar{\Gs}^*,\gamma^*(u))=\det\Delta^{*,\outer}(u).
\end{equation}

\paragraph{From $\outer$-directed spanning trees to rooted spanning forests.}

The last step consists in going from $\outer$-directed spanning trees of $\bar{\Gs}^*$ counted by $\det\Delta^{*,\outer}(u)$ to 
rooted spanning forests of $\Gs^*$ counted by $\det\Delta^{m,*}$ using gauge equivalences on weighted adjacency matrices of digraphs, see
Section~\ref{sec:app_gauge_1} of Appendix~\ref{app:gauge}. 

\begin{lem}
The weighted adjacency matrices $(k')^{-\frac{1}{2}}\Delta^{*,\outer}(u)$ and $\Delta^{m,*}$ are gauge equivalent and we have,
\begin{equation}\label{lem:D*Dm*}
\det\Delta^{*,\outer}(u)=(k')^{\frac{|\Vs^*|}{2}}\det\Delta^{m,*}.
\end{equation}
\end{lem}
\begin{proof}
Both matrices $(k')^{-\frac{1}{2}}\Delta^{*,\outer}(u)$ and $\Delta^{m,*}$ have the same associated digraph which is the restricted dual 
$\Gs^*$ where a loop is added at every vertex, and each undirected edge is replaced by the two possible directed edges. 
The graph $\Gs^*$ being connected, the associated digraph is strongly connected. 
We use Lemma~\ref{lem:crit_gauge} to prove gauge equivalence of the matrices. 

For every vertex $f$ of $\Gs^*$, define $q_{f,f}=1$. Next, consider two distinct vertices $f_1,f_2$ of $\Gs^*$ and a simple
di-path $\gamma$ from $f_1$ to $f_2$. Set,
\[
q_{f_1,f_2}=\prod_{e\in\gamma}[(k')^{-1}\dn(u_\alpha)\dn(u_{\beta})].
\]
The function $q$ is well defined because it is the exponential function of~\cite{BdTR1} evaluated at $u-2K-2iK'$, 
see also Equation~\eqref{equ:expo_function}. Indeed,
\begin{align*}
(k')^{-1}\dn(u_\alpha)\dn(u_\beta)&=(k')^{-1}\dn((u-2K-2iK')_\alpha +K+iK')\dn((u-2K-2iK')_\beta +K+iK')\\
&=[i(k')^{\frac{1}{2}}\sc((u-2K-2iK')_\alpha) ][i(k')^{\frac{1}{2}}\sc((u-2K-2iK')_\beta)].
\end{align*}
Fix a vertex $f_0$ of $\Gs^*$, and define $D^{f_0}$ to be the diagonal matrix whose diagonal coefficient $D^{f_0}_{f,f}$ corresponding 
to the vertex $f$ of $\Gs^*$ is $q_{f_0,f}$. Let us prove that
\begin{equation}\label{equ:Lapo_Lapm}
\Delta^{m,*}=(k')^{-\frac{1}{2}}D^{f_0}\,\Delta^{*,\outer}(u)(D^{f_0})^{-1}.
\end{equation}
Recall the definition of the Laplacian matrix $\Delta^{*,\outer}(u)$, see~\eqref{equ:Lap_oriente}
and~\eqref{equ:gamma_star}. For the diagonal coefficient corresponding to a vertex $f$ of $\Gs^*$, we have
\begin{align*}
[(k')^{-\frac{1}{2}}D^{f_0}\,\Delta^{*,\outer}(u)(D^{f_0})^{-1}]_{f,f}&=
(k')^{-\frac{1}{2}}\Delta^{*,\outer}(u)_{f,f}=(k')^{-1}\sum_{j=1}^d \sc(\theta_j^*)\dn(u_{\alpha_j})\dn(u_{\beta_j}), 
\\
&=\sum_{j=1}^d A(\theta_j^*)=\Delta^{m,*}_{f,f}, \text{ by Proposition 11 of~\cite{BdTR1}}.
\end{align*}
For an edge $(f,f')$ of $\Gs^*$, we have
\begin{align*}
[(k')^{-\frac{1}{2}}D^{f_0}\,\Delta^{*,\outer}(u)(D^{f_0})^{-1}]_{f,f'}&=
(k')^{-\frac{1}{2}}\frac{q_{f_0,f}}{q_{f_0,f'}}\Delta^{*,\outer}(u)_{f,f'}
=(k')^{-\frac{1}{2}}\frac{1}{q_{f,f'}}\Delta^{*,\outer}(u)_{f,f'}\\
&=-(k')^{-\frac{1}{2}}(k')\nd(u_{\alpha})\nd(u_{\beta})(k')^{-\frac{1}{2}}\sc(\theta^*)\dn(u_{\alpha})\dn(u_{\beta})\\
&=-\sc(\theta^*)=\Delta^{m,*}_{f,f'}.
\end{align*}
By Lemma~\ref{lem:crit_gauge}, Equation~\eqref{equ:Lapo_Lapm} implies that the matrices 
$(k')^{-\frac{1}{2}}\Delta^{*,\outer}(u)$ and $\Delta^{m,*}$ are gauge equivalent, and we have
equality of the determinants:
$\det\Delta^{m,*}=\det[(k')^{-\frac{1}{2}}\Delta^{*,\outer}(u)]$. 
\end{proof}

\begin{proof}[Proof of Theorem~\ref{thm:det}]
Combining Equation~\eqref{equ:lem_K_Kg}, \eqref{equ:Kg_arbres}, \eqref{equ:egalite_KgD} and \eqref{lem:D*Dm*}, we obtain,
\begin{align*}
|\det\KD(u)|&= \Bigl(\prod_{w\in W }\sc(\theta_w)^\frac{1}{2}\Bigr)\Bigl(\prod_{e=wx\in\Ms_0} [\dn(u_{\alpha_e})\dn(u_{\beta_e})]^{\frac{1}{2}}  \Bigr)|\det\KD^\gs(u)|\\
&=\Bigl(\prod_{w\in W }\sc(\theta_w)^\frac{1}{2}\Bigr)\Bigl(\prod_{e=wx\in\Ms_0} [\dn(u_{\alpha_e})\dn(u_{\beta_e})]^{\frac{1}{2}}  \Bigr) \det\Delta^{*,\outer}(u)\\
&=(k')^{\frac{|\Vs^*|}{2}}\Bigl(\prod_{w\in W }\sc(\theta_w)^\frac{1}{2}\Bigr)\Bigl(\prod_{e=wx\in\Ms_0} [\dn(u_{\alpha_e})\dn(u_{\beta_e})]^{\frac{1}{2}}  \Bigr)
\det\Delta^{m,*}.\qedhere
\end{align*}
\end{proof}

\subsection{Z-invariance of the $Z^u$-Dirac operator}\label{sec:Z_invariance}

According to Baxter~\cite{Baxter:8V,Baxter:Zinv,Baxter:exactly} a model of statistical mechanics is \emph{$Z$-invariant} if, 
when decomposing the partition function according to the possible 
configurations outside of the hexagon of the diamond graph $\GR$ defining the star/triangle, it only changes by a constant
independent of the outer configurations when performing a $\mathsf{Y}$-$\Delta$ move.

Suppose that the isoradial graph $\Gs$ is finite. The $Z^u$-Dirac operator $\KD(u)$ is the bipartite Kasteleyn matrix of the dimer model on the 
double graph $\GDro$ with weight function $c(u)$ on the edges given by~\eqref{equ:defKD}. By Section~\ref{sec:determinant_pf} 
this dimer model is in bijection with pairs of dual directed spanning trees of 
$\T^{\rs,\outer}(\Gs,\bar{\Gs}^*)$, with weight function $\tilde{c}(u)$ on directed edges of $\Gs,\bar{\Gs}^*$:
\begin{equation*}
\Zdimer(\GDro,c(u))=|\det\KD(u)|=\Ztreepair((\Gs,\bar{\Gs}^*),\tilde{c}(u)).
\end{equation*}
By Lemma~\ref{lem:K_Kg} and Equation~\eqref{equ:Kg_arbres}, $|\det \KD(u)|$ is equal up to a constant to $|\det \KD^{\,\gs}(u)|$
which counts $\outer$-directed spanning trees of $\T^\outer(\bar{\Gs}^*)$ with conductances $\gamma^* (u)$ on directed edges of 
$\bar{\Gs}^*$ given by~\eqref{equ:gamma_star}; that is:
\begin{equation*}
\Zdimer(\GDro,c(u))=C(u)\cdot \Ztree^\outer(\bar{\Gs}^*,\gamma^*(u)),
\end{equation*}
where $C(u)$ is given in Lemma~\ref{lem:K_Kg}. By Temperley's bijection again, there is a one-to-one correspondence between 
dimer configurations of $\GDro$ and $\outer$-directed spanning trees of $\T^\outer(\bar{\Gs}^*)$ (such that the primal tree is $\rs$-rooted). 
We prove $Z$-invariance of this $\outer$-rooted directed spanning tree model on $\bar{\Gs}^*$;
using the above, the decomposition of the partition function has a direct 
interpretation in terms of the dimer model on $\GDro$. 

Since duality preserves isoradiality, we actually show $Z$-invariance of 
the $\rs$-directed spanning tree model on $\Gs$, where directed edges of $\Gs$ are assigned conductances $\gamma(u)$
given by, for every directed edge $(v,v')$ of $\Gs$, 
\begin{equation*}
\gamma(u)_{v,v'}=(k')^{-\frac{1}{2}}\sc(\theta)\dn(u_\alpha)\dn(u_\beta), 
\end{equation*}
where $2e^{i\bar{\alpha}},2e^{i\bar{\beta}},\bar{\theta}$ are the rhombus vectors and half-angle of $\GR$ associated to the edge $(v,v')$;
$\gamma(u)$ is the primal version of the conductances $\gamma^*(u)$ of~\eqref{equ:gamma_star}. 

\begin{prop}\label{prop:Zinv}
Consider a finite isoradial graph $\Gs$, and let $u\in\Re(\TT(k))$. Then, the model of $\rs$-directed spanning trees on $\Gs$, with weight function $\gamma(u)$ on the 
edges is $Z$-invariant.
\end{prop}
\begin{rem}$\;$
\begin{itemize}
\item By Theorem~\ref{thm:det}, we have:
\begin{equation*}
\Zdimer(\GDro,c(u))=C'(u)\cdot \ZforestA(\Gs^*,\rho^*,m^*),
\end{equation*}
and in the paper~\cite{BdTR1}, the model of rooted spanning forests with these weights
is shown to be $Z$-invariant. But, since the proof of Theorem~\ref{thm:det}
does not provide a bijection between directed spanning trees and rooted spanning forests, 
the two decompositions of the partition functions are not directly comparable; they should nevertheless be 
compatible. Note that the computations of the proof of Proposition~\ref{prop:Zinv} are reminiscent but
much simpler than those of~\cite{BdTR1} (the latter have been removed from the published version).
\item The critical spanning tree model of~\cite{Kenyon3} with conductances $\tan(\theta)$ is
$Z$-invariant~\cite{Kennelly}. The result of~\cite{BdTR1} extends this to rooted spanning forest while Proposition~\ref{prop:Zinv} 
extends it to directed spanning trees.
\end{itemize}
\end{rem}

\begin{proof}
Let $\Gsstar$ and $\Gstriang$ be two finite isoradial graphs differing by a star-triangle transformation, and let $\rs$ be a fixed root on the 
boundary of the graph, outside of the hexagon defining the star/triangle. Let $\gammas(u)$, resp. $\gammat(u)$, be the weight function on 
$\rs$-dST of $\Gstriang$, resp. $\Gsstar$. We use the notation of Figure~\ref{fig:Zinv}, and write $\gammas_{i,j}$/$\gammat_{i,j}$
for the weight of the edge $(v_i,v_j)$. Using the identities $\dn(u+K)=k'\nd(u)$, $\sc(\theta^*)=\sc(K-\theta)=
(k')^{-1}\cs(\theta)$, we have the following: for every $j\in\{1,2,3\}$, with cyclic notation for indices:
\[
\begin{array}{clcl}
\gammat_{0,j}&=(k')^{-\frac{1}{2}}\sc(\theta_j)\dn(u_{\alpha_{j}})\dn(u_{\alpha_{j+1}}),\quad     
    &  \gammas_{j,j+1}&=(k')^{-\frac{1}{2}}\cs(\theta_{j-1}) \nd(u_{\alpha_{j}})  \dn(u_{\alpha_{j-1}}),\quad   \\
\gammat_{j,0}&=(k')^{\frac{3}{2}}\sc(\theta_j) \nd(u_{\alpha_{j}})  \nd(u_{\alpha_{j+1}}),    
   &\gammas_{j+1,j}&=(k')^{-\frac{1}{2}}\cs(\theta_{j-1})\dn(u_{\alpha_{j}}) \nd(u_{\alpha_{j-1}}).
\end{array}
\]

\begin{figure}[ht]
\begin{minipage}[b]{0.5\linewidth}
\begin{center}
\begin{overpic}[width=4.5cm]{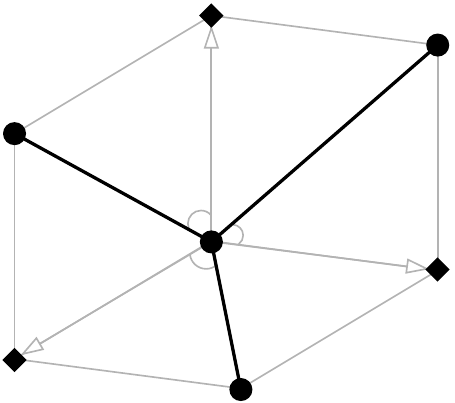}
  \put(51,-5){\scriptsize $v_1$}
  \put(100,78){\scriptsize $v_2$}
  \put(-7,58.5){\scriptsize $v_3$}
  \put(50,28){\scriptsize $v_0$}
  \put(40,23){\scriptsize $\bar{\theta}_1$}
  \put(56,37){\scriptsize $\bar{\theta}_2$}
  \put(40,44){\scriptsize $\bar{\theta}_3$}
  \put(48,65){\scriptsize $e^{i\bar{\alpha}_3}$}
  \put(15,22.5){\scriptsize $e^{i\bar{\alpha}_1}$}
  \put(66,33.5){\scriptsize $e^{i\bar{\alpha}_2}$}
\end{overpic}
\end{center}
\end{minipage}
\begin{minipage}[b]{0.5\linewidth}
\begin{center}
\begin{overpic}[width=4.5cm]{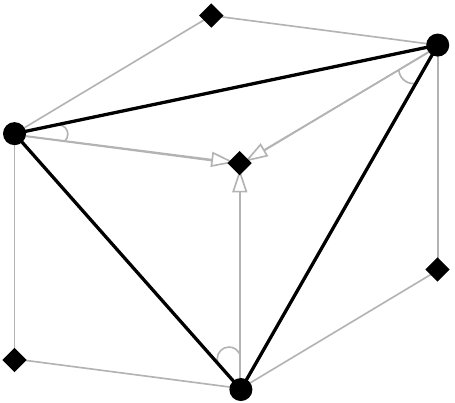}
\put(51,-5){\scriptsize $v_1$}
\put(100,78){\scriptsize $v_2$}
\put(-5,58.5){\scriptsize $v_3$}
\put(24,58){\scriptsize $\bar{\theta}_1^*$}
\put(45,16){\scriptsize $\bar{\theta}_2^*$}
\put(80,63){\scriptsize $\bar{\theta}_3^*$}
\put(54,33){\scriptsize $e^{i\bar{\alpha}_3}$}
\put(58,62){\scriptsize $e^{i\bar{\alpha}_1}$}
\put(28,48){\scriptsize $e^{i\bar{\alpha}_2}$}
\end{overpic}
\end{center}
\end{minipage}
\caption{Notation for vertices, angles and rhombus vectors in a star-triangle transformation.}
\label{fig:Zinv}
\end{figure}
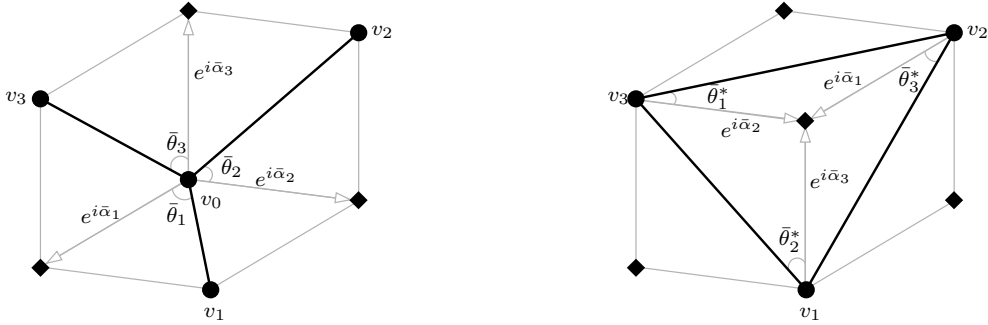

Consider a directed edge configuration $\Ts$ outside of the hexagon which can be extended to an $\rs$-dST of $\Gsstar$/$\Gstriang$. 
Then given $\Ts$, the configurations inside the $\mathsf{Y}$-$\Delta$ only depend on connection properties of $\Ts$ outside. 
As in the non-directed case~\cite{Kennelly}, we thus have three kinds of configurations to consider.
\begin{itemize}
 \item[I.] $\{\Ts:\{v_1,v_2,v_3\}$ are connected to $\rs\}$. 
 \item[II.] $\{\Ts:\{v_j,v_{j+1}\}$ are connected to   $\{\rs\},\,j\in\{1,2,3\}\}$.
 \item[III.] $\{\Ts:\{v_j\}$ is connected to  $\{\rs\}$, $j\in\{1,2,3\}\}$.
\end{itemize}

Let $Z(\Gsstar|\,i\,)$, resp. $Z(\Gstriang|\,i\,)$, be the partition function of $\Gsstar$, resp. $\Gstriang$, 
restricted to an outer configuration $\Ts$ belonging to the set defined in $i$,
$i\in\{\mathrm{I},\mathrm{II},\mathrm{III}\}$, divided by the contribution
of the configuration $\Ts$. Then, proving $Z$-invariance amounts to showing that there exists a constant $\C$, such that
\begin{equation*}
\forall\,i\in\{\mathrm{I,II,III}\},\quad Z(\Gsstar|\,i\,)=\C\cdot Z(\Gstriang|\,i\,). 
\end{equation*}
Let us prove that this is indeed the case, with $\C=(k')^{\frac{3}{2}}\prod_{j=1}^3 \sc(\theta_j)$.

\paragraph{Restriction $\mathrm{I}.$}
We have:
\begin{align}
&Z(\Gsstar|\,\mathrm{I}\,)=(k')^{-\frac{1}{2}}\sum_{j=1}^3 \sc(\theta_j)\dn(u_{\alpha_{j}}) \dn(u_{\alpha_{j+1}})=
(k')^{\frac{1}{2}}\sum_{j=1}^3 A(\theta_j)\nonumber\\
&Z(\Gstriang|\,\mathrm{I}\,)=1,\nonumber\\
\text{implying that }\ &\C=(k')^{\frac{1}{2}}\sum_{j=1}^3 A(\theta_j)=(k')^{\frac{3}{2}}\prod_{j=1}^3 \sc(\theta_j)\label{equ:cst_C},
\end{align}
by Equation~(71) of Lemma 47 of~\cite{BdTR1}.
\paragraph{Restriction $\mathrm{II}.$}
We have, $\forall\,j\in\{1,2,3\}$,
\begin{align*}
Z(\Gsstar|\,\mathrm{II}\,)&=\gammat_{j+2,0}(\gammat_{0,j}+\gammat_{0,j+1})\\
&=k'\sc(\theta_{j+2})\nd(u_{\alpha_{j+2}})\nd(u_{\alpha_j})
  [\sc(\theta_j)\dn(u_{\alpha_j})\dn(u_{\alpha_{j+1}})  + \\
&\hspace{6.5cm}  +\sc(\theta_{j+1})\dn(u_{\alpha_{j+1}})\dn(u_{\alpha_{j+2}})]\\
  &=k'\sc(\theta_{j+2})[\sc(\theta_j)\nd(u_{\alpha_{j+2}})\dn(u_{\alpha_{j+1}})+
  \sc(\theta_{j+1})\nd(u_{\alpha_j})\dn(u_{\alpha_{j+1}})],\\
Z(\Gstriang|\,\mathrm{II}\,)&=\gammas_{j+2,j}+\gammas_{j+2,j+1}\\
&=(k')^{-\frac{1}{2}}[\cs(\theta_{j+1})\nd(u_{\alpha_{j+2}})\dn(u_{\alpha_{j+1}})+
\cs(\theta_{j})\dn(u_{\alpha_{j+1}})\nd(u_{\alpha_j})].
\end{align*}
Using Equation~\eqref{equ:cst_C}, it is straightforward that $Z(\Gsstar|\,\mathrm{II}\,)=\C \,Z(\Gstriang|\,\mathrm{II}\,)$.

\paragraph{Restriction $\mathrm{III}.$}
We have, $\forall\,j\in\{1,2,3\}$,
\begin{align*}
Z(\Gsstar|\,\mathrm{III}\,)&=\gammat_{j+1,0}\gammat_{j+2,0}\gammat_{0,j}\\
&=(k')^{\frac{5}{2}}\sc(\theta_{j+1})\nd(u_{\alpha_{j+1}})\nd(u_{\alpha_{j+2}})\times
\sc(\theta_{j+2})\nd(u_{\alpha_{j+2}})\nd(u_{\alpha_j})\times \\
&\hspace{6.1cm} \times \sc(\theta_j) \dn(u_{\alpha_j})\dn(u_{\alpha_{j+1}})\\
&=(k')^{\frac{5}{2}}\sc(\theta_j)\sc(\theta_{j+1})\sc(\theta_{j+2})\nd^2(u_{\alpha_{j+2}}),
\end{align*}
\begin{align*}
Z(\Gstriang|\,\mathrm{III}\,)&=\gammas_{j+1,j+2}\gammas_{j+2,j}+\gammas_{j+2,j+1}\gammas_{j+1,j}+
\gammas_{j+1,j}\gammas_{j+2,j}\\
&=(k')^{-1}[\cs(\theta_j)\nd(u_{\alpha_{j+1}})\dn(u_{\alpha_j})\cs(\theta_{j+1})\nd(u_{\alpha_{j+2}})\dn(u_{\alpha_{j+1}})+\\
&\hspace{1.6cm}  +\cs(\theta_j)\dn(u_{\alpha_{j+1}})\nd(u_{\alpha_j})   \cs(\theta_{j+2})\dn(u_{\alpha_j})\nd(u_{\alpha_{j+2}})+\\
&\hspace{1.6cm}  +\cs(\theta_{j+2})\dn(u_{\alpha_{j}})\nd(u_{\alpha_{j+2}}) \cs(\theta_{j+1})\nd(u_{\alpha_{j+2}})\dn(u_{\alpha_{j+1}})]\\
&=\frac{(k')^{-1}}{\prod_{j=1}^3 \sc(\theta_j)}\nd^2(u_{\alpha_{j+2}})
[\sc(\theta_{j+2})\dn(u_{\alpha_{j+2}}) \dn(u_{\alpha_j}) +\\
&\hspace{2cm}+ \sc(\theta_{j+1})\dn(u_{\alpha_{j+1}}) \dn(u_{\alpha_{j+2}}) +
\sc(\theta_{j})\dn(u_{\alpha_{j}}) \dn(u_{\alpha_{j+1}})
]\\
&=k'\nd^2(u_{\alpha_{j+2}}),\text{ by the proof of Restriction I}.
\end{align*}
The proof that $Z(\Gsstar|\,\mathrm{III}\,)=\C\, Z(\Gstriang|\,\mathrm{III}\,)$ is concluded by using Equation~\eqref{equ:cst_C} again. 
\end{proof}

\subsection{Inverse $Z^u$-Dirac operator and dimer model on the double graph}\label{sec:KD_Lap_inv}

Using Theorem~\ref{prop:KDtKD}, in Corollaries~\ref{cor:KD_G} and~\ref{cor:KD_G_finite}, we express the inverse $Z^u$-Dirac operators
$\KD(u)^{-1}$ and $\KD^\spartial(u)^{-1}$ using the $Z$-massive and dual $Z$-massive Green functions.
In the infinite case, this allows to prove in 
Theorem~\ref{thm:Gibbs_KD} an explicit~\emph{local} expression for a Gibbs measure for the dimer model on the double graph with weight 
function $c(u)$.

\subsubsection{Inverse $Z^u$-Dirac operator and $Z$-massive Green functions}

\emph{Infinite case}.
Consider an infinite isoradial graph $\Gs$, and the $Z^u$-Dirac operator $\KD(u)$. Consider also the $Z$-massive Laplacian $\Delta^m$ on $\Gs$,
the dual $Z$-massive Laplacian $\Delta^{m,*}$ of the dual $\Gs^*$~\cite{BdTR1}. When 
$k\neq 0$, let $G^m$ and $G^{m,*}$ be the $Z$-massive and dual massive Green functions of $\Gs$ and $\Gs^*$ of~\cite{BdTR1}, 
whose definition is recalled in
Section~\ref{sec:def_Lapmass_0}. When $k=0$, the mass is 0 and we let $G^0$ and $G^{0,*}$ be the Green and dual Green functions of~\cite{Kenyon3}.

\paragraph{Notation for coefficients of Corollary~\ref{cor:KD_G}.}
 Let $\ubar{v}$, resp. $\ubar{f}$, be a vertex of $\Gs$, resp. $\Gs^*$. 
 Let $w$ be a white vertex of $\GD$, its neighbors in $\GD$ are $v_1,f_1,v_2,f_2$, and let 
 $e^{i\alphafb},e^{i\betafb},\thetafb$ be the rhombus vectors and half-angle of $\GRR$ associated to the edge $(w,v_2)$,
 where the subscript ``$\mathrm{f}$'' stands for ``final'', see Figure~\ref{fig:corKDG}.

\begin{figure}[ht]
\begin{minipage}[b]{0.5\linewidth}
\begin{center}
 \begin{overpic}[height=3.2cm]{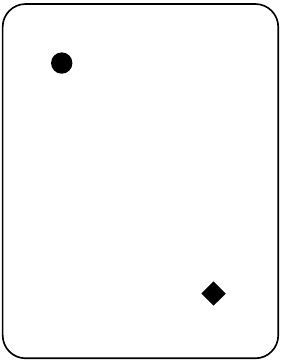}
\put(24,81){\scriptsize $\ubar{v}$}
\put(45,15){\scriptsize $\ubar{f}$}
\end{overpic}
\end{center}
\end{minipage}
\begin{minipage}[b]{0.5\linewidth}
\begin{center}
\begin{overpic}[height=3.2cm]{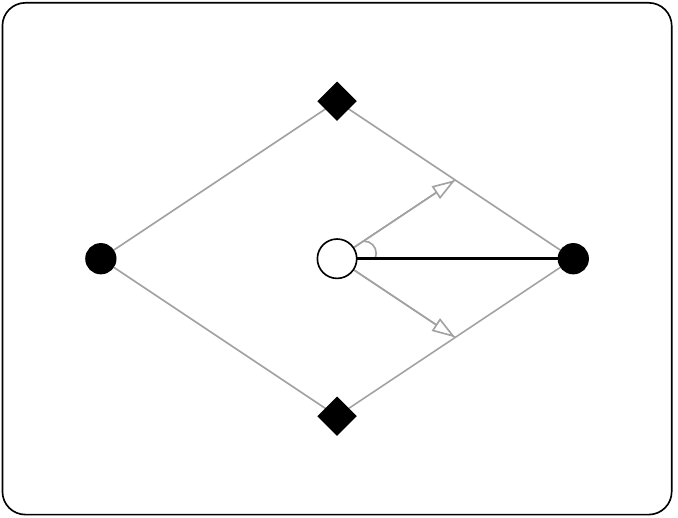}
  \put(5,37){\scriptsize $v_1$}
  \put(89,37){\scriptsize $v_2$}
  \put(47,4){\scriptsize $f_1$}
  \put(47,68){\scriptsize $f_2$}
  \put(52,23){\scriptsize $e^{i\alphafb}$}
  \put(52,47){\scriptsize $e^{i\betafb}$}
  \put(41,37){\scriptsize $w$}
  \put(61,40){\scriptsize $\thetafb$}
\end{overpic} 
\end{center}
\end{minipage}
\caption{Notation for coefficients of Corollary~\ref{cor:KD_G} and~\ref{cor:KD_G_finite}.}
\label{fig:corKDG}
\end{figure}
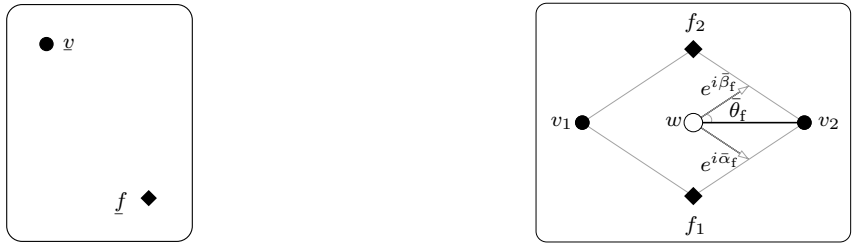

As a consequence of Theorem~\ref{prop:KDtKD} (infinite case) we obtain,
\begin{cor}\label{cor:KD_G}
For every $u\in\Re(\TT(k))$, consider the operator $\KD(u)^{-1}$
mapping $\CC^{W}$ to $\CC^{B}$ defined by:

$\bullet$ \emph{Matrix form.} 
$\displaystyle
\KD(u)^{-1}=(k')^{-1}
\begin{pmatrix}
G^m & 0\\
0 & G^{m,*}
\end{pmatrix}
\overline{\KD(u)}^{\,t\,}.
$

$\bullet$ \emph{Coefficients.} For every $\ubar{v},\ubar{f},w$ as in the notation above, 
\begin{align*}
\KD(u)^{-1}_{\ubar{v},w}&=e^{-i\frac{\alphafb+\betafb}{2}}(k')^{-1}\sc(\thetaf)^{\frac{1}{2}}
\Bigl([\dn(u_{\alphaf})\dn(u_{\betaf})]^{\frac{1}{2}} G^m_{\ubar{v},v_2}-[\dn(u_{\alphaf+2K})\dn(u_{\betaf+2K})]^{\frac{1}{2}}G^m_{\ubar{v},v_1}\Bigr)\\
\KD(u)^{-1}_{\ubar{f},w}&=-ie^{-i\frac{\alphafb+\betafb}{2}}(k')^{-1}\sc(\thetaf^*)^{\frac{1}{2}}
\Bigl([\dn((u_{\betaf})^*)\dn((u_{\alphaf+2K})^*)]^{\frac{1}{2}}G^{m,*}_{\ubar{f},f_2}+\\
&\hspace{7.5cm} -[\dn((u_{\betaf-2K})^*)\dn((u_{\alphaf})^*)]^{\frac{1}{2}}G^{m,*}_{\ubar{f},f_1} \Bigr).
\end{align*}
Then $\KD(u)^{-1}$ is an inverse of the Kasteleyn operator $\KD(u)$. When the graph $\GD$ is moreover $\ZZ^2$-periodic, it is 
the unique inverse decreasing to zero at infinity.
\end{cor}

\begin{rem}\label{rem:localKD_m_1}$\,$
\begin{itemize}
\item[$\bullet$] When $k=0$, then $\dn\equiv 1$ and we recover Corollary 7.2 of~\cite{Kenyon3}. When the graph 
is $\ZZ^2$ and the weights are specific, Chhita~\cite{Chhita} has a result in the same flavor, relating the inverse Kasteleyn 
operator of pairs of dual directed spanning trees to a massive Green function.
\item[$\bullet$] The operator $\KD(u)^{-1}$ is \emph{local}. Indeed the expression for $\KD(u)^{-1}_{\ubar{x},w}$, with $\ubar{x}=\ubar{v}\in \Vs$ or 
$\ubar{x}=\ubar{f}\in\Vs^*$, 
only depends on: two paths of the diamond graph $\GR$ from $\ubar{x}$ to the two neighbors $x_1,x_2$ of $w$ in $\Gs$ or $\Gs^*$,
where the two paths are those used in computing the massive exponential function of coefficients of the massive Green function of~\cite{BdTR1};
and from the rhombus vectors and half-angles of $\GRR$ associated to the edges $(w,x_1), (w,x_2)$. 
\end{itemize}
\end{rem}

\begin{proof}
By definition, see Section~\ref{sec:dimer_infinite}, to show that $\KD(u)^{-1}$ is an inverse, we need to prove that 
$\KD(u)^{-1}\KD(u)=\Id$, and that $\KD(u)^{-1}_{x,w}\rightarrow 0$ as $|w-x|\rightarrow \infty$. 

For the second part, when $k\neq 0$ we prove in~\cite{BdTR1} that the Green function $G^m$ and $G^{m,*}$ decrease exponentially fast
to 0 at infinity; since the function $\dn$ is uniformly bounded, the same holds for $\KD(u)^{-1}$. When $k=0$, the Green function 
explodes like $-\log|x-x'|$ at infinity, but the difference converges polynomially fast to 0 as is proved in~\cite{Kenyon3}. 

For the first part, we apply $\KD(u)$ to the right of the definition of $\KD(u)^{-1}$ in matrix form and use Theorem~\ref{prop:KDtKD}.
Note that this step also uses associativity of the infinite matrix product, which holds in this case.
\end{proof}

\emph{Finite case}.
Consider a finite isoradial graph $\Gs$ and the $Z^u$-Dirac operators $\KD(u)$ and $\KD^\spartial(u)$ of 
Section~\ref{sec:def_Dirac} (finite case). Consider also the $Z$-massive Laplacian $\Delta^{m,\spartial}(u)$ and dual Laplacian $\Delta^{m,*}$
of Section~\ref{sec:def_Lap} (finite case). For the purpose of handling the Ising model, our goal is to 
obtain an explicit expression for the inverse of the operator $\KD^\spartial(u)$. In order to ensure that $\KD^\spartial(u)$ and 
$\Delta^{m,\spartial}(u)$ are invertible, we restrict $u$ to:
\begin{align}
\Re(\TT(k))''&=\Re(\TT(k))\setminus\{\alpha^\ell[2K],\ \beta^{\ell}[2K]:e^{i\bar{\alpha}^\ell},\,e^{i\bar{\beta}^\ell}\in\R^\spartial\}.
\end{align}
Indeed, by the forthcoming Remark~\ref{rem:part_function} this ensures that $\det \KD^\spartial(u)\neq 0$ and by Corollary~\ref{cor:det} that 
$\det\Delta^{m,\spartial}(u)\neq 0$ (since $\nd$ is positive when $u$ is real).
Denote by $G^{m,*}$ the inverse of the (finite) matrix $\Delta^{m,*}$ and by
$G^{m,\spartial}(u)$ the inverse of $\Delta^{m,\spartial}(u)$. 
Note that when some of the conductances $(\rho^\spartial_{v^c,v^\ell}(u))$ of the boundary rhombus 
pairs of $\R^{\spartial,\rs}$ are negative, there is a twist in defining the associated random walk and in giving
the random walk interpretation of the Green function, but this might at most happen for half of the boundary edges.

\paragraph{Notation for coefficients of Corollary~\ref{cor:KD_G_finite}.} Let $w$ be a white vertex of $\GDro$; if 
$w\notin\{w^\ell,w^r\in\R^\spartial\}$, then we use the notation of the infinite case, see Figure~\ref{fig:corKDG};
if $w\in\{w^\ell,w^r\in\R^\spartial\}$, then the vertex $f_2$ is absent.
As an immediate consequence of Theorem~\ref{prop:KDtKD} (finite case) we obtain
\begin{cor}\label{cor:KD_G_finite}
For every $u\in\Re(\TT(k))''$, the inverse matrix $\KD^\spartial(u)^{-1}$ has the following explicit expression.

$\bullet$ \emph{Matrix form.}
$\displaystyle
\KD^\spartial(u)^{-1}=(k')^{-1}
\Bigl[\overline{\KD(u)}
\begin{pmatrix}
G^{m,\spartial}(u) & -G^{m,\spartial}(u)\overline{Q(u)}G^{m,*}\\
0 & G^{m,*}
\end{pmatrix}\Bigr]^t.
$

$\bullet$ \emph{Coefficients.} For every vertex $\ubar{v}$ of $\Gs\ro$, $\ubar{f}$ of $\Gs^*$, and every white vertex $w$ of $\GDro$, 
using the notation of Figure~\ref{fig:corKDG},
\begin{align*}
\KD^\spartial(u)^{-1}_{\ubar{v},w}&\textstyle =e^{-i\frac{\alphafb+\betafb}{2}}(k')^{-1}\sc(\thetaf)^{\frac{1}{2}}
\Bigl([\dn(u_{\alphaf})\dn(u_{\betaf})]^{\frac{1}{2}} G^{m,\spartial}(u)_{v_2,\ubar{v}}+\\
&\hspace{7cm}-
[\dn(u_{\alphaf+2K})\dn(u_{\betaf+2K})]^{\frac{1}{2}}G^{m,\spartial}(u)_{v_1,\ubar{v}}\Bigr)\\
\KD^\spartial(u)^{-1}_{\ubar{f},w}&\textstyle =-ie^{-i\frac{\alphafb+\betafb}{2}}(k')^{-1}\sc(\thetaf^*)^{\frac{1}{2}}
\Bigl(\II_{\{w\notin\{w^\ell,w^r\in \R^\spartial\}\}}[\dn((u_{\betaf})^*)\dn((u_{\alphaf+2K})^*)]^{\frac{1}{2}}G^{m,*}_{f_2,\ubar{f}}+\\
&\textstyle\hspace{7cm}-[\dn((u_{\betaf-2K})^*)\dn((u_{\alphaf})^*)]^{\frac{1}{2}}G^{m,*}_{f_1,\ubar{f}}\Bigl)+\\
&\hspace{0.4cm}\textstyle+i \sum\limits_{(v^c,f^c)\in\R^{\spartial,\rs}}\frac{\nd(u_{\beta^\ell})}{\cd(u_{\alpha^\ell})}\bigl(\cd(u_{\beta^r})-\cd(u_{\alpha^\ell})\bigr)
\KD^{\spartial}(u)^{-1}_{v^c,w}\cdot G^{m,*}_{f^c,\ubar{f}},
\end{align*}
where $\II_{\{w\notin \{w^\ell,w^r\in \R^\spartial\}\}}$ is equal to 0 if $w$ is a boundary vertex of $\GDro$ and 1 otherwise.
In the sum over $(v^c,f^c)\in\R^{\spartial,\rs}$, we use the notation of 
Figure~\ref{fig:def_Kpartial} for vertices and rhombus vectors of the boundary rhombus pairs of $\R^{\spartial,\rs}$;
$\KD^{\spartial}(u)^{-1}_{v^c,w}$ in the formula for $\KD^\spartial(u)^{-1}_{\ubar{f},w}$ is given by the first formula with $\ubar{v}=v^c$.
\end{cor}

\begin{exm}\label{ex:KD_u_v_special}
Let us compute the explicit values of Corollaries~\ref{cor:KD_G} and~\ref{cor:KD_G_finite} in the case where
$u=\ubf:=\frac{\alphaf+\betaf}{2}+K$. This will be used again in Sections~\ref{sec:inv_KQ_KD} and~\ref{sec:GQ_GF_Zinv}.
In the finite case, we will also need to evaluate it at $u=\vbf:=\frac{\alphaf+\betaf}{2}-K$. We have,
\begin{align}
&\textstyle \vbf=\ubf-2K,\ \ \ubf_{\alphaf}=\frac{K+\thetaf}{2},\ \ \ubf_{\betaf}=\frac{K-\thetaf}{2},\ \,\vbf_{\alphaf}=-\ubf_{\betaf},\ \,\vbf_{\betaf}
=-\ubf_{\alphaf}\nonumber\\
&\ubf_{\alphaf}+\ubf_{\betaf}=K,\ \ \ubf_{\alphaf+2K}+\ubf_{\betaf+2K}=-K,\ \
\ubf_{\alphaf+2K}+\ubf_{\betaf}=0,\ \ \ubf_{\alphaf}+\ubf_{\betaf-2K}=2K \label{equ:relation_u_v}\\
&\vbf_{\alphaf}+\vbf_{\betaf}=-K,\ \ \vbf_{\alphaf+2K}+\vbf_{\betaf+2K}=-3K,\ \
\vbf_{\alphaf+2K}+\vbf_{\betaf}=-2K,\ \ \vbf_{\alphaf}+\vbf_{\betaf-2K}=0.\nonumber
\end{align}

Using that $\dn(u\pm K)=k'\nd(u)$, $\dn(u\pm 2K)=\dn(u)$, $\dn(-u)=\dn(u)$ and the above relations, we obtain, for $u\in\{\ubf,\vbf\}$,
\begin{align*}
&[\dn(u_{\alphaf})\dn(u_{\betaf})]^\frac{1}{2}=[\dn(u_{\alphaf+2K})\dn(u_{\betaf+2K})]^\frac{1}{2}=(k')^\frac{1}{2}\\
&[\dn((u_{\betaf})^*)\dn((u_{\alphaf+2K})^*)]^\frac{1}{2}=k' \nd(u_{\betaf})=k'
\begin{cases}
\nd\bigl(\frac{K-\thetaf}{2}\bigr)& \text{ if $u=\ubf$} \\
\nd\bigl(\frac{K+\thetaf}{2}\bigr)& \text{ if $u=\vbf$}
\end{cases}\\
&[\dn((u_{\betaf-2K})^*)\dn((u_{\alphaf})^*)]^\frac{1}{2}=k'\nd(u_{\alphaf})=k'
\begin{cases}
\nd\bigl(\frac{K+\thetaf}{2}\bigr)& \text{ if $u=\ubf$} \\
\nd\bigl(\frac{K-\thetaf}{2}\bigr)& \text{ if $u=\vbf$},
\end{cases}
\end{align*}
an thus, with the notation of Corollaries~\ref{cor:KD_G} and~\ref{cor:KD_G_finite}: 

$\bullet$ \emph{Infinite case.}
\begin{align*}\textstyle
\KD(\ubf)^{-1}_{\ubar{v},w}&=e^{-i\frac{\alphafb+\betafb}{2}}\sc(\thetaf)^{\frac{1}{2}}(k')^{-\frac{1}{2}}
\bigl(G^m_{\ubar{v},v_2}-G^m_{\ubar{v},v_1}\bigr)\\
\KD(\ubf)^{-1}_{\ubar{f},w}&\textstyle=-i e^{-i\frac{\alphafb+\betafb}{2}}\sc(\thetaf^*)^{\frac{1}{2}}
\Bigl(\nd\bigl(\frac{K-\thetaf}{2}\bigr)G^{m,*}_{\ubar{f},f_2}-\nd\bigl(\frac{K+\thetaf}{2}\bigr)G^{m,*}_{\ubar{f},f_1} \Bigr).
\end{align*}
$\bullet$ \emph{Finite case.}
\begin{align*}
\KD^\spartial(\ubf)^{-1}_{\ubar{v},w}&\textstyle =e^{-i\frac{\alphafb+\betafb}{2}}(k')^{-\frac{1}{2}}\sc(\thetaf)^{\frac{1}{2}}
\Bigl(G^{m,\spartial}(\ubf)_{v_2,\ubar{v}}-G^{m,\spartial}(\ubf)_{v_1,\ubar{v}}\Bigr)\\
\KD^\spartial(\ubf)^{-1}_{\ubar{f},w}&\textstyle =-ie^{-i\frac{\alphafb+\betafb}{2}}\sc(\thetaf^*)^{\frac{1}{2}}
\Bigl(\II_{\{w\notin\{w^\ell,w^r\in \R^\spartial\}\}}\nd\bigl(\frac{K-\thetaf}{2}\bigr)G^{m,*}_{f_2,\ubar{f}}-
\nd\bigl(\frac{K+\thetaf}{2}\bigr)G^{m,*}_{f_1,\ubar{f}}\Bigl)+\\
&\ \ \  \textstyle+i \sum\limits_{(v^c,f^c)\in\R^{\spartial,\rs}}\frac{\nd(\ubf_{\beta^\ell})}{\cd(\ubf_{\alpha^\ell})}
\bigl(\cd(\ubf_{\beta^r})-\cd(\ubf_{\alpha^\ell})\bigr) \KD^{\spartial}(\ubf)^{-1}_{v^c,w}\cdot G^{m,*}_{f^c,\ubar{f}},
\end{align*}
\begin{align*}
\KD^\spartial(\vbf)^{-1}_{\ubar{v},w}&\textstyle =e^{-i\frac{\alphafb+\betafb}{2}}(k')^{-\frac{1}{2}}\sc(\thetaf)^{\frac{1}{2}}
\Bigl(G^{m,\spartial}(\vbf)_{v_2,\ubar{v}}-G^{m,\spartial}(\vbf)_{v_1,\ubar{v}}\Bigr)\\
\KD^\spartial(\vbf)^{-1}_{\ubar{f},w}&\textstyle =-ie^{-i\frac{\alphafb+\betafb}{2}}\sc(\thetaf^*)^{\frac{1}{2}}
\Bigl(\II_{\{w\notin\{w^\ell,w^r\in \R^\spartial\}\}}\nd\bigl(\frac{K+\thetaf}{2}\bigr)G^{m,*}_{f_2,\ubar{f}}-
\nd\bigl(\frac{K-\thetaf}{2}\bigr)G^{m,*}_{f_1,\ubar{f}}\Bigl)+\\
&\ \ \  \textstyle+i \sum\limits_{(v^c,f^c)\in\R^{\spartial,\rs}}\frac{\nd(\vbf_{\beta^\ell})}{\cd(\vbf_{\alpha^\ell})}
\bigl(\cd(\vbf_{\beta^r})-\cd(\vbf_{\alpha^\ell})\bigr) \KD^{\spartial}(\vbf)^{-1}_{v^c,w}\cdot G^{m,*}_{f^c,\ubar{f}}.
\end{align*}

\end{exm}

\subsubsection{Dimer model on an infinite isoradial double graph $\GD$}

Suppose that the isoradial graph $\Gs$ is infinite; when it is moreover $\ZZ^2$-periodic, we consider the 
natural exhaustion $(\GD_n)_{n\geq 1}$ of $\GD$ by toroidal graphs, where $\GD_n=\GD/n\ZZ^2$. 
Let $\F$ denote the $\sigma$-field generated by cylinder sets of $\GD$.
Using arguments of~\cite{CKP,KOS,deTiliere:quadri},
we obtain an explicit, \emph{local} expression for a Gibbs measure of the dimer model on $\GD$ with 
weight function $c(u)$. Since the proof closely follows that done in the papers~\cite{deTiliere:quadri,BoutillierdeTiliere:iso_gen,BdtR2}, we do not
repeat it here. The key requirements are that the operator $\KD(u)^{-1}$ is local and unique in 
the $\ZZ^2$-periodic case.

\begin{thm}\label{thm:Gibbs_KD}
For every $u\in\Re(\TT(k))$,
there exists a unique probability measure on $(\M(\GD),\F)$, denoted $\PPdimerD$, such that the probability of occurrence of a subset of edges 
$\E=\{w_1x_1,\cdots,w_lx_l\}$ in a dimer configuration of $\GD$ is given by:
\begin{equation*}
\PPdimerD(w_1x_1,\dots,w_lx_l)=\Bigl(\prod_{j=1}^l \KD(u)_{w_j,x_j}\Bigr)\det(\KD(u)^{-1})_\E,
\end{equation*}
where $(\KD(u)^{-1})_\E$ is the sub-matrix of $\KD(u)^{-1}$ given by Corollary~\ref{cor:KD_G}, whose rows are indexed by $x_1,\dots,x_l$
and columns by $w_1,\dots,w_l$. The measure $\PPdimerD$ is a Gibbs measure.
Moreover, when the graph $\GD$ is $\ZZ^2$-periodic, the probability measure $\PPdimerD$ is obtained as weak limit of the 
Boltzmann measures on the toroidal exhaustion $(\GD_n)_{n\geq 1}$.
\end{thm}
\begin{rem}$\,$
\begin{itemize}
\item[$\bullet$] As mentioned in the introduction to Section~\ref{sec:Z_Dirac_Z_Lap}, this theorem is a directed version of the 
transfer impedance theorem of~\cite{BurtonPemantle}. A result in the same flavor, \emph{i.e.}, computing probabilities of 
pairs of directed spanning trees using the massive Green functions of massive~\emph{non-directed} random walks, is obtained 
by~\cite{Chhita} in the case where $\Gs=\ZZ^2$ with specific weights.
\item[$\bullet$] A version of this theorem in the finite case can be obtained using Remark~\ref{rem:Lap}.
\end{itemize}
\end{rem}

\begin{exm}\label{ex:Prob_GD} As an example of application we express the probability of single edges occurring in dimer configurations of $\GD$
chosen with respect to the measure $\PPdimerD$, using the $Z$-massive and dual $Z$-massive Green functions of~\cite{BdTR1}. 
We use the notation of Figure~\ref{fig:corKDG} and omit the subscript ``$\mathrm{f}$'' since there is no confusion
possible between the initial and final vertices.
Details of computations are given in Appendix~\ref{app:dimers_double}.
\begin{align*}
\PPdimerD(wv_2)&=\frac{\sc(\theta)}{k'}
\left[\dn(u_{\alpha})\dn(u_{\beta})G^m_{v_2,v_2}-k'G^m_{v_2,v_1}\right]\\
&=H(2u_{\alpha})-H(2u_{\beta}), \\
\PPdimerD(wf_2)&=\frac{\sc(\theta^*)}{k'}
\left[
\dn((u_{\beta})^*)\dn((u_{\alpha+2K})^*)G^{m,*}_{f_2,f_2}-k'G^{m,*}_{f_2,f_1}\right]\\
&=H(2(u_{\alpha+2K})^*)-H(2(u_{\beta})^*),
\end{align*}
where $H(u|k)=-\frac{ik'}{\pi}\mathrm{A}\bigl(\frac{iu}{2}\vert k'\bigr)$, and $\mathrm{A}(u|k)=\frac{1}{k'}\Bigl(\int_0^u \dc^2(v|k)\ud v + \frac{E-K}{K}u\Bigr)$,
see \cite[(9)]{BdTR1}.
\end{exm}


\section{Kasteleyn operator of the graph $\GQ$ and $Z^u$-Dirac operator} \label{sec:GQ_Z_Dirac}

Let $\Gs$ be an isoradial graph, infinite or finite. We consider the isoradial embedding of the bipartite graph 
$\GQ=(\VQ,\EQ)$ given in Section~\ref{sec:iso_GDGQ}, and 
the weight function $\nu^\Js$ of Equation~\eqref{eq:def_nu_Zinv} arising from the $Z$-invariant Ising model.
Let $\KQ$ be the associated complex, bipartite Kasteleyn matrix defined in Section~\ref{sec:Kast_complex}, with rows indexed by 
black vertices. In the finite case we moreover consider the diagonal matrix $\DQB$, resp. $\DQW$, whose rows/columns are indexed by 
black/white vertices of $\GQ$, and whose diagonal coefficients are:
\begin{align*}
\forall\,\bs\in\Bs,\quad \DQB_{\bs,\bs}&=
\begin{cases}
\sn(\theta^\spartial) & \text{ if $\bs\in\{\bs^\ell,\bs^r\in\R^\spartial\}$}\\
1 & \text{ otherwise},
\end{cases}\\
\forall\,\ws\in\Ws,\quad
\DQW_{\ws,\ws}&=
\begin{cases}
\sn(\theta^\spartial)^{-1} & \text{ if $\ws\in\{\ws^c\in \R^\spartial\}$}\\
1&\text{ otherwise}.
\end{cases}
\end{align*}
Let $\KQu$ be the modified, complex, bipartite Kasteleyn matrix defined by 
\begin{equation}\label{equ:KQ_KQu}
\KQu=\DQB \KQ \DQW,
\end{equation}
that is, $\KQu$ is obtained from $\KQ$ by 
multiplying the weight of the edges $\bs^\ell \ws^\ell$, $\bs^r \ws^r$ of all boundary rhombus pairs of $\R^\spartial$ 
by $\sn(\theta^\spartial)$.

Consider also the isoradial embedding of the double graph $\GD$ of Section~\ref{sec:iso_GDGQ}. For every $u\in\Re(\TT(k))$,
let $\KD(u)$ be the $Z^u$-Dirac operator and, in the finite case, for every $u\in\Re(\TT(k))'$,
let $\KD^\spartial(u)$ be the $Z^u$-Dirac operator with specific boundary conditions, as defined in 
Equation~\eqref{equ:defKD}.

The main result of this section, and actually the key result of this paper, is Theorem~\ref{thm:main} of 
Section~\ref{sec:KQ_KD_relation} proving, for every $u\in\Re(\TT(k))$,
explicit linear relations between the matrices $\KQ$ and $\KD(u)$ in the infinite case, and between $\KQu$ and $\KD^\spartial(u)$ in the finite case.

In Section~\ref{sec:det_KQ_KD} we restrict to the finite case;
using Theorem~\ref{thm:main} and a combinatorial argument, we prove in Theorem~\ref{thm:part_function} that the determinant of the Kasteleyn
matrix $\KQ$ is equal, up to an explicit multiplicative constant depending on $u$, to the determinant of the $Z^u$-Dirac operator $\KD^\spartial(u)$. 
Combining this with Theorem~\ref{thm:det} proves that the determinant of $\KQ$ is equal, up to an explicit constant, to the 
determinant of the $Z$-massive Laplacian $\Delta^{m,\spartial}(u)$. Interpreting these determinants as partition functions, 
Theorem~\ref{thm:part_function} proves that the squared partition function of the $Z$-invariant Ising model with $+$ boundary conditions
is equal, up to an explicit constant, to the partition function of weighted rooted directed spanning forests counted by
$\Delta^{m,\spartial}(u)$, where the dependence in $u$ is along the boundary only. This generalizes to the full $Z$-invariant case the results
of~\cite{deTiliere:mapping,deTiliere:partition} proved in the $Z$-invariant \emph{critical} case, and to the case of simply connected domains 
the result of~\cite{BdtR2} proved in the toroidal case. 
The proof we provide here has a slight combinatorial flavor but is mainly based on matrix relations, 
so quite different from~\cite{deTiliere:mapping,deTiliere:partition}. Note that the combinatorics argument of~\cite{deTiliere:partition} can be generalized to the full $Z$-invariant case and would 
give an alternative proof. Note also that the boundary trick of Chelkak and Smirnov~\cite{ChelkakSmirnov:ising} allows us 
to remove dual trees along the boundary which we could not do in~\cite{deTiliere:partition}.

Using Theorem~\ref{thm:main}, Corollaries~\ref{cor:KD_KQ} and~\ref{cor:KD_KQ_finite}
of Section~\ref{sec:inv_KQ_KD} prove linear relations between the inverse operator $(\KQ)^{-1}$ and the inverse $Z^u$-Dirac operator. 
Choosing specific values of $u$ allows us to express
the dimer measure of the graph $\GQ$ using the inverse $Z^u$-Dirac operator and the $Z$-massive Green functions,
see Corollaries~\ref{cor:gibbsGQ_GD} and~\ref{cor:BoltzmannGQ_GD}. This also provides an alternative direct way of finding a local 
formula for $(\KQ)^{-1}$~\cite{BdtR2}, explicitly relating it to the $Z$-massive Green functions.

\subsection{Relating the Kasteleyn operator $\KQ$ and the $Z^u$-Dirac operator}\label{sec:KQ_KD_relation}

The main result of this section is Theorem~\ref{thm:main}
proving an explicit relation between the matrices $\KQ$ and $\KD(u)$ in the infinite case, and between $\KQu$ and $\KD^\spartial(u)$
in the finite case. In order to state this theorem we need to introduce two additional matrices $S(u)$ and $T(u)$.
Both of them are ``rectangular'' with ``twice'' more rows than columns. 

The \emph{matrix $S(u)$} has rows indexed by black vertices of $\GQ$ and columns by white vertices of $\GD$, resp. of $\GDro$,
in the \emph{infinite} case, resp. \emph{finite} case. If $\GQ$ is \emph{infinite}, let 
$\bs$ be a black vertex; if $\GQ$ is \emph{finite},
let $\bs$ be a black vertex of an inner quadrangle. 
Let $w$ be the 
white vertex of $\GD$ corresponding to the quadrangle 
to which $\bs$ belongs. Then, the only non-zero coefficient of the row corresponding to $\bs$ is:
\begin{equation}\label{def:sbw}
s(u)_{\bs,w}=e^{-i\frac{\bar{\beta}}{2}}\cn(u_\beta)[\sn(\theta)\cn(\theta)\nd(u_\alpha)\nd(u_\beta)]^{\frac{1}{2}},
\end{equation}
with the following notation, see Figure~\ref{fig:def_ST} (left): $\ws$ is the white vertex of $\GQ$ such that 
$\bs\ws$ is parallel to an edge of $\Gs$; $e^{i\bar{\alpha}}$, $e^{i\bar{\beta}},\bar{\theta}$ are the rhombus vectors and half-angle
of $\GRR$ assigned to the edge $(\bs,\ws)$.

Suppose that $\GQ$ is \emph{finite}, and let $\bs$ be a black vertex of a boundary quadrangle of $\GQ$. 
Then $\bs\in\{\bs^r,\bs^\ell\}$ for some boundary rhombus pair of $\R^\spartial$,
see Figure~\ref{fig:def_ST} (right), and the non-zero coefficient of the row corresponding to $\bs$ is:
\begin{align}\label{def:sbl_boundary}
\begin{split}
s(u)_{\bs^r,w^r}&=e^{-i\frac{\bar{\beta}^r}{2}}\cn(u_{\beta^r})[\sn(\theta^\spartial)\cn(\theta^\spartial)\nd(u_{\alpha^r})\nd(u_{\beta^r})]^{\frac{1}{2}}\\
s(u)_{\bs^\ell,w^\ell}&=e^{-i\frac{\bar{\alpha}^\ell}{2}}\cn(u_{\alpha^\ell})[\sn(\theta^\spartial)\cn(\theta^\spartial)
\nd(u_{\alpha^\ell})\nd(u_{\beta^\ell})]^{\frac{1}{2}},
\end{split}
\end{align}
where for $\bs^r$ the definition is coherent with that of the non-boundary case, since the rhombus vectors of $\GRR$ assigned to the edge
$(\bs^r,\ws^r)$ are $e^{i\bar{\alpha}^r},\,e^{i\bar{\beta}^r}$. For $\bs^\ell$, the rhombus vectors assigned to the edge
$(\bs^\ell,\ws^\ell)$ are $e^{i\bar{\alpha}^\ell},\,e^{i\bar{\beta}^\ell}$ so that there is a change of definition; 
this comes from our choice of embedding of $\GQ$ which exchanges the bipartite coloring of vertices in the left rhombus of the pair.

\begin{figure}[ht]
\begin{minipage}[b]{0.5\linewidth}
\begin{center}
\begin{overpic}[width=5cm]{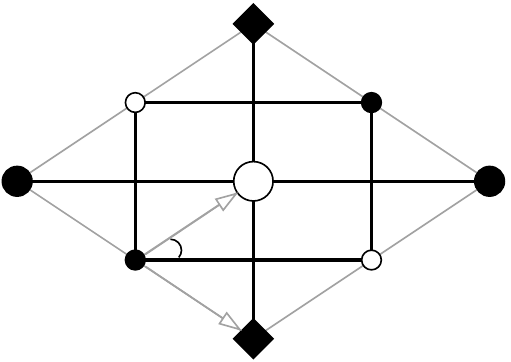}
  \put(22,15){\scriptsize $\bs$}
  \put(74,15){\scriptsize $\ws$}
  \put(31,6.5){\scriptsize $e^{i\bar{\alpha}}$}
  \put(31,26){\scriptsize $e^{i\bar{\beta}}$}
  \put(101,34){\scriptsize $v$}
  \put(47,-5){\scriptsize $f$}
  \put(51,28){\scriptsize $w$}
  \put(37,21){\scriptsize $\bar{\theta}$}
\end{overpic}
\end{center}
\end{minipage}
\begin{minipage}[b]{0.5\linewidth}
\begin{center}
\begin{overpic}[width=3.8cm]{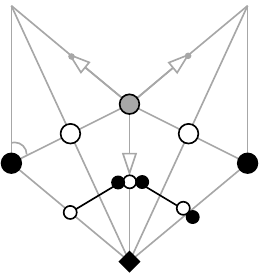}
\put(-7,40){\scriptsize $v^\ell$}
\put(44,68){\scriptsize $v^c$}
\put(95,39){\scriptsize $v^r$}
\put(14,53){\scriptsize $w^\ell$}
\put(72,52){\scriptsize $w^r$}
\put(22,16){\scriptsize $\ws^\ell$}
\put(36,24){\scriptsize $\bs^\ell$}
\put(43,26){\scriptsize $\ws^c$}
\put(52,23){\scriptsize $\bs^r$}
\put(66,28){\scriptsize $\ws^r$}
\put(44,-5){\scriptsize $f^c$}
\put(26,64){\scriptsize $e^{i\bar{\alpha}^\ell}$}
\put(34,49){\scriptsize $e^{i\bar{\beta}^\ell}$}
\put(48,49){\scriptsize $e^{i\bar{\alpha}^r}$}
\put(56,64){\scriptsize $e^{i\bar{\beta}^r}$}
\put(6,48.5){\scriptsize $\bar{\theta}^\spartial$}
\end{overpic}
\end{center}
\end{minipage}
\caption{Notation around a black vertex $\bs$ and a white vertex $\ws$ of $\GQ$: non-boundary case
(left), boundary case (right).}
\label{fig:def_ST}
\end{figure}
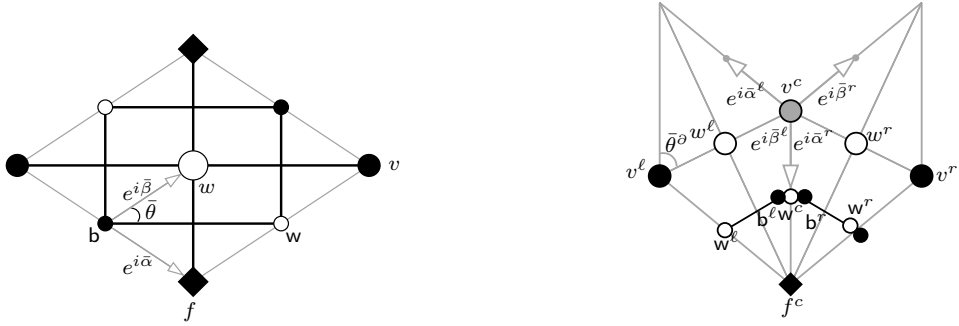

The \emph{matrix $T(u)$} has rows indexed by white vertices of $\GQ$ and columns 
by black vertices of $\GD$, resp. of $\GDro$, in the \emph{infinite} case, resp. \emph{finite} case. 
If $\GQ$ is \emph{infinite}, let $\ws$ be a white vertex; if $\GQ$ is \emph{finite}, let $\ws$ be a white vertex 
such that $\ws\neq\ws^c$ for all boundary rhombus pairs of $\R^\spartial$.
The vertex $\ws$ is on a rhombus edge $vf$ of the diamond
graph $\GR$, where $v$ is a vertex of $\Gs$ and $f$ a vertex of $\Gs^*$, see Figure~\ref{fig:def_ST} (left). 
Then the row of $T(u)$ corresponding to $\ws$ has two non-zero coefficients defined by,
\begin{align}
\label{def:twvf}
\begin{split}
&t(u)_{\ws,v}=e^{-i\frac{\bar{\beta}}{2}}\cn(u_\beta)\\
&t(u)_{\ws,f}=e^{-i\frac{\bar{\beta}+\pi}{2}}\sn((u_{\beta+2K})^*)=e^{-i\frac{\bar{\beta}+\pi}{2}}\cd(u_{\beta+2K}),
\end{split}
\end{align}
where $e^{i\bar{\beta}}$ is the rhombus vector $(\ws,v)$ and $e^{i(\bar{\beta}+\pi)}$ is the rhombus vector $(\ws,f)$.

Suppose that $\GQ$ is \emph{finite}, and let $\ws=\ws^c$ for some boundary rhombus pair of $\R^{\spartial}$, 
then $\ws^c$ is on a rhombus edge $v^c f^c$
of $\GR$, see Figure~\ref{fig:def_ST} (right). As long as the rhombus pair is not the root one, \emph{i.e.}, the one where $v^c=\rs$,
the row of $T(u)$ corresponding to $\ws^c$ has non-zero coefficients
given by,
\begin{align}\label{def:twv_boundary}
\begin{split}
t(u)_{\ws^c,v^c}&=
-i k' e^{-i\frac{\bar{\alpha}^r}{2}}\sn(\theta^\spartial)\nd(u_{\alpha^r})\cd(u_{\beta^r}),\\
t(u)_{\ws^c,f^c}&=e^{-i\frac{\bar{\beta}^\ell}{2}}\sn((u_{\beta^\ell})^*)=
e^{-i\frac{\bar{\beta}^\ell}{2}}\cd(u_{\beta^\ell}).
\end{split}
\end{align}
When the boundary rhombus pair is the root pair, then only the term $t(u)_{\ws^c,f^c}$ is defined.
Note that the definition is specific for $v^c$, and coherent with the non-boundary case for $f^c$ since the rhombus vector 
$(\ws^c,f^c)$ is $e^{i\bar{\beta}^\ell}$.

\begin{rem}
Rhombus angles are well defined mod $2\pi$ implying that half-angles are well defined mod $4\pi$,
but the coefficients of $S(u)$ and $T(u)$ are nevertheless well defined. Indeed, keeping
in mind that by definition $\frac{\bar{\beta}-\bar{\alpha}}{2}\in(\eps,\frac{\pi}{2}-\eps)$, one has 
for example, see Equation~\eqref{def:sbw}: 
\begin{align*}
e^{-i\frac{\bar{\beta}+2\pi}{2}}\cn(u_{\beta+4K})[\sn(\theta)\cn(\theta)\nd(u_{\alpha+4K})\nd(u_{\beta+4K})]^{\frac{1}{2}}&=\\
=\bigl(-e^{-i\frac{\bar{\beta}}{2}}\bigr)\bigl(-\cn(u_\beta)\bigr)&
  [\sn(\theta)\cn(\theta)\nd(u_\alpha)\nd(u_\beta)]^{\frac{1}{2}}=s_{\bs,w},
\end{align*}
using that $u_{\beta+4K}=u_{\beta}-2K$ and $\cn(u-2K)=-\cn(u)$, $\nd(u-2K)=\nd(u)$. 
Similar arguments hold for other coefficients.
\end{rem}

We are now ready to state our main result.

\begin{thm}\label{thm:main}$\,$
\begin{itemize}
\item[$\bullet$] \emph{Infinite case.} Let $u\in\Re(\TT(k))$, then the Kasteleyn matrix $\KQ$, the $Z^u$-Dirac operator
$\KD(u)$ and the matrices $S(u)$, $T(u)$ are related by the following identity:
\begin{equation}
\KQ\,T(u) = S(u)\, \KD(u). \label{equ:thm_main_3}
\end{equation}
\item[$\bullet$] \emph{Finite case.} Let $u\in\Re(\TT(k))'$, then the Kasteleyn matrix $\KQu$, the $Z^u$-Dirac operator $\KD^\spartial(u)$
and the matrices $S(u)$, $T(u)$ are related by the following identity:
\begin{equation}
\KQu\,T(u) = S(u)\, \KD^\spartial(u). \label{equ:thm_main_4}
\end{equation}
\end{itemize}
\end{thm}

\begin{proof}
In the whole of the proof, we omit the argument $u$ from the matrices.

\emph{Infinite case and finite non-boundary case}. Figure~\ref{fig:notations} below sets the notation. 
Let $\bs$ be a black vertex of $\GQ$, then $\bs$ belongs to a quadrangle corresponding to a vertex $w$ of $\GD$.
If $\GQ$ is finite, suppose further that the quadrangle is not a boundary one, or equivalently that $w$ is not a boundary vertex 
of $\GD$. Let $v_1,f_1,v_2,f_2$ be the four black vertices
of $\GD$ incident to $w$. 
Denote by $\ws_1,\ws_2,\ws_3$ the three white vertices of $\GQ$ incident to $\bs$, and 
let $e^{i\bar{\alpha}}$, $e^{i\bar{\beta}}$ be the rhombus vectors of the edge $(\bs,\ws_1)$.

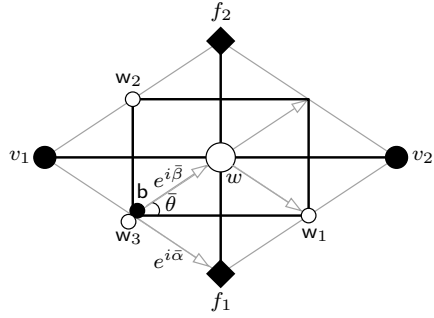
\begin{figure}[ht]
\centering
\begin{overpic}[width=5cm]{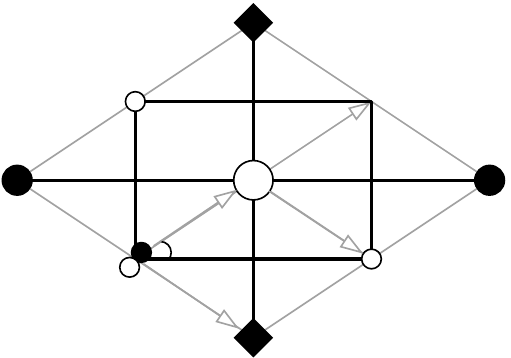}
  \put(27.5,24){\scriptsize $\bs$}
  \put(71,13.5){\scriptsize $\ws_1$}
  \put(22,12.5){\scriptsize $\ws_3$}
  \put(22,54.5){\scriptsize $\ws_2$}
  \put(32,5.5){\scriptsize $e^{i\bar{\alpha}}$}
  \put(32,27){\scriptsize $e^{i\bar{\beta}}$}
  \put(-6,34){\scriptsize $v_1$}
  \put(100,34){\scriptsize $v_2$}
  \put(47,-5){\scriptsize $f_1$}
  \put(47,72){\scriptsize $f_2$}
  \put(51,28){\scriptsize $w$}
  \put(35,20.5){\scriptsize $\bar{\theta}$}
\end{overpic}
\caption{Notation around a black vertex $\bs$ of $\GQ$.}
\label{fig:notations}
\end{figure}

The coefficient $[\KQu\,T]_{\bs,x}$ of the LHS of~\eqref{equ:thm_main_3} and~\eqref{equ:thm_main_4}  is 
non-zero only when $x\in\{v_1,v_2,f_1,f_2\}$, and
\begin{align*}
[\KQu\,T]_{\bs,v_1}=&\KQ_{\bs,\ws_2}t_{\ws_2,v_1}+ \KQ_{\bs,\ws_3}t_{\ws_3,v_1}\\
[\KQu\,T]_{\bs,v_2}=&\KQ_{\bs,\ws_1}t_{\ws_1,v_2}\\
[\KQu\,T]_{\bs,f_1}=&\KQ_{\bs,\ws_1}t_{\ws_1,f_1}+ \KQ_{\bs,\ws_3}t_{\ws_3,f_1}\\
[\KQu\,T]_{\bs,f_2}=&\KQ_{\bs,\ws_2}t_{\ws_2,f_2}.
\end{align*}
The coefficients of $T$ involved are, see definition~\eqref{def:twvf}, 
\begin{align*}
&t_{\ws_2,v_1}=-ie^{-i\frac{\bar{\beta}}{2}}\cn(u_{\beta+2K}),\,
t_{\ws_3,v_1}=-ie^{-i\frac{\bar{\alpha}}{2}}\cn(u_{\alpha+2K}),\,
t_{\ws_1,v_2}=e^{-i\frac{\bar{\beta}}{2}}\cn(u_{\beta})\\
&t_{\ws_1,f_1}=ie^{-i\frac{\bar{\beta}}{2}}\sn((u_{\beta-2K})^*),\,
t_{\ws_3,f_1}=e^{-i\frac{\bar{\alpha}}{2}}\sn((u_{\alpha})^*),\,
t_{\ws_2,f_2}=e^{-i\frac{\bar{\beta}}{2}}\sn((u_{\beta})^*).
\end{align*}
Replacing coefficients of $\KQ$ yields for the LHS:
\begin{align*}
[\KQu\,T]_{\bs,v_1}=&e^{i\frac{\bar{\beta}+\bar{\alpha}+\pi}{2}}\cn(\theta) t_{\ws_2,v_1}+e^{i\frac{\bar{\alpha}-\pi+\bar{\alpha}}{2}}t_{\ws_3,v_1}=
e^{i\frac{\bar{\alpha}}{2}}[\cn(\theta)\cn(u_{\beta+2K})-\cn(u_{\alpha+2K})]\\
[\KQu\,T]_{\bs,v_2}=&e^{i\frac{\bar{\beta}+\bar{\alpha}}{2}}\sn(\theta) t_{\ws_1,v_2}=
e^{i\frac{\bar{\alpha}}{2}}\sn(\theta)\cn(u_{\beta})\\
[\KQu\,T]_{\bs,f_1}=&e^{i\frac{\bar{\beta}+\bar{\alpha}}{2}}\sn(\theta)t_{\ws_1,f_1}+ e^{i\frac{\bar{\alpha}-\pi+\bar{\alpha}}{2}}t_{\ws_3,f_1}
=ie^{i\frac{\bar{\alpha}}{2}}[\sn(\theta)\sn((u_{\beta-2K})^*)-\sn((u_{\alpha})^*)]\\
[\KQu\,T]_{\bs,f_2}=&e^{i\frac{\bar{\beta}+\bar{\alpha}+\pi}{2}}\cn(\theta) t_{\ws_2,f_2}=
ie^{i\frac{\bar{\alpha}}{2}}\cn(\theta)\sn((u_{\beta})^*).
\end{align*}
Note that in the \emph{finite case}, we have $\KD^\spartial=\KD$ for coefficients involved.
The coefficient $[S\, \KD]_{\bs,x}$ of the RHS of~\eqref{equ:thm_main_3} and \eqref{equ:thm_main_4} is also non-zero when $x\in\{v_1,v_2,f_1,f_2\}$
and we have $[S\, \KD]_{\bs,x}=s_{\bs,w}\KD_{w,x}$. 
To compute these terms, we
first express the part $\cn(u_\beta)[\nd(u_\alpha)\nd(u_\beta)]^{\frac{1}{2}}:=(\diamond)$
in $s_{\bs,w}$, see~\eqref{def:sbw}, using the angles and parameters involved in the four coefficients of $\KD$.
\begin{align*}
(\diamond)&=-\sn(u_{\beta+2K})[\dn(u_{\alpha+2K})\nd(u_{\beta+2K})]^{\frac{1}{2}}\\
&=\cn(u_\beta)[\nd(u_\alpha)\nd(u_\beta)]^{\frac{1}{2}}\\
&=(k')^{-\frac{1}{2}}\cn((u_{\beta-2K})^*)[\nd((u_{\beta-2K})^*)\dn((u_\alpha)^*)]^{\frac{1}{2}}\\
&=(k')^{\frac{1}{2}}\sn((u_{\beta})^*)[\nd((u_\beta)^*)\nd((u_{\alpha+2K})^*)]^{\frac{1}{2}},
\end{align*}
where in the first line we used that $u_{\alpha}=K+u_{\alpha+2K}$, in the third that
$u_\beta=-(u_{\beta-2K})^*$, $u_{\alpha}=K-(u_\alpha)^*$, and in the fourth that
$u_\beta=K-(u_{\beta})^*$, $u_{\alpha}=2K-(u_{\alpha+2K})^{*}$. 
Replacing coefficients of $\KD$ by their definition gives for the RHS 
of~\eqref{equ:thm_main_3},
\begin{align*}
[S\, \KD]_{\bs,v_1}&=s_{\bs,w}\KD_{w,v_1}=e^{i\frac{\bar{\alpha}}{2}} \sn(\theta)\sn(u_{\beta+2K})\dn(u_{\alpha+2K})\\
[S\, \KD]_{\bs,v_2}&=s_{\bs,w}\KD_{w,v_2}=e^{i\frac{\bar{\alpha}}{2}} \sn(\theta)\cn(u_\beta)\\
[S\, \KD]_{\bs,f_1}&=s_{\bs,w}\KD_{w,f_1}=-i(k')^{-1}e^{i\frac{\bar{\alpha}}{2}}\cn(\theta)\cn((u_{\beta-2K})^*)\dn((u_\alpha)^*)\\
[S\, \KD]_{\bs,f_2}&=s_{\bs,w}\KD_{w,f_2}=ie^{i\frac{\bar{\alpha}}{2}}\cn(\theta)\sn((u_{\beta})^*).
\end{align*}
The equality $[\KQ\,T]_{\bs,x}=[S\, \KD]_{\bs,x}$ is then straightforward when $x=v_2$ and $x=f_2$.
When $x=v_1$, this is a consequence of the identity, see~\cite[chap.2, ex.32 (i)]{Lawden},
\begin{equation}\label{equ:Lawden1}
\cn u\cn v-\cn(u+v)=\sn u\sn v\dn(u+v),
\end{equation}
evaluated at $u=\theta,\,v=u_{\beta+2K},\,u+v=u_{\alpha+2K}$.

We are left with proving the case $x=f_1$. Multiplying the identity~\eqref{equ:Lawden1} by $\nd(u-v)$, using that
$\nd(u-v)=(k')^{-1}\dn(K-(u-v))$, $\cd(u-v)=\sn(K-(u-v))$ we obtain,
\[
(k')^{-1}\cn u\cn v\dn(K-(u-v))=\sn(K-(u-v))-\sn u \sn v.
\]
The proof is concluded by evaluating the above at $u=\theta$, $v=(u_{\beta-2K})^*$, $K-(u-v)=(u_{\alpha})^*$.

\emph{Finite boundary case.} The notation used are those of Figure~\ref{fig:def_ST} (right).
Let $\bs\in\{\bs^\ell,\bs^r\}$ be a black vertex of some boundary rhombus pair of $\R^{\spartial}$.
Suppose first that this rhombus pair is not the root one. Then,
if $\bs=\bs^r$, resp. $\bs=\bs^\ell$, the coefficients of the LHS and RHS of~\eqref{equ:thm_main_4} 
are non-zero when $x\in\{v^c,v^r,f^c\}$, resp. $x\in\{v^c,v^\ell,f^c\}$. We need to prove:
\[
\begin{array}{llcl}
&\KQu_{\bs^r,\ws^c}t_{\ws^c,v^c}=s_{\bs^r,w^r}\KD^\spartial_{w^r,v^c}&&\KQu_{\bs^\ell,\ws^c}t_{\ws^c,v^c}=s_{\bs^\ell,w^\ell}\KD^\spartial_{w^\ell,v^c}\\
&\KQu_{\bs^r,\ws^r}t_{\ws^r,v^r}=s_{\bs^r,w^r}\KD^\spartial_{w^r,v^r}& \text{ resp. }\,\,&\KQu_{\bs^\ell,\ws^\ell}t_{\ws^\ell,v^\ell}=s_{\bs^\ell,w^\ell}\KD^\spartial_{w^\ell,v^\ell}\\
&\KQu_{\bs^r,\ws^r}t_{\ws^r,f^c}+ \KQu_{\bs^r,\ws^c}t_{\ws^c,f^c}=s_{\bs^r,w^r}\KD^\spartial_{w^r,f^c}&&
\KQu_{\bs^\ell,\ws^\ell}t_{\ws^\ell,f^c}+ \KQu_{\bs^\ell,\ws^c}t_{\ws^c,f^c}=s_{\bs^\ell,w^\ell}\KD^\spartial_{w^\ell,f^c}.
\end{array}
\]
In both cases, the last two equalities are as in the full plane with the appropriate change of notation.
We thus need to check the first equality of each case. Returning to the definition of $\KQu$ and boundary values of the matrices $S$ and $T$ defined in
Equations~\eqref{def:sbl_boundary} and\eqref{def:twv_boundary} we have,
\begin{align*}
\frac{\KQu_{\bs^r,\ws^c}t_{\ws^c,v^c}}{s_{\bs^r,w^r}}&=
\frac{e^{i\frac{\bar{\alpha}^r-\pi+\bar{\alpha}^r}{2}}\times (-ik'e^{-i\frac{\bar{\alpha}^r}{2}}\sn(\theta^\spartial)\nd(u_{\alpha^r})\cd(u_{\beta^r}))}
{e^{-i\frac{\bar{\beta}^r}{2}}\cn(u_{\beta^r})[\sn(\theta^\spartial)\cn(\theta^\spartial)\nd(u_{\alpha^r})\nd(u_{\beta^r})]^{\frac{1}{2}}}
\\
&=-k'e^{i\frac{\bar{\alpha}^r+\bar{\beta}^r}{2}}[\sc(\theta^\spartial)\nd(u_{\alpha^r})\nd(u_{\beta^r})]^{\frac{1}{2}}\\
&=e^{i\frac{\bar{\alpha}^r+\pi+\bar{\beta}^r+\pi}{2}}[\sc(\theta^\spartial)\dn(u_{\alpha^r+2K})\dn(u_{\beta^r+2K})]^{\frac{1}{2}}=
\KD_{w^r,v^c}=\KD^\spartial_{w^r,v^c}.
\end{align*}
In a similar way,
\begin{align*}
\frac{\KQu_{\bs^\ell,\ws^c}t_{\ws^c,v^c}}{s_{\bs^\ell,w^\ell}}&=
\frac{e^{i\frac{\bar{\beta}^\ell+\bar{\beta}^\ell+\pi}{2}}\times (-ik'e^{-i\frac{\bar{\alpha}^r}{2}}\sn(\theta^\spartial)\nd(u_{\alpha^r})\cd(u_{\beta^r}))}
{e^{-i\frac{\bar{\alpha}^\ell}{2}}\cn(u_{\alpha^\ell})[\sn(\theta^\spartial)\cn(\theta^\spartial)\nd(u_{\alpha^\ell})\nd(u_{\beta^\ell})]^{\frac{1}{2}}}\\
&=k'e^{i\frac{\bar{\alpha}^\ell+\bar{\beta}^\ell+(\bar{\beta}^\ell-\bar{\alpha}^r)}{2}}
[\sc(\theta^\spartial)\dn(u_{\alpha^\ell})\dn(u_{\beta^\ell})]^{\frac{1}{2}}
\nd(u_{\alpha^r})\cd(u_{\beta^r})\nc(u_{\alpha^\ell}),\\
&=-k'e^{i\frac{\bar{\alpha}^\ell+\bar{\beta}^\ell}{2}}
[\sc(\theta^\spartial)\nd(u_{\alpha^\ell})\nd(u_{\beta^\ell})]^{\frac{1}{2}}\cd(u_{\beta^r})\dc(u_{\alpha^\ell}), 
\text{ since $\bar{\beta}^\ell-\bar{\alpha}^r=2\pi$}\\
&=e^{i\frac{\bar{\alpha}^\ell+\pi+\bar{\beta}^\ell+\pi}{2}}
[\sc(\theta^\spartial)\dn(u_{\alpha_{\ell+2K}})\nd(u_{\beta_{\ell+2K}})]^{\frac{1}{2}}\cd(u_{\beta^r})\dc(u_{\alpha^\ell})\\
&=\KD_{w^\ell,v^c} \cd(u_{\beta^r})\dc(u_{\alpha^\ell})=\KD^\spartial_{w^\ell,v^c}.
\end{align*}
Note that we always have $s_{\bs^r,w^r},s_{\bs^\ell,w^\ell}\neq 0$ on $\Re(\TT(k))'$, see~\eqref{equ:Re_Tk'},
so that it makes sense to divide by these quantities.

The last case we need to consider is if $\bs$ belongs to the root boundary rhombus pair of $\R^{\spartial}$.
But then, of the three equations above, the first one is absent since we have $v^c=\rs$, so we are left with the last two which are as in 
the full plane case. This ends the proof of~\eqref{equ:thm_main_4} and thus finishes the proof of Theorem~\ref{thm:main}.
\end{proof}


\subsection{Determinants of the Kasteleyn matrix $\KQ$ and of the $Z^u$-Dirac operator}\label{sec:det_KQ_KD}

We restrict to the finite case. Theorem~\ref{thm:part_function} proves that the determinants of the matrices $\KQ$ and $\KD^\spartial(u)$
are equal up to an explicit multiplicative constant depending on $u$. By~\cite{Dubedat}, see also Equation~\eqref{equ:PF_Ising_2}, the determinant 
of $\KQ$ is equal, up to a constant, to the squared partition function of the Ising model with + boundary conditions. Using Corollary~\ref{cor:det},
the determinant of $\KD^\spartial(u)$ is equal, up to an explicit constant, to the determinant of $\Delta^{m,\spartial}(u)$ which counts
weighted rooted directed spanning forests. Theorem~\ref{thm:part_function} thus implies identities between partition functions
made explicit in Corollary~\ref{cor:partition_function}.

In the whole of this section, we restrict the domain of $u$ to
\begin{align}
\Re(\TT(k))''&=\Re(\TT(k))\setminus\{\alpha^\ell[2K],\ \beta^{\ell}[2K]:e^{i\bar{\alpha}^\ell},\,e^{i\bar{\beta}^\ell}\in\R^\spartial\}
\label{equ:Re_Tk_2}\\
&=\Re(\TT(k))\setminus\{\alpha^r[2K],\ \beta^r[2K]:\,e^{i\bar{\alpha}^r},\,e^{i\bar{\beta}^r}\in\R^\spartial\}.\nonumber 
\end{align}
Note that by Section~\ref{sec:train_tracks} on train-tracks, this amounts to removing all the parallel directions of the train-tracks of the isoradial graph
$\Gs$.

\subsubsection{Results}

Every edge $e=(w,x)$ of $\GDro$ is assigned two rhombus vectors $e^{i\bar{\alpha}_e}$, $e^{i\bar{\beta}_e}$
of the diamond graph $\GRR$ and a half-angle $\bar{\theta}_e$. We partition the set of white vertices $W^\rs=W$ of $\GDro$ as $W^\spartial\cup W^\scirc$, where 
$W^\spartial$ consists of boundary vertices of $W$ and $W^\scirc$ of inner ones.
Every white vertex $w$ of $\GDro$ is in the
center of a rhombus of the diamond graph $\GR$; we let $\bar{\theta}_w$ be the half-angle of this rhombus at one of the two 
\emph{primal} vertices. 

In the statement below, a specific role is played by the boundary rhombus pair of $\R^\spartial$ containing the root $\rs$.
We use the notation
of Figure~\ref{fig:def_ST} (right) and add a superscript $\rs$ to specify vertices/angles of this root pair.

\begin{thm}\label{thm:part_function}
Let $\Ms_1$ be a dimer configuration of $\GDro$. Then, for every $u\in\Re(\TT(k))''$, we have
\begin{align*}
|\det\KQ|=\tilde{C}(u)\cdot|\det \KD^\spartial(u)|,
\end{align*}
where $\tilde{C}(u)$ is equal to:
\begin{equation*}
\textstyle
(k')^{\frac{|\Vs^*|}{2}}\Bigl(\prod\limits_{w\in W^\spartial}\sn(\theta_w)^{-1} \Bigr)
\Bigl(\prod\limits_{w\in W}[\cn(\theta_w)\sn(\theta_w)]^{\frac{1}{2}}\Bigr) 
\Bigl(\prod\limits_{wx\in\Ms_1:w\in W}|\sc(u_{\alpha_e})\sc(u_{\beta_e})|^{\frac{1}{2}}\Bigr)
\sn(\theta^{\spartial,\rs})\left|\frac{\cd(u_{\beta^{r,\rs}})}{\sn(u_{\alpha^{r,\rs}})}\right|.
\end{equation*}
\end{thm}

\begin{rem}\label{rem:part_function}$\,$
\begin{itemize}
 \item[$\bullet$] Since for $u\in\Re(\TT(k))''$, $\sc(u_{\alpha_e})\sc(u_{\beta_e})\neq 0$, and since $\det\KQ\neq 0$, we have 
 that the matrix $\KD^\spartial(u)$ is invertible for these values of $u$.
 \item[$\bullet$] The quantity $\prod_{wx\in\Ms_1:w\in W}|\sc(u_{\alpha_e})\sc(u_{\beta_e})|^{\frac{1}{2}}$ is independent 
 of the choice of perfect matching $\Ms_1$. Similarly to Remark~\ref{rem:partition_funciont_KD_Lap}, this consists in showing that 
 the alternating product around every inner face of $\GDro$ is equal to 1; the proof is similar to that of 
 Lemma~\ref{lem:K_Kg}. 
 \item A surprising fact is that the LHS is independent of $u$ while the RHS does not seem to be, also the RHS seems to depend on 
 the choice of root vertex $\rs$. It is not straightforward to see why this indeed not the case. An alternative way of proving
 this theorem is to extend to the full $Z$-invariant case the combinatorial argument of~\cite{deTiliere:partition}; one then better sees the 
 parameter $u$ and the choice of root vertex $\rs$ appearing.
\end{itemize}
\end{rem}

Combining Theorems~\ref{thm:part_function} and~\ref{thm:det}, we deduce that the squared partition function of the $Z$-invariant Ising model
with + boundary conditions is equal, up to an explicit constant, to the determinant of the massive Laplacian $\Delta^{m,\spartial}(u)$,
\emph{i.e.}, to the partition function of rooted directed spanning forests. This generalizes to the full $Z$-invariant case 
the result of~\cite{deTiliere:mapping,deTiliere:partition} and to simply connected domains the result of~\cite{BdtR2} proved 
for the characteristic polynomial in the toroidal case in an abstract way.
\begin{cor}\label{cor:partition_function}
For every $u\in\Re(\TT(k))''$, we have
\begin{align*}\textstyle 
1.\ [\Zising^+(\Gs,\Js)]^2 =&\textstyle2^{|\Vs|-1}(k')^{|\Es|}
\Bigl(\prod\limits_{w\in W^\spartial}\frac{1+\sn(\theta_w)}{2\sn(\theta_w)}\Bigr)
\bigl(\prod\limits_{wx\in\Ms_1:w\in W}|\sc(u_{\alpha_e})\sc(u_{\beta_e})\nd(u_{\alpha_e})\nd(u_{\beta_e})|^{\frac{1}{2}}\bigr)\times\\
\textstyle &\times \sn(\theta^{\spartial,\rs})\left|\frac{\cd(u_{\beta^{r,\rs}})}{\sn(u_{\alpha^{r,\rs}})}\right|
|\det \Delta^{m,\spartial}(u)|.\\
\textstyle 2.\ [\Zising^+(\Gs,\Js)]^2=&\textstyle2^{|\Vs|-1}(k')^{-|\Es|}
\Bigl(\prod\limits_{w\in W^\spartial}\frac{1+\sn(\theta_w)}{2\sn(\theta_w)}\Bigr)
\sn(\theta^{\spartial,\rs})
\left|\frac{\cd(u_{\beta^{r,\rs}})\sn(u_{\beta^{r,\rs}})}{
\cd(u_{\alpha^{r,\rs}})\sn(u_{\alpha^{r,\rs}})}\right|^{\frac{1}{2}}\times\\
&\textstyle \times |\det \Delta^{m,\spartial}(u)\det \Delta^{m,\spartial}(u+2K)|^{\frac{1}{2}}.
\end{align*}
\end{cor}
\begin{proof}
Equality 2. is obtained by writing $[\Zising^+(\Gs,\Js)]^2=([\Zising^+(\Gs,\Js)]^2[\Zising^+(\Gs,\Js)]^2)^{\frac{1}{2}}$, using 
Equality 1. evaluated once at $u$ and once at $u+2K$, and using the identities $|\sc(u-K)|=|(k')^{-1}\cs(u)|$, 
$|\nd(u-K)|=|(k')^{-1}\dn(u)|$, $|\cd(u-K)|=|\sn(u)|$, $|\sn(u-K)|=\cd(u)$.

Let us prove Equality 1. From Equation~\eqref{equ:PF_Ising_2} we have, for all coupling constants $\Js$,
\begin{equation*}
\textstyle
[\Zising^+(\Gs,\Js)]^2=2^{|\Vs|-1}
\Bigl(\prod\limits_{e\in\Es^\spartial} \frac{e^{2\Js_e}}{2}\Bigr) 
\Bigl(\prod\limits_{e^*\in\Es^*}\cosh(2\Js_e)\Bigr) \Zdimer(\GQ,\nuJ).
\end{equation*}
Returning to the definition of the $Z$-invariant weights~\eqref{equ:def_J_Zinv} and~\eqref{eq:def_nu_Zinv} gives,
\begin{align*}
\textstyle 
\Bigl(\prod\limits_{e\in\Es^\spartial} \frac{e^{2\Js_e}}{2}\Bigr) 
\Bigl(\prod\limits_{e^*\in\Es^*}\cosh(2\Js_e)\Bigr)&=\textstyle
\Bigl(\prod\limits_{e\in\Es^\spartial} \frac{1+\sn(\theta_e)}{2\cn(\theta_e)}\Bigr) 
\Bigl(\prod\limits_{e^*\in\Es^*}\frac{1}{\cn(\theta_e)}\Bigr)\\
&\textstyle=\Bigl(\prod\limits_{w\in W^\spartial}\frac{1+\sn(\theta_w)}{2}\Bigr)
\Bigl(\prod\limits_{w\in W}\frac{1}{\cn(\theta_w)}\Bigr),
\end{align*}
where in the second equality we use the notation introduced before the statement of Theorem~\ref{thm:part_function}.
Since $\Zdimer(\GQ,\nuJ)=|\det\KQ|$, Theorem~\ref{thm:part_function} yields
\begin{align*}
[\Zising^+(\Gs,\Js)]^2&=\textstyle 2^{|\Vs|-1}
(k')^{\frac{|\Vs^*|}{2}}\prod\limits_{w\in W}[\sc(\theta_w)]^{\frac{1}{2}}
\Bigl(\prod\limits_{w\in W^\spartial}\frac{1+\sn(\theta_w)}{2\sn(\theta_w)}\Bigr) 
\Bigl(\prod\limits_{wx\in\Ms_1:w\in W}|\sc(u_{\alpha_e})\sc(u_{\beta_e})|^{\frac{1}{2}}\Bigr)\times \\
&\textstyle \times\sn(\theta^{\spartial,\rs})\left|\frac{\cd(u_{\beta^{r,\rs}})}{\sn(u_{\alpha^{r,\rs}})}\right|
|\det \KD^\spartial(u)|.
\end{align*}
The proof is concluded by using Corollary~\ref{cor:det}, choosing
$\Ms_1$ as perfect matching of $\GDro$ and using that $|\Vs|-1=-|\Vs^*|+|\Es|$.
\end{proof}

\subsubsection{Proof of Theorem~\ref{thm:part_function}}

Let us prove that for the modified Kasteleyn matrix $\KQu$, we have
\begin{align*}
\det\KQu&=(k')^{\frac{|\Vs^*|}{2}}\Bigl(\prod_{w\in W}[\sn(\theta_w)\cn(\theta_w)]^{\frac{1}{2}} \Bigr)
\Bigl(\prod_{w\in W^\spartial}\sn(\theta_w)^{-\frac{1}{2}}\Bigr) 
\Bigl(\prod_{wx\in\Ms_1:w\in W}|\sc(u_{\alpha_e})\sc(u_{\beta_e})|^{\frac{1}{2}}\Bigr)\times \\
&\times \sn(\theta^{\spartial,\rs})\left|\frac{\cd(u_{\beta^{r,\rs}})}{\sn(u_{\alpha^{r,\rs}})}\right|
|\det \KD^\spartial(u)|.
\end{align*}
To obtain the result for $\KQ$, we use Relation~\eqref{equ:KQ_KQu}
which implies that 
\[
\det\KQ = \Bigl(\prod_{w\in W^\spartial}\sn(\theta_w)^{-\frac{1}{2}}\Bigr)\det \KQu,
\]
because $|\{\bs_\ell,\bs_r\in \R^\spartial\}|=2|\{\ws^c\in\R^\spartial\}|=2|W^\spartial|$.
The proof has three main steps: the first
consists in using a partition of black/white vertices of $\GQ$ and Theorem~\ref{thm:main}
for comparing $\det\KQu$ and $\det\KD^\spartial(u)$; in the second step, we specify this partition so
that the computation required by the first step become tractable; finally, we perform these computations in the third step.

\paragraph{Partition of the vertices of $\GQ$ and Theorem~\ref{thm:main}.} We partition black vertices of $\GQ$ into two subsets: 
$\Bs=\Bs_1\cup \Bs_2$, where $\Bs_1$ has one black vertex per quadrangle of $\GQ$ and $\Bs_2=\Bs\setminus \Bs_1$.
Since boundary quadrangles of $\GQ$ are reduced to edges, $\Bs_1$ contains 
all black vertices of boundary quadrangles, and half of the black vertices of inner quadrangles. As a consequence, $\Bs_1$ has a 
natural partition as $\Bs_1^\spartial\cup \Bs_1^\scirc$, where $\Bs_1^\spartial$, resp. $\Bs_1^\scirc$, consists of boundary quadrangle, resp. inner quadrangles,
black vertices. Then $\Bs_2$ has no boundary quadrangle black vertices. 
In a similar way, we partition white vertices of $\GQ$: $\Ws=\Ws_1\cup \Ws_2$,
where $\Ws_1=\Ws_1^\spartial\cup \Ws_1^\scirc$.

Recalling that there is a natural bijection between quadrangles of $\GQ$ 
and white vertices of $\GDro$, we have the following natural bijections: 
\begin{equation*}
\Ws_1^\spartial \leftrightarrow W^\spartial, \quad \Ws_1^\scirc \leftrightarrow W^\scirc,\quad 
\Ws_2 \leftrightarrow W^\scirc.
\end{equation*}
Here are some notation for sub-matrices of the matrices $S(u)$ and $T(u)$ defined in Section~\ref{sec:KQ_KD_relation}, and for sub-matrices
of the matrix $\KQu$.
\begin{align*}
&S_1(u)=S(u)_{\Bs_1}^{{W\ro}},\quad S_2(u)=S(u)_{\Bs_2}^{W^\scirc}\\
&S_1^\spartial (u)=S(u)_{\Bs_1^\spartial}^{{W^\spartial}},\quad S_1^\scirc (u)=S(u)_{\Bs_1^\circ}^{{W^\scirc}}\\
&T_1(u)=T(u)_{\Ws_1}^{B\ro},\quad T_2(u)=T(u)_{\Ws_2}^{B\ro}\\
\forall\,i,j&\in\{1,2\},\quad \KQu_{ij}=(\KQu)_{\Bs_i}^{\Ws_j},\quad \KQu_{1^\scirc j}=(\KQu)_{\Bs_1^\scirc}^{\Ws_j}.
\end{align*}
Using Theorem~\ref{thm:main}, we obtain the following lemma.
\begin{lem}\label{lem:KQ_KD_det}
Consider a partition of the black and white vertices of $\GQ$ as above. Then, for every $u\in\Re(\TT(k))''$,
\begin{equation*}
\det\KQu\cdot \det T_1(u)=\det S_1^\spartial(u)\cdot\det S_1^\scirc(u)\cdot\det S_2(u)\cdot\det R(u)\cdot \det \KD^\spartial(u),
\end{equation*}
where $R(u):=S_2(u)^{-1}\KQu_{22}-S_1^\circ(u)^{-1}\KQu_{1^\scirc 2}$.
\end{lem}
\begin{proof}
Note that $S_2(u)$ and $S_1^\circ(u)$ are invertible because we choose $u\in\Re(\TT(k))''$.
By Theorem~\ref{thm:main}, we have the following identity:
\begin{equation*}
\begin{blockarray}{cccc}
     &\scriptstyle{\Ws_1^\spartial} & \scriptstyle{\Ws_1^\scirc} & \scriptstyle{\Ws_2}\\
\begin{block}{c(cc|c)}
  \scriptstyle{\Bs_1^\spartial}\! & \BAmulticolumn{2}{c|}{\multirow{2}{*}{$\KQu_{11}$}}&\multirow{2}{*}{$\KQu_{12}$}\\ 
  \scriptstyle{\Bs_1^\scirc}\! & & &  \\
  \cline{2-4}
  \multirow{2}{*}{$\scriptstyle{\Bs_2}\!$}& \BAmulticolumn{2}{c|}{\multirow{2}{*}{$\KQu_{21}$}} & \multirow{2}{*}{$\KQu_{22}$}\\
  &&&\\
\end{block}
\end{blockarray}
\begin{blockarray}{ccc}
     &\scriptstyle{B\ro} & \scriptstyle{\Ws_2}\\
\begin{block}{c(c|c)}
  \scriptstyle{\Ws_1^\spartial}\! &\multirow{2}{*}{$T_1(u)$}&\multirow{2}{*}{$0$}\\ 
  \scriptstyle{\Ws_1^\scirc}\! & &   \\
  \cline{2-3}
  \multirow{2}{*}{$\scriptstyle{\Ws_2}\!$}&\multirow{2}{*}{$T_2(u)$}&\multirow{2}{*}{$I$}\\
  &&\\
\end{block}
\end{blockarray}=
\begin{blockarray}{cccc}
     &\scriptstyle{W^\spartial} & \scriptstyle{W^\scirc} & \scriptstyle{W^\scirc}\\
\begin{block}{c(cc|c)}
  \scriptstyle{\Bs_1^\spartial}\! & S_1^\spartial(u) &0&0\\ 
  \scriptstyle{\Bs_1^\scirc}\! & 0 & S_1^\scirc(u)   & 0 \\
  \cline{2-4}
  \multirow{2}{*}{$\scriptstyle{\Bs_2}\!$}& \multirow{2}{*}{$0$} & \multirow{2}{*}{$S_2(u)$} & \multirow{2}{*}{$S_2(u)$}\\
  &&&\\
\end{block}
\end{blockarray}
\begin{blockarray}{ccc}
     &\scriptstyle{B\ro} & \scriptstyle{\Ws_2} \\
\begin{block}{c(c|c)}
  \scriptstyle{W^\spartial} & \multirow{2}{*}{$\KD^\spartial(u)$} &\multirow{2}{*}{$\tilde{R}(u)$}\\ 
  \scriptstyle{W^\scirc} & &\\
  \cline{2-3}
  \multirow{2}{*}{$\scriptstyle{W^\scirc}$}& \multirow{2}{*}{$0$} & \multirow{2}{*}{$R(u)$}\\
  &&\\
\end{block}
\end{blockarray}\,,
\end{equation*}
with
$
\begin{cases}
\KQu_{1^\scirc 2}&=S_1^\scirc(u) \tilde{R}(u)_{W^\scirc}^{\Ws_2}\\ 
\KQu_{22}&=S_2(u) \tilde{R}(u)_{W^\scirc}^{\Ws_2} +S_2(u) R(u).
\end{cases}
$

We extract $\tilde{R}(u)_{W^\scirc}^{\Ws_2}$ from the first equation.
Plugging $\tilde{R}(u)_{W^\scirc}^{\Ws_2}$ in the second equation gives $R(u)$;
taking the determinant ends the proof.
\end{proof}

\paragraph{Combinatorial partition of the vertices of $\GQ$.} Let us specify the partition of the black/white vertices of $\GQ$. It is
constructed from a well chosen perfect matching $\Ms_1$ of $\GDro$, which we now define. Recall that by Temperley's bijection,
perfect matchings of $\GDro$ are in bijection with pairs of dual directed spanning trees of $\T^{\rs,\outer}(\Gs,\bar{\Gs}^*)$, see 
Section~\ref{sec:def_double_graph}. Consider a spanning tree of the restricted dual $\Gs^*$, and root it at the vertex $f^{c,\rs}$ incident to 
the root vertex $v^c=\rs$ in the diamond graph $\GR$. To this spanning tree, add the edge $(f^{c,\rs},\outer)$ which is 
the dual of the edge $\rs v^{\ell,\rs}$ of $\Gs$. This defines an $\outer$-directed spanning tree $\Ts^*$ of $\bar{\Gs}^*$. Consider the 
$\rs$-directed spanning tree $\Ts$ of $\Gs$ which is the dual of $\Ts^*$, rooted at the vertex $\rs$. Then $(\Ts,\Ts^*)$ is a pair of dual 
spanning trees of $\T^{\rs,\outer}(\Gs,\bar{\Gs}^*)$, and we let $\Ms_1$ be the corresponding perfect matching of $\GDro$,
see Figure~\ref{fig:G_Gdual_3}.

\begin{figure}[ht]
\begin{minipage}[b]{0.5\linewidth}
\begin{center}
\begin{overpic}[width=7.8cm]{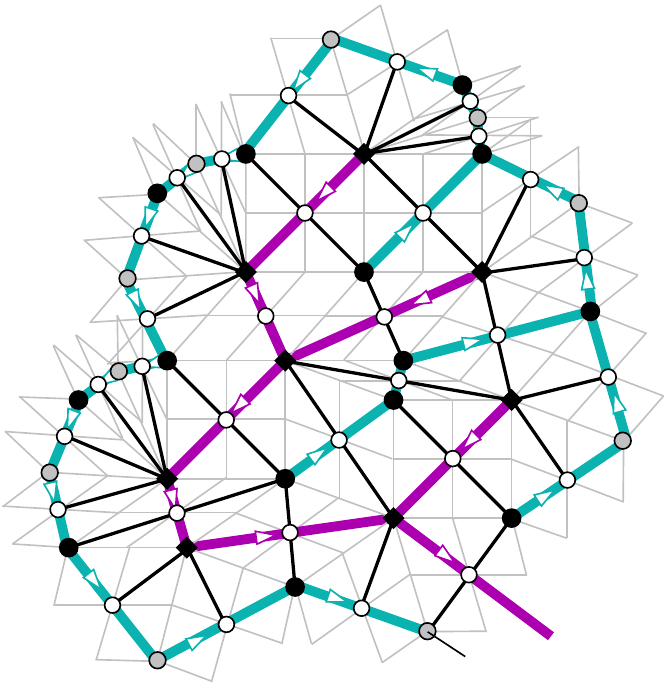}
\put(69.5,3){\scriptsize $\rs$}
\put(75,20){\scriptsize $v^{\ell,\rs}$}
\put(61,23.5){\scriptsize $f^{c,\rs}$}
\put(82,6){\scriptsize  $\outer$}
\end{overpic}
\end{center}
\end{minipage}
\begin{minipage}[b]{0.5\linewidth}
\begin{center}
\begin{overpic}[width=7.8cm]{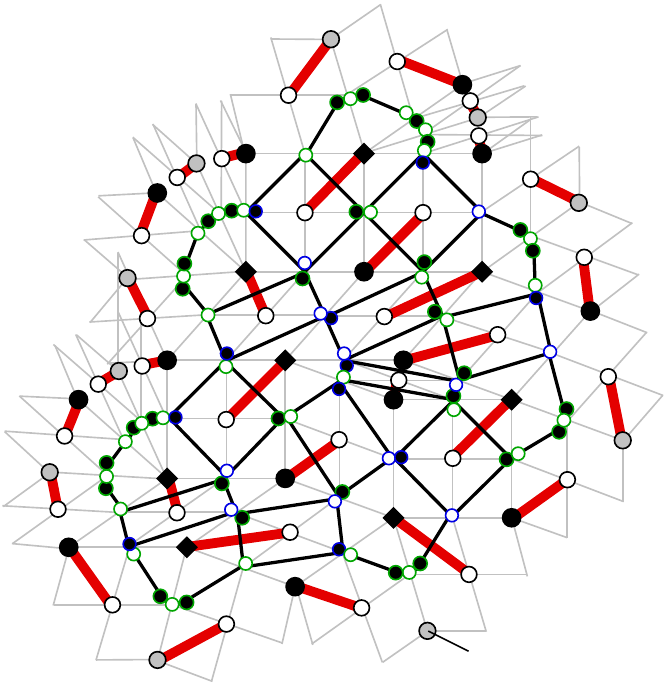}
\put(69.5,3){\scriptsize $\rs$}
\end{overpic}
\end{center}
\end{minipage}
\caption{Left: a pair $(\Ts,\Ts^*)$ of dual directed spanning trees of $\T^{\rs,\outer}(\Gs,\bar{\Gs}^*)$ such that $\Ts^*$ arises from a 
spanning tree of the restricted dual $\Gs^*$. Right: corresponding dimer configuration $\Ms_1$ of $\GDro$ (red); associated partition of black/white
vertices of $\GQ$: vertices of $\Bs_1,\Ws_1$ are in green, those of $\Bs_2,\Ws_2$ are in blue.}
\label{fig:G_Gdual_3}
\end{figure}

We now define the \emph{partition of black/white vertices of $\GQ$ arising from $\Ms_1$}; this amounts to specifying the partition 
for vertices of inner quadrangle since those of boundary quadrangles define $\Bs_1^\spartial$/$\Ws_1^\spartial$. 
For $j\in\{1,2\}$, $\bs_j$/$\ws_j$ denotes a black/white vertex of $\Bs_j/\Ws_j$.
Consider an inner quadrangle of $\GQ$ corresponding to a white vertex $w$ of $\GDro$. 
Then, exactly one edge of the quadrangle is crossed by an edge $wx$ of the perfect matching $\Ms_1$. The partition is defined as follows:
if $x=v\in \Vs\ro$, then $\ws_1$ is the white vertex on the right of the edge $(w,v)$, and $\bs_1$ is the black vertex on the left;
if $x=f\in \Vs^*$, then $\ws_1$ is the white vertex on the left of the edge $(w,f)$, and $\bs_1$ is the black vertex on the right.
An example is provided in Figure~\ref{fig:G_Gdual_3} (right) and Figure~\ref{fig:preuve_detKQ}.

Let us prove a combinatorial lemma. Note that because $u\in\Re(\TT(k))''$, the fact that a coefficient of $T_1(u)$ or $R(u)$ is non-zero is 
independent of $u$. As in Section~\ref{sec2:app} of Appendix~\ref{app:gauge}, to the matrix $T_1(u)$ corresponds a bipartite graph 
$\Gs(T_1)=(\Ws_1\cup B^\rs,\Es(T_1))$, where there is an edge $\ws_1 x$ iff $t(u)_{\ws_1,x}\neq 0$, with $\ws_1\in\Ws_1$, $x\in B\ro=\Vs\ro\cup\Vs^*$.
In a similar way, to the matrix $R(u)$ of Lemma~\ref{lem:KQ_KD_det} corresponds a bipartite graph $\Gs(R)=(W^\scirc\cup \Ws_2,\Es(R))$ where there is an edge $w\ws_2$ iff 
$r(u)_{w,\ws_2}\neq 0$,
with $w\in W^\scirc$, $\ws_2\in \Ws_2$. The matrix $T_1(u)$, resp. $R(u)$, is then a bipartite, weighted adjacency matrix of the graph $\Gs(T_1)$, resp. $\Gs(R)$. 

\begin{figure}[ht]
\begin{center}
\begin{overpic}[width=7.8cm]{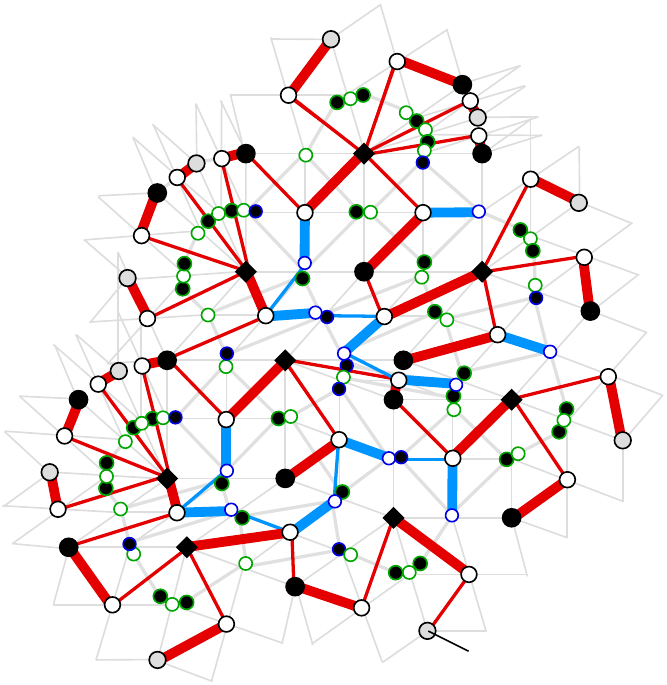}
\put(69.5,3){\scriptsize $\rs$}
\end{overpic}
\caption{Example of graph $\Gs(T_1)$ (red edges) and of graph $\Gs(R)$ (blue edges).}
\label{fig:spanning_trees}
\end{center}
\end{figure}

\begin{lem}
Consider a partition of the black/white vertices of $\GQ$ arising from the perfect matching $\Ms_1$ of $\GDro$. 
Then, the graph $\Gs(T_1)$ is a spanning tree on the vertex set $\Ws_1\cup B^\rs$, and the graph $\Gs(R)$ is a union of trees on the vertex set 
$W^\scirc\cup \Ws_2$, spanning all vertices of $W^\scirc\cup \Ws_2$.
\end{lem}
\begin{proof}
Let us prove that $\Gs(T_1)$ is a spanning tree. Every vertex $\ws_1$ of $\Ws_1$ has degree 2 in $\Gs(T_1)$: 
it is adjacent to a vertex $x$ of $B^\rs$ such 
that $wx$ is an edge of the perfect matching $\Ms_1$, and to a vertex $x'$ such that $xx'$ is an edge of the diamond graph $\GR$ with 
$\ws_1$ in its middle. Using the natural bijection $\Ws_1\leftrightarrow W$, that is, identifying every vertex $\ws_1$ with the corresponding vertex 
$w$ of $W^\rs$ does not change the combinatorics of the graph $\Gs(T_1)$. With this identification, the graph $\Gs(T_1)$ contains all edges of the 
perfect
matching $\Ms_1$, and the second edge $wx'$ incident to $w$ is such that $(w,x')$ is on the right, resp. left, of $(x,w)$ if $x\in\Vs^*$, resp.
$x\in\Vs^\rs$. Then, by Proposition~7.3. of~\cite{deTiliere:partition}, the graph $\Gs(T_1)$ is a spanning tree. An example of $\Gs(T_1)$ is pictured
with red edges in Figure~\ref{fig:spanning_trees}.

Let us prove that $\Gs(R)$ is a union of trees spanning all vertices of $W^\scirc\cup \Ws_2$. By definition, a white vertex $w$ of $W^\scirc$
is adjacent to all white vertices of $\Ws_2$ that are adjacent to the black vertices $\bs_1$ and $\bs_2$ of the quadrangle $w$ corresponds to; that is,
$w$ is adjacent to the vertex $\ws_2$ of the quadrangle of $w$ and maybe to another vertex of $\Ws_2$. As a consequence $\Gs(R)$ is spanning all 
vertices of $W^\scirc\cup \Ws_2$. It cannot contain a cycle for otherwise it would mean that there is a vertex of $B^\rs$ which does not belong 
to $\Ms_1$, which is a contradiction with it being a perfect matching. An example of $\Gs(R)$ is pictured with blue edges in 
Figure~\ref{fig:spanning_trees}.
\end{proof}

\begin{cor}\label{cor:det_T_C}
Consider a partition of the black/white vertices of $\GQ$ arising from the perfect matching $\Ms_1$ of $\GDro$. 
\begin{itemize}
\item[$\bullet$] Using the identification $\Ws_1\leftrightarrow W\ro=W$, we have,
\begin{equation*}
\left|\det T_1(u)\right|=\prod_{wx\in \Ms_1:\,w\in W} |t(u)_{\ws_1,x}|
=\prod_{wx\in M:\,w\in W^\scirc} |t(u)_{\ws_1,x}|\prod_{wx\in M:\,w\in W^\spartial} |t(u)_{\ws_1,x}|.
\end{equation*}
\item[$\bullet$] Using the identification $\Ws_2\leftrightarrow W^\scirc$, we have,
\begin{equation*}
\left|\det R(u)\right|=\prod_{w\in W^\scirc}|r(u)_{w,\ws_2}|,
\end{equation*}
where for every $w\in W^\scirc$, $r(u)_{w,\ws_2}=s(u)_{\bs_2,w}^{-1}\KQu_{\bs_2,\ws_2}-s(u)_{\bs_1,w}^{-1}\KQu_{\bs_1,\ws_2}$.
\end{itemize}
\end{cor}
\begin{proof}
Writing the determinant as a sum over permutations, we have that non-zero terms in the expansion of $\det T_1(u)$, resp.
$\det R(u)$, correspond to perfect matchings of the bipartite graph $\Gs(T_1)$, resp. $\Gs(R)$. Since these two graphs
are trees or union of trees spanning all vertices, they have at most one perfect matching; indeed if they had more, the union of two
different ones would yield a cycle which is in contradiction with being a tree. 
Using the identification $\Ws_1\leftrightarrow W$, The graph $\Gs(T_1)$ has one perfect matching given by edges of $\Ms_1$ (pictured
in thick red lines in Figure~\ref{fig:spanning_trees}), while the graph $\Gs(R)$ has one perfect matching given by the natural identification
$W^\scirc\leftrightarrow \Ws_2$ (pictured in thick blue lines in Figure~\ref{fig:spanning_trees}). Since, the contribution of  
a perfect matching to the determinant is the product of the edge-weights (up to a sign), this ends the proof of the corollary.
\end{proof}

Since $S_1^\spartial(u),S_1^\scirc(u),S_2(u)$ are diagonal matrices, their determinant is the product of the diagonal terms. 
Combining Lemma~\ref{lem:KQ_KD_det} and Corollary~\ref{cor:det_T_C} we thus obtain the following.

\begin{cor}\label{cor:det_KQ_caca}
Consider a partition of the black/white vertices of $\GQ$ arising from the perfect matching $\Ms_1$ of $\GDro$.
Then, for every $u\in\Re(\TT(k))''$,
\begin{equation}\label{equ:detKQ}
|\det\KQu|=\underbrace{\left|\frac{\prod_{w\in W^\spartial}s(u)_{\bs_1,w}}{\prod_{wx\in\Ms_1:\,w\in W^\spartial}t(u)_{\ws_1,x}}\right|}_{(\mathrm{I})}\,
\underbrace{\left|\frac{\prod_{w\in W^\scirc} r'(u)_{w,\ws_2}}{
\prod_{wx\in\Ms_1:\,w\in W^\scirc}t(u)_{\ws_1,x}}\right|}_{(\mathrm{II})} \left|\det \KD^\spartial(u)\right|,
\end{equation}
where $r'(u)_{w,\ws_2}=s(u)_{\bs_1,w}\KQu_{\bs_2,\ws_2}-s(u)_{\bs_2,w}\KQu_{\bs_1,\ws_2}$.
\end{cor}

\paragraph{Computation of (I)(II) in Identity~\eqref{equ:detKQ}.}  
For every $u\in\Re(\TT(k))''$,
define the weight function $\eta(u)$ on edges of $\GDro$ as follows.
\begin{equation}\label{equ:def_eta_prime}
\eta(u)_{wx}=
\begin{cases}
 \left|\frac{\sn(u_{\alpha_e})}{\cn(u_{\beta_e})}\right|[\nd(u_{\alpha_e})\dn(u_{\beta_e})]^{\frac{1}{2}} &\text{if $x=v$}\\
 (k')^{\frac{1}{2}}\left|\frac{\sn(u_{\beta_e})}{\cn(u_{\alpha_e})}\right|[\dn(u_{\alpha_e})\nd(u_{\beta_e})]^{\frac{1}{2}} &\text{if $x=f$}.
\end{cases}
\end{equation}

\begin{lem}\label{lem:prod_I_II}
The product $\mathrm{(I)(II)}$ is equal to:
\begin{equation*}
\sn(\theta^{\spartial,\rs})\left|\frac{\cd(u_{\beta^{r,\rs}})}{\sn(u_{\alpha^{r,\rs}})}\right|
\Biggl(\prod_{w\in W^\spartial}\sn(\theta_w)^{-\frac{1}{2}}\Biggr)
\Biggl(\prod_{w\in W}|\sn(\theta_w)\cn(\theta_w)|^{\frac{1}{2}}\Biggr)
\prod_{wx\in \Ms_1:\,w\in W} \eta(u)_{wx}.
\end{equation*}
\end{lem}

\begin{proof}
In the whole of the proof, we simply denote $\bar{\theta}_w$ by $\bar{\theta}$, and omit the argument $u$ from matrix coefficients.

We first handle Part (II) involving inner vertices. Let $w\in W^\scirc$ and $wx$ be an edge of the perfect matching
$\Ms_1$, with $x=v$ or $f$. 
By definition of the partition of black/white vertices arising from $\Ms_1$, we have $\bs_1,\bs_2,\ws_1,\ws_2$ as in 
Figure~\ref{fig:preuve_detKQ}. Let $v'$ be the primal vertex such that the edge $v'v$ crosses the quadrangle, and let $2e^{i\bar{\alpha}}$,
$2e^{i\bar{\beta}}$ be the two rhombus vectors of $\GR$ associated to the edge $(v',v)$. 

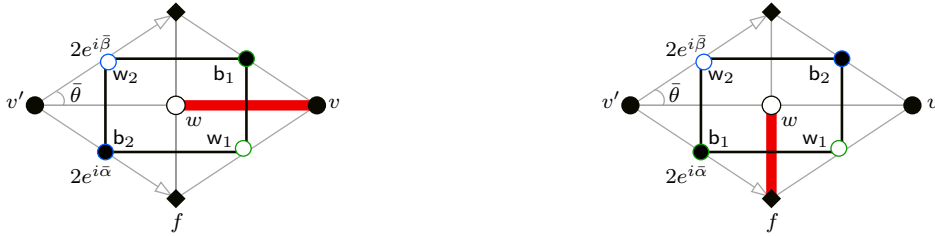
\begin{figure}[ht]
\begin{minipage}[b]{0.5\linewidth}
\begin{center}
\begin{overpic}[width=4cm]{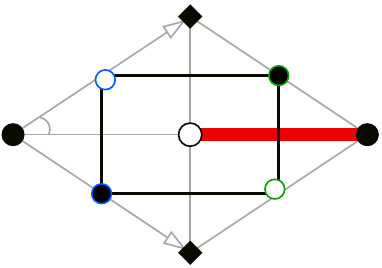}
  \put(15,9){\scriptsize $2e^{i\bar{\alpha}}$}
  \put(15,51){\scriptsize $2e^{i\bar{\beta}}$}
  \put(-6,34){\scriptsize $v'$}
  \put(100,34){\scriptsize $v$}
  \put(53,28){\scriptsize $w$}
  \put(48,-6){\scriptsize $f$}
  \put(15,36){\scriptsize $\bar{\theta}$}
  \put(29,22){\scriptsize $\bs_2$}
  \put(29,43){\scriptsize $\ws_2$}
  \put(60,22){\scriptsize $\ws_1$}
  \put(62,43){\scriptsize $\bs_1$}
\end{overpic}
\end{center}
\end{minipage}
\begin{minipage}[b]{0.5\linewidth}
\begin{center}
\begin{overpic}[width=4cm]{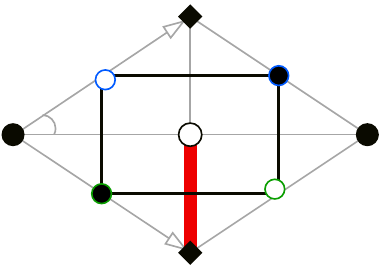}
  \put(15,9){\scriptsize $2e^{i\bar{\alpha}}$}
  \put(15,51){\scriptsize $2e^{i\bar{\beta}}$}
  \put(-6,34){\scriptsize $v'$}
  \put(101,34){\scriptsize $v$}
  \put(48,-6){\scriptsize $f$}
  \put(53,28){\scriptsize $w$}
  \put(16,36){\scriptsize $\bar{\theta}$}
  \put(29,22){\scriptsize $\bs_1$}
  \put(29,43){\scriptsize $\ws_2$}
  \put(60,22){\scriptsize $\ws_1$}
  \put(62,43){\scriptsize $\bs_2$}
\end{overpic}
\end{center}
\end{minipage}
\caption{Notation for computing $\frac{r'_{w,\ws_2}}{t_{\ws_1,x}}$, when $w\in W^\scirc$ and $x=v$ (left), $x=f$ (right).}
\label{fig:preuve_detKQ}
\end{figure}
We compute the term $r'_{w,\ws_2}$ when $x=v$. Using the notation of Figure~\ref{fig:preuve_detKQ},
the rhombus vectors of $\GRR$ assigned to the edge $(w,v)$ are $e^{i\bar{\alpha}_e}=e^{i\bar{\alpha}}$, $e^{i\bar{\beta}_e}=e^{i\bar{\beta}}$.
Returning to the definition of the matrices $\KQu$ and $S$, see~\eqref{def:sbw}, we have
\begin{align*}
s_{\bs_1,w}\KQu_{\bs_2,\ws_2}&= e^{-i\frac{\bar{\beta}+\pi}{2}}\cn(u_{\beta+2K})
  [\sn(\theta)\cn(\theta)\nd(u_{\alpha+2K})\nd(u_{\beta+2K})]^{\frac{1}{2}}\times
  e^{i\frac{\bar{\beta}+\bar{\alpha}+\pi}{2}}\cn(\theta)\\
  &=e^{i\frac{\bar{\alpha}}{2}}\cn(\theta)\sn(u_{\beta})\dn(u_\alpha)[\sn(\theta)\cn(\theta)\nd(u_{\alpha})\nd(u_{\beta})]^{\frac{1}{2}}\\
s_{\bs_2,w}\KQu_{\bs_1,\ws_2}&=
   e^{-i\frac{\bar{\beta}}{2}}\cn(u_{\beta})
  [\sn(\theta)\cn(\theta)\nd(u_{\alpha})\nd(u_{\beta})]^{\frac{1}{2}}\times
  e^{i\frac{\bar{\alpha}+\pi+\bar{\beta}+\pi}{2}}\sn(\theta)\\
&=-e^{i\frac{\bar{\alpha}}{2}}\sn(\theta)\cn(u_\beta)[\sn(\theta)\cn(\theta)\nd(u_\alpha)\nd(u_\beta)]^{\frac{1}{2}}.
\end{align*}
As a consequence,
\begin{align*}
r'_{w,\ws_2}&=
e^{i\frac{\bar{\alpha}}{2}}[\sn(\theta)\cn(\theta)\nd(u_{\alpha})\nd(u_{\beta})]^{\frac{1}{2}}
\bigl(\cn(\theta)\sn(u_{\beta})\dn(u_\alpha)+\sn(\theta)\cn(u_\beta)\bigr)\nonumber\\
&=e^{i\frac{\bar{\alpha}}{2}}[\sn(\theta)\cn(\theta)\nd(u_{\alpha})\nd(u_{\beta})]^{\frac{1}{2}}\dn(u_\beta)\sn(u_\alpha)\nonumber\\
&=e^{i\frac{\bar{\alpha}}{2}}\sn(u_\alpha)[\sn(\theta)\cn(\theta)\nd(u_\alpha)\dn(u_\beta)]^{\frac{1}{2}},
\end{align*}
using the identity 
$\dn(u)\sn(u+v)=\cn(u)\sn(v)+\sn(u)\cn(v)\dn(u+v)$~\cite[chap.2, ex.32 (ii)]{Lawden}, evaluated at $u=u_\beta,v=\theta,u+v=u_\alpha$, in the penultimate line.
Since $e^{i\bar{\alpha}_e}=e^{i\bar{\alpha}}$, $e^{i\bar{\beta}_e}=e^{i\bar{\beta}}$, we have
\begin{equation}\label{equ:cuv}
r'_{w,\ws_2}=e^{i\frac{\bar{\alpha}_e}{2}}\sn(u_{\alpha_e})[\sn(\theta)\cn(\theta)\nd(u_{\alpha_e})\dn(u_{\beta_e})]^{\frac{1}{2}}.
\end{equation}
By definition of $T$, see~\eqref{def:twvf}, we have $t_{\ws_1,v}=e^{-i\frac{\bar{\beta}}{2}}\cn(u_{\beta})=e^{-i\frac{\bar{\beta}_e}{2}}\cn(u_{\beta_e})$,
and we conclude that 
\begin{equation*}
\left|\frac{r'_{w,\ws_2}}{t_{\ws_1,v}}\right|=[\sn(\theta)\cn(\theta)]^\frac{1}{2}
\left|\frac{\sn(u_{\alpha_e})}{\cn(u_{\beta_e})}\right|[\nd(u_{\alpha_e})\dn(u_{\beta_e})]^{\frac{1}{2}}.
\end{equation*}

We now turn to the term $r'_{w,\ws_2}$ in the case where $x=f$.
The rhombus vectors of $\GRR$
assigned to the edge $(w,f)$ are $e^{i\bar{\alpha}_e}=e^{i(\bar{\beta}-\pi)}$, $e^{i\bar{\beta}_e}=e^{i\bar{\alpha}}$. Moreover,
referring to Figure~\ref{fig:preuve_detKQ}, we see that taking $x=f$ has the effect of exchanging $\bs_1$ and
$\bs_2$ and leaving $\ws_1,\ws_2$ fixed. The quantity $r'_{w,\ws_2}$ being skew-symmetric in $\bs_1,\bs_2$, we have that $r'_{w,\ws_2}$ is equal 
to the opposite of~\eqref{equ:cuv}. As a consequence,
\begin{align*}
r'_{w,\ws_2}&=-e^{i\frac{\bar{\alpha}}{2}}\sn(u_\alpha)[\sn(\theta)\cn(\theta)\nd(u_\alpha)\dn(u_\beta)]^{\frac{1}{2}}\nonumber\\
&=-e^{i\frac{\bar{\beta}_e}{2}}\sn(u_{\beta_e})[\sn(\theta)\cn(\theta)\nd(u_{\beta_e})\dn(u_{\alpha_e+2K})]^{\frac{1}{2}}
\nonumber\\
&=-(k')^{\frac{1}{2}}e^{i\frac{\bar{\beta}_e}{2}}\sn(u_{\beta_e})[\sn(\theta)\cn(\theta)\nd(u_{\beta_e})\nd(u_{\alpha_e})]^{\frac{1}{2}}.
\end{align*}
using that $e^{i\bar{\alpha}_e}=e^{i(\bar{\beta}-\pi)}$, $e^{i\bar{\beta}_e}=e^{i\bar{\alpha}}$, and that $\dn(u-K)=k'\nd(u)$.
By definition of $T$, we have $t_{\ws_1,f}=e^{-i\frac{\bar{\beta}-\pi}{2}}\cd(u_{\beta-2K})=e^{-i\frac{\bar{\alpha}_e}{2}}\cd(u_{\alpha_e})$,
and we deduce that
\begin{equation*}
\left|\frac{r'_{w,\ws_2}}{t_{\ws_1,v}}\right|=(k')^{\frac{1}{2}}[\sn(\theta)\cn(\theta)]^\frac{1}{2}
\left|\frac{\sn(u_{\beta_e})}{\cn(u_{\alpha_e})}\right|[\nd(u_{\beta_e})\dn(u_{\alpha_e})]^{\frac{1}{2}}.
\end{equation*}
Summarizing, we have proved that, for every edge $wx\in\Ms_1$ such that $w\in W^\scirc$,
\begin{align*}
\left|\frac{r'_{w,\ws_2}}{t_{\ws_1,x}}\right|&=
[\sn(\theta)\cn(\theta)]^{\frac{1}{2}}\times
\begin{cases}
\left|\frac{\sn(u_{\alpha_e})}{\cn(u_{\beta_e})}\right|[\nd(u_{\alpha_e})\dn(u_{\beta_e})]^{\frac{1}{2}},&\text{ if $x=v$}\\
(k')^{\frac{1}{2}}\left|\frac{\sn(u_{\beta_e})}{\cn(u_{\alpha_e})}\right|[\nd(u_{\beta_e})\dn(u_{\alpha_e})]^{\frac{1}{2}},&\text{ if $x=f$}
\end{cases}\\
&=[\sn(\theta)\cn(\theta)]^{\frac{1}{2}}\,\eta_{wx};
\end{align*}
and thus,
\begin{equation}\label{equ:II}
\mathrm{(II)}=\prod_{w\in W^\scirc}[\sn(\theta)\cn(\theta)]^{\frac{1}{2}}\prod_{wx\in \Ms_1:\,w\in W^\scirc}\eta_{wx}. 
\end{equation}

We now compute Part (I) involving boundary vertices of $W^\spartial$.
We will be using the notation of Figure~\ref{fig:def_ST} 
for vertices, rhombus vectors and angles of boundary rhombus pairs of $\R^\spartial$, and add a superscript $\rs$ 
when the pair is the root pair, \emph{i.e.}, the one where $v^c=\rs$. With our choice of perfect matching 
$\Ms_1$, the boundary contribution (I) of~\eqref{equ:detKQ} can be rewritten as, see also Figure~\ref{fig:G_Gdual_3} (right):
\begin{equation*}
(\mathrm{I})=\Bigl(\prod_{w^r,w^\ell\in W^{\spartial}\setminus\{w^{r,\rs},w^{\ell,\rs}\}}
\Bigl|\frac{s_{\bs^r,w^r}s_{\bs^\ell,w^\ell} }{t_{\ws^r,v^r}t_{\ws^c,v^c}}\Bigr|\Bigr)\times
\Bigl|\frac{s_{\bs^{r,\rs},w^{r,\rs}}s_{\bs^{\ell,\rs},w^{\ell,\rs}}}{
t_{\ws^{r,\rs},v^{r,\rs}}t_{\ws^{c,\rs},f^{c,\rs}}}\Bigr|.
\end{equation*}

Suppose first that $w^r,w^\ell\in W^\spartial\setminus\{w^{r,\rs},w^{\ell,\rs}\}$. Recalling the definition of $S$ and $T$ along the boundary, 
see~\eqref{def:sbl_boundary} and~\eqref{def:twv_boundary}, we have
\begin{align*}
\left|\frac{s_{\bs^r,w^r}s_{\bs^\ell,w^\ell} }{t_{\ws^r,v^r}t_{\ws^c,v^c}}\right|&=
\left|\frac{\sn(\theta^\spartial)\cn(\theta^\spartial)\cn(u_{\beta^r})
\cn(u_{\alpha^\ell})[\nd(u_{\alpha^r})\nd(u_{\beta^r})\nd(u_{\alpha^\ell})\nd(u_{\beta^\ell})]^{\frac{1}{2}}
}{\cn(u_{\beta^r})\times k'\sn(\theta^\spartial)\nd(u_{\alpha^r})\cd(u_{\beta^r})}\right|\\
&=\left|
\frac{\cn(\theta^\spartial)\cn(u_{\alpha^\ell})}{k'\cn(u_{\beta^r})}
[\dn(u_{\alpha^r})\dn(u_{\beta^r})\nd(u_{\alpha^\ell})\nd(u_{\beta^\ell})]^{\frac{1}{2}}
\right|.
\end{align*}

On the other hand, by definition of $\eta$, the product of weights of the edges $w^r v^r,w^\ell v^c$ of the perfect matching $\Ms_1$ is equal to,
\begin{align*}
\eta_{w^r v^r}\eta_{w^\ell v^c}&=
\left|
\frac{\sn(u_{\alpha^r})}{\cn(u_{\beta^r})}\frac{\sn(u_{\alpha^\ell+2K})}{\cn(u_{\beta^\ell+2K})}
[\nd(u_{\alpha^r})\dn(u_{\beta^r})\nd(u_{\alpha^\ell+2K})\dn(u_{\beta^\ell+2K})]^{\frac{1}{2}}
\right|\\
&=\left|
\frac{\cn(u_{\alpha^\ell})}{k'\cn(u_{\beta^r})}
[\dn(u_{\alpha^r})\dn(u_{\beta^r})\nd(u_{\alpha^\ell})\nd(u_{\beta^\ell})]^{\frac{1}{2}}
\right|,
\end{align*}
writing $u_{\alpha^\ell+2K}=u_{\alpha_{\ell}}-K$, using elliptic trigonometric identities and the fact that $\bar{\alpha}^r=\bar{\beta}^\ell[2\pi]$.

so that, $\left|\frac{s_{\bs^r,w^r}s_{\bs^\ell,w^\ell} }{t_{\ws^r,v^r}t_{\ws^c,v^c}}\right|=\cn(\theta^\spartial)\eta_{w^\ell v^c}\eta_{w^r v^r}$.

Let us now consider the term $\Bigl|\frac{s_{\bs^{r,\rs},w^{r,\rs}}s_{\bs^{\ell,\rs},w^{\ell,\rs}}}{
t_{\ws^{r,\rs},v^{r,\rs}}t_{\ws^{c,\rs},f^{c,\rs}}}\Bigr|$. For the purpose of this computation, it is useful to imagine that the vertex 
$v^c=\rs$ is present
and that $t_{\ws^{c,\rs},v^{c,\rs}}, \eta_{w^{\ell,\rs} v^{c,\rs}}$ are defined as for the other pairs of rhombi. We then have, omitting to write the superscript
$\rs$,
\begin{align*}
\left|\frac{s_{\bs^r,w^r}s_{\bs^\ell,w^\ell}}{
t_{\ws^r,v^r}t_{\ws^c,f^c}}\frac{1}{\eta_{w^r v^r}\eta_{w^\ell f^c}}\right|&=
\left|\frac{s_{\bs^r,w^r}s_{\bs^\ell,w^\ell} }{t_{\ws^r,v^r}t_{\ws^c,v^c}}
\times \frac{1}{\eta_{w^r v^r}\eta_{w^\ell v^c}}\times\frac{\eta_{w^\ell,v^c}}{\eta_{w^\ell f^c}}\times\frac{t_{\ws^c,v^c}}{t_{\ws^c,f^c}}\right|.
\end{align*}
The product of the first two terms is equal to $|\cn(\theta^\spartial)|$ by the above computation. Then, returning to the definition of 
the weight function $\eta$, we have
\begin{align*}
\left|\frac{\eta_{w^\ell v^c}}{\eta_{w^\ell f}}\right|&=\left|
\frac{\sn(u_{\alpha^\ell+2K})}{\cn(u_{\beta^\ell+2K})}[\nd(u_{\alpha^\ell+2K})\dn(u_{\beta^\ell+2K})]^\frac{1}{2}
\frac{\cn(u_{\beta^\ell})}{\sn(u_{\alpha^\ell+2K})}[\nd(u_{\beta^\ell})\dn(u_{\alpha^\ell+2K})]^\frac{1}{2}(k')^{-\frac{1}{2}}
\right|\\
&=\left|(k')^{-1}\cs(u_{\beta^\ell})\right|.
\end{align*}
Returning to the definition of the matrix $T$ along the boundary, we have
\begin{align*}
\left|\frac{t_{\ws^c,v^c}}{t_{\ws^c,f^c}}\right|&=\left|
\frac{k'\sn(\theta^\spartial)\nd(u_{\alpha^r})\cd(u_{\beta^r})}{\cd(u_{\beta^\ell})}\right|=\left|\frac{k'\sn(\theta^\spartial)\cd(u_{\beta^r})
}{\cn(u_{\alpha^r})}\right|,
\end{align*}
using that $\bar{\beta}^\ell=\bar{\alpha}^r[2\pi]$. Putting the three computations together, and writing the superscript $\rs$ again, we deduce that 
$\left|\frac{s_{\bs^{r,\rs},w^{r,\rs}}s_{\bs^{\ell,\rs},w^{\ell,\rs}}}{
t_{\ws^{r,\rs},v^{r,\rs}}t_{\ws^{c,\rs},f\ro}}\frac{1}{\eta_{w^{r,\rs} v^{r,\rs}}\eta_{w^{\ell,\rs} f\ro}}\right|=
\left|\sn(\theta^{\spartial,\rs})\cn(\theta^{\spartial,\rs})\frac{\cd(u_{\beta^{r,\rs}})}{\sn(u_{\alpha^{r,\rs}})}\right|$, and thus
\begin{align}
\mathrm{(I)}&=\bigl(\sn(\theta^{\spartial,\rs})\cn(\theta^{\spartial,\rs})\bigr)\left|\frac{\cd(u_{\beta^{r,\rs}})}{\sn(u_{\alpha^{r,\rs}})}\right|
\Bigl(\prod_{w^r,w^\ell\in W^\spartial\setminus\{w^{r,\rs},w^{\ell,\rs}\}}\cn(\theta^\spartial)^\frac{1}{2}\Bigr)
\Bigl(\prod_{wx\in\Ms_1:\,w\in W^\spartial }\eta_{w x}\Bigr)\nonumber\\
&=\sn(\theta^{\spartial,\rs})\left|\frac{\cd(u_{\beta^{r,\rs}})}{\sn(u_{\alpha^{r,\rs}})}\right|
\Bigl(\prod_{w\in W^\spartial}\cn(\theta_w)^\frac{1}{2}\Bigr)
\Bigl(\prod_{wx\in\Ms_1:\,w\in W^\spartial }\eta_{w x}\Bigr),
\label{equ:I} 
\end{align}
using that $W^\spartial=\{w_\ell,w_r\in \R^\spartial\}$.
Combining \eqref{equ:II} and \eqref{equ:I} allows to conclude the proof of Lemma~\ref{lem:prod_I_II}.
\end{proof}

The next lemma proves a simplified expression for the product of the weights $\eta_{w x}$ in Lemma~\ref{lem:prod_I_II}.
\begin{lem}\label{lem:simplified}
For every $u\in\Re(\TT(k))''$, we have the following identity,
\begin{equation*}
\prod_{wx\in\Ms_1:w\in W}\eta_{w x}=(k')^{\frac{|\Vs^*|}{2}}\prod_{wx\in\Ms_1:w\in W}|\sc(u_{\alpha_e})\sc(u_{\beta_e})|^{\frac{1}{2}},
\end{equation*}
where the weight function $\eta$ is defined in~\eqref{equ:def_eta_prime}.
\end{lem}
\begin{proof}
We have the following identities:
\begin{align*}
|\sn(u_{\alpha})|\nd(u_{\alpha})^{\frac{1}{2}}&=|\sn(u_{\alpha})\sn(u_{\alpha+2K})|^\frac{1}{2} |\sc(u_{\alpha})|^\frac{1}{2}\\
|\cn(u_{\alpha})|\nd(u_{\alpha})^{\frac{1}{2}}&=|\sn(u_{\alpha})\sn(u_{\alpha+2K})|^\frac{1}{2}|\cs(u_{\alpha})|^\frac{1}{2}.
\end{align*}
As a consequence, for every edge $wx$ of $\GDro$, the weight function $\eta_{w x}$ can be rewritten as:
\begin{equation*}
\eta_{wx}=(k')^{\frac{1}{2}\II_{\{x\in \Vs^*\}}}|\sc(u_{\alpha_e})\sc(u_{\beta_e})|\eta'_{w x}, 
\text{ where } \eta'_{w x}=
\begin{cases}
\left|\frac{\sn(u_{\alpha_e})\sn(u_{\alpha_e+2K})}{\sn(u_{\beta_e})\sn(u_{\beta_e+2K})}\right|^{\frac{1}{2}}&\text{ if $x=v$}\\
\left|\frac{\sn(u_{\beta_e})\sn(u_{\beta_e+2K})}{\sn(u_{\alpha_e})\sn(u_{\alpha_e+2K})}\right|^{\frac{1}{2}}&\text{ if $x=f$}.
\end{cases}
\end{equation*}

Now consider a vertex $w$ of $\GDro$ and, using the notation of Figure~\ref{fig:notations}, the corresponding rhombus $v_1,f_1,v_2,f_2$
of the diamond graph $\GR$. Introduce the following notation for the rhombus vectors:
\begin{equation}\label{equ:star}
(v_1,f_1)=e^{i\bar{\alpha}_1(w)},\ (v_2,f_2)=e^{i\bar{\alpha}_2(w)},\ (v_1,f_2)=e^{i\bar{\beta}_1(w)},\ (v_2,f_1)=e^{i\bar{\beta}_2(w)},
\end{equation}
that is, the notation $\alpha$, resp. $\beta$, is for vectors on the right, resp. left, of the primal edge of the rhombus.
With this notation, the weight $\eta'_{wx}$ can be written as
\[
\eta'_{wx}=\left|\frac{\sn(u_{\alpha_1(w)})\sn(u_{\alpha_2(w)})}{\sn(u_{\beta_1(w)})\sn(u_{\beta_2(w)})}\right|^{\frac{1}{2}},
\]
and this, independently of whether $x=v$ or $f$. Let us prove that
\[
\prod_{wx\in\Ms_1:w\in W}\eta'_{wx}=
\prod_{w\in W}\left|\frac{\sn(u_{\alpha_1(w)})\sn(u_{\alpha_2(w)})}{\sn(u_{\beta_1(w)})\sn(u_{\beta_2(w)})}\right|^{\frac{1}{2}}=1.
\]
Because of~\eqref{equ:star}, the product over white vertices $W$ can be seen as a product over rhombus vectors of $\GR$. 
Then, every inner 
rhombus vector of $\GR$ occurs twice exactly and contributes once to the numerator and once to the denominator, so that the contributions cancel.
Boundary rhombus vectors of $\GR$ occur once but, referring to Section~\ref{sec:train_tracks} on train-tracks, we know that they come in parallel pairs and contribute
once to the numerator and once to the denominator; the contributions thus also compensate ending the proof of this lemma.
\end{proof}
Putting together Corollary~\ref{cor:det_KQ_caca}, Lemmas~\ref{lem:prod_I_II} and~\ref{lem:simplified} ends the proof of Theorem~\ref{thm:part_function}.

\subsection{Dimer model on the graph $\GQ$ and inverse $Z^u$-Dirac operator}\label{sec:inv_KQ_KD}

Using Theorem~\ref{thm:main}, in Corollaries~\ref{cor:KD_KQ} and~\ref{cor:KD_KQ_finite}, we prove linear relations satisfied by 
the inverse Kasteleyn operator $(\KQ)^{-1}$ and the inverse of the $Z^u$-Dirac operators $\KD(u)$ and $\KD^\spartial(u)$. 
Section~\ref{sec:dimer_KQ} is about applications
of these results to the dimer model on $\GQ$. In particular, when the graph $\GQ$ is infinite, we prove an alternative way 
of obtaining a local formula for the inverse~\cite{BdtR2} which is seen as directly related to the $Z$-massive Green functions.

\subsubsection{Inverse Kasteleyn operator $(\KQ)^{-1}$ and inverse $Z^u$-Dirac operator}

\emph{Infinite case}.
In the paper~\cite{BdtR2} we prove an explicit 
\emph{local} expression for an inverse $(\KQ)^{-1}$ of the operator $\KQ$, which decreases to 0 exponentially fast in the distance 
when $k\neq 0$, and as the inverse distance when $k=0$. When $k=0$, the local expression is actually computed in~\cite{Kenyon3}.
When the graph $\GQ$ is $\ZZ^2$-periodic, the operator $(\KQ)^{-1}$ is the unique inverse decreasing to 0 at infinity.
In Corollary~\ref{cor:KD_KQ} below, we use the existence and uniqueness of this inverse operator but not the explicit expression; 
we also need the following notation.

\paragraph{Notation for coefficients of Corollary~\ref{cor:KD_KQ}.} 
Let $\ubar{\ws}$ be a white vertex of $\GQ$ and $\ubar{v},\ubar{f}$ be its adjacent vertices in the diamond graph $\GRR$, 
such that $\ubar{v}\in\Vs$ and $\ubar{f}\in\Vs^*$.
Denote by $e^{i\betaib}$ the rhombus vector corresponding to the edge $(\ubar{\ws},\ubar{v})$. 
Let $w$ be a white vertex of $\GD$ and $\bs,\bs'$ be the black vertices of $\GQ$ of the corresponding quadrangle. 
To the vertex $\bs$, we assign the rhombus vectors $e^{i\alphafb}$, $e^{i\betafb}$ of $\GRR$ of the edge $(\bs,\ws)$, where 
the vertex $\ws$ is such that the edge $\bs\ws$ is parallel to an edge of $\Gs$. 
Then, the rhombus vectors assigned to the vertex $\bs'$ are $e^{i\alphafb+\pi}$, $e^{i\betafb+\pi}$, see Figure~\ref{fig:cor};
the subscripts ``$\mathrm{i}$'' and ``$\mathrm{f}$'' stand for ``initial'' and ``final''.

\begin{figure}[ht]
\begin{minipage}[b]{0.5\linewidth}
\begin{center}
\begin{overpic}[height=3.2cm]{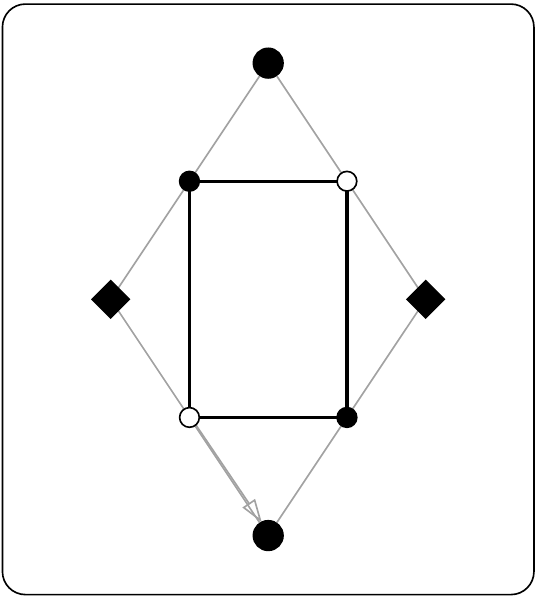}
  \put(21,27){\scriptsize $\ubar{\ws}$}
  \put(23,15){\scriptsize $e^{i\betaib}$}
  \put(50,8){\scriptsize $\ubar{v}$}
  \put(15,38){\scriptsize $\ubar{f}$}
\end{overpic}
\end{center}
\end{minipage}
\begin{minipage}[b]{0.5\linewidth}
\begin{center}
\begin{overpic}[height=3.2cm]{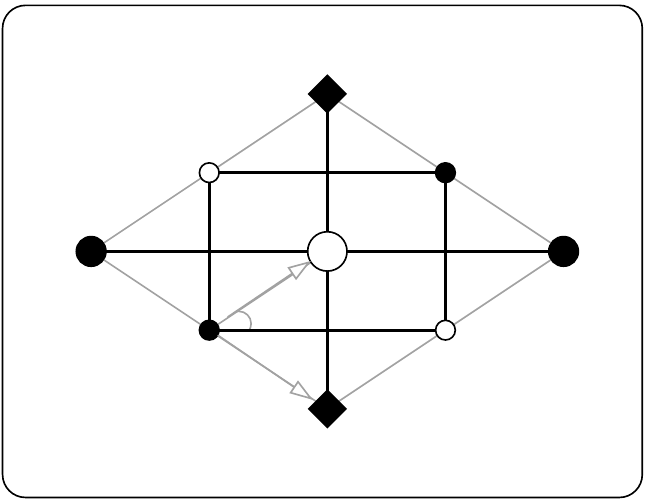}
  \put(29,19){\scriptsize $\bs$}
  \put(69,54){\scriptsize $\bs'$}
  \put(69,19){\scriptsize $\ws$}
  \put(31,14){\scriptsize $e^{i\alphafb}$}
  \put(33,33){\scriptsize $e^{i\betafb}$}
  \put(53,42){\scriptsize $w$}
  \put(42,27.5){\scriptsize $\thetafb$}
\end{overpic}
\end{center}
\end{minipage}

\caption{Notation for coefficients of Corollary~\ref{cor:KD_KQ} and~\ref{cor:KD_KQ_finite}.}
\label{fig:cor}
\end{figure}

As a consequence of Theorem~\ref{thm:main} (infinite case), we obtain the following.

\begin{cor}\label{cor:KD_KQ}
For every $u\in\Re(\TT(k))$, as long as they are unique (which happens for sure
in the $\ZZ^2$-periodic case), the inverse operators $(\KQ)^{-1}$, $\KD(u)^{-1}$, and the matrices $S(u),T(u)$ satisfy the following 
identity.

$\bullet$ \emph{Matrix form.}
\begin{equation*}
(\KQ)^{-1}S(u)=T(u)\KD(u)^{-1}.
\end{equation*}
$\bullet$ \emph{Coefficients.} For every $\ubar{\ws},\ubar{v},\ubar{f}$, and every $w,\bs,\bs'$ as in the notation above, we have:
\begin{equation*}
(\KQ)^{-1}_{\ubar{\ws},\bs}s(u)_{\bs,w}+ (\KQ)^{-1}_{\ubar{\ws},\bs'}s(u)_{\bs',w}=
t(u)_{\ubar{\ws},\ubar{v}}\KD(u)^{-1}_{\ubar{v},w}+t(u)_{\ubar{\ws},\ubar{f}}\KD(u)^{-1}_{\ubar{f},w}.
\end{equation*}
Or equivalently,
\begin{equation}\label{equ:KQKD_gen}
\textstyle
(\KQ)^{-1}_{\ubar{\ws},\bs}\cn(u_{\betaf})-i(\KQ)^{-1}_{\ubar{\ws},\bs'}\sn(u_{\betaf})\dn(u_{\alphaf})
=\frac{e^{i\frac{\betafb-\betaib}{2}}}{\Lambda(u_{\alphaf},u_{\betaf})}
\left[\cn(u_{\betai})\KD(u)^{-1}_{\ubar{v},w}-i\sn({u_{\betai}})\KD(u)^{-1}_{\ubar{f},w}\right],
\end{equation}
where $\Lambda(u_\alpha,u_\beta)=[\sn\theta\cn\theta\nd(u_{\alpha})\nd(u_{\beta})]^{\frac{1}{2}}$, and 
$\theta=u_\alpha-u_\beta$.
\end{cor}
\begin{proof}
The matrix form is obtained by left multiplying by $(\KQ)^{-1}$ and right multiplying by $\KD(u)^{-1}$ Equation~\eqref{equ:thm_main_3}
of Theorem~\ref{thm:main}. We are allowed to do so because, coefficients of the inverses decrease to 0 at infinity and the 
other matrices involved only have finitely many non-zero terms per row and column, implying associativity of the infinite 
matrix products. We also use uniqueness of the right (or left) inverse. Indeed, together with 
the fact that the products $\KQ (\KQ)^{-1}\KQ$ and $\KD(u)\KD(u)^{-1}\KD(u)$ are associative, this implies that they each are 
inverses on both sides~\cite{Cooke}.

For coefficients, we return to the definition of the matrix $S(u)$, see \eqref{def:sbw}, and obtain
\begin{align*}
s(u)_{\bs,w}&=e^{-i\frac{\betafb}{2}} 
[\sn(\thetaf)\cn(\thetaf)\nd(u_{\alphaf})\nd(u_{\betaf})]^{\frac{1}{2}}\cn(u_{\betaf})
=e^{-i\frac{\betafb}{2}}\cn(u_{\betaf})\Lambda(u_{\alphaf},u_{\betaf}) \\
s(u)_{\bs',w}&=
e^{-i\frac{\betafb+\pi}{2}}[\sn(\thetaf)\cn(\thetaf)\nd(u_{\alphaf+2K})\nd(u_{\betaf+2K})]^{\frac{1}{2}} \cn(u_{\betaf+2K})\\
&=-i e^{-i\frac{\betafb}{2}} [\sn(\thetaf)\cn(\thetaf)\nd(u_{\alphaf})\nd(u_{\betaf})]^{\frac{1}{2}} 
\sn(u_{\betaf})\dn(u_{\alphaf})\\
&=-i e^{-i\frac{\betafb}{2}} \sn(u_{\betaf})\dn(u_{\alphaf}) \Lambda(u_{\alphaf},u_{\betaf}),
\end{align*}
using that $\cn(u-K)=k'\sd(u-K)$, $\nd(u-K)=(k')^{-1}\dn(u)$ in the penultimate line.
Returning to the definition of coefficients of the matrix $T$, see~\eqref{def:twvf}, we have
\begin{align*}
t(u)_{\ubar{\ws},\ubar{v}}&= e^{-i\frac{\betaib}{2}}\cn(u_{\beta_i})\\
t(u)_{\ubar{\ws},\ubar{f}}&=-i e^{-i\frac{\betaib}{2}}\cd(u_{\betai+2K})=
-i e^{-i\frac{\betaib}{2}}\sn(u_{\betai}),
\end{align*}
using that $\cd(u-K)=\sn(u)$ in the last equality. This ends the proof of the formula for coefficients.
\end{proof}

\emph{Finite case.} We restrict to $u\in\Re(\TT(k))''$ so that the $Z^u$-Dirac operator $\KD^\spartial(u)$ is invertible.
\paragraph{Notation for coefficients of Corollary~\ref{cor:KD_KQ_finite}.}
As in the notation for coefficients of Corollary~\ref{cor:KD_KQ}, we add a subscript $\mathrm{i}$/$\mathrm{f}$ for rhombus vectors and 
half-angles of initial/final vertices. If $\ubar{\ws}\in\{\ws^c\in\R^\spartial\}$, or $w\in\{w^\ell,w^r\in\R^\spartial\}$, we thus have the 
notation of Figure~\ref{fig:cor_finite}.
\begin{figure}[ht]
\begin{minipage}[b]{0.5\linewidth}
\begin{center}
\begin{overpic}[height=3.6cm]{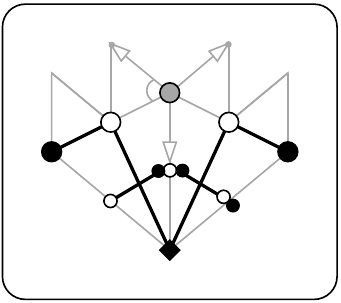}
\put(47,67){\scriptsize $\ubar{v}^c$}
\put(47,30){\scriptsize $\ubar{\ws}^c$}
\put(47,5){\scriptsize $\ubar{f}^c$}
\put(36,60){\scriptsize $\bar{\theta}^\spartial_{\mathrm{i}}$}
\end{overpic}
\end{center}
\end{minipage}
\begin{minipage}[b]{0.5\linewidth}
\begin{center}
\begin{overpic}[height=3.6cm]{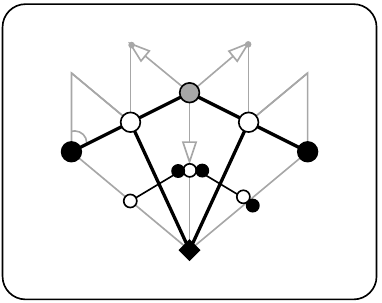}
\put(27,52.5){\scriptsize $w^\ell$}
\put(65.5,52.5){\scriptsize $w^r$}
\put(44,27){\scriptsize $\bs^\ell$}
\put(50,27){\scriptsize $\bs^r$}
\put(35.5,54){\scriptsize $e^{i\bar{\alpha}^\ell_{\mathrm{f}}}$}
\put(53,54){\scriptsize $e^{i\bar{\beta}^r_{\mathrm{f}}}$}
\put(37,42){\scriptsize $e^{i\bar{\beta}^\ell_{\mathrm{f}}}$}
\put(50.5,42){\scriptsize $e^{i\bar{\alpha}^r_{\mathrm{f}}}$}
\put(21,46){\scriptsize $\bar{\theta}^\spartial_{\mathrm{f}}$}
\end{overpic}
\end{center}
\end{minipage}
\caption{Notation for boundary coefficients of Corollary~\ref{cor:KD_KQ_finite}.}
\label{fig:cor_finite}
\end{figure}
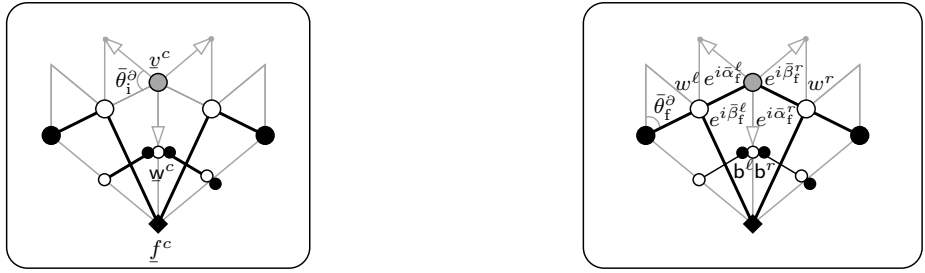

As a consequence of Theorem~\ref{thm:main} (finite case), we obtain the following.
\begin{cor}\label{cor:KD_KQ_finite}
For every $u\in\Re(\TT(k))''$, the inverse operators $(\KQu)^{-1},(\KQ)^{-1},\KD^\spartial(u)^{-1}$ and the matrices $S(u),T(u)$ satisfy the following identity.

$\bullet$ \emph{Matrix form.}
\begin{equation*}
(\KQu)^{-1}S(u)=T(u)\KD^\spartial(u)^{-1}  \quad \Leftrightarrow \quad (\DQW)^{-1}(\KQ)^{-1}(\DQB)^{-1}S(u)=T(u)\KD^\spartial(u)^{-1}.
\end{equation*}
$\bullet$ \emph{Coefficients.} We have two cases to consider. 

1. For every $\ubar{\ws},\ubar{v},\ubar{f}$; for every $w,\bs,\bs'$ such that $w\notin\{w^\ell,w^r\in\R^\spartial\}$,
using the notation of Figure~\ref{fig:cor} and~\ref{fig:cor_finite} (left), we have
\begin{align*}
(\KQ)^{-1}_{\ubar{\ws},\bs}\cn(u_{\betaf})-i&(\KQ)^{-1}_{\ubar{\ws},\bs'}\sn(u_{\betaf})\dn(u_{\alphaf})=\\
&=\frac{e^{i\frac{\betafb}{2}}\sn(\theta^{\spartial}_\irm)^{\II_{\{\ubar{\ws}\in\{\ws^c\in\R^\spartial\}\}}}}{\Lambda(u_{\alphaf},u_{\betaf})}
\left(\II_{\{\ubar{\ws}\neq \ws^{c,\rs}\}}
t(u)_{\ubar{\ws},\ubar{v}}\KD^\spartial(u)^{-1}_{\ubar{v},w}+t(u)_{\ubar{\ws},\ubar{f}}\KD^\spartial(u)^{-1}_{\ubar{f},w}
\right).
\end{align*}
2. For every $\ubar{\ws},\ubar{v},\ubar{f}$, for every $w$ such that $w\in\{w^\ell,w^r\}$ for one of the rhombus pairs of 
$\R^\spartial$, then using the notation of Figures~\ref{fig:cor} and~\ref{fig:cor_finite}, we have
\begin{align*}
(\KQ)^{-1}_{\ubar{\ws},\bs^\ell}\cn(u_{\alphaf^\ell})&=
\frac{e^{i\frac{\bar{\alpha}^\ell_\frm}{2}}\sn(\theta^\spartial_\frm)^{-1}\sn(\theta^{\spartial}_\irm)^{\II_{\{\ubar{\ws}\in\{\ws^c\in\R^\spartial\}\}}}
}{\Lambda(u_{\alpha^\ell_\frm},u_{\beta^\ell_\frm})}
\left(\II_{\{\ubar{\ws}\neq \ws^{c,\rs}\}}
t(u)_{\ubar{\ws},\ubar{v}}\KD^\spartial(u)^{-1}_{\ubar{v},w^\ell}+t(u)_{\ubar{\ws},\ubar{f}}\KD^\spartial(u)^{-1}_{\ubar{f},w^\ell}\right)\\
(\KQ)^{-1}_{\ubar{\ws},\bs^r}\cn(u_{\betaf^r})&=
\frac{e^{i\frac{\bar{\beta}^r_\frm}{2}}\sn(\theta^\spartial_\frm)^{-1}\sn(\theta^{\spartial}_\irm)^{\II_{\{\ubar{\ws}\in\{\ws^c\in\R^\spartial\}\}}}
}{\Lambda(u_{\alpha^r_\frm},u_{\beta^r_\frm})}
\left(\II_{\{\ubar{\ws}\neq \ws^{c,\rs}\}}
t(u)_{\ubar{\ws},\ubar{v}}\KD^\spartial(u)^{-1}_{\ubar{v},w^r}+t(u)_{\ubar{\ws},\ubar{f}}\KD^\spartial(u)^{-1}_{\ubar{f},w^r}\right),
\end{align*}
where coefficients of $t(u)$ are given by Equation~\eqref{def:twvf} or~\eqref{def:twv_boundary} when $\ubar{\ws}=\ws^c$ for some boundary 
rhombus pair.
\end{cor}

\begin{exm}\label{ex:KQ_KD_u_v_special}
Let us compute the explicit values of Corollary~\ref{cor:KD_KQ} and of Point 1. of Corollary~\ref{cor:KD_KQ_finite} in the case where
$u=\ubf:=\frac{\alphaf+\betaf}{2}+K$. This will be used again in Sections~\ref{sec:dimer_KQ} and~\ref{sec:GQ_GF_Zinv}.
In the finite case, we also need to evaluate it at $u=\vbf:=\frac{\alphaf+\betaf}{2}-K$. Returning to Relations~\eqref{equ:relation_u_v} of
Example~\ref{ex:KD_u_v_special}, we have,
\begin{align*}
\sn(\ubf_{\betaf})\dn(\ubf_{\alphaf})&=\sn(K-\ubf_{\alphaf})\dn(\ubf_{\alphaf})=\cn(\ubf_{\alphaf})\\
\sn(\vbf_{\betaf})\dn(\vbf_{\alphaf})&=\sn(-K-\vbf_{\alphaf})\dn(\vbf_{\alphaf})=-\cn(\vbf_{\alphaf})\\
\Lambda(\ubf_{\alphaf},\ubf_{\betaf})&=[(k')^{-1}\sn(\thetaf)\cn(\thetaf)]^{\frac{1}{2}}=
\Lambda(\vbf_{\alphaf},\vbf_{\betaf})\\
\ubf_{\betaf}=\frac{K-\thetaf}{2},&\ \ubf_{\alphaf}=\frac{K+\thetaf}{2},\ \vbf_{\betaf}=-\ubf_{\alphaf},\ \vbf_{\alphaf}=-\ubf_{\betaf}.
\end{align*}
Then, 

$\bullet$ \emph{Infinite case.} For every $\ubar{\ws},v,f$; for every $w,\bs,\bs'$, with the notation of Figure~\ref{fig:cor}, we have
\begin{equation*}
\textstyle
(\KQ)^{-1}_{\ubar{\ws},\bs}\cn\bigl(\frac{K-\thetaf}{2}\bigr)-i(\KQ)^{-1}_{\ubar{\ws},\bs'}\sn\bigl(\frac{K+\thetaf}{2}\bigr)
=\frac{e^{i\frac{\betafb-\betaib}{2}}(k')^{\frac{1}{2}}}{[\cn(\thetaf)\sn(\thetaf)]^{\frac{1}{2}}}
\left(\cn(\ubf_{\betai})\KD(\ubf)^{-1}_{\ubar{v},w}-i\sn({\ubf_{\betai}})\KD(\ubf)^{-1}_{\ubar{f},w}\right).
\end{equation*}
$\bullet$ \emph{Finite case. Point 1.}

For every $\ubar{\ws},\ubar{v},\ubar{f}$; for every $w,\bs,\bs'$ such that $w\notin\{w^\ell,w^r\in\R^\spartial\}$,
with the notation of Figures~\ref{fig:cor} and~\ref{fig:cor_finite} (left), we have
\begin{align*}
\textstyle
(\KQ)^{-1}_{\ubar{\ws},\bs}\cn\bigl(\frac{K-\thetaf}{2}\bigr)&-i\textstyle (\KQ)^{-1}_{\ubar{\ws},\bs'}\sn\bigl(\frac{K+\thetaf}{2}\bigr)=\\
&\textstyle=\frac{e^{i\frac{\betafb}{2}}(k')^{\frac{1}{2}}\sn(\theta^{\spartial}_\irm)^{\II_{\{\ubar{\ws}\in\{\ws^c\in\R^\spartial\}\}}}}{
[\cn(\thetaf)\sn(\thetaf)]^{\frac{1}{2}}}
\left(\II_{\{\ubar{\ws}\neq \ws^{c,\rs}\}}
t(\ubf)_{\ubar{\ws},\ubar{v}}\KD^\spartial(\ubf)^{-1}_{\ubar{v},w}+t(\ubf)_{\ubar{\ws},\ubar{f}}\KD^\spartial(\ubf)^{-1}_{\ubar{f},w}
\right)\\
\textstyle (\KQ)^{-1}_{\ubar{\ws},\bs}\cn\bigl(\frac{K+\thetaf}{2}\bigr)&+i\textstyle (\KQ)^{-1}_{\ubar{\ws},\bs'}\sn\bigl(\frac{K-\thetaf}{2}\bigr)=\\
&\textstyle=\frac{e^{i\frac{\betafb}{2}}(k')^{\frac{1}{2}}\sn(\theta^{\spartial}_\irm)^{\II_{\{\ubar{\ws}\in\{\ws^c\in\R^\spartial\}\}}}}{
[\cn(\thetaf)\sn(\thetaf)]^{\frac{1}{2}}}
\left(\II_{\{\ubar{\ws}\neq \ws^{c,\rs}\}}
t(\vbf)_{\ubar{\ws},\ubar{v}}\KD^\spartial(\vbf)^{-1}_{\ubar{v},w}+t(\vbf)_{\ubar{\ws},\ubar{f}}\KD^\spartial(\vbf)^{-1}_{\ubar{f},w}
\right).
\end{align*}
\end{exm}

\subsubsection{The dimer model on an isoradial graph $\GQ$ and the $Z^u$-Dirac operator}\label{sec:dimer_KQ}

\emph{Infinite case.}
In the paper~\cite{BdtR2}, we prove an explicit \emph{local} expression for an inverse $(\KQ)^{-1}$; as a byproduct we obtain a \emph{local}
formula for a Gibbs measure $\PPdimerQ$ on $(\M(\GQ),\F)$, involving the operators $\KQ$ and $(\KQ)^{-1}$; we refer to the paper~\cite{BdtR2}
for the explicit formula for $(\KQ)^{-1}$ and to Section~\ref{sec:dimer_infinite} for the explicit formula of the Gibbs measure $\PPdimerQ$.

Using Corollary~\ref{cor:KD_KQ} and Corollary~\ref{cor:KD_G}, we provide an alternative \emph{direct} way of finding the \emph{local} formula 
for the inverse
operator $(\KQ)^{-1}$ of~\cite{BdtR2}, where the locality property is directly seen as inherited from that of the $Z$-massive and dual massive Green 
functions: 
we first express coefficients of $(\KQ)^{-1}$ using the inverse $Z^{u}$-Dirac operator
for appropriate values of $u$, and then express the latter using the $Z$-massive and dual massive Green function of~\cite{BdTR1}.
Note that it is not immediate to see equality between the formulas of~\cite{BdtR2} and Corollay~\ref{cor:gibbsGQ_GD}; it
probably requires to use elliptic trigonometric identities. The approach we propose here also extends to the finite case. 

\paragraph{Notation for Corollary~\ref{cor:gibbsGQ_GD}.}

Let $\ubar{\ws}$ be a white vertex of $\GQ$ and $\bs$ be a black one. Consider 
$\ubar{v},\ubar{f},e^{i\betaib}$, and $w,\bs,\ws,e^{i\alphafb},e^{i\betafb}$ as in Figure~\ref{fig:cor}. 
The quadrangle of the vertex $\bs$ corresponds to a rhombus $v_1,f_1,v_2,f_2$ of the 
diamond graph $\GR$, where vertices are labeled so that the edge $(v_1,v_2)$ of $\Gs$ is parallel to the edge $(\bs,\ws)$ of $\GQ$, and $f_1$ is on the right of 
$(v_1,v_2)$, see Figure~\ref{fig:cor_dimerKQ}. 

\begin{figure}[ht]
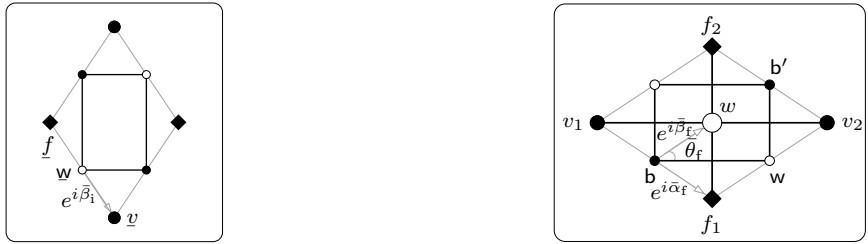

\begin{minipage}[b]{0.5\linewidth}
\begin{center}
\begin{overpic}[height=3.2cm]{fig_corKQKD_0.pdf}
  \put(21,27){\scriptsize $\ubar{\ws}$}
  \put(23,15){\scriptsize $e^{i\betaib}$}
  \put(50,8){\scriptsize $\ubar{v}$}
  \put(15,38){\scriptsize $\ubar{f}$}
\end{overpic}
\end{center}
\end{minipage}
\begin{minipage}[b]{0.5\linewidth}
\begin{center}
\begin{overpic}[height=3.2cm]{fig_corKQKD.pdf}
  \put(29,19){\scriptsize $\bs$}
  \put(69,54){\scriptsize $\bs'$}
  \put(69,19){\scriptsize $\ws$}
  \put(31,14){\scriptsize $e^{i\alphafb}$}
  \put(33,33){\scriptsize $e^{i\betafb}$}
  \put(53,42){\scriptsize $w$}
  \put(42,27.5){\scriptsize $\thetafb$}
  \put(3,37){\scriptsize $v_1$}
  \put(92,37){\scriptsize $v_2$}
  \put(47,4){\scriptsize $f_1$}
  \put(47,68){\scriptsize $f_2$}
\end{overpic}
\end{center}
\end{minipage}
\caption{Notation for Corollary~\ref{cor:gibbsGQ_GD}.}
\label{fig:cor_dimerKQ}
\end{figure}

\begin{cor}\label{cor:gibbsGQ_GD}
For every white vertex $\ubar{\ws}$ and every black vertex $\bs$ of $\GQ$, 
using the notation of Figure~\ref{fig:cor_dimerKQ}, we have
\begin{align*}
(\KQ)^{-1}_{\ubar{\ws},\bs}&=\textstyle\frac{
e^{i\frac{\betafb-\betaib}{2}}}
{[\cn(\thetaf)\sn(\thetaf)\nd(\thetaf)]^{\frac{1}{2}}}
\left(\cn\bigl(\frac{\betaf-\betai}{2}\bigr)\KD(\betaf)^{-1}_{\ubar{v},w}
-i\sn\bigl(\frac{\betaf-\betai}{2}\bigr)\KD(\betaf)^{-1}_{\ubar{f},w}
\right)\\
(\KQ)^{-1}_{\ubar{\ws},\bs}&=\textstyle e^{-i\frac{\betaib+\alphafb}{2}}(k')^{-1}
\Bigl(\frac{\cn\bigl(\frac{\betaf-\betai}{2}\bigr)}{\cn(\thetaf)}
\bigl[\dn(\thetaf) G^m_{\ubar{v},v_2}-k'G^m_{\ubar{v},v_1}\bigr]
-\frac{\sn\bigl(\frac{\betaf-\betai}{2}\bigr)}{\sn(\thetaf)}
\bigl[\dn(\thetaf)G^{m,*}_{\ubar{f},f_2}-G^{m,*}_{\ubar{f},f_1}\bigr]\Bigr).
\end{align*}
\end{cor}
\begin{proof}
We set $u=\betaf$ in Corollary~\ref{cor:KD_KQ}. Then, $u_{\alphaf}=\thetaf$, $u_{\betaf}=0$, and we have $\cn(u_{\betaf})=1$, 
$\sn(u_{\betaf})=0$, $\dn(u_{\betaf})=1$. This gives
$\Lambda(u_{\alphaf},u_{\betaf})\vert_{\betaf}=[\sn\thetaf\cn\thetaf\nd(\thetaf)]^{\frac{1}{2}}$, thus explaining the first equality of the Corollary.

We now set $u=\betaf$ in Corollary~\ref{cor:KD_G}. Using that
$\dn(-u)=\dn(u)$, $\dn(K\pm u)=k'\nd(u)$, $\dn(K)=k'$, $\sc(\thetaf^*)=(k')^{-1}\cs(\thetaf)$, we obtain:
\begin{align*}
\KD(\betaf)^{-1}_{\ubar{v},w}&=e^{-i\frac{\alphafb+\betafb}{2}}(k')^{-1}\sc(\thetaf)^{\frac{1}{2}}
\Bigl(\dn(\thetaf)^{\frac{1}{2}} G^m_{\ubar{v},v_2}-[(k')^{2}\nd(\thetaf)]^{\frac{1}{2}}G^m_{\ubar{v},v_1}\Bigr)\\
\KD(\betaf)^{-1}_{\ubar{f},w}&=-ie^{-i\frac{\alphafb+\betafb}{2}}(k')^{-1}\sc(\thetaf^*)^{\frac{1}{2}}
\Bigl([\dn(K)\dn(\thetaf)]^{\frac{1}{2}}G^{m,*}_{\ubar{f},f_2}-[\dn(0)\dn(\thetaf-K)]^{\frac{1}{2}}G^{m,*}_{\ubar{f},f_1} \Bigr)\\
&=-ie^{-i\frac{\alphafb+\betafb}{2}}(k')^{-1}\cs(\thetaf)^{\frac{1}{2}}
\Bigl(\dn(\thetaf)^{\frac{1}{2}}G^{m,*}_{\ubar{f},f_2}-\nd(\thetaf)^{\frac{1}{2}}G^{m,*}_{\ubar{f},f_1}
\Bigr).
\end{align*}
Plugging this into the first equality of the corollary yields the second and concludes the proof.
\end{proof}

\begin{exm}\label{ex:proba_GQ}

As an example of application we express the probability of single edges occurring in dimer configurations of $\GQ$ chosen 
with respect to the measure $\PPdimerQ$,
as a function of single edge probabilities of the dimer model on $\GD$ with Gibbs measure
$\PPdimerDbeta$. We then use Example~\ref{ex:Prob_GD} evaluated at $u=\beta$ to obtain explicit expressions using the function $H$;
details of computations are given in Appendix~\ref{app:dimers_GQ}. Using the notation of Figure~\ref{fig:notations}, we have
\begin{align*}
\PPdimerQ(\bs\ws_1)&=\PPdimerDbeta(wv_2)=H(2\theta)\\ 
\PPdimerQ(\bs\ws_2)&=\PPdimerDbeta(wf_2)=\frac{1}{2}-H(2\theta)\\ 
\PPdimerQ(\bs\ws_3)&=\frac{1}{2}.
\end{align*}
\end{exm}
We recover the results of the computations of Theorem 37 of~\cite{BdtR2} after using addition formulas.

\emph{Finite case.} An explicit expression for the Boltzmann measure $\PPdimerQ$
as a function of the matrix $\KQ$ and its inverse $(\KQ)^{-1}$ is given in Section~\ref{sec:dimer_model}. We now express the inverse Kasteleyn matrix
$(\KQ)^{-1}$ as a function of the $Z^u$-Dirac operator for some choices of $u$.
Note that we cannot proceed as in the infinite case using $u=\betaf$ since then the matrix $\KD^\spartial(\betaf)$ is not invertible.

\begin{cor}\label{cor:BoltzmannGQ_GD}$\,$

\emph{1.} For every white vertex $\ubar{\ws}$ of $\GQ$ and every black vertex $\bs\notin\{\bs^\ell,\bs^r\in\R^{\spartial}\}$,
using the notation of Figures~\ref{fig:cor} and~\ref{fig:cor_finite} (left), let
$\ubf=\frac{\alphaf+\betaf}{2}+K$, $\vbf=\frac{\alphaf+\betaf}{2}-K$. Then,
\begin{equation*}
\textstyle
(\KQ)^{-1}_{\ubar{\ws},\bs}=e^{i\frac{\betafb}{2}}(k')^{\frac{1}{2}}\sn(\theta^{\spartial}_\irm)^{\II_{\{\ubar{\ws}\in\{\ws^c\in\R^\spartial\}\}}}
\frac{1+(k')^{-1}\dn(\thetaf)}{2[\cn(\thetaf)\sn(\thetaf)]^{\frac{1}{2}}}
\Bigl[\cn\bigl(\frac{K-\thetaf}{2}\bigr)\Gamma(\ubf)+ \cn\bigl(\frac{K+\thetaf}{2}\bigr)\Gamma(\vbf)\Bigr],
\end{equation*}
where, $\Gamma(u)=\II_{\{\ubar{\ws}\neq \ws^{c,\rs}\}}
t(u)_{\ubar{\ws},\ubar{v}}\KD^\spartial(u)^{-1}_{\ubar{v},w}+t(u)_{\ubar{\ws},\ubar{f}}\KD^\spartial(u)^{-1}_{\ubar{f},w}$, and coefficients 
of $t(u)$ are given by~\eqref{def:twvf} and~\eqref{def:twv_boundary}.

\emph{2.} For every white vertex $\ubar{\ws}$ and every black vertex $\bs\in\{\bs^\ell,\bs^r\}$ for one of the rhombus pairs of $\R^{\spartial}$,
using the notation of Figures~\ref{fig:cor} and~\ref{fig:cor_finite}, let $\ubfl=\frac{\alphaf^\ell+\betaf^\ell}{2}+K$ and 
$\ubfr=\frac{\alphaf^r+\betaf^r}{2}+K$. Then, 
\begin{align*}
\textstyle
(\KQ)^{-1}_{\ubar{\ws},\bs^\ell}
=&\textstyle
\frac{e^{i\frac{\bar{\alpha}^\ell_\frm}{2}}(k')^{\frac{1}{2}}\sn(\theta^\spartial_\frm)^{-1}\sn(\theta^{\spartial}_\irm)^{\II_{\{\ubar{\ws}\in\{\ws^c\in\R^\spartial\}\}}}
}{\cn\bigl(\frac{K+\thetaf^\spartial}{2}\bigr)[\cn(\thetaf^\spartial)\sn(\thetaf^\spartial)]^{\frac{1}{2}}}\times\\
&\hspace{2.5cm}\times\left(\II_{\{\ubar{\ws}\neq \ws^{c,\rs}\}}
t(\ubfl)_{\ubar{\ws},\ubar{v}}\KD^\spartial(\ubfl)^{-1}_{\ubar{v},w^\ell}+t(\ubfl)_{\ubar{\ws},\ubar{f}}\KD^\spartial(\ubfl)^{-1}_{\ubar{f},w^\ell}\right)\\
\textstyle
(\KQ)^{-1}_{\ubar{\ws},\bs^r}\textstyle
=&\textstyle \frac{e^{i\frac{\bar{\beta}^r_\frm}{2}}(k')^{\frac{1}{2}}\sn(\theta^\spartial_\frm)^{-1}\sn(\theta^{\spartial}_\irm)^{\II_{\{\ubar{\ws}\in\{\ws^c\in\R^\spartial\}\}}}
}{\cn\bigl(\frac{K-\thetaf^\spartial}{2}\bigr)[\cn(\thetaf^\spartial)\sn(\thetaf^\spartial)]^{\frac{1}{2}}}\times\\
&\hspace{2.5cm}\times \left(\II_{\{\ubar{\ws}\neq \ws^{c,\rs}\}}
t(\ubfr)_{\ubar{\ws},\ubar{v}}\KD^\spartial(\ubfr)^{-1}_{\ubar{v},w^r}+t(\ubfr)_{\ubar{\ws},\ubar{f}}\KD^\spartial(\ubfr)^{-1}_{\ubar{f},w^r}\right).
\end{align*}
\end{cor}

\begin{rem}\label{rem:final_0}
Example~\ref{ex:KD_u_v_special} expresses coefficients of $\KD(\ubf)^{-1}, \KD(\vbf)^{-1}$ using the 
$Z$-massive and dual massive Green function. Plugging this into Corollary~\ref{cor:BoltzmannGQ_GD} allows to 
write coefficients of $(\KQ)^{-1}$ using the Green functions $G^{m,\spartial}(\ubf),G^{m,\spartial}(\vbf),G^{m,*}$.
\end{rem}

\begin{proof}
Consider $\ubar{\ws},\ubar{v},\ubar{f}$ and $w,\bs,\bs'$ such that $w\notin\{w^\ell,w^r\in\R^\spartial\}$ as in the statement of Corollary~\ref{cor:KD_KQ_finite}. Then,
from Example~\ref{ex:KQ_KD_u_v_special}, we have the following linear system of equations: 
\begin{align*}
&\textstyle
(\KQ)^{-1}_{\ubar{\ws},\bs}\cn\bigl(\frac{K-\thetaf}{2}\bigr)-i\textstyle (\KQ)^{-1}_{\ubar{\ws},\bs'}\sn\bigl(\frac{K+\thetaf}{2}\bigr)
=C\times 
\Gamma(\ubf)\\
&\textstyle
(\KQ)^{-1}_{\ubar{\ws},\bs}\cn\bigl(\frac{K+\thetaf}{2}\bigr)+i\textstyle (\KQ)^{-1}_{\ubar{\ws},\bs'}\sn\bigl(\frac{K-\thetaf}{2}\bigr)
=C\times 
\Gamma(\vbf),
\end{align*}
where $C=\frac{e^{i\frac{\betafb}{2}}(k')^{\frac{1}{2}}\sn(\theta^{\spartial}_\irm)^{\II_{\{\ubar{\ws}\in\{\ws^c\in\R^\spartial\}\}}}}{
[\cn(\thetaf)\sn(\thetaf)]^{\frac{1}{2}}}$,
$\Gamma(u)=\II_{\{\ubar{\ws}\neq \ws^{c,\rs}\}}
t(u)_{\ubar{\ws},\ubar{v}}\KD^\spartial(u)^{-1}_{\ubar{v},w}+t(u)_{\ubar{\ws},\ubar{f}}\KD^\spartial(u)^{-1}_{\ubar{f},w}$.

Solving for $(\KQ)^{-1}_{\ubar{\ws},\bs}$ gives,
\begin{equation*}
\textstyle
(\KQ)^{-1}_{\ubar{\ws},\bs}=\frac{C}{\cn^2\bigl(\frac{K-\thetaf}{2}\bigr)+\cn^2\bigl(\frac{K+\thetaf}{2}\bigr)}
[\cn\bigl(\frac{K-\thetaf}{2}\bigr)\Gamma(\ubf)+ \cn\bigl(\frac{K+\thetaf}{2}\bigr)\Gamma(\vbf)].
\end{equation*}
Now, using~\cite[2.4.8]{Lawden} and the fact that $\dn(2K-u)=\dn(u),\, \cn(2K-u)=-\cn(u)$, we have for every $u\in\TT(k)$,
\begin{align}
\textstyle 
&\cn^2(u)+\cn^2(K-u)=\textstyle \frac{\cn(2u)+\dn(2u)}{1+\dn(2u)}+\frac{\cn(2K-2u)+\dn(2K-2u)}{1+\dn(2K-2u)}=\frac{2\dn(2u)}{1+\dn(2u)}
=\frac{2}{1+\nd(2u)}.\nonumber\\
&\text{This implies that, } \quad 
\textstyle
\cn^2\bigl(\frac{K-\thetaf}{2}\bigr)+\cn^2\bigl(\frac{K+\thetaf}{2}\bigr)=\frac{2}{1+\nd\bigl(2\frac{K-\thetaf}{2}\bigr)}
=\frac{2}{1+(k')^{-1}\dn(\thetaf)}.\label{equ:id_somme_cn_carre}
\end{align}
Putting everything together ends the proof of Point 1. Point 2. directly follows from Point 2. of Corollary~\ref{cor:KD_KQ_finite}.
\end{proof}

\section{Dimer model on the Fisher graph $\GF$ and the Kasteleyn matrix $\KQ$}\label{sec:KFKQ}

In the whole of this section, $\Gs$ is a planar, simply connected graph defined as in Section~\ref{sec:defIsing}; it is infinite or finite in which case
it has an additional vertex on every boundary edge. Unless specified, 
definitions and results hold for the finite and infinite cases; until Section~\ref{sec:GQ_GF_Zinv}, we do not suppose that $\Gs$ is isoradial.

Consider the dimer model on the Fisher graph $\GF$ with weight function $\muJ$ arising from the LTE of the Ising model on $\Gs$
with coupling constants $\Js$ and in the finite case, + boundary conditions, see Section~\ref{sec:def_Fisher_graph}. 
Consider also the dimer model on the bipartite graph $\GQ$ with weight function
$\nuJ$ arising for the XOR-Ising model, see Section~\ref{sec:def_graph_GQ}. 

Suppose that edges of $\GF$ are oriented according to a Kasteleyn orientation, and denote by $\KF$ the corresponding Kasteleyn matrix,
as defined in Section~\ref{sec:dimer_model}. Following Dub\'edat~\cite{Dubedat}, we partition vertices of $\GF$ as $\VF=A\cup B$, where $A$ consists of 
vertices incident to four internal edges of $\GF$, and $B$ consists of those incident to either two internal and one external edges or
to two internal edges (this possibility only occurs in the finite case). Vertices of type $A$, resp. $B$, will be denoted by $a$, resp. 
$b$, with or without sub/super-scripts. Up to a reordering of the rows and columns, the matrix $\KF$ can be written
in block form as

\begin{equation*}
\KF=
\begin{pmatrix}
\KF_{\sB,\sB}&\KF_{\sB,\sA}\\
\KF_{\sA,\sB}&\KF_{\sA,\sA}
\end{pmatrix}.
\end{equation*}

Recall that in the finite case, boundary quadrangles of $\GQ$ are degenerate and consist of edges in bijection with boundary edges of $\Gs$. 
\emph{Boundary $B$-vertices} of $\GF$ are defined to be $B$-vertices incident to two internal edges only; they are in natural bijection with black, resp. 
white, vertices of boundary quadrangles of $\GQ$, and also with boundary white vertices of the double graph $\GD$.

In~\cite{Dubedat} Dub\'edat shows how, in the case where $\Gs$ is $\ZZ^2$, a Kasteleyn orientation on $\GF$ induces a Kasteleyn 
orientation on $\GQ$; this generalizes to the case where $\Gs$ is planar: using the notation of Figure~\ref{fig:matrixKQKF} 
below, define
\begin{align}\label{equ:def_orientation}
\eps_{\bs,\ws}=\eps_{b,a},\quad
\eps_{\bs,\ws'}=\eps_{b,a'},\quad
\eps_{\bs,\ws''}=\eps_{b,b'}\eps_{b',a''},
\end{align}
then it is straightforward to check that the orientation so defined on $\GQ$ is Kasteleyn. Note that when $\Gs$ is finite, the case
$\eps_{\bs,\ws''}$ is not present when $\bs$ belongs to a boundary quadrangle.

\begin{figure}[H]
\begin{minipage}[b]{0.5\linewidth}
\begin{center}
\begin{overpic}[width=5cm]{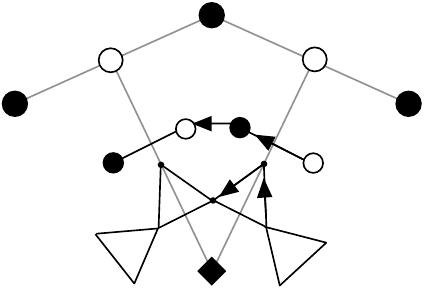}
 \put(58,29){\scriptsize $b$}
 \put(40,29){\scriptsize $\tilde{b}$}
 \put(49,23){\scriptsize $a$}
 \put(64,15){\scriptsize $a'$}
 \put(53,42){\scriptsize $\bs$}
 \put(42,42){\scriptsize $\ws$}
 \put(72,33){\scriptsize $\ws'$}
\end{overpic}
\end{center}
\end{minipage}
\begin{minipage}[b]{0.5\linewidth}
\begin{center}
\begin{overpic}[width=5cm]{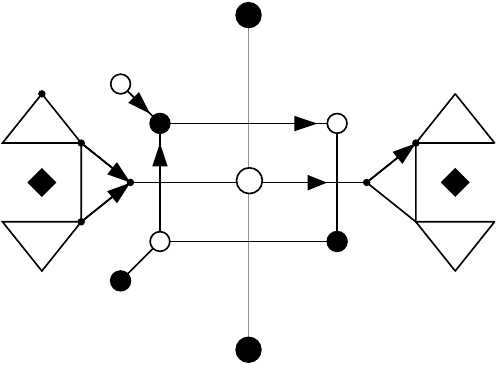}
 \put(26,31){\scriptsize $b$}
 \put(70,31){\scriptsize $b'$}
 \put(18,45){\scriptsize $a$}
 \put(18,24){\scriptsize $a'$}
 \put(76,45){\scriptsize $a''$}
 \put(26,46){\scriptsize $\bs$}
 \put(22,60){\scriptsize $\ws$}
 \put(8,55){\scriptsize $\tilde{b}$}
 \put(30,18){\scriptsize $\ws'$}
 \put(67,52){\scriptsize $\ws''$}
 \put(67,18){\scriptsize $\bs'$}
 \put(51,10){\scriptsize $e\rightsquigarrow\Js_e$}
 \put(56,31){\scriptsize $e^*$}
\end{overpic}
\end{center}
\end{minipage}
\caption{Notation used to relate the Kasteleyn orientations on $\GF$ and $\GQ$. }\label{fig:matrixKQKF}
\end{figure}
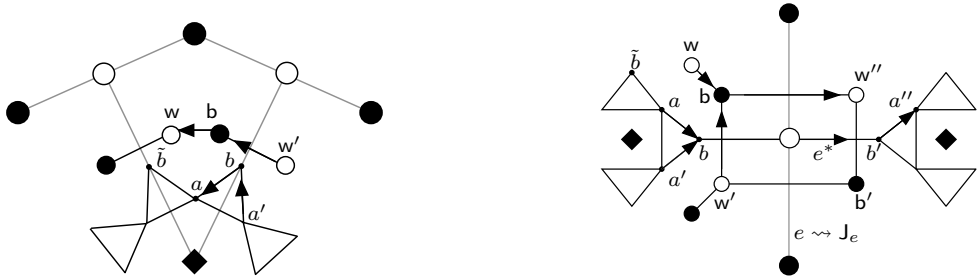

Denote by $\KQt$ the bipartite Kasteleyn matrix
corresponding to the weight function $\nuJ$ and to the Kasteleyn orientation constructed in~\eqref{equ:def_orientation}
with rows indexed by black vertices.

Wrapping up and using the notation of Figure~\ref{fig:matrixKQKF}, we have: for the Kasteleyn matrix~$\KF$,
\begin{equation}
\label{table:KF}
\begin{tabular}{|ll|}
\hline
&$\KF_{u,v}=\eps_{u,v},\quad\quad\quad\, \text{ if $u\sim v$, and $u$ or $v$ is of type $A$ (non exlusive ``or'')}$\\
&$\KF_{b,b'}=\eps_{b,b'}e^{-2\Js_e},\quad \text{ if the external edge $bb'$ corresponds to a dual edge $e^*$ of $\Gs^*$}$\\
\hline
\end{tabular}
\end{equation}
and for the bipartite Kasteleyn matrix $\KQt$,
\begin{equation}
 \label{table:KQ}
\begin{tabular}{|ll|ll|}
\hline
&Infinite case \& finite case when $\bs$&\quad & Finite case when $\bs$ belongs to a\\
&does not belong to a boundary quadrangle&&boundary quadrangle\\
\hline
&$\KQt_{\bs,\ws}=\eps_{\bs,\ws}=\eps_{b,a}$&\quad&$\KQt_{\bs,\ws}=\eps_{\bs,\ws}=\eps_{b,a}$\\
&$\KQt_{\bs,\ws'}=\eps_{\bs,\ws'}\tanh(2\Js_e)= \eps_{b,a'}\tanh(2\Js_e)$&\quad&$\KQt_{\bs,\ws'}=\eps_{\bs,\ws'}=\eps_{b,a'}$\\
&$\KQt_{\bs,\ws''}=\eps_{\bs,\ws''} \cosh^{-1}(2\Js_e)= \eps_{b,b'}\eps_{b',a''}\cosh^{-1}(2\Js_e)$&\quad& n.a.\\
\hline
\end{tabular}
\end{equation}

The first contribution of this section is Theorem~\ref{thm:KFKQinv} of Section~\ref{sec:relation_KF_KQ_inv} expressing the inverse Kasteleyn operator $(\KF)^{-1}$ using the 
inverse \emph{bipartite} Kasteleyn operator $(\KQt)^{-1}$; proving that the contour Ising Boltzmann/Gibbs measures can be computed from the bipartite
dimer model on $\GQ$; note that this result is not restricted to the $Z$-invariant case. The proof of Theorem~\ref{thm:KFKQinv} builds on matrix 
relations of~\cite{Dubedat}; this is the subject of Section~\ref{sec:relation_KF_KQ}. In Section~\ref{sec:GQ_GF_Zinv} we restrict to the $Z$-invariant case
and obtain Corollary~\ref{thm:KFKQinv_Zinv}, one of the main results of this paper, expressing the inverse Kasteleyn operator $(\KF)^{-1}$ using the 
inverse $Z^u$-Dirac operator and also using the $Z$-massive and dual Green functions. This shows that the contour Ising Boltzmann/Gibbs measures can be 
computed using information from random walks only (with specific boundary conditions in the finite case). 
As written in the introduction to this paper, this has 
implications for other observables of the Ising model,
as the spin-Ising observable of~\cite{ChelkakSmirnov:ising} or the fermionic spinor observable of~\cite{KadanoffCeva}.
Note that in the infinite case, this also gives an alternative direct way of finding the local formula for $(\KF)^{-1}$ of~\cite{BdtR2}, 
explicitly relating it to the Green functions.

\subsection{Relating Kasteleyn matrices of the Fisher graph $\GF$ and the bipartite graph $\GQ$}\label{sec:relation_KF_KQ}

Dub\'edat~\cite{Dubedat} establishes a matrix relation between the matrix $\KF$ and a block \emph{triangular}
matrix containing $\KQt$ as one of the diagonal blocks.
Using this matrix relation, he proves that the 
squared dimer partition function of $\GF$ is equal, up to a constant, to the dimer partition function of $\GQ$ in the finite case, and that 
the characteristic polynomials of the two models are equal in the case of infinite $\ZZ^2$-periodic graphs. By adding defects to Ising coupling constants,
this allows him to prove bozonisation identities. 

We attribute the forthcoming Proposition~\ref{prop:Dub}, consisting of two matrix relations, to Dub\'edat. The first is the actual identity of~\cite{Dubedat}; 
it is appropriate for comparing determinants of the matrices $\KF$ and $\KQt$ (related matrix relations can also be found in~\cite{CCK}). 
The second proves an identity between $\KF$ and a 
block \emph{diagonal} matrix containing $\KQt$ in both diagonal blocks; it is not present in the paper~\cite{Dubedat} but does not require much more 
work; it is useful for comparing matrix inverses. For convenience of the reader we provide a proof because: we write weights in a different way,
directly write the proof for all planar graphs (and not $\ZZ^2$), handle the boundary conditions very carefully, and the second identity needs an additional
argument. 

In order to state Proposition~\ref{prop:Dub}, we need to introduce the following matrices, all of which are ``square''. 

The \emph{matrix $I_{\sWQ,\sA}$} has rows indexed by white vertices of $\GQ$ and columns by $A$-vertices of $\GF$. It is the identity matrix 
associated to the following bijection between $\WQ$ and $A$. Using the notation of Figure~\ref{fig:matrixKQKF},
a vertex $\ws$ of $\GQ$ is incident to a unique external edge $\ws\bs$ of $\GQ$; the edge $\ws\bs$
is naturally ``parallel'' to a path from $\tilde{b}$ to $b$ in the external cycle of the closest decoration of $\GF$; then the vertex 
$a$ in bijection with $\ws$ is the unique vertex of type $A$ in the path from $\tilde{b}$ to $b$.

The \emph{matrix $X$} has rows indexed by $B$-vertices of $\GF$ and columns by black vertices of $\GQ$.
It is block diagonal, with blocks of size $2\times 2$ corresponding to edges of $\Gs^*$. For each such edge, the 
rows are indexed by $b$ and $b'$, the two corresponding adjacent $B$-vertices, and the
columns are indexed by the two black vertices $\bs,\bs'$ of the quadrangle of $\GQ$ traversed by the edge $bb'$, with $\bs$ closest 
to $b$, see Figure~\ref{fig:matrixKQKF}. 
The non-zero coefficients of the row corresponding to the vertex $b$ are,
\begin{equation*}
\begin{pmatrix}
x_{b,\bs}&x_{b,\bs'} 
\end{pmatrix}
=
\begin{pmatrix}
1&\KF_{b',b}
\end{pmatrix}.
\end{equation*}

In the \emph{finite} case, the matrix $X$ also has size 1 blocks corresponding to boundary $B$-vertices of $\GF$. For such a vertex 
$b$, let $\bs$ be the closest black vertex of $\GQ$. Then, the only non-zero coefficient of the row corresponding to $b$ is:
\begin{align*}
x_{b,\bs}=1.
\end{align*}

The \emph{matrix $M$} has rows indexed by $B$-vertices and columns by $A$-vertices of $\GF$. It is block diagonal, with blocks 
corresponding to decorations, each block having per size the number of $B$-vertices times the number of $A$-vertices of the decoration.
The matrix $M$ is the matrix $\KF_{\sB,\sA}$ with some signs reversed. That is,
for a $B$-vertex $b$, denote by $a,a'$ its two neighbors of type $A$ so that in cclw order around the triangle we have $a,a',b$, see 
Figure~\ref{fig:matrixKQKF}. Then the non-zero coefficients of the row corresponding to the vertex $b$ are,
\begin{equation*}
\begin{pmatrix}
m_{b,a}&m_{b,a'} 
\end{pmatrix}
=
\begin{pmatrix}
-\eps_{b,a}&\eps_{b,a'}
\end{pmatrix}
\end{equation*}

The \emph{matrix $M'$} has rows indexed by $A$-vertices and columns by $B$-vertices of $\GF$. It is defined as,
\[
M'=-(\KF_{\sB,\sA})^{-1}\KF_{\sB,\sB}.
\]

\begin{rem}\label{rem:KABinv}
The matrix $\KF_{\sB,\sA}=-\KF_{\sA,\sB}$ is invertible. Indeed, it is block diagonal, with blocks corresponding 
to decorations; for each decoration, the block is a directed adjacency matrix of the bipartite graph consisting of the outer cycle of 
the decoration. The orientation on this cycle is Kasteleyn because it is on the whole graph and this cycle contains no vertex 
in its interior~\cite{Kasteleyn2}. As a consequence, the determinant of this block is equal to $\pm$ the number of dimer configurations
of this cycle, that is $\pm 2$, and the block is thus invertible.

In the sequel, it would have been tempting to sometimes use the inverse of the matrix 
$\KF_{\sA,\sA}$, but this matrix is not always invertible. Indeed, it is block diagonal with blocks corresponding to decorations, and 
when the decoration is associated to a dual vertex of odd degree, then the corresponding block of $\KF_{\sA,\sA}$ is not invertible.
\end{rem}

\begin{prop}\emph{\cite{Dubedat}}\label{prop:Dub}
The Kasteleyn matrix $\KF$ of the Fisher graph $\GF$ and the bipartite Kasteleyn matrix $\KQt$ of the graph $\GQ$ are related by 
the following identities:
\begin{equation*}
\KF
\begin{pmatrix}
I&M\\
0&I
\end{pmatrix}
=
\begin{pmatrix}
\KF_{\sB,\sB}& (X \KQt I_{\sWQ,\sA})\\
\KF_{\sA,\sB}&0
\end{pmatrix},\quad \quad 
\KF
\begin{pmatrix}
I&M\\
M'&I
\end{pmatrix}
=
\begin{pmatrix}
0& (X \KQt I_{\sWQ,\sA})\\
-(X \KQt I_{\sWQ,\sA})^t&0
\end{pmatrix}. 
\end{equation*}

\end{prop}

\begin{proof}
The first identity is an easy consequence of the second, so let us prove the second;
unless specified, the arguments hold in the infinite and finite cases.
We need to show the four identities below. 
\begin{align}
&\KF_{\sB,\sB}+\KF_{\sB,\sA}M'=0\label{equ:proofKF1}\\
&\KF_{\sA,\sB}M+\KF_{\sA,\sA}=0\label{equ:proofKF4}\\
&\KF_{\sB,\sB}M+\KF_{\sB,\sA}=X \KQt I_{\sWQ,\sA}\label{equ:proofKF2}\\
&\KF_{\sA,\sB}+\KF_{\sA,\sA}M'=-(X \KQt I_{\sWQ,\sA})^t.\label{equ:proofKF3}
\end{align}
Note that even in the infinite case, they all make sense since matrices involved have 
finitely many non-zero coefficients per rows and columns.

Identity~\eqref{equ:proofKF1} is immediate by definition of $M'$. We now prove~\eqref{equ:proofKF4} and 
\eqref{equ:proofKF2} and then show that~\eqref{equ:proofKF3} follows. 

\emph{Proof of \eqref{equ:proofKF4}.} Let us show that $\KF_{\sA,\sB}M=-\KF_{\sA,\sA}$.
Consider an $A$-vertex $a$ of $\GF$, and let $a_1,a_2,b_1,b_2$ be its four neighbors in $\GF$ 
with the notation of Figure~\ref{fig:proofM}. 
\begin{figure}[H]
\begin{center}
\begin{overpic}[width=2.1cm]{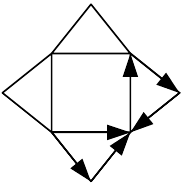}
\put(102,48){\scriptsize $b_1$}
\put(47,-10){\scriptsize $b_2$}
\put(73,73){\scriptsize $a_1$}
\put(73,20){\scriptsize $a$}
\put(17,20){\scriptsize $a_2$}
\end{overpic}
\caption{Notation for the proof of~\eqref{equ:proofKF4}.}\label{fig:proofM}
\end{center}
\end{figure}
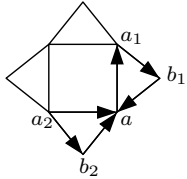
Then the coefficient $(\KF_{\sA,\sB}M)_{a,a'}$ is a priori non zero
when $a'\in\{a,a_1,a_2\}$. Returning to the definition of $\KF$ and $M$, we have 
\[
(\KF_{\sA,\sB}M)_{a,a}=\KF_{a,b_1}m_{b_1,a}+\KF_{a,b_2}m_{b_2,a}=\eps_{a,b_1}\eps_{b_1,a}-\eps_{a,b_2}\eps_{b_2,a}=0
=-\KF_{a,a}.
\]
When $a'\in\{a_1,a_2\}$, using moreover that the orientation around the triangles $a,b_1,a_1$ and $a,a_2,b$ is Kasteleyn, we have
\begin{align*}
(\KF_{\sA,\sB}M)_{a,a_1}&=\KF_{a,b_1}m_{b_1,a_1}=-\eps_{a,b_1}\eps_{b_1,a_1}=-\eps_{a,a_1}=-\KF_{a,a_1}\\
(\KF_{\sA,\sB}M)_{a,a_2}&=\KF_{a,b_2}m_{b_2,a_2}=\eps_{a,b_2}\eps_{b_2,a_2}=-\eps_{a,a_2}=-\KF_{a,a_2},
\end{align*}
thus ending the proof of~\eqref{equ:proofKF4}.

\emph{Proof of \eqref{equ:proofKF2}.}
\emph{Infinite case.}
Figure~\ref{fig:MX2} (left) below sets the notation and labeling:
$b,b'$ are adjacent $B$-vertices of $\GF$, and $a_1,a_2,a_3,a_4$ are their neighbors of type $A$ in $\GF$; 
$\bs,\bs'$ are the two black vertices of the quadrangle of $\GQ$ traversed by the edge $bb'$, and $\ws_1,\dots,\ws_4$
are their neighboring white vertices in $\GQ$.

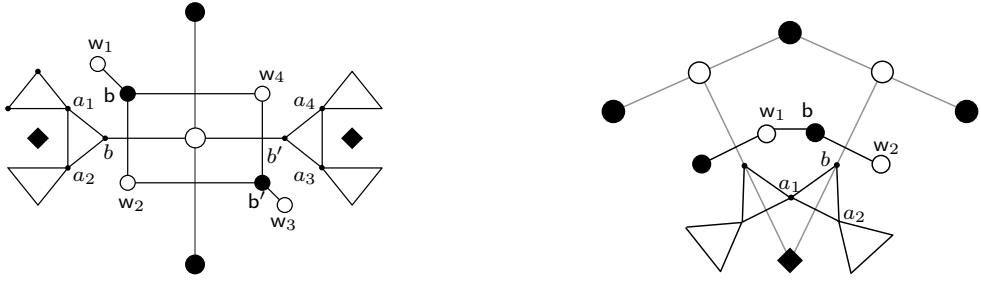
\begin{figure}[H]

\begin{minipage}[b]{0.5\linewidth}
\begin{center}
\begin{overpic}[width=5cm]{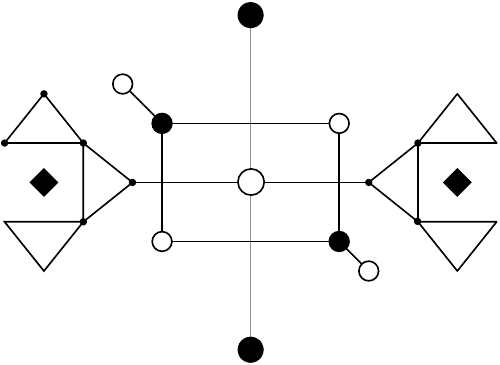}
 \put(26,31){\scriptsize $b$}
 \put(69,30){\scriptsize $b'$}
 \put(18,45){\scriptsize $a_1$}
 \put(18,25){\scriptsize $a_2$}
 \put(76,25){\scriptsize $a_3$}
 \put(76,45){\scriptsize $a_4$}
 \put(26,46){\scriptsize $\bs$}
 \put(64,18){\scriptsize $\bs'$}
 \put(22,60){\scriptsize $\ws_1$}
 \put(30,18){\scriptsize $\ws_2$}
 \put(70,13){\scriptsize $\ws_3$}
 \put(67,52){\scriptsize $\ws_4$}
\end{overpic}
\end{center}
\end{minipage}
\begin{minipage}[b]{0.5\linewidth}
\begin{center}
 \begin{overpic}[width=5cm]{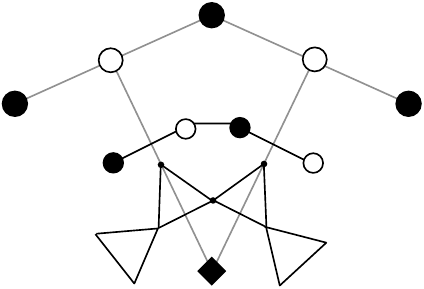}
 \put(58,29){\scriptsize $b$}
 \put(47,23){\scriptsize $a_1$}
 \put(64,15){\scriptsize $a_2$}
 \put(53,42){\scriptsize $\bs$}
 \put(42,42){\scriptsize $\ws_1$}
 \put(72,33){\scriptsize $\ws_2$}
\end{overpic}
\end{center}
\end{minipage}

\caption{Notation for the proof of~\eqref{equ:proofKF2}.}\label{fig:MX2}
\end{figure}

Consider a $B$-vertex $b$ of $\GF$.
Then the coefficient $(\KF_{\sB,\sB}M+\KF_{\sB,\sA})_{b,a}$ of the LHS of~\eqref{equ:proofKF2} is non-zero 
when $a\in\{a_1,\dots,a_4\}$, and
\begin{align*}
(\KF_{\sB,\sB}M+\KF_{\sB,\sA})_{b,a}&=
\begin{cases}
\KF_{b,a}=\eps_{b,a}&\text{ if $a\in\{a_1,a_2\}$}\\
\KF_{b,b'}m_{b',a}&\text{ if $a\in\{a_3,a_4\}$},
\end{cases}
\end{align*}
where recall $m_{b',a_3}=-\eps_{b',a_3}$ and $m_{b',a_4}=\eps_{b',a_4}$.

The coefficient $(X \KQt I_{\sWQ,\sA})_{b,a}$ of the RHS of~\eqref{equ:proofKF2} is non-zero for the same choices of $a$. Returning to the definition
of $X$ and $\KQt$ we have,
\begin{align}
(X \KQt I_{\sWQ,\sA})_{b,a_1}&=x_{b,\bs}\KQt_{\bs,\ws_1}=1\cdot \eps_{b,a_1},\nonumber\\
(X \KQt I_{\sWQ,\sA})_{b,a_3}&=x_{b,\bs'}\KQt_{\bs',\ws_3}=\KF_{b',b}\eps_{b',a_3}=-\KF_{b,b'}\eps_{b',a_3},\nonumber\\
(X \KQt I_{\sWQ,\sA})_{b,a_2}&=x_{b,\bs}\KQt_{\bs,\ws_2}+x_{b,\bs'}\KQt_{\bs',\ws_2}\nonumber\\
&=1\cdot \eps_{b,a_2}\tanh(2\Js) + \KF_{b',b}\eps_{b',b}\eps_{b,a_2}\cosh^{-1}(2\Js)\nonumber\\
&=\eps_{b,a_2}\left(\tanh(2\Js)+|\KF_{b',b}|\cosh^{-1}(2\Js)\right)\nonumber\\
&=\eps_{b,a_2},\text{ since }|\KF_{b',b}|=e^{-2\Js},\label{equ:reltanh}\\
(X \KQt I_{\sWQ,\sA})_{b,a_4}&=x_{b,\bs}\KQt_{\bs,\ws_4}+x_{b,\bs'}\KQt_{\bs',\ws_4}\nonumber\\
&=1\cdot \eps_{b,b'}\eps_{b',a_4}\cosh^{-1}(2\Js)+\KF_{b',b}\eps_{b',a_4}\tanh(2\Js)\nonumber\\
&=\KF_{b,b'}\eps_{b',a_4}
\left(\frac{\eps_{b,b'}}{\KF_{b,b'}}\cosh^{-1}(2\Js)-\tanh(2\Js)\right)\nonumber\\
&=\KF_{b,b'}\eps_{b',a_4}
\left(e^{2\Js}\cosh^{-1}(2\Js)-\tanh(2\Js)\right)=\KF_{b,b'}\eps_{b',a_4}.\label{equ:reltanh2}
\end{align}
\emph{Finite case.} The proof is as in the infinite case as long as $b$ is not a boundary $B$-vertex of $\GF$, so let
$b$ be a boundary $B$-vertex and refer to Figure~\ref{fig:MX2} (right) for notation. The coefficient $(\KF_{\sB,\sB}M+\KF_{\sB,\sA})_{b,a}$ of 
the LHS is non-zero when $a\in\{a_1,a_2\}$. Similarly to the infinite case computation, we have:
\begin{align*}
(\KF_{\sB,\sB}M+\KF_{\sB,\sA})_{b,a}=\KF_{b,a}=\eps_{b,a},\quad \text{ if $a\in\{a_1,a_2\}$}
\end{align*}
The coefficient $(X \KQt I_{\sWQ,\sA})_{b,a}$ of the RHS is non-zero for the same choices of $a$. Returning to the definition
of $X$ and $\KQt$ (boundary case) we have:
\begin{align*}
(X \KQt I_{\sWQ,\sA})_{b,a_1}&=x_{b,\bs}\KQt_{\bs,\ws_1}=1\cdot \eps_{b,a_1},\\
(X \KQt I_{\sWQ,\sA})_{b,a_2}&=x_{b,\bs}\KQt_{\bs,\ws_2}=1\cdot \eps_{b,a_2}.
\end{align*}
This ends the proof of~\eqref{equ:proofKF2}.

\emph{Proof of \eqref{equ:proofKF3}.} From Remark~\ref{rem:KABinv} the matrix $K_{\sA,\sB}$ is invertible, thus from~\eqref{equ:proofKF4} we have 
$M=-(\KF_{\sA,\sB})^{-1}\KF_{\sA,\sA}$. Plugging this into the LHS of~\eqref{equ:proofKF2} gives that it is equal to:
\[
\mathrm{LHS}~\eqref{equ:proofKF2}=-\KF_{\sB,\sB}(\KF_{\sA,\sB})^{-1}\KF_{\sA,\sA}+\KF_{\sB,\sA}.
\]
Returning to the definition of $M'$ (or to~\eqref{equ:proofKF1}), we have that the LHS of\eqref{equ:proofKF3} is 
\[
\mathrm{LHS}\eqref{equ:proofKF3}=\KF_{\sA,\sB}-\KF_{\sA,\sA}(\KF_{\sB,\sA})^{-1}\KF_{\sB,\sB}.
\]
Using that the matrix $\KF$ is skew-symmetric, we deduce that 
$\mathrm{LHS}~\eqref{equ:proofKF3}=-[\mathrm{LHS}~\eqref{equ:proofKF2}]^t$.
The same clearly holds for the RHS of the two equations; they are thus equivalent and we have proved~\eqref{equ:proofKF2}. 
\end{proof}

\begin{cor}[\cite{Dubedat}]\label{cor:det_KF_KQ}
Suppose that the graph $\Gs$ is finite. Then,
\begin{equation*}
[\Zdimer(\GF,\muJ)]^2=2^{|\Vs^*|}\Bigl(\prod_{e^*\in\Es^*}(1+e^{-4\Js_e})\Bigr)\Zdimer(\GQ,\nuJ),
\end{equation*}
where $\Zdimer(\GF,\muJ)=|\det \KF|$, $\Zdimer(\GQ,\nuJ)=|\det \KQt|$.
\end{cor}
\begin{proof}
By the first identity of Proposition~\ref{prop:Dub}, we have $|\det\KF|=|\det\KF_{\sA,\sB} \det X \det \KQ|$. By Remark~\ref{rem:KABinv}, 
we know that $|\det \KF_{\sA,\sB}|=2^{|\Vs^*|}$. The determinant of $X$ is computed by calculating that of its blocks. 
Let $bb'$ be an edge of $\GF$ corresponding to an edge $e^*$ of $\Gs^*$ and let $\bs,\bs'$ be the black vertices of the quadrangle of 
$\GQ$ traversed by the edge $bb'$. Then by definition, the corresponding block of $X$ is:
\begin{equation*}
\begin{pmatrix}
1&\KF_{b',b}\\ 
\KF_{b,b'}&1
\end{pmatrix}_{(b,b')}^{(\bs,\bs')}.
\end{equation*}
Its determinant is equal to $1+|\KF_{b,b'}|^2=1+e^{-4\Js_e}$, thus ending the proof of the corollary.
\end{proof}


\subsection{Relating inverse Kasteleyn matrices of $\GF$ and $\GQ$}\label{sec:relation_KF_KQ_inv}

Consider the inverse Kasteleyn operators $(\KF)^{-1}$ and $(\KQ)^{-1}$.
When the graph $\Gs$ is \emph{infinite} and $\ZZ^2$-periodic, these inverses denote the unique ones decreasing to 0 at infinity~\cite{KOS,BoutillierdeTiliere:iso_perio}. 
When the graph $\Gs$ is infinite and isoradial (not necessarily $\ZZ^2$-periodic) and the corresponding dimer weights on $\GQ$ and $\GF$ are $Z$-invariant, 
then $(\KQt)^{-1}$ and $(\KF)^{-1}$ are the operators decreasing to 0 at infinity with local expression given 
in~\cite{Kenyon3,BoutillierdeTiliere:iso_gen,BdtR2}.

From Proposition~\ref{prop:Dub} and proving additional relations (not present in the paper~\cite{Dubedat}) we show, in 
Theorem~\ref{thm:KFKQinv} below, identities relating the inverse operator $(\KF)^{-1}$ to the inverse operator $(\KQt)^{-1}$. 
Using Section~\ref{sec:dimer_infinite}, Theorem~\ref{thm:KFKQinv} allows to express the dimer Boltzmann measure (finite case)
and the Gibbs measure (infinite case) of the Fisher graph $\GF$, denoted $\PPdimerF$, using coefficients of the matrix 
$\KF$ and of the inverse~\emph{bipartite} Kasteleyn operator $(\KQt)^{-1}$, see also Example~\ref{ex:prob_GF}.
To state Theorem~\ref{thm:KFKQinv}, we need to define 
two additional matrices $D_{\sBQ,\sA}$ and $\kappa$. 

The \emph{matrix $D_{\sBQ,\sA}$} has rows indexed by 
black vertices of $\GQ$ and columns by $A$-vertices of $\GF$. It is a diagonal matrix associated to the following bijection
between $\Bs$ and $A$. Using the notation of Figure~\ref{fig:matrixKQKF}, a vertex $\bs$ of $\GQ$ belongs to a unique external
edge $\bs\ws$ of $\GQ$; then the vertex $a$ in bijection with $\bs$ is the vertex in bijection with $\ws$ in the construction of 
the matrix $I_{\sWs,\sA}$. The corresponding diagonal coefficient $d_{\bs,a}$ is,
\begin{equation*}
d_{\bs,a}=\eps_{b,a}.
\end{equation*}

The \emph{matrix $\kappa$} has rows and columns indexed by $A$-vertices of $\GF$. It is block diagonal with blocks corresponding
to decorations, each block having per size the number of $A$-vertices of the decoration.
Given two vertices $a,a'$ of a decoration of $\GF$, we have,
\begin{equation*}
\kappa_{a,a'}=-\frac{1}{4}
\begin{cases}
-1&\text{ if } a=a'\\
(-1)^{n(a,a')}&\text{ if }a\neq a',
\end{cases}
\end{equation*}
where $n(a,a')$ is the number of edges oriented cw in the cclw path going from $a$ to $a'$ in the $A$-cycle of the decoration. 

\textbf{Notation for coefficients $(\KF)^{-1}_{\ubar{u},v}$ of Theorem~\ref{thm:KFKQinv}.} 
If $\ubar{u}=\ubar{a}$ is an $A$-vertex, then $\ubar{\ws}$ is the white vertex of $\GQ$ corresponding
to $\ubar{a}$ in the bijection defining $I_{\sWs,\sA}$. If $\ubar{u}=\ubar{b}$
is a $B$-vertex, then $\ubar{a}_1$ and $\ubar{a}_2$ are the two $A$-vertices belonging to the same triangle, with $\ubar{b},\ubar{a}_1,\ubar{a}_2$
in cclw order around the triangle. Note that this definition also holds if $\ubar{b}$ is a boundary vertex.

If $v=b$ is a $B$-vertex, we let $\bs$ be the closest black vertex of $\GQ$. When moreover $b$ is not a boundary vertex, 
we let $b'$ be the $B$-vertex such that $bb'$ defines an edge $e^*$ of $\Gs^*$; we let $\bs'$ be the black vertex of $\GQ$ closest to $b'$ ($\bs$ and 
$\bs'$ are the black vertices of the quadrangle of $\GQ$ traversed by the edge $bb'$); the coupling constant of the edge $e$, dual of $e^*$, is
denoted $\Jf_e$, where ``$\mathrm{f}$'' stands for ``final''. 
If $v=a$ is an $A$-vertex, 
we let $\bs$ and $b$ be as defined in the matrix $D_{\sBQ,\sA}$, see Figure~\ref{fig:thmKFKQinv}.

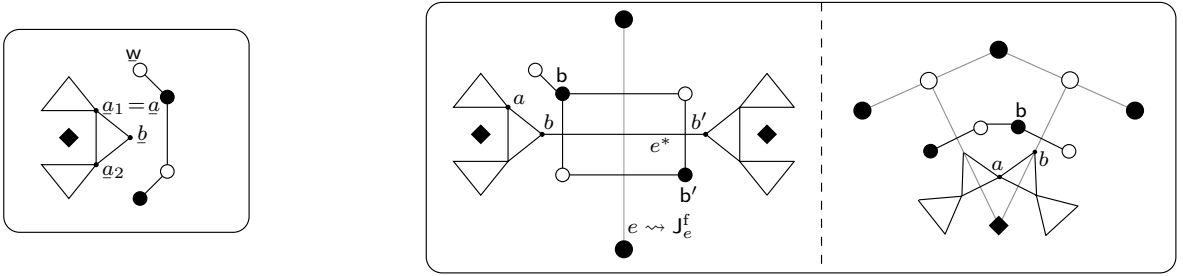
\begin{figure}[H]
\centering
\begin{overpic}[width=\linewidth]{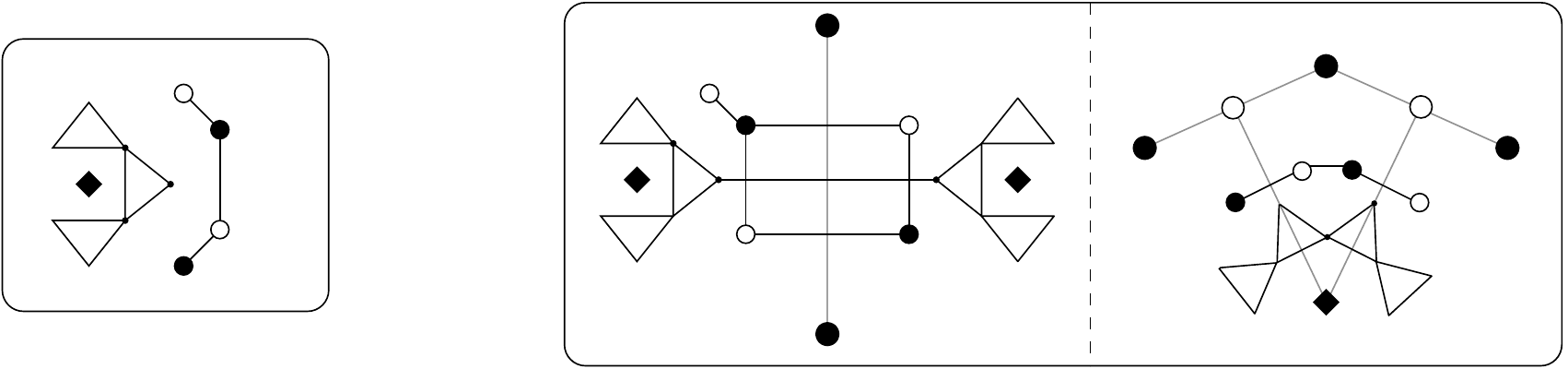}
\put(8.5,14){\scriptsize $\ubar{a}_1\!=\!\ubar{a}$}
\put(8.5,8.5){\scriptsize $\ubar{a}_2$}
\put(11.5,11.5){\scriptsize $\ubar{b}$}
\put(10.5,18.5){\scriptsize $\ubar{\ws}$}

\put(47,16.5){\scriptsize $\bs$}
\put(86,13.5){\scriptsize $\bs$}
\put(57.5,6.2){\scriptsize $\bs'$}

\put(46,12.5){\scriptsize $b$}
\put(88,9.3){\scriptsize $b$}
\put(58.5,12.5){\scriptsize $b'$}

\put(43.4,14.6){\scriptsize $a$}
\put(84,9){\scriptsize $a$}

\put(53.1,3.6){\scriptsize $e\rightsquigarrow \Jf_e$}
\put(55,10.5){\scriptsize $e^*$}
\end{overpic}
\caption{Left: notation for initial vertices. Right: notation for final vertices
when $b$ is not a boundary vertex and when it is.}
\label{fig:thmKFKQinv}
\end{figure}

\begin{thm}\label{thm:KFKQinv}
As long as they are unique (which happens for sure in the finite and $\ZZ^2$-periodic cases), 
the inverse Kasteleyn operator $(\KF)^{-1}$ can be expressed using the inverse bipartite Kasteleyn operator $(\KQt)^{-1}$ as follows.

$\bullet$ \emph{Matrix form.}
\begin{equation}
\label{equ:thmKFKQinv}
(\KF)^{-1}=
\begin{pmatrix}
(\KF)^{-1}_{\sB,\sB}&(\KF)^{-1}_{\sB,\sA}\\ 
(\KF)^{-1}_{\sA,\sB}&(\KF)^{-1}_{\sA,\sA}
\end{pmatrix}=
\begin{pmatrix}
M(X\KQt I_{\sWQ,\sA})^{-1}&-[(X\KQt I_{\sWQ,\sA})^t]^{-1}\\
(X\KQt I_{\sWQ,\sA})^{-1}&-\frac{1}{2}I_{\sA,\Ws}(\KQt)^{-1}D_{\sBQ,\sA}+\kappa_{\sA,\sA}
\end{pmatrix}.
\end{equation}
$\bullet$ \emph{Coefficients}. We have four cases to consider and use the notation of~Figure~\ref{fig:thmKFKQinv}.

\emph{1.} For every $\ubar{a}\in A$ and every $b\in B$ such that, when the graph $\GF$ is moreover finite, $b$ is not a boundary vertex:
\begin{align}\label{form:KABm1}
(\KF)^{-1}_{\ubar{a},b}&=\frac{1}{1+e^{-4\Jf_e}}
\bigl[(\KQt)^{-1}_{\ubar{\ws},\bs} + (\KQt)^{-1}_{\ubar{\ws},\bs'}\eps_{b',b}e^{-2\Jf_e}\bigr].
\end{align}
\emph{2.} When the graph $\GF$ is finite, for every $\ubar{a}\in A$ and every boundary vertex $b$ of $B$, we have
\begin{equation}\label{form:KABm1_finite}
(\KF)^{-1}_{\ubar{a},b}=(\KQt)^{-1}_{\ubar{\ws},\bs}.
\end{equation}
\emph{3.} For every $\ubar{a},\,a\in A$,
\begin{align}\label{form:KAAm1}
(\KF)^{-1}_{\ubar{a},a}&=-\frac{1}{2}(\KQt)^{-1}_{\ubar{\ws},\bs}\eps_{b,a}+\kappa_{\ubar{a},a}.
\end{align}
\emph{4.} For every $\ubar{b},\,b\in B$,
\begin{align}\label{form:KBBm1}
(\KF)^{-1}_{\ubar{b},b}&=-\eps_{\ubar{b},\ubar{a}_1}(\KF)^{-1}_{\ubar{a}_1,b}+\eps_{\ubar{b},\ubar{a}_2}(\KF)^{-1}_{\ubar{a}_2,b},
\end{align}
where $(\KF)^{-1}_{\ubar{a}_1,b},(\KF)^{-1}_{\ubar{a}_2,b}$ are given by~\eqref{form:KABm1}.
\end{thm}

\begin{rem}\label{rem:positivity_dimer}$\,$
\begin{itemize}
\item When proving the local formula for $(\KF)^{-1}_{\ubar{a},a}$~\cite{BoutillierdeTiliere:iso_gen,BdtR2} in the $Z$-invariant case,
we obtained a formula of the form~\eqref{form:KAAm1} - with the constant $\kappa_{\ubar{a},a}$ - without explicitly relating it to a coefficient of $(\KQ)^{-1}$. It is quite 
remarkable that this formula holds in the \emph{full} planar case (without assuming $Z$-invariance), in the finite and infinite cases.
\item In the finite case, we do not need positivity of the coupling constants $\Js$.
In particular if the coupling constants are all negative, Theorem~\ref{thm:KFKQinv} expresses probabilities of the dimer model on the 
\emph{non-bipartite} graph $\GF$ with positive weights 
$e^{-2\Js_e}>1$ on external edges, as a function of the inverse Kasteleyn operator of a ``dimer model'' on the \emph{bipartite} graph $\GQ$ with 
\emph{negative weights} $(\tanh(\Js_e)<0)$ on quadrangle edges parallel to edges of $\Gs$. Having a negative weight for an edge amounts to reversing 
its orientation. From a physics point of view, this amounts to adding defects or creating vortices
at each face where this new orientation is 
not Kasteleyn. The physics paper~\cite{Nash_Oconnor} considers bipartite models with negative weights and somehow describes the non-Harnacity of the 
associated spectral curves. 
\item Coefficients of the inverse Kasteleyn operator $(\KF)^{-1}$ not only allow to express the dimer Boltzmann/Gibbs measure $\PPdimerF$,
but are also related to important observables of the Ising model.
By~\cite{CCK}, the coefficient $(\KF)^{-1}_{\ubar{b},b}$ is essentially 
the \emph{spin-observable} of~\cite{ChelkakSmirnov:ising} when fixing one vertex to be on the boundary of the domain; 
by~\cite{Dubedat} the coefficient $(\KF)^{-1}_{\ubar{a},a}$ is the 
\emph{fermionic spinor correlator} of~\cite{KadanoffCeva} and by~\cite{NK}, it is related to the 
\emph{FK-Ising observable}~\cite{Smirnov3,Smirnov2,ChelkakSmirnov:ising} when fixing one vertex on the boundary of the domain,
taking appropriate boundary conditions, up to normalization.
\item Consider a fixed $A$-vertex $a$ and $\ubar{a}_1,\ubar{a}_2,\ubar{a}_3$ as in Figure~\ref{fig:Dotsenko} below, such that 
the decorations of 
$\ubar{a}_1,\ubar{a}_2,\ubar{a}_3$ are distinct from that of $a$. Then using the identity $[\KQt (\KQt)^{-1}]_{\ubar{\bs},\bs}=0$,
from~\eqref{form:KAAm1} we immediately obtain the Dotsenko~\emph{three-terms relation}~\cite{Dotsenko,Mercat:ising,CCK}, see also Definition 2.1. 
of~\cite{Chelkak:embedding}.
\begin{equation*}
\eps_{\ubar{b},\ubar{a}_1}(\KF)^{-1}_{\ubar{a}_1,a}+ 
\eps_{\ubar{b},\ubar{a}_2}\tanh(2\Js_e)(\KF)^{-1}_{\ubar{a}_2,a}+
\eps_{\ubar{b},\ubar{b}'}\eps_{\ubar{b'},\ubar{a}_3}\cosh^{-1}(2\Js_e)(\KF)^{-1}_{\ubar{a}_3,a}=0.
\end{equation*}

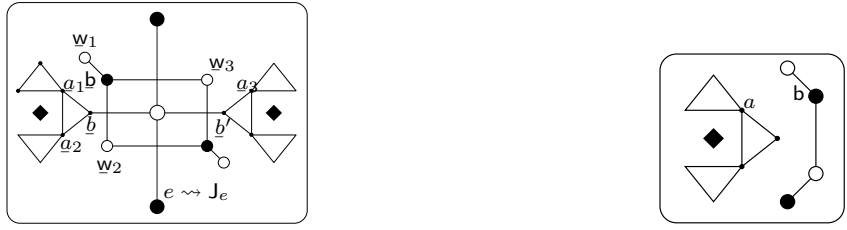
\begin{figure}[H]
\begin{minipage}[b]{0.5\linewidth}
\begin{center}
\begin{overpic}[width=4cm]{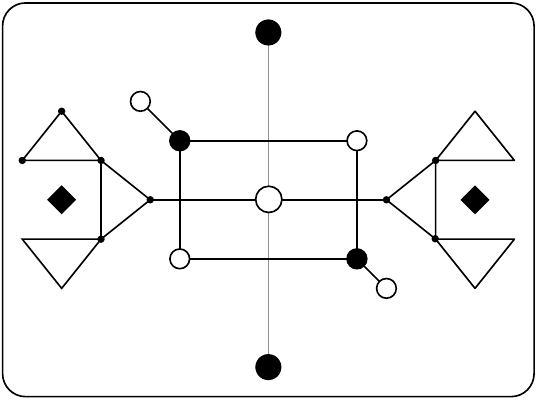}
 \put(26.5,31){\scriptsize $\ubar{b}$}
 \put(69,30){\scriptsize $\ubar{b}'$}
 \put(19,45){\scriptsize $\ubar{a}_1$}
 \put(18,25){\scriptsize $\ubar{a}_2$}
 \put(76,45){\scriptsize $\ubar{a}_3$}
 \put(26,46){\scriptsize $\ubar{\bs}$}
 \put(22,60){\scriptsize $\ubar{\ws}_1$}
 \put(30,18){\scriptsize $\ubar{\ws}_2$}
 \put(67,52){\scriptsize $\ubar{\ws}_3$}
 \put(52,9){\scriptsize $e\rightsquigarrow\Js_e$}
\end{overpic}
\end{center}
\end{minipage}
\begin{minipage}[b]{0.5\linewidth}
\begin{center}
\begin{overpic}[width=2.5cm]{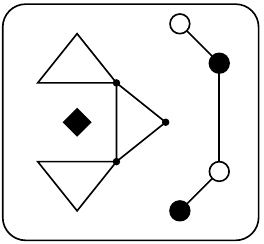}
 \put(45,62){\scriptsize $a$}
 \put(71,66){\scriptsize $\bs$}
\end{overpic}
\end{center}
\end{minipage}
\caption{Notation for the Dotsenko three-terms relation.}\label{fig:Dotsenko}
\end{figure}
\end{itemize}
\end{rem}

\begin{proof}
The expressions for coefficients are obtained from Formula~\eqref{equ:thmKFKQinv} and 
by returning to the definition of the matrices $I_{\sWs,\sA}, D_{\sBs,\sA}$,
$M$ and $X$. The inverse of $X$ is computed by blocks. 
For an edge $bb'$ of $\GF$ corresponding to an edge $e^*$ of $\Gs^*$, let $\bs,\bs'$ be the black vertices of the quadrangle of 
$\GQ$ traversed by the edge $bb'$. Then the inverse of the corresponding block is,
\begin{equation}\label{equ:Xinv}
\begin{pmatrix}
\begin{pmatrix}
1&\KF_{b',b}\\ 
\KF_{b,b'}&1
\end{pmatrix}^{-1}
\end{pmatrix}_{(\bs,\bs')}^{(b,b')}=
\frac{1}{1+|\KF_{b,b'}|^2}
\begin{pmatrix}
1&\KF_{b,b'}\\
\KF_{b',b}&1
\end{pmatrix}_{(\bs,\bs')}^{(b,b')}.
\end{equation}
In the finite case, the matrix $X$ also has a size 1, identity block for each boundary $B$-vertex $b$ of $\GF$ and its closest black vertices $\bs$
of $\GQ$. This ends the proof of formulas~\eqref{form:KABm1},~\eqref{form:KAAm1},
\eqref{form:KBBm1} for coefficients and we now turn to the proof of~\eqref{equ:thmKFKQinv}.

The expressions for $(\KF)^{-1}_{\sB,\sB},(\KF)^{-1}_{\sB,\sA},(\KF)^{-1}_{\sA,\sB}$ are a direct consequence of Proposition~\ref{prop:Dub}.
This is not the case of $(\KF)^{-1}_{\sA,\sA}$ which requires proving additional identities. From Proposition~\ref{prop:Dub}
we know that,
\[
(\KF)^{-1}_{\sA,\sA}=-M'[(X\KQt I_{\sWQ,\sA})^t]^{-1}=(X\KQt I_{\sWQ,\sA})^{-1}(M')^t=
-I_{\sA,\Ws}(\KQt)^{-1}X^{-1}\KF_{\sB,\sB}(\KF_{\sA,\sB})^{-1},
\]
where in the second and third equalities we used skew-symmetry of $\KF$, and in the third the definition of $M'$.
We thus need to prove that 
\begin{align}\label{equ:proofKF5}
&-I_{\sA,\Ws}(\KQt)^{-1}X^{-1}\KF_{\sB,\sB}(\KF_{\sA,\sB})^{-1}=-\frac{1}{2}I_{\sA,\Ws}(\KQt)^{-1}D_{\sBQ,\sA}+\kappa_{\sA,\sA}\nonumber\\
\Leftrightarrow \quad & 
-I_{\sA,\Ws}(\KQt)^{-1}X^{-1}\KF_{\sB,\sB}(\KF_{\sA,\sB})^{-1}=
-I_{\sA,\Ws}(\KQt)^{-1}\left[\frac{1}{2}D_{\sBQ,\sA}-\KQt I_{\sWQ,\sA}\kappa_{\sA,\sA}\right]\nonumber\\
\Leftrightarrow \quad& X^{-1}\KF_{\sB,\sB}(\KF_{\sA,\sB})^{-1}=\frac{1}{2}D_{\sBQ,\sA}-\KQt I_{\sWQ,\sA}\kappa_{\sA,\sA}\nonumber\\
\Leftrightarrow \quad& 2X^{-1}\KF_{\sB,\sB}=D_{\sBQ,\sA}\KF_{\sA,\sB}-2 \KQt I_{\sWQ,\sA} \kappa_{\sA,\sA}\KF_{\sA,\sB},
\end{align}
so let us prove~\eqref{equ:proofKF5}. We will be using the notation of Figure~\ref{fig:MX2_1} below.

\begin{figure}[H]
\begin{minipage}[b]{0.5\linewidth}
\begin{center}
\begin{overpic}[width=5cm]{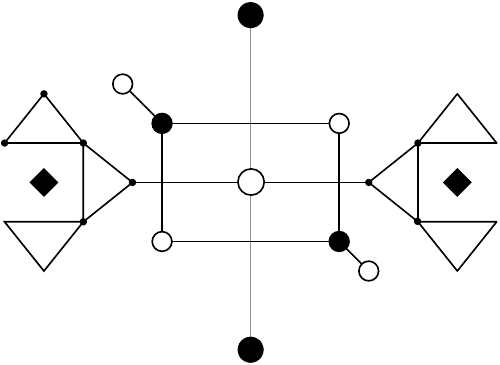}
 \put(26,31){\scriptsize $b$}
 \put(70,31){\scriptsize $b'$}
 \put(8,56){\scriptsize $b_1$}
 \put(18,45){\scriptsize $a_1$}
 \put(18,25){\scriptsize $a_2$}
 \put(76,25){\scriptsize $a_3$}
 \put(76,45){\scriptsize $a_4$}
 \put(-7,45){\scriptsize $a_5$}
 \put(26,46){\scriptsize $\bs$}
 \put(64,18){\scriptsize $\bs'$}
 \put(22,60){\scriptsize $\ws_1$}
 \put(30,18){\scriptsize $\ws_2$}
 \put(70,13){\scriptsize $\ws_3$}
 \put(67,52){\scriptsize $\ws_4$}
\end{overpic}
\end{center}
\end{minipage}
\begin{minipage}[b]{0.5\linewidth}
\begin{center}
 \begin{overpic}[width=5cm]{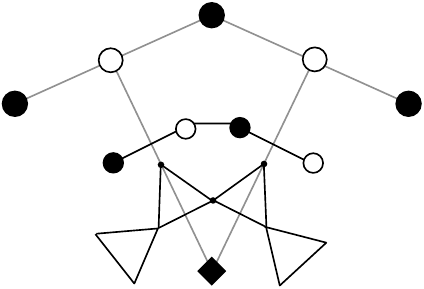}
 \put(58,29){\scriptsize $b$}
 \put(47,23){\scriptsize $a_1$}
 \put(64,15){\scriptsize $a_2$}
 \put(53,42){\scriptsize $\bs$}
 \put(42,42){\scriptsize $\ws_1$}
 \put(72,33){\scriptsize $\ws_2$}
\end{overpic}
\end{center}
\end{minipage}
\caption{Notation for the proof of Equation~\eqref{equ:proofKF5}.}\label{fig:MX2_1}
\end{figure}

We need to introduce an additional matrix, the \emph{matrix $D_{\sA,\sB}$}, which has rows indexed by $A$-vertices and columns by $B$-vertices of $\GF$. It is diagonal:
to an $A$-vertex $a$ corresponds the unique $B$-vertex $b$ such that $b$ comes before $a$ in the cw ordering of the triangle
containing $a$ and $b$. The diagonal coefficient $d_{a,b}$ is:
\[
d_{a,b}=\frac{1}{2}\eps_{b,a}.
\]
For example to $a_2$ of Figure~\ref{fig:MX2_1} corresponds the vertex $b$, and the coefficient $d_{a_2,b}=\frac{1}{2}\eps_{b,a_2}$.
Let us first show that
\begin{equation} \label{equ:proof_ident_3}
\kappa_{\sA,\sA}\KF_{\sA,\sB}=-D_{\sA,\sB}.
\end{equation}
Consider a $B$-vertex $b$ of $\GF$. Then, the coefficient $(\kappa_{\sA,\sA}\KF_{\sA,\sB})_{a,b}$ of the LHS 
of~\eqref{equ:proof_ident_3} is a priori non-zero
for all $A$-vertices $a$ belonging to the same decoration as $b$. We have,
\[
(\kappa_{\sA,\sA}\KF_{\sA,\sB})_{a,b}=\kappa_{a,a_1}\KF_{a_1,b}+\kappa_{a,a_2}\KF_{a_2,b}=
\kappa_{a,a_1}\eps_{a_1,b}+\kappa_{a,a_2}\eps_{a_2,b}.
\]
Returning to the definition of the matrix $\Lambda$, as long as $a\neq a_2$, we have that
$\kappa_{a,a_1}=\kappa_{a,a_2}\eps_{a_2,a_1}$ implying,
\[
(\kappa_{\sA,\sA}\KF_{\sA,\sB})_{a,b}=
\kappa_{a,a_2}[\eps_{a_2,a_1}\eps_{a_1,b}+\eps_{a_2,b}]=0=-d_{a,b},
\]
since the orientation around the triangle $a_1,a_2,b$ is Kasteleyn and using the definition of $D_{\sA,\sB}$. 
When $a=a_2$, then $\kappa_{a_2,a_1}=-\frac{1}{4}\eps_{a_2,a_1},
\kappa_{a_2,a_2}=\frac{1}{4}$; thus
\[
(\kappa_{\sA,\sA}\KF_{\sA,\sB})_{a_2,b}=-\frac{1}{4}\eps_{a_2,a_1}\eps_{a_1,b}+\frac{1}{4}\eps_{a_2,b}=
-\frac{1}{2}\eps_{b,a_2}=-d_{b,a_2},
\]
using again the Kasteleyn orientation around the triangle and the definition of $D_{\sA,\sB}$, thus ending the proof of~\eqref{equ:proof_ident_3}.
Plugging~\eqref{equ:proof_ident_3} into~\eqref{equ:proofKF5} leaves us with showing the equivalent
\begin{equation}\label{proof_ident_6}
2X^{-1}\KF_{\sB,\sB}=D_{\sBQ,\sA}\KF_{\sA,\sB}+2 \KQt I_{\sWQ,\sA}D_{\sA,\sB}.
\end{equation}
\emph{Infinite case.} 
Let $\bs$ be a vertex of $\GQ$, then 
the coefficient $(2X^{-1}\KF_{\sB,\sB})_{\bs,\tilde{b}}$ of the LHS of~\eqref{proof_ident_6} is non-zero when 
$\tilde{b}\in\{b,b'\}$. 
Recalling the computation of $X^{-1}$ given in~\eqref{equ:Xinv}, we have
\begin{align*}
2(X^{-1}\KF_{\sB,\sB})_{\bs,b}&=\frac{2}{1+|\KF_{b,b'}|^2}\times(\KF_{b,b'})\KF_{b',b}=e^{2\Js}\cosh^{-1}(2\Js)\eps_{b,b'}\eps_{b',b}e^{-4\Js}\\
&=-e^{-2\Js}\cosh^{-1}(2\Js)=-1+\tanh(2\Js),\text{ using~\eqref{equ:reltanh}}.\\
2(X^{-1}\KF_{\sB,\sB})_{\bs,b'}&=\frac{2}{1+|\KF_{b,b'}|^2}\times 1\times \KF_{b,b'}=e^{2\Js}\cosh^{-1}(2\Js)\eps_{b,b'}e^{-2\Js}\\
&=\eps_{b,b'}\cosh^{-1}(2\Js).
\end{align*}
The RHS $(D_{\sBQ,\sA}\KF_{\sA,\sB}+2 \KQt I_{\sWQ,\sA}D_{\sA,\sB})_{\bs,\tilde{b}}$ a priori has non-zero coefficients when 
$\tilde{b}\in\{b,b',b_1\}$, and
\begin{align*}
(D_{\sBQ,\sA}\KF_{\sA,\sB}+2 \KQt I_{\sWQ,\sA}D_{\sA,\sB})_{\bs,b}&=
\eps_{b,a_1}\eps_{a_1,b}+2\KQ_{\bs,\ws_2}I_{\ws,a_2}\eps_{b,a_2}\\
&=-1+\eps_{b,a_2}\tanh(2\Js_e)\eps_{b,a_2}
=-1+\tanh(2\Js_e).\\
(D_{\sBQ,\sA}\KF_{\sA,\sB}+2 \KQt I_{\sWQ,\sA}D_{\sA,\sB})_{\bs,b'}&=0+2\KQ_{\bs,\ws_4}I_{\ws_4,a_4}\eps_{b',a_4}\\
&=0+\eps_{b,b'}\eps_{b',a_4}\cosh^{-1}(2\Js_e) \eps_{b',a_4}=\eps_{b,b'}\cosh^{-1}(2\Js_e).\\
(D_{\sBQ,\sA}\KF_{\sA,\sB}+2 \KQt I_{\sWQ,\sA}D_{\sA,\sB})_{\bs,b_1}&=
\eps_{b,a_1}\eps_{a_1,b_1}+\KQt_{\bs,\ws_1}I_{\ws_1,a_1}\eps_{b_1,a_1}\\
&=-\eps_{b,a_1}\eps_{a_1,b_1}+\eps_{b,a_1}\eps_{b_1,a_1}=0,
\end{align*}
and hence Equation~\eqref{proof_ident_6} is proved in the infinite case.

\emph{Finite, boundary case}. Consider a black vertex $\bs$ of $\GQ$. 
As long as $\bs$ is not a boundary vertex of $\GQ$, the argument is as in the infinite case, so
we suppose that $\bs$ is a boundary vertex. The coefficient $(2X^{-1}\KF_{\sB,\sB})_{\bs,\tilde{b}}$ of the LHS of~\eqref{proof_ident_6}
is zero for all $\tilde{b}\in B$. The coefficient 
$(D_{\sBQ,\sA}\KF_{\sA,\sB}+2 \KQt I_{\sWQ,\sA}D_{\sA,\sB})_{\bs,\tilde{b}}$ of the RHS of~\eqref{proof_ident_6}
is a priori non-zero when $\tilde{b}\in\{b,b_1\}$, and we have
\begin{align*}
(D_{\sBQ,\sA}\KF_{\sA,\sB}+2 \KQt I_{\sWQ,\sA}D_{\sA,\sB})_{\bs,b}&=
\eps_{b,a_1}\eps_{a_1,b}+2\KQt_{\bs,\ws_2}I_{\ws,a_2}\eps_{b,a_2}\\
&=-1+\eps_{b,a_2}\eps_{b,a_2}=0,\\
(D_{\sBQ,\sA}\KF_{\sA,\sB}+2 \KQt I_{\sWQ,\sA}D_{\sA,\sB})_{\bs,b_1}&=0,
\end{align*}
where in the computation for $\tilde{b}=b$ we have used that the boundary coefficient $\KQt_{\bs,\ws_2}=\eps_{b,a_2}$. 
The computation for $\tilde{b}=b_1$
is as in the infinite case. This ends the proof of the finite, boundary case of~\eqref{proof_ident_6} and the proof of 
Theorem~\ref{thm:KFKQinv}.
\end{proof}

\begin{exm}\label{ex:prob_GF}
As an example we give the probability of single edges occurring in dimer configurations of $\GF$ chosen with respect to the 
Boltzmann measure $\PPdimerF$ in the finite case, or the Gibbs measure $\PPdimerF$ in the infinite case,
as a function of edge probabilities of the dimer model on $\GQ$. Details of computations are given in Appendix~\ref{app:dimers_GF}. Using the notation of 
Figure~\ref{fig:matrixKQKF}, we have
\begin{align*}
\PPdimerF(aa')&=\frac{1}{4}-\frac{\PPdimerQ(\bs\ws')}{2\tanh(2\Js)}\\
\PPdimerF(ab)&=\PPdimerF(a'b)=\frac{1}{4}+\frac{\PPdimerQ(\bs\ws')}{2\tanh(2\Js)}\\
\PPdimerF(bb')&=\frac{1}{2}-\frac{\PPdimerQ(\bs\ws')}{\tanh(2\Js)}.
\end{align*}
\end{exm}

\subsection{In the $Z$-invariant case}\label{sec:GQ_GF_Zinv}

We restrict to the case where the graph $\Gs$ is isoradial, finite or infinite, with Ising coupling constants $\Js$ given 
by~\eqref{equ:def_J_Zinv}, dimer weights on $\GF$ by~\eqref{equ:def_mu_Zinv} and dimer weights on $\GQ$ by~\eqref{eq:def_nu_Zinv}. 
The main result of this section is Corollary~\ref{thm:KFKQinv_Zinv} relating the inverse Kasteleyn operator $(\KF)^{-1}$ to the inverse $Z^u$-Dirac operator
using Theorem~\ref{thm:KFKQinv}, Example~\ref{ex:KQ_KD_u_v_special}, Corollaries~\ref{cor:BoltzmannGQ_GD} and~\ref{cor:gibbsGQ_GD}, and 
to the $Z$-massive Green functions. This proves one of the main results of this paper, namely that the contour Ising Boltzmann/Gibbs measures can be computed 
from the inverse $Z^u$-Dirac operator, and also from the $Z$-massive Green functions. As a byproduct, in the infinite this also gives 
a \emph{direct} alternative way of proving the local formula of~\cite{BdtR2} for $(\KF)^{-1}$, where the locality is seen as directly inherited from the 
Green function. Note that as for the local expression of $(\KQ)^{-1}$, it is not immediate to see equality with the expression of~\cite{BdtR2}.

We first relate the real and complex bipartite Kasteleyn matrices $\KQt$ and $\KQ$ of the graph $\GQ$; we use Appendix~\ref{sec2:app}.
Define the following function $q$ on pairs of vertices of $\GQ$, inductively on edges.
For every edge $\bs\ws$ of $\GQ$, let
\[
q_{\bs,\ws}=\frac{\KQ_{\bs,\ws}}{\KQt_{\bs,\ws}}=\frac{e^{i\frac{\bar{\alpha}+\bar{\beta}}{2}}}{\eps_{\bs,\ws}}, \quad q_{\ws,\bs}=q_{\bs,\ws}^{-1}.
\]
For every pair of vertices $\xs,\ys$ of $\GQ$, let 
$
q_{\xs,\ys}=\prod_{(\xs',\ys')\in\gamma_{\xs,\ys}}q_{\xs',\ys'},
$
where $\gamma_{\xs,\ys}$ is an edge-path of $\GQ$ from $\xs$ to $\ys$. Then, since the matrices 
$\KQt$ and $\KQ$ satisfy the alternating product condition around every
face/inner face of $\GQ$ if the graph is infinite/finite~\cite{Kuperberg,Kenyon3}, the function $q$ is well defined.

Consider a fixed vertex $\xs_0$ of $\GQ$, and define the diagonal matrices $D^{\xs_0,\sBs}$, $D^{\xs_0,\sWs}$
on black, resp. white vertices, of $\GQ$ by:
\[
\forall\,\bs\in\Bs,\quad D^{\xs_0,\sBs}_{\bs,\bs}=q_{\xs_0,\bs},\ \forall\,\ws\in\Ws,\quad \D^{\xs_0,\sWs}_{\ws,\ws}=q_{\xs_0,\ws}^{-1}.
\]
By~\cite{Kuperberg,Kenyon3}, see also Appendix~\ref{sec2:app}, we have
\[
\KQt=D^{\xs_0,\sBs} \KQ D^{\xs_0,\sWs},
\]
thus implying the following lemma:
\begin{lem}[\cite{Kuperberg,Kenyon3}]\label{lem:Kup}
Consider the function $q$ as defined above. Then, coefficients of the inverse of the matrices $\KQt$ and $\KQ$ are related by the following, for every
white vertex $\ubar{\ws}$ and every black vertex $\bs$ of $\GQ$,
\[
(\KQt)^{-1}_{\ubar{\ws},\bs}=q_{\bs,\ubar{\ws}}(\KQ)^{-1}_{\ubar{\ws},\bs}.
\]
\end{lem}
Note that the above matrix relation holds in the finite and infinite cases because coefficients of the diagonal matrices are finite and
uniformly bounded away from 0.

Cases 2. and 3. of Theorem~\ref{thm:KFKQinv} directly relate coefficients of $(\KF)^{-1}$ to $(\KQt)^{-1}_{\ubar{\ws},\bs}$. Using
Lemma~\ref{lem:Kup}, Corollary~\ref{cor:gibbsGQ_GD} (infinite case), Corollary~\ref{cor:BoltzmannGQ_GD} and Remark~\ref{rem:final_0} (finite case), these coefficients of
$(\KF)^{-1}$ are easily expressed using the 
inverse $Z^u$-Dirac operator and the $Z$-massive Green functions.

Case 4. of Theorem~\ref{thm:KFKQinv} uses Case 1. so we are left with considering Case 1. The notation used are summarized in 
Figure~\ref{fig:defalphabeta} below. 

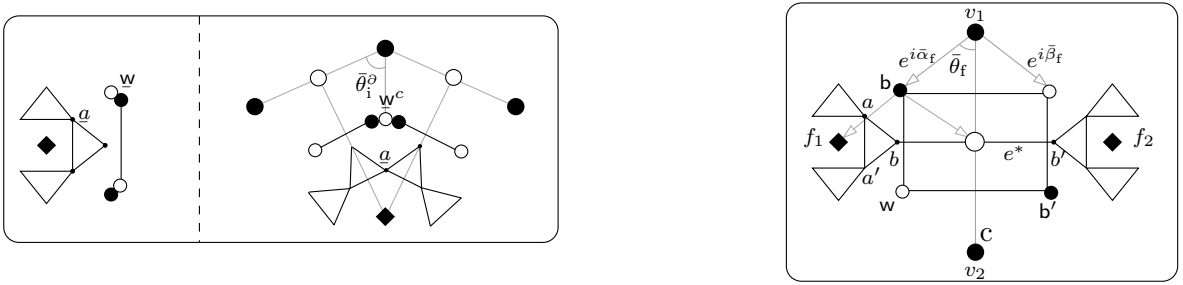
\begin{figure}[H]
\centering
\begin{overpic}[width=\linewidth]{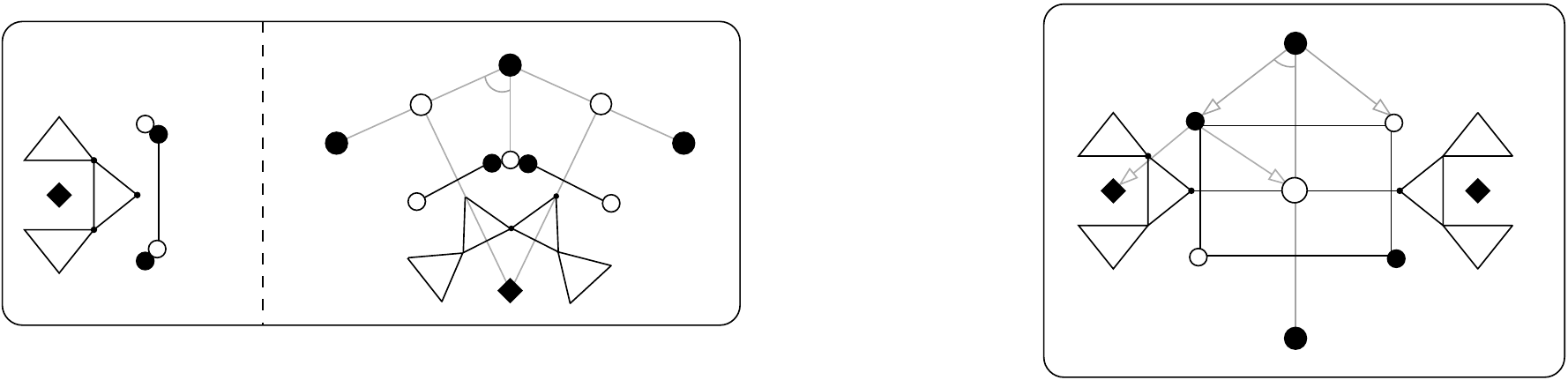}
\put(6.5,14){\scriptsize $\ubar{a}$}
\put(10,17){\scriptsize $\ubar{\ws}$}

\put(32,15.2){\scriptsize $\ubar{\ws}^c$}
\put(32,10.5){\scriptsize $\ubar{a}$}
\put(30,16.4){\scriptsize $\bar{\theta}_\irm^\spartial$}

\put(74.5,16.6){\scriptsize $\bs$}
\put(88,5.6){\scriptsize $\bs'$}
\put(74.5,6.2){\scriptsize $\ws$}

\put(75.3,10){\scriptsize $b$}
\put(89,10){\scriptsize $b'$}

\put(73,15){\scriptsize $a$}
\put(73,8.5){\scriptsize $a'$}

\put(83,3.6){c}
\put(85,10.5){\scriptsize $e^*$}
\put(76.1,18.5){\scriptsize $e^{i\bar{\alpha}_\frm}$}
\put(87,18.5){\scriptsize $e^{i\bar{\beta}_\frm}$}
\put(80.4,18){\scriptsize $\bar{\theta}_\frm$}

\put(68,12){\scriptsize $f_1$}
\put(96,12){\scriptsize $f_2$}
\put(81.7,22.7){\scriptsize $v_1$}
\put(81.7,0.6){\scriptsize $v_2$}
\end{overpic}
\caption{Left: notation for initial vertices. Right: notation for final vertices.}
\label{fig:defalphabeta}
\end{figure}

\begin{cor}[Case 1]\label{thm:KFKQinv_Zinv}$\,$

$\bullet$ For every $\ubar{a}\in A$ and every $b\in B$ such that, when the graph $\GF$ is moreover finite, $b$ is not a boundary vertex:
\begin{align*}
(\KF)^{-1}_{\ubar{a},b}&=\textstyle
q_{\bs,\ubar{\ws}}\cn\bigl(\frac{K-\thetaf}{2}\bigr)\frac{1+(k')^{-1}\dn(\thetaf)}{2}
\left[(\KQ)^{-1}_{\ubar{\ws},\bs}\cn\bigl(\frac{K-\thetaf}{2}\bigr) -i(\KQ)^{-1}_{\ubar{\ws},\bs'}\cn\bigl(\frac{K-\thetaf}{2}\bigr)\right].
\end{align*}
$\bullet$ Moreover, as a function of the inverse $Z^{\ubf}$-Dirac operator, where $\ubf=\frac{\alphaf+\betaf}{2}+K$, we have:

\emph{Infinite case.}
\begin{align*}
\textstyle
(\KF)^{-1}_{\ubar{a},b}
=&\textstyle q_{\bs,\ubar{\ws}} e^{i\frac{\betafb-\betaib}{2}}(k')^{\frac{1}{2}}  \frac{\cn\bigl(\frac{K-\thetaf}{2}\bigr)(1+(k')^{-1}\dn(\thetaf))}{
2[\cn(\thetaf)\sn(\thetaf)]^{\frac{1}{2}}}
\left(\cn(\ubf_{\betai})\KD(\ubf)^{-1}_{\ubar{v},w}-i\sn({\ubf_{\betai}})\KD(\ubf)^{-1}_{\ubar{f},w}\right)\\
=&\textstyle q_{\bs,\ubar{\ws}}e^{-i\frac{\alphafb+\betaib}{2}} \frac{\cn\bigl(\frac{K-\thetaf}{2}\bigr)(1+(k')^{-1}\dn(\thetaf))}{
2}\times\\
&\textstyle\times\left(\frac{\cn(\ubf_{\betai})}{\cn(\thetaf)}(G^m_{\ubar{v},v_2}-G^m_{\ubar{v},v_1})
-\frac{\sn({\ubf_{\betai}})}{\sn(\thetaf)}
\Bigl(\nd\bigl(\frac{K-\thetaf}{2}\bigr)G^{m,*}_{\ubar{f},f_2}-\nd\bigl(\frac{K+\thetaf}{2}\bigr)G^{m,*}_{\ubar{f},f_1}\Bigr)\right).
\end{align*}
\emph{Finite case.} If $\ubar{a}$ corresponds to a vertex $\ubar{\ws}^c$ for some rhombus pair of $\R^{\spartial}$, we then use
the notation of Figure~\ref{fig:defalphabeta} (2nd quadrant on the left):
\begin{align*}
\textstyle
(\KF)^{-1}_{\ubar{a},b}
\textstyle=q_{\bs,\ubar{\ws}}
e^{i\frac{\betafb}{2}}(k')^{\frac{1}{2}}\sn(\theta^{\spartial}_\irm)^{\II_{\{\ubar{\ws}\in\{\ws^c\in\R^\spartial\}\}}}&
\frac{\cn\bigl(\frac{K-\thetaf}{2}\bigr)(1+(k')^{-1}\dn(\thetaf))}{
2[\cn(\thetaf)\sn(\thetaf)]^{\frac{1}{2}}}\times\\
&\textstyle \times \left(\II_{\{\ubar{\ws}\neq \ws^{c,\rs}\}}
t(\ubf)_{\ubar{\ws},\ubar{v}}\KD^\spartial(\ubf)^{-1}_{\ubar{v},w}+t(\ubf)_{\ubar{\ws},\ubar{f}}\KD^\spartial(\ubf)^{-1}_{\ubar{f},w}
\right),
\end{align*}
where the coefficients $\KD^\spartial(\ubf)^{-1}_{\ubar{v},w},\KD^\spartial(\ubf)^{-1}_{\ubar{f},w}$ are expressed using the Green functions $G^{m,\spartial}(u)$
and $G^{m,*}$ in Example~\ref{ex:KD_u_v_special}.
\end{cor}


\begin{proof}
Let us prove the first point. We first compare $(\KQt)^{-1}_{\ubar{\ws},\bs}$ and $(\KQt)^{-1}_{\ubar{\ws},\bs'}$ of Theorem~\ref{thm:KFKQinv}.
Using the notation of Figure~\ref{fig:defalphabeta} (right), we have
\[
(\KQt)^{-1}_{\ubar{\ws},\bs'}\eps_{b',b}=-iq_{\bs,\ubar{\ws}}(\KQ)^{-1}_{\ubar{\ws},\bs},
\]
that is because, omitting the subscript ``$\mathrm{f}$'', 
\begin{align*}
q_{\bs',\ubar{\ws}}\eps_{b',b}&=q_{\bs',\ws}q_{\ws,\bs} q_{\bs,\ubar{\ws}}\eps_{b',b}
=\frac{e^{i\frac{\bar{\beta}-\pi+\bar{\alpha}}{2}}}{\eps_{\bs',\ws}}
\frac{\eps_{\bs,\ws}}{e^{i\frac{\bar{\beta}+\bar{\alpha}}{2}}}\eps_{b',b}q_{\bs,\ubar{\ws}}
=-i\frac{\eps_{b,a'}}{\eps_{b',b'}\eps_{b,a'}}\eps_{b',b}q_{\bs,\ubar{\ws}}
=-iq_{\bs,\ubar{\ws}}.
\end{align*}
We thus have,
\begin{align}\label{equ:fin_0}
(\KF)^{-1}_{\ubar{a},b}&=q_{\bs,\ubar{\ws}}\frac{1}{1+e^{-4\Jf_e}}
\bigl[(\KQ)^{-1}_{\ubar{\ws},\bs} -i (\KQ)^{-1}_{\ubar{\ws},\bs'}e^{-2\Jf_e}\bigr].
\end{align}
We are left with computing the terms involving the coupling constants $\Jf_e$ in the $Z$-invariant case.
By definition, $\Jf_e=\frac{1}{2}\ln\frac{1+\sn\thetaf}{\cn\thetaf}$.
Set $u=\frac{K-\thetaf}{2}$, $v=\frac{K+\thetaf}{2}$, then $u+v=K,u-v=-\thetaf$, so that
\begin{equation}\label{equ:fin_1}
\textstyle 
e^{-2\Jf_e}=\frac{\cn\thetaf}{1+\sn\thetaf}=\frac{\cn(u-v)-\cn(u+v)}{\sn(u+v)-\sn(u-v)}=\sc(u)\dn(v)=
\frac{\sn(K-v)\dn(v)}{\cn(u)}=\frac{\cn\bigl(\frac{K+\thetaf}{2}\bigr)}{\cn\bigl(\frac{K-\thetaf}{2}\bigr)},
\end{equation}
where in the third equality we used~\cite[chap.2, ex.14 (iii)]{Lawden} and in the fourth that $\sn(v+K)=\cd(v)$.
From this and Identity~\eqref{equ:id_somme_cn_carre}, we obtain,
\begin{align}\label{equ:fin_2}
&\frac{1}{1+e^{-4\Jf_e}}=\textstyle
\frac{\cn^2\bigl(\frac{K-\thetaf}{2}\bigr)}{\cn^2\bigl(\frac{K-\thetaf}{2}\bigr)+\cn^2\bigl(\frac{K+\thetaf}{2}\bigr)}
=\cn^2\bigl(\frac{K-\thetaf}{2}\bigr)\frac{1+(k')^{-1}\dn(\thetaf)}{2}.
\end{align}
Putting together Equation~\eqref{equ:fin_0}, ~\eqref{equ:fin_1} and ~\eqref{equ:fin_2} ends the proof of the first point; let us now prove the second.

In the infinite case, by Example~\ref{ex:KQ_KD_u_v_special}, we have
\begin{equation*}
\textstyle
(\KQ)^{-1}_{\ubar{\ws},\bs}\cn\bigl(\frac{K-\thetaf}{2}\bigr)-i(\KQ)^{-1}_{\ubar{\ws},\bs'}\sn\bigl(\frac{K+\thetaf}{2}\bigr)
=\frac{e^{i\frac{\betafb-\betaib}{2}}(k')^{\frac{1}{2}}}{[\cn(\thetaf)\sn(\thetaf)]^{\frac{1}{2}}}
\left(\cn(\ubf_{\betai})\KD(\ubf)^{-1}_{\ubar{v},w}-i\sn({\ubf_{\betai}})\KD(\ubf)^{-1}_{\ubar{f},w}\right),
\end{equation*}
thus proving the first line. The second line is obtained by using Example~\ref{ex:KD_u_v_special} to express $\KD(\ubf)^{-1}_{\ubar{v},w}$
and $\KD(\ubf)^{-1}_{\ubar{f},w}$ using the $Z$-massive Green functions $G^m$ and $G^{m,*}$.

In the finite case, by Example~\ref{ex:KQ_KD_u_v_special}, we have
\begin{align*}
\textstyle
(\KQ)^{-1}_{\ubar{\ws},\bs}\cn\bigl(\frac{K-\thetaf}{2}\bigr)&-i\textstyle (\KQ)^{-1}_{\ubar{\ws},\bs'}\sn\bigl(\frac{K+\thetaf}{2}\bigr)=\\
&\textstyle=\frac{e^{i\frac{\betafb}{2}}(k')^{\frac{1}{2}}\sn(\theta^{\spartial}_\irm)^{\II_{\{\ubar{\ws}\in\{\ws^c\in\R^\spartial\}\}}}}{
[\cn(\thetaf)\sn(\thetaf)]^{\frac{1}{2}}}
\left(\II_{\{\ubar{\ws}\neq \ws^{c,\rs}\}}
t(\ubf)_{\ubar{\ws},\ubar{v}}\KD^\spartial(\ubf)^{-1}_{\ubar{v},w}+t(\ubf)_{\ubar{\ws},\ubar{f}}\KD^\spartial(\ubf)^{-1}_{\ubar{f},w}
\right),
\end{align*}
thus concluding the proof.
\end{proof}

\section{Examples}\label{sec:Examples}

In this section we specify some of our results to two cases of interest: the \emph{critical} $Z$-invariant Ising model, and
the \emph{full} $Z$-invariant Ising model when the underlying isoradial graph $\Gs$ is $\ZZ^2$ with the regular embedding.

\subsection{$Z$-invariant critical case}\label{sec:ex_Z_inv_critical}

The $Z$-invariant Ising model is critical when $k=0$ ($k'=1$)~\cite{Li:critical,CimasoniDuminil,Lis}. In this case, the elliptic functions 
$\sn,\cn$ are the trigonometric functions $\sin,\cos$ and $\dn\equiv 1$. 

Returning to Section~\ref{sec:def_Dirac}, the finite
$Z^u$-Dirac operator $\KD^\spartial(u)$ with boundary conditions arising from the Ising model is:
\begin{equation*}
\forall\,\text{edge $wx$ of $\GD$},\quad 
\KD^\spartial(u)_{w,x}=
e^{i\frac{\bar{\alpha}_e+\bar{\beta}_e}{2}}
\begin{cases}
\tan(\theta)^{\frac{1}{2}}& \text{ if $x\in\Vs$ and $(w,x)\notin(w^\ell,v^c)\in\R^{\spartial,\rs}$}\\
\cot(\theta)^{\frac{1}{2}}& \text{ if $x\in\Vs^*$}\\
\tan(\theta^\spartial)\frac{\cos(u_{\beta^r})}{\cos(u_{\alpha^\ell})}& \text{ if $(w,x)\in(w^\ell,v^c)\in\R^{\spartial,\rs}$}.
\end{cases}
\end{equation*}
Away from the boundary, we recover the Dirac operator of~\cite{Kenyon3}.
\paragraph{Dimer model on the bipartite graph $\GQ$, finite case.}
Corollary~\ref{cor:BoltzmannGQ_GD} expresses coefficients of $(\KQ)^{-1}$ as a function of the inverse $Z^u$-Dirac operator. 
When $\bs$ does not belong to a boundary quadrangle (Case 1.), we have
\begin{equation*}\textstyle
(\KQ)^{-1}_{\ubar{\ws},\bs}=e^{i\frac{\betafb}{2}}(k')^{\frac{1}{2}}\sin(\theta^{\spartial}_\irm)^{\II_{\{\ubar{\ws}\in\{\ws^c\in\R^\spartial\}\}}}
\frac{1}{[\cos(\thetaf)\sin(\thetaf)]^{\frac{1}{2}}}
\Bigl[\cos\bigl(\frac{K-\thetaf}{2}\bigr)\Gamma(\ubf)+ \cos\bigl(\frac{K+\thetaf}{2}\bigr)\Gamma(\vbf)\Bigr],
\end{equation*}
where, $\Gamma(u)=\II_{\{\ubar{\ws}\neq \ws^{c,\rs}\}}
t(u)_{\ubar{\ws},\ubar{v}}\KD^\spartial(u)^{-1}_{\ubar{v},w}+t(u)_{\ubar{\ws},\ubar{f}}\KD^\spartial(u)^{-1}_{\ubar{f},w}$, and coefficients 
of $t(u)$ are given by~\eqref{def:twvf} and~\eqref{def:twv_boundary}. 
The mass of the $Z$-massive Laplacian is equal to 0 away from the boundary~\cite{BdTR1}; then
Example~\ref{ex:KD_u_v_special} expresses coefficients of
$\KD^\spartial(\ubf)^{-1}$ using the $Z$-Green function $G^{\, 0,\spartial}(\ubf)$ and dual $Z$-Green
function $G^{0,*}$; we have
\begin{align*}
\KD^\spartial(\ubf)^{-1}_{\ubar{v},w}&\textstyle =e^{-i\frac{\alphafb+\betafb}{2}}\tan(\thetaf)^{\frac{1}{2}}
\Bigl(G^{\, 0,\spartial}(\ubf)_{v_2,\ubar{v}}-G^{\, 0,\spartial}(\ubf)_{v_1,\ubar{v}}\Bigr)\\
\KD^\spartial(\ubf)^{-1}_{\ubar{f},w}&\textstyle =-ie^{-i\frac{\alphafb+\betafb}{2}}\cot(\thetaf)^{\frac{1}{2}}
\Bigl(\II_{\{w\notin\{w^\ell,w^r\in \R^\spartial\}\}}G^{\, 0,*}_{f_2,\ubar{f}}-
G^{\, 0,*}_{f_1,\ubar{f}}\Bigl)+\\
&\textstyle\quad \quad \quad \quad 
\ \ \ \ \ \ +
i \sum\limits_{(v^c,f^c)\in\R^{\spartial,\rs}}\frac{\cos(\ubf_{\beta^r})-\cos(\ubf_{\alpha^\ell})}{\cos(\ubf_{\alpha^\ell})}
\KD^{\spartial}(\ubf)^{-1}_{v^c,w}\cdot G^{\, 0,*}_{f^c,\ubar{f}}.
\end{align*}
A similar expression holds for $\KD^\spartial(\vbf)^{-1}$.
\paragraph{Dimer model on the Fisher graph $\GF$, finite case.} Corollary~\ref{thm:KFKQinv_Zinv} (finite case) simply becomes:
\begin{align*}
\textstyle
(\KF)^{-1}_{\ubar{a},b}
\textstyle=
e^{i\frac{\betafb}{2}}\sin(\theta^{\spartial}_\irm)^{\II_{\{\ubar{\ws}\in\{\ws^c\in\R^\spartial\}\}}}&
\frac{\cos\bigl(\frac{K+\thetaf}{2}\bigr)}{
[\cos(\thetaf)\sin(\thetaf)]^{\frac{1}{2}}}\times\\
&\textstyle \times \left(\II_{\{\ubar{\ws}\neq \ws^{c,\rs}\}}
t(\ubf)_{\ubar{\ws},\ubar{v}}\KD^\spartial(\ubf)^{-1}_{\ubar{v},w}+t(\ubf)_{\ubar{\ws},\ubar{f}}\KD^\spartial(\ubf)^{-1}_{\ubar{f},w}
\right).
\end{align*}
\paragraph{Partition function of the Ising model with + boundary conditions.} Corollary~\ref{cor:partition_function}, 
expressing the squared $Z$-invariant
Ising partition function, holds for every $u\in\Re(\TT(k))''$. 
When $k=0$, a nice expression is obtained by setting $u=iu'$ and taking the limit $u'\rightarrow-\infty$. We have 
\[
\lim_{u'\rightarrow-\infty} \frac{\cos((iu')_\alpha)}{\cos((iv)_\beta)}=e^{i\frac{\beta-\alpha}{2}}, \quad 
\lim_{u'\rightarrow-\infty} \tan((iu')_{\alpha})=1.
\]
As a consequence, 
\begin{equation*}
\Delta^{m,\spartial}(-i\infty)_{v,v'}=
\begin{cases}
-\tan(\theta)& \text{  if $(v,v')\notin\{(v^c,v^\ell)\in \R^{\spartial,\rs}\}$}\\
-e^{i\frac{\alpha^\ell-\beta^r}{2}}\tan(\theta^\spartial)&\text{ if $(v,v')\in\{(v^c,v^\ell)\in \R^{\spartial,\rs}\}$}\\
\sum_{j=1}^d \tan(\theta_j)&\text{ if $v'=v=v^\spartial\in\{v^r,v^\ell\in\R^\spartial\}$}\\
\tan(\theta^\spartial)(e^{i\frac{\alpha^\ell-\beta^r}{2}}-1)&\text{ if $v'=v=v^c\in\{v^c\in \R^{\spartial,\rs}\}$},
\end{cases}
\end{equation*}
where we recover the critical Laplacian of~\cite{Kenyon3} away from the boundary. Corollary~\ref{cor:partition_function} becomes:
\begin{equation*}
[\Zising^+(\Gs,\Js)]^2=2^{|\Vs|-1}
\Bigl(\prod_{w\in W^\spartial}\frac{1+\sin(\theta_w)}{2\sin(\theta_w)}\Bigr)
\sn(\theta^{\spartial,\rs})|\det \Delta^{m,\spartial}(-i\infty)|.\\
\end{equation*}
We essentially recover the main result of~\cite{deTiliere:partition} proving that the squared critical $Z$-invariant Ising model partition function 
is equal, up to an explicit constant, to the partition function of spanning trees with specific boundary conditions. The difference is that we here consider + boundary conditions instead of free ones, and 
more importantly, we use the boundary trick of Chelkak and Smirnov~\cite{ChelkakSmirnov:ising} allowing us to remove all contributions from 
dual spanning trees, which we could not do in~\cite{deTiliere:partition}.

\subsection{Full $Z$-invariant case when the isoradial graph $\Gs$ is $\ZZ^2$}\label{sec:ex_Z_inv_ZZ2}

The goal of this example is to relate our results to the papers~\cite{Messikh,BeffaraDuminil}.
When the isoradial graph $\Gs=\ZZ^2$ in its regular embedding, all rhombus half-angles $\bar{\theta}_e$ are equal to $\frac{\pi}{4}$. Using that 
$\sn\bigl(\frac{K}{2}\bigr)=\frac{1}{(1+k')^{\frac{1}{2}}}$, $\cn\bigl(\frac{K}{2}\bigr)=\frac{(k')^{\frac{1}{2}}}{(1+k')^{\frac{1}{2}}}$,
$\dn\bigl(\frac{K}{2}\bigr)=(k')^{\frac{1}{2}}$~\cite[16.5.2]{AS}, we obtain the following parametrization of the Ising coupling constants:
\[
\forall\,e\in\Es,\quad \Js_e=\Js_e(k')=\frac{1}{2}\ln\left(\frac{1+(1+k')^{\frac{1}{2}}}{(k')^{\frac{1}{2}}}\right),
\]
which is a $\mathcal{C}^\infty$-bijection from $(0,\infty)$ to $(\infty,0)$. 

The masses of the $Z$-massive Laplacian $\Delta^m$ are equal to~\cite{BdTR1},
$
\forall\,v\in \Vs,\ m_v=\frac{2(k')^{\frac{1}{2}}}{1+k'},
$
and the survival probabilities $\ss$ of the random walk associated to the $Z$-invariant conductances $\rho$ and masses $m$ are
\[
\forall\, v\in\Vs,\quad \ss_v=\frac{4\sc\bigl(\frac{K}{2}\bigr)}{4\sc\bigl(\frac{K}{2}\bigr)+m_v}=\frac{2(k')^{\frac{1}{2}}}{1+k'}.
\]
As in the papers~\cite{Messikh,BeffaraDuminil} the survival probabilities satisfy
\[
\forall\, v\in\Vs, \quad \ss_v=\frac{2}{\sinh(2\Js_e)+\sinh^{-1}(2\Js_e)}.
\]
In the infinite case, by Theorem~\ref{thm:KFKQinv} and Corollary~\ref{cor:gibbsGQ_GD}, the coefficient $(\KF)^{-1}_{\ubar{a},a}$ of the inverse Kasteleyn operator $(\KF)^{-1}$ is equal to, as long as $\ubar{a}$ and $a$
do not belong to the same decoration, 
\[\textstyle 
|(\KF)^{-1}_{\ubar{a},a}|=(k')^{-1}(1+(k'))^{\frac{1}{2}}
\Bigl(\cn\bigl(\frac{\betaf-\betai}{2}\bigr)\bigl[
G^m_{\ubar{v},v_2}-(k')^{\frac{1}{2}}G^m_{\ubar{v},v_1}\bigr]
-\sn\bigl(\frac{\betaf-\betai}{2}\bigr)
\bigl[(k')^{\frac{1}{2}}G^{m,*}_{\ubar{f},f_2}-G^{m,*}_{\ubar{f},f_1}\bigr]\Bigr).
\]
A similar expression holds in the finite case, see Corollary~\ref{cor:BoltzmannGQ_GD} and Example~\ref{ex:KD_u_v_special}. According 
to~\cite{Dubedat,NK}, this coefficient is the fermionic spinor observable of~\cite{KadanoffCeva} and, up to normalization, the FK-spin observable 
of~\cite{Smirnov2,ChelkakSmirnov:ising}. Now, the proof of~\cite{BeffaraDuminil} for showing the occurrence of 
large deviation estimates of the massive random walk in the correlation length of the spin correlations~\cite{Messikh} consists in proving that, in the super-critical 
regime, spin correlations can be approximated by the FK-spin observable, and then using massive harmonicity of the latter to relate it to the massive
random walk. Our explicit expression for $|(\KF)^{-1}_{\ubar{a},a}|$ in the finite case gives
a direct explanation of the occurrence of these large deviation estimates, and up to handling boundary terms, should give a rather direct proof valid in the whole 
super-critical $Z$-invariant case.

\appendix

\section{Gauge equivalence revisited}\label{app:gauge}

We consider gauge equivalence of weighted adjacency matrices of digraphs and rephrase gauge equivalence of bipartite 
weighted adjacency matrices as defined in \cite{Kuperberg,KOS} in this context. 

\subsection{Definitions}\label{sec:app_gauge_1}

In the whole of this section, we consider square matrices of size $n\times n$;
let $M$ be such a matrix. We associate two graphs to $M$, a non-directed one $\Gs(M)=(\Vs(M),\Es(M))$
and a directed one $\Ds(M)=(\Vs(M),\As(M))$, both having the same vertex set of cardinality $n$ in bijection 
with rows/columns of $M$. 
An edge
$x_j x_\ell$ is in $\Es(M)$ iff $M_{x_j,x_\ell}\neq 0$ or $M_{x_\ell,x_j}\neq 0$. A directed edge (or simply an edge) $(x_j,x_\ell)$ is in
$\As(M)$ iff the coefficient $M_{x_j,x_\ell}\neq 0$. The matrix $M$ is a \emph{weighted adjacency matrix} of the digraph $\Ds(M)$.
Whenever no confusion occurs, we remove the argument $M$ of the graphs. 

Let us recall a few definitions.  A \emph{di-path} of a digraph
$\Ds=(\Vs,\As)$ is a sequence $(x_0,\dots,x_n)$ of vertices such that for every $j\in\{0,\dots,n-1\}$,
$(x_j,x_{j+1})$ is an edge of $\As$. A \emph{simple di-path} is a path with pairwise disjoint vertices.
A \emph{di-cycle} is a di-path such that the first and last vertices are the same. A \emph{simple di-cycle}
is a cycle whose only common vertices are the first and the last. Note that loops and length-two di-cycles are simple.
A digraph is \emph{strongly connected} if any two pairs of vertices are joined by a di-path.
The above definitions are easily adapted in the case of non-directed graphs.

\begin{defi}\label{def:jauge}
Consider two matrices $M$ and $N$ having the same associated digraph $\Ds$.
The matrices $M$ and $N$ are said to be \emph{gauge equivalent} if,
\[\forall\,
\text{simple di-cycle $c$ of $\Ds$},\quad
\prod_{e=(x,y)\in c} M_{x,y}=\prod_{e=(x,y)\in c} N_{x,y}.
\]
\end{defi}
\begin{rem}\label{rem:gauge1}$\,$
\begin{itemize}
\item  Having the same associated digraph is equivalent to asking that $M_{x,y}\neq 0\,\Leftrightarrow\,N_{x,y}\neq 0$.
\item Since loops are simple di-cycles, if $M$ and $N$ are gauge equivalent, they have equal diagonal coefficients.
\item Definition~\ref{def:jauge} holds if and only if the product condition holds for every di-cycle of $\Ds$. 
\end{itemize}
\end{rem}

\begin{lem}\label{lem:gauge_det_1}
Let $M,N$ be two gauge equivalent matrices, then
\[
\det M=\det N.
\]
\end{lem}
\begin{proof}
This is proved by writing the determinant as a sum over permutations, and doing the cyclic decomposition of permutations. 
\end{proof}

Let us suppose that $M$ and $N$ are gauge equivalent and that the associated digraph $\Ds$ is \emph{strongly connected}.
Define the function $q\in \CC^{\Vs\times \Vs}$ on pairs of vertices of $\Vs$ taking values in $\CC$ as follows. 
For vertices $x,y$ such that $(x,y)$ is an edge of $\Ds$, set 
\[
q_{x,y}=\frac{N_{x,y}}{M_{x,y}}.
\]
For vertices $x,y$ of $\Ds$, since the digraph is strongly connected, there exists a di-path
$\gamma$ from $x$ to $y$; set
\[
q_{x,y}=\prod_{e=(x',y')\in\gamma}q_{x',y'}.
\]
Note that if $y=x$, then $\gamma$ is a di-cycle and we have $q_{x,x}=1$.
\begin{rem}
The function $q$ is well defined, \emph{i.e.}, independent of the choice of path from $x$ to $y$. If $y=x$, then $q_{x,x}=1$ 
independently of the choice of di-cycle from $x$ to $x$. If $y\neq x$,
consider two di-paths $\gamma_1,\gamma_2$ from $x$ to $y$. Since the digraph $\Ds$ is strongly connected, there exists a simple 
di-path $\tilde{\gamma}$ from $y$ to $x$. Then, $\gamma_1$  (resp. $\gamma_2$) followed by $\tilde{\gamma}$ is a di-cycle and by 
definition of gauge equivalence we have,
\[
\Bigl(\prod_{e=(x',y')\in\gamma_1}q_{x',y'}\Bigr)\Bigl(\prod_{e=(x',y')\in\tilde{\gamma}}q_{x',y'}\Bigr)=1=
\Bigl(\prod_{e=(x',y')\in\gamma_2}q_{x',y'}\Bigr)\Bigl(\prod_{e=(x',y')\in\tilde{\gamma}}q_{x',y'}\Bigr),
\]
implying that
\[
\prod_{e=(x',y')\in\gamma_1}q_{x',y'}=\prod_{e=(x',y')\in\gamma_2}q_{x',y'},
\]
and the function $q$ is thus well defined.
\end{rem}

The following lemma proves that if the associated digraph is strongly connected, gauge equivalence amounts to having the two 
matrices related through a diagonal matrix.

\begin{lem}\label{lem:crit_gauge}
Let $M,N$ be two matrices having the same associated digraph $\Ds$ and suppose that $\Ds$ is strongly connected. Then 
$M$ and $N$ are gauge equivalent if and only if there exists an invertible diagonal matrix $D$ such that,
\[
M=D N D^{-1}.
\]
\end{lem}
\begin{proof}
Suppose that $M$ and $N$ are gauge equivalent. Fix a vertex $x_0$ of $\Ds$, and define $D^{x_0}$ to be the diagonal matrix whose
diagonal coefficient $D^{x_0}_{x,x}$ corresponding to the vertex $x$ is $q_{x_0,x}$. Since the digraph $\Ds$ is strongly 
connected,
the matrix $D^{x_0}$ is invertible. Let us prove that $M=D^{x_0}N(D^{x_0})^{-1}$. Non zero coefficients of $M$ and $N$ correspond to edges
of $\Ds$; let $(x,y)$ be such an edge. Consider a di-path $\gamma$ from $x_0$ to $x$, then $\gamma'=\gamma\cup\{(x,y)\}$
is a di-path from $x_0$ to $y$. Using $\gamma$ to compute $q_{x_0,x}$ and $\gamma'$ to compute $q_{x_0,y}$, we deduce that,
\begin{equation*}
(D^{x_0} N (D^{x_0})^{-1})_{x,y}=\frac{q_{x_0,x}}{q_{x_0,y}}N_{x,y}=\frac{1}{q_{x,y}}N_{x,y}=\frac{M_{x,y}}{N_{x,y}}N_{x,y}=M_{x,y}.
\end{equation*}
Suppose that $M=D N D^{-1}$, with $D$ an invertible diagonal matrix, and
let us prove that $M$ and $N$ are gauge equivalent.
Consider a simple di-cycle $c$ of $\Ds$, then
\begin{equation*}
\prod_{e=(x,y)\in c} M_{x,y}=\prod_{e=(x,y)\in c} D_{x,x} N_{x,y} D_{y,y}^{-1}=\prod_{e=(x,y)\in c}N_{x,y} \text{ (telescopic product)},
\end{equation*}
thus concluding the proof.
\end{proof}

\begin{rem}
Consider two matrices $M,N$ having the same associated graphs $\Gs$ and $\Ds$, such that $\Gs$ is connected and 
such that for every undirected edge $e$ of $\Gs$, the two possible oriented
edges are present in $\Ds$; then $\Ds$ is strongly connected. If moreover $\Gs$ is planar and embedded in a planar way, then
every simple di-cycle $c$ of length $\geq 3$ is the union of face di-cycles, where edges not in $c$ are traversed in both directions. 
As a consequence, in this case
proving gauge equivalence for $M$ and $N$ is equivalent to proving that,
\begin{align*}
&\forall\,x\in\Vs,\quad M_{x,x}=N_{x,x},\\
&\forall\text{ edge $e$ of $\Gs$},\quad M_{x,y}M_{y,x}=N_{x,y}N_{y,x},\\
&\forall\text{ face di-cycle } c,\quad \prod_{e=(x,y)\in c} M_{x,y}=\prod_{e=(x,y)\in c} N_{x,y}.
\end{align*}
\end{rem}

\subsection{The bipartite case}\label{sec2:app}

Consider a non-directed, finite, bipartite graph $\Gs=(\Vs=\Ws\cup\Bs,\Es)$, such that $|\Ws|=|\Bs|=n$, having at least one perfect matching.
Note that $\Gs$ being bipartite, it cannot have loops; we furthermore suppose that it has no multiple edges, \emph{i.e.},
that it is \emph{simple}.

Fix a perfect matching $\Ms_0=\{\bs_1\ws_1,\dots,\bs_n\ws_n\}$ of $\Gs$. From $\Gs$ and $\Ms_0$, construct a digraph 
$\Ds^0$ in the following way: vertices $\bs_j$ and $\ws_j$ are 
merged into a single vertex $x_j$, and the corresponding edge becomes a loop. The vertex set of $\Ds^0$ is $\Vs^0=\{x_1,\dots,x_n\}$.
Edges not in the perfect matching $\Ms_0$ remain in $\Ds^0$ and are directed from their black vertex to the white one, defining the 
directed edges of $\Ds^0$. 

A \emph{bipartite, weighted adjacency matrix} $\Ks$ associated to the graph $\Gs$ has rows indexed by vertices of $\Bs$, 
columns by those of $\Ws$, and non-zero
coefficients correspond to edges of~$\Gs$. Up to a reordering of the rows and columns, we can suppose that rows of 
$\Ks$ are indexed by $\bs_1,\dots,\bs_n$ and 
columns by $\ws_1,\dots,\ws_n$. 

Instead of seeing $\Ks$ as a bipartite adjacency matrix, we can interpret it as an adjacency matrix of the digraph $\Ds^0$.
In this interpretation, rows and columns are indexed by vertices of $\Vs^0$ and diagonal coefficients correspond to edges 
of the perfect matching, they now represent loops.

Consider the diagonal matrix $D_\Ks^{0}$ whose $j$-th diagonal coefficients is the coefficient $\Ks_{b_j,w_j}$ corresponding to 
the $j$-th edge of $\Ms_0$. Define the matrix 
$\Ks^0$ to be,
\begin{equation*}
\Ks^0=(D_\Ks^0)^{-1}\Ks. 
\end{equation*}
Note that the matrix $\Ks^0$ has ones on the diagonal.

\begin{defi}\label{def:gauge_bip}
Let $\Gs$ be a finite, bipartite graph, and let $\Ks,\Ls$ be two associated bipartite, weighted adjacency matrices.
Fix a perfect matching $\Ms_0$ of $\Gs$, and let $\Ds^0$ be the directed 
graph constructed from $\Gs$ and $\Ms_0$. 
Then, $\Ks$ and $\Ls$ are said to be \emph{gauge equivalent} if the matrices $\Ks^0$ and $\Ls^0$, 
seen as weighted adjacency matrices of the digraph $\Ds^0$, are gauge equivalent.
\end{defi}

Rephrasing Lemma~\ref{lem:gauge_det_1} in the context of bipartite graphs, we obtain

\begin{cor}\label{cor:gauge_bip_det}
Let $\Ks,\Ls$ be two gauge equivalent bipartite, weighted adjacency matrices, then
\begin{align*}
\det\Ks =\Bigl(\prod_{e=\bs\ws\in\Ms_0}\frac{\Ks_{\bs,\ws}}{\Ls_{\bs,\ws}}\Bigr) \det \Ls.
\end{align*}
\end{cor}

We now rephrase Definition~\ref{def:gauge_bip} in the more usual form~\cite{Kuperberg,KOS}. 
Consider the bipartite graph $\Gs$ together with the reference perfect matching $\Ms_0$. 
An \emph{alternating cycle of $\Gs$ and $\Ms_0$} is a \emph{simple} cycle of $\Gs$ whose edges alternate
between edges in $\Ms_0$ and edges in $\Es\setminus\Ms_0$. 

\begin{lem}
The matrices $\Ks$ and $\Ls$ are gauge equivalent iff for every alternating cycle $c$ of $\Gs$ and $\Ms_0$ of length $>2$, 
we have
\begin{equation*}
\frac{\prod_{e=bw\in c\setminus \Ms_0}\Ks_{b,w}}{\prod_{e=bw\in c\cap \Ms_0}\Ks_{b,w}}=
\frac{\prod_{e=bw\in c\setminus \Ms_0}\Ls_{b,w}}{\prod_{e=bw\in c\cap \Ms_0}\Ls_{b,w}}.
\end{equation*}
\end{lem}
\begin{proof}
By definition, $\Ks$ and $\Ls$ are gauge equivalent if $\Ks^0$ and $\Ls^0$ seen as adjacency matrices of the digraph $\Gs^0$ are 
gauge equivalent, see Definition~\ref{def:jauge}. Length one di-cycles of $\Ds^0$ 
are loops corresponding to diagonal coefficients of $\Ks^0,
\Ls^0$. The latter are all equal to 1 by definition, and thus equal. 
Consider a simple di-cycle $c$ of $\Ds^0$ of length $m\geq 2$. Up to a relabeling of the vertices, it can be denoted as
$c=\{x_1,\dots,x_l=x_1\}$. Then, $\forall\,j$, $x_j$ corresponds to an edge $b_j w_j$ of the perfect matching $\Ms_0$. 
By construction of $\Ds^0$, the cycle $c$ is in correspondence with an alternating cycle 
$\{w_1,b_1,w_2,\dots,w_m,b_m,w_1\}$ of length $2m$ of $\Gs$. By definition of $\Ks^0$ we have,
\[
\prod_{j=1}^m \Ks^0_{x_j,x_{j+1}}=\frac{\prod_{j=1}^m \Ks_{b_j,w_{j+1}}}{\prod_{j=1}^m \Ks_{b_j,w_{j}}},
\]
with cyclic notation for indices. A similar equality holds for the matrices $\Ls^0$ and $\Ls$, thus ending the proof.
\end{proof}

\begin{rem}
Definition~\ref{def:gauge_bip} is independent of the choice of $\Ms_0$. To prove this, use the fact that if $\Ms_1$ is another
reference perfect matching, then the superimposition of $\Ms_0$ and $\Ms_1$ consists of alternating cycles of length $>2$ and 
doubled edges. 
\end{rem}

From now on we suppose that the bipartite graph $\Gs$ is finite, \emph{planar and simply connected}, and we let $\Ks$, $\Ls$ be two bipartite,
weighted adjacency matrices of $\Gs$. 

\begin{lem}[\cite{Kuperberg}]\label{lem:gauge_bip_equiv}
If the alternating products of the matrices $\Ks$ and $\Ls$ are equal around 
every inner face-cycle of $\Gs$, then $\Ks$ and $\Ls$ are gauge equivalent.
\end{lem}
\begin{proof}
This is proved by induction on the number of faces contained in an alternating cycle of length $>2$.
\end{proof}

Suppose that $\Ks$ and $\Ls$ satisfy the alternating product condition around every inner face-cycle of $\Gs$. Similarly
to the directed case, define the function $q\in\CC^{\Vs\times\Vs}$ as follows. For every edge $bw$ of $\Gs$, 
\[
q_{b,w}=\frac{\Ls_{b,w}}{\Ks_{b,w}},\quad q_{w,b}=q_{b,w}^{-1}.
\]
For every pair of vertices $x,y$ of $\Gs$, let $\gamma$ be a path from $x$ to $y$ and set $q_{x,y}=\prod_{(x',y')\in\gamma}q_{x',y'}$. 
The function $q$ is well defined~\cite{Kuperberg} and

\begin{lem}[\cite{Kuperberg}]
The matrices $\Ks$ and $\Ls$ satisfy the alternating product condition around every inner face-cycle if and only if there exist
diagonal matrices $D^{\sBs},D^{\sWs}$ such that 
\[
\Ks=D^{\sBs}\,\Ls\,D^{\sWs}.
\]
\end{lem}
\begin{proof}
The proof can be found in~\cite{Kuperberg}. For the purpose of Section~\ref{sec:GQ_GF_Zinv}, we make the diagonal matrices explicit assuming the alternating
product condition is satisfied. Fix a vertex $x_0$ of $\Gs$ and set,
\begin{equation*}
\forall\,b\in \Bs,\quad D^{x_0,\sBs}_{b,b}=q_{x_0,b}, \quad \forall\,w\in \Ws,\quad D^{x_0,\sWs}_{b,b}=q_{x_0,w}^{-1},
\end{equation*}
where $q$ is given above. Then, $\Ks=D^{x_0,\sBs}\,\Ls\,D^{x_0,\sWs}$.
\end{proof}


\section{Computations of probabilities of single edges.}

\subsection{Dimers model on an infinite isoradial double graph $\GD$}\label{app:dimers_double}

We compute the probability of single edges occurring in dimer configurations of $\GD$ chosen with respect to 
the measure $\PPdimerD$. Notation are recalled in Figure~\ref{fig:corGDK_1} below; since no confusion occurs, we omit the subscripts
$\mathrm{f}$ from $\alpha,\beta,\theta$.

\begin{figure}[H]
\centering
\begin{overpic}[height=2.8cm]{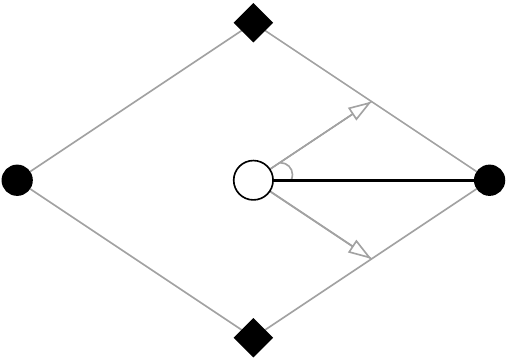}
  \put(101,34){\scriptsize $v_2$}
  \put(47,73){\scriptsize $f_2$}
  \put(51,20){\scriptsize $e^{i\bar{\alpha}}$}
  \put(51,42.5){\scriptsize $e^{i\bar{\beta}}$}
  \put(39,33){\scriptsize $w$}
  \put(60,35){\scriptsize $\bar{\theta}$}
\end{overpic}
\caption{Notation for computing $\PPdimerD(wx)$ when $x=v_2$ or $f_2$.}
\label{fig:corGDK_1}
\end{figure}

For every edge $wx$ of $\GD$, we have $\PPdimerD(wx)=\KD(u)_{w,x}\KD(u)^{-1}_{x,w}$,
where $x$ is a vertex of $\Gs$ or $\Gs^*$. Setting $\ubar{v}=v_2$ in $\KD(u)^{-1}_{v,w}$ and $\ubar{f}=f_2$ in $\KD(u)^{-1}_{f,w}$ of
Corollary~\ref{cor:KD_G}, and using that $\dn(u-K)=k'\nd(u)$, gives
\begin{align*}
\KD(u)^{-1}_{v_2,w}&=
e^{-i\frac{\bar{\alpha}+\bar{\beta}}{2}}(k')^{-1}[\sc(\theta)\nd(u_{\alpha})\nd(u_{\beta})]^{\frac{1}{2}}
\left[
\dn(u_{\alpha})\dn(u_{\beta})G^m_{v_2,v_2}-k'G^m_{v_2,v_1}\right]\\
\KD(u)^{-1}_{f_2,w}&\textstyle=
-ie^{-i\frac{\bar{\alpha}+\bar{\beta}}{2}}(k')^{-1}[\sc(\theta^*)\nd((u_{\beta})^*)\nd((u_{\alpha+2K})^*)]^{\frac{1}{2}}\times\\
&\hspace{3cm}\times \bigl[\dn((u_{\beta})^*)\dn((u_{\alpha+2K})^*)G^{m,*}_{f_2,f_2}-k'G^{m,*}_{f_2,f_1}\bigr].
\end{align*}
Multiplying by $\KD(u)_{w,v_2}$ and $\KD(u)_{w,f_2}$ respectively yields the first equalities of Example~\ref{ex:Prob_GD}:
\begin{align*}
\KD(u)_{w,v_2}\KD(u)^{-1}_{v_2,w}&=
(k')^{-1}\sc(\theta)
\left[
\dn(u_{\alpha})\dn(u_{\beta})G^m_{v_2,v_2}-k'G^m_{v_2,v_1}\right]\\
\KD(u)_{w,f_2}\KD(u)^{-1}_{f_2,w}&=
(k')^{-1}\sc(\theta^*)
\left[\dn((u_{\beta})^*)\dn((u_{\alpha+2K})^*)G^{m,*}_{f_2,f_2}-k'G^{m,*}_{f_2,f_1}
\right].
\end{align*}
Now by \cite[Lemma 46]{BdTR1}, for every vertex $v$ of $\Gs$, $G^m_{v,v}=\frac{k'K'}{\pi}$ and, for every $u\in\Re(\TT(k))$, 
\begin{align*}
G^m_{v_2,v_1}&=G^m_{v_1,v_2}=
\frac{H(2u_{\beta})-H(2u_{\alpha})}{\sc(\theta)}+\frac{K'}{\pi}\dn(u_\alpha)\dn(u_\beta)\\
G^{m,*}_{f_2,f_1}&=G^{m,*}_{f_1,f_2}=
\frac{H(2(u_{\beta})^*)-H(2(u_{\alpha+2K})^*)}{\sc(\theta^*)}+\frac{K'}{\pi}\dn((u_\beta)^*)\dn((u_{\alpha+2K})^*).
\end{align*}
Combining this with the expressions of $\KD(u)^{-1}_{x,w}$ concludes the proof of the second equalities of Example~\ref{ex:Prob_GD}.

In the critical case, by \cite[Lemma 45]{BdTR1} we have $\lim_{k\rightarrow 0} H(u)=\frac{u}{2\pi}$, 
and we recover that $\PPdimerD(e)=\frac{\theta}{\pi}$ independently of $u$~\cite{Kenyon3}.

\subsection{Dimers on $\GQ$}\label{app:dimers_GQ}

We compute the probability of single edges occurring in dimer configurations of $\GQ$ chosen with respect to the measure $\PPdimerQ$.
Notation used are those of Figure~\ref{fig:notations_1} below.

\begin{figure}[H]
\centering
\begin{overpic}[width=5.3cm]{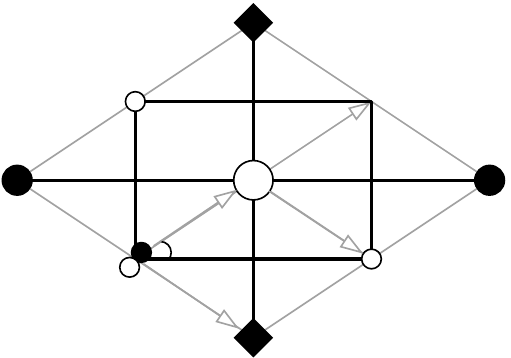}
  \put(28,23){\scriptsize $\bs$}
  \put(74,15){\scriptsize $\ws_1$}
  \put(22,13){\scriptsize $\ws_3$}
  \put(22,54){\scriptsize $\ws_2$}
  \put(32,8){\scriptsize $e^{i\bar{\alpha}}$}
  \put(32,27){\scriptsize $e^{i\bar{\beta}}$}
  \put(-6,34){\scriptsize $v_1$}
  \put(101,34){\scriptsize $v_2$}
  \put(48,-4){\scriptsize $f_1$}
  \put(48,72){\scriptsize $f_2$}
  \put(51,28){\scriptsize $w$}
  \put(35,20){\scriptsize $\bar{\theta}$}
\end{overpic}
\caption{Notation around a vertex $\bs$ of $\GQ$.}
\label{fig:notations_1}
\end{figure}
For every edge $\ws\bs$ of $\GQ$ we have, $\PPdimerQ(\bs\ws)=\KQ_{\bs,\ws}(\KQ)^{-1}_{\ws,\bs}$. We compute 
$(\KQ)^{-1}_{\ws_i,\bs}$ using Corollary~\ref{cor:gibbsGQ_GD}, $i\in\{1,2\}$. 
In both cases we have $\betafb=\bar{\beta}$; when
$\ws=\ws_1$ then $\bar{\betai}=\bar{\beta}$ and when $\ws=\ws_2$ then $\bar{\betai}=\bar{\beta}-\pi$. Thus,
\begin{align*}
(\KQ)^{-1}_{\ws_1,\bs}&=
\frac{1}{[\cn(\theta)\sn(\theta)\nd(\theta)]^{\frac{1}{2}}}
\KD(\beta)^{-1}_{v_2,w}=
\frac{e^{-i\frac{\bar{\alpha}+\bar{\beta}}{2}}}{\sn(\theta)}\KD(\beta)_{w,v_2}\KD(\beta)^{-1}_{v_2,w}=
\frac{e^{-i\frac{\bar{\alpha}+\bar{\beta}}{2}}}{\sn(\theta)}\PPdimerDbeta(wv_2)
\\
(\KQ)^{-1}_{\ws_2,\bs}&=
\frac{i}{[\cn(\theta)\sn(\theta)\nd(\theta)]^{\frac{1}{2}}}
(-i)\KD(\beta)^{-1}_{f_2,w}=
\frac{-i e^{-i\frac{\bar{\alpha}+\bar{\beta}}{2}}}{\cn(\theta)}\KD(\beta)_{w,f_2}\KD(\beta)^{-1}_{f_2,w}
=\frac{-i e^{-i\frac{\bar{\alpha}+\bar{\beta}}{2}}}{\cn(\theta)}\PPdimerDbeta(wf_2),
\end{align*}
where in the second equalities of each line we used the definition of $\KD(\beta)$, see Section~\ref{sec:def_Dirac} and 
in the third equalities we used the formulas for edge-probabilities on $\GD$. 

Returning to the definition of $\KQ$, this immediately gives,
\[
\PPdimerQ(\bs\ws_1)=\PPdimerDbeta(wv_2),\quad \PPdimerQ(\bs\ws_2)=\PPdimerDbeta(wf_2).
\]
We now use the computations of Example~\ref{ex:Prob_GD} evaluated at $u=\beta$. In the second line we use that 
$(u_{\alpha+2K})^*\vert_{\beta}=(2K-u_{\alpha})\vert_{\beta}=2K-\theta$, $(u_{\beta})^*\vert_{\beta}=(K-u_\beta)\vert_{\beta}=K$,
\begin{align*}
\PPdimerQ(\ws_1\bs)&=\PPdimerDbeta(wv_2)=H(2\theta)-H(0)=H(2\theta),\text{ since $H(0)=0$}\\
\PPdimerQ(\ws_2\bs)&=\PPdimerDbeta(wf_2)=H(2(2K-\theta))-H(2K)=\frac{1}{2}-H(2\theta),
\end{align*}
where in the last equality we used \cite[Lemma 45]{BdTR1} to obtain:
\begin{align*}
H(4K-2\theta)&=H(-2\theta)+1=-H(2\theta)+1\\
H(2K)=H(4K-2K)&=H(-2K)+1=-H(2K)+1\quad \Rightarrow \quad H(2K)=\frac{1}{2}.
\end{align*}
Since probabilities around the vertex $\bs$ sum to 1, we have $\PPdimerQ(\ws_3\bs)=\frac{1}{2}$.

\subsection{Dimers on $\GF$}\label{app:dimers_GF}

We now compute single edge probabilities for dimer configurations of $\GF$ chosen with respect to the Boltzmann measure $\PPdimerF$.

\begin{figure}[H]
\begin{center}
\begin{overpic}[width=5cm]{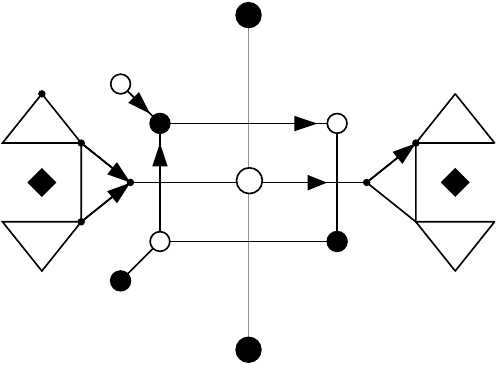}
 \put(26,31){\scriptsize $b$}
 \put(70,31){\scriptsize $b'$}
 \put(18,45){\scriptsize $a$}
 \put(18,24){\scriptsize $a'$}
 \put(76,45){\scriptsize $a''$}
 \put(26,46){\scriptsize $\bs$}
 \put(20,9){\scriptsize $\bs''$}
 \put(22,60){\scriptsize $\ws$}
 \put(30,18){\scriptsize $\ws'$}
 \put(67,52){\scriptsize $\ws''$}
\end{overpic}
\caption{Notation}\label{fig:matrixKQKF_1}
\end{center}
\end{figure}

Setting $\ubar{a}=a'$ in Formula~\eqref{form:KAAm1} gives,
\begin{align*}
(\KF)^{-1}_{a',a}&=
-\frac{1}{2}(\KQt)^{-1}_{\ws',\bs}\eps_{b,a}+\kappa_{a',a}=
-\frac{1}{2\tanh(2\Js)}\PPdimerQ(\bs\ws')\eps_{b,a'}\eps_{b,a}-\frac{1}{4}\eps_{a',a},
\quad \intertext{where we used that 
$\PPdimerQ(\bs\ws')=\KQt_{\bs,\ws'}(\KQt)^{-1}_{\ws',\bs}=\eps_{b,a'}\tanh(2\Js)(\KQt)^{-1}_{\ws',\bs}$,}
&=\eps_{a',a}\left[-\frac{1}{4}+ \frac{\PPdimerQ(\bs\ws')}{2\tanh(2\Js)} \right], \text{ since the orientation of the triangle is Kasteleyn.}
\end{align*}
We deduce that $\PPdimerF(aa')=\KF_{a,a'}(\KF)^{-1}_{a',a}=\frac{1}{4}-\frac{\PPdimerQ(\bs\ws')}{2\tanh(2\Js)}$.

Setting $\ubar{a}=a'$ and in Formula~\eqref{form:KABm1} gives,
\begin{align*}
(\KF)^{-1}_{a',b}&=\frac{1}{2}\left[e^{2\Js}\cosh^{-1}(2\Js)(\KQ)^{-1}_{\ws',\bs}+\eps_{b',b}\cosh^{-1}(2\Js)(\KQ)^{-1}_{\ws',\bs'}\right]\\
&=\frac{1}{2}\left[(1+\tanh(2\Js))(\KQ)^{-1}_{\ws',\bs}+\eps_{b',b}\cosh^{-1}(2\Js)(\KQ)^{-1}_{\ws',\bs'}\right],\text{ using \eqref{equ:reltanh2}}\\
&=\frac{1}{2}\left[(1+\tanh(2\Js))\eps_{b,a'}\frac{\PPdimerQ(\bs\ws')}{\tanh(2\Js)}+\eps_{b,a'}\PPdimerQ(\bs'\ws')\right],
\quad \intertext{where, as before, we returned to the definition of $\KQ$ and of single edge probabilities of $\GQ$}
&=\frac{\eps_{b,a'}}{2}\left[\PPdimerQ(\bs\ws')+\PPdimerQ(\bs'\ws') +\frac{\PPdimerQ(\bs\ws')}{\tanh(2\Js)}\right]\\
&=\frac{\eps_{b,a'}}{2}\left[1-\PPdimerQ(\bs''\ws') +\frac{\PPdimerQ(\bs\ws')}{\tanh(2\Js)}\right],
\end{align*}
using that the sum of probabilities is 1 around $\ws'$ for $\PPdimerQ$.
We deduce that $\PPdimerF(ba')=\KF_{b,a'}(\KF)^{-1}_{a',b}=\frac{1}{2}-\frac{\PPdimerQ(\bs''\ws')}{2} +\frac{\PPdimerQ(\bs\ws')}{2\tanh(2\Js)}$.

As a consequence also, we have: $\PPdimerF(ba')+\PPdimerF(aa')=\frac{3}{4}-\frac{\PPdimerQ(\bs''\ws')}{2}$. Using that the sum of probabilities around $a'$ is 1, we deduce
that $\PPdimerF(ab)=\frac{\PPdimerQ(\bs\ws)}{2}+\frac{\PPdimerQ(\bs\ws')}{2\tanh(2\Js)}$. Using that the sum of probabilities around the vertex $b$ is 1, we
obtain
\[
\PPdimerF(bb')=\frac{1}{2}+\frac{\PPdimerQ(\bs''\ws')}{2}-\frac{\PPdimerQ(\bs\ws)}{2}-\frac{\PPdimerQ(\bs\ws')}{2}.
\]
By symmetry $\PPdimerQ(\bs\ws)=\frac{1}{2}=\PPdimerQ(\bs''\ws')$, implying the expressions of Example~\ref{ex:prob_GF}.

\bibliographystyle{alpha}
\bibliography{survey}

\end{document}